%% file: TR_Nominal_Recursors_as_Epirecursors.tex
\documentclass[acmsmall,screen  
]{acmart}

\usepackage{etex}

\usepackage{booktabs}   
\usepackage{subcaption} 
\usepackage{etex}

\renewcommand\simeq\cong

\usepackage[all,cmtip,2cell]{xy}

\usepackage{booktabs}   
\usepackage{subcaption} 

\usepackage[normalem]{ulem}
\usepackage{stmaryrd}
\usepackage{tikz}
\usetikzlibrary{shapes,calc}
\usepackage{url}
\usepackage{multirow}
\usepackage{multicol}
\usepackage{rotating}
\usepackage{listings}
\usepackage[utf8]{inputenc}
\usepackage[T1]{fontenc}
\usepackage{bussproofs}
\usepackage{hyphenat}

\usepackage{enumitem} 
\usepackage{fancyvrb}
\usepackage{listings}

\usepackage{color}

\definecolor{light-gray}{gray}{0.85}

\include{myCommands}

\makeatletter

\EnableBpAbbreviations

\begin{document}
	
	\title
	[Nominal 
	Recursors as Epi-
	Recursors]
	{Nominal 
		Recursors as Epi-
		Recursors: Extended Technical Report} 
	\thanks{This is an extended version of the paper ``Nominal 
		Recursors as Epi-Recursors'' published in POPL 2024. 
It includes an appendix that gives more details about the results and their proofs.	
}


	\author{Andrei Popescu}
	\affiliation{
		\department{Department of Computer Science}             
		\institution{University of Sheffield}           
		\city{Sheffield}
		\country{United Kingdom}
	}
	\email{a.popescu@sheffield.ac.uk}        

	\begin{abstract}
			 	We study nominal recursors from the literature on syntax with bindings and compare them with respect to expressiveness. 
		The term ``nominal'' refers to the fact that these recursors operate on a syntax representation 
		where the names of bound variables appear explicitly, as in nominal logic.  
		%
		We argue that nominal recursors can be viewed as \emph{epi-recursors}, a concept that captures abstractly the distinction between the constructors on which one actually recurses, and other operators and properties that further 
		underpin recursion.  
		%
		We develop an abstract framework for comparing epi-recursors and instantiate it to the existing nominal recursors, and also to several recursors obtained from them by 
		cross-pollination.   
		The resulted expressiveness hierarchies 
		depend on how strictly we perform this comparison, and bring insight into the relative merits of different axiomatizations of syntax.  
		We also apply our methodology to produce an expressiveness hierarchy of nominal \emph{corecursors}, which are principles for defining functions targeting infinitary non-well-founded terms (
		which underlie $\lambda$-calculus semantics  concepts such as B\"{o}hm trees).  
		%
		%
		%
		Our results are validated with the Isabelle/HOL theorem prover. 
	\end{abstract}



	\maketitle

\section{Introduction}

Syntax with bindings is pervasive in $\lambda$-calculi, logics and programming languages. 
Powerful  
mechanisms for performing definitions and reasoning involving bindings are important for formalizing 
the meta-theory of such systems  \cite{POPLmark,poplmarkReloaded,DBLP:journals/mscs/FeltyMP18}.  
Central among these mechanisms are \emph{recursion principles}  (\emph{recursors} for short), allowing one 
to define functions by recursing over the syntax---e.g., 
for syntactic translations, semantic interpretations, and 
static analysis. 
%

A large amount of 
research has been dedicated to devising such mechanisms, within three main paradigms: nominal / nameful, nameless / De Bruijn, and 
higher-order abstract syntax (HOAS). 
Each of the three paradigms has pros and cons discussed at length in the literature (e.g., \cite{DBLP:journals/entcs/BerghoferU07,DBLP:conf/tphol/NorrishV07,poplmarkReloaded,momFelty-Hybrid4,DBLP:journals/pacmpl/BlanchetteGPT19}). 
%
A major selling point of the nominal paradigm, of which the most prominent representative is  nominal logic   \cite{DBLP:conf/lics/GabbayP99,UrbanTasson,nominalCoq}, is that it employs a formal representation that is close to the one used in textbooks and informal descriptions, where on the one hand the names of bound variables are shown explicitly, and on the other hand their particular choice is irrelevant. Moreover, definitions and reasoning within this paradigm mimic informal practice,   
such as avoiding the capturing of bound variables by conveniently choosing their names in definition and proof contexts \cite{pitts-AlphaStructural,urban-Barendregt,nominalInAgda}.  

A 
delicate subject, where the nominal paradigm must walk a tightrope to achieve its goals, is 
the recursion principles. 
The specific challenge for recursion here is 
that terms with bindings, 
which are equated modulo (i.e., quotiented to)  $\alpha$-equivalence 
(\S\ref{subsec-terms}),  
do not form a free, hence standardly recursable datatype. 
To overcome this problem, various nominal recursors have been proposed and successfully deployed in formal developments (e.g.,  \cite{DBLP:conf/lics/GabbayP99,pitts-AlphaStructural,primrecFOAS-Norrish04,urbanNominalRec,DBLP:conf/icfp/PopescuG11}). 
These recursors come in a variety of formats and flavors: 
they use different operators and have different features that enhance their cores  (\S\ref{subsec-nominalRec}).  
%

This paper contributes a general, systematic account of nominal recursors, 
highlighting their underlying principles and 
inter-connections. 
We ask two questions. First,  
\emph{what is a nominal recursor?}  
In particular, what are the essential features that 
nominal recursors from the literature have in common (\S\ref{sec-nomRecAsEpiRec})?  
  After an 
  analysis of what the existing recursors aim to achieve and 
  how they operate 
  (\S\ref{subsec-purposeNomRec}) 
and 
the uniform rephrasing of their original presentations  
using signatures and models (\S\ref{subsec-sigMod}),  
we synthesize the concept of an \emph{epi-recursor} (\S\ref{subsec-epiRec}). This concept captures abstractly 
their essential behavior, which can be summarized as follows: On top of the constructor 
infrastructure specific to standard recursion, these recursors take advantage of additional infrastructure employing non-constructor operators, to make the recursive definitions go through. And indeed, all the considered nominal recursors, and others obtained by cross-pollinating them, are particular cases of epi-recursors  (\S\ref{subsec-nomrecAsEpirecFormally}).  


Second, \emph{what does it mean for a nominal recursor to be more expressive than another, and how do the existing recursors compare?}  (\S\ref{sec-compareNominalRec}). 
Apart from its theoretical interest, this question is of practical importance for designers and developers of formal reasoning frameworks. 
%
%
We answer it by introducing two relations 
for comparing the strength of epi-recursors, which differ in the amount of effort 
required 
in simulating one recursor by another.  The first, stricter relation (\S\ref{subsec-compareHeadToHead}) follows naturally from the definition of epi-recursors. The second, laxer relation (\S\ref{subsec-moreGentle}) is more elaborate, and was inspired by previous efforts to make a nominal recursor work on a 
brittle 
terrain where syntax 
meets semantics (\S\ref{subsec-semInt}). 
Instantiating the two relations to compare the nominal recursors  yields two different hierarchies of strength. The comparisons reveal some interesting phenomena about the relative merits of considering various combinations of operations and axioms. Quite surprisingly given the wide variability of the underlying infrastructures, the laxer comparison yields an almost flat hierarchy, revealing that most of the recursors have the same strength---but still revealing that the  symmetric operators (swapping and permutation) fare better than the asymmetric ones (renaming and  substitution). 

Analogous questions make sense when moving from the inductive to the coinductive world (\S\ref{sec-co}). Here, we deal with infinitary non-well-founded $\lambda$-terms
where we allow an infinite number of constructor applications (\S\ref{subsec-infTerms}) and we study \emph{corecursors}, which are principles for defining functions not from but \emph{to} the set of infinitary terms. 
While our abstract notion of epi-corecursor (\S\ref{subsec-epicoRec}) is perfectly dual to that of epi-recursor, 
this is far from the case 
with the nominal corecursor versus recursor instances. 
However, there are 
elements of duality 
between these instances which we explore systematically,  
establishing 
a similar but different nominal corecursor expressiveness hierarchy (\S\ref{subsec-hiarNomCorec}). 

We have mechanized 
the discussed nominal (co)recursors and their comparison results 
in the Isabelle/HOL theorem prover \cite{LNCS2283} (\S\ref{sec-mechResults}).  
App.~\ref{app-isa} gives extensive details on the mechanization. 

\section{Background} 
\label{sec-nomRec}

This section provides 
background on syntax with bindings (\S\ref{subsec-terms}) and 
recalls 
several 
nominal recursors 
recursion from the literature 
(\S\ref{subsec-nominalRec}).    

\vspace*{-0.5ex}\subsection{Terms with bindings}
\label{subsec-terms}

We work with the paradigmatic syntax of 
lambda-calculus, but our results 
generalize to arbitrary binding syntaxes, as in 
 \cite{pitts-AlphaStructural, nominalTwo}.  
%
Let $\Var$ be a countably infinite set of variables, ranged over 
by $x,y,z$. 
The set $\Trm$ of $\lambda$-terms, ranged over by $t,s$, 
is defined by the grammar: 
\looseness=-1
  \vspace*{-1.3ex}
	$$
	t \;::=\; \Vr\;x  \;\mid\; \Ap\;t_1\;t_2  \;\mid\; \Lm\;x\;t
		\vspace*{-1ex}
	$$
with the proviso that terms are equated (identified) modulo $\alpha$-equivalence (a.k.a.\ naming equivalence). 
Thus, for example, $\Lm\;x\;(\Ap\;(\Vr\;x)\;(\Vr\;x))$ and 
$\Lm\;y\;(\Ap\;(\Vr\;y)\;(\Vr\;y))$ are considered to be the same term. 
We will often omit writing the injection $\Vr$ of variables into terms. 

In more detail, 
the above definition means 
the following: One first defines  the set $\PTrm$ of 
\emph{preterms} (also called ``raw terms'') to be freely generated by the 
grammar 
$
p ::= \PVr\;x  \mid \PAp\;p_1\;p_2  \mid \PLm\;x\;p 
$. 
Then one defines $\alpha$-equivalence 
$\equiv\; : 
\PTrm \ra \PTrm \ra 
\Bool$  
inductively 
and defines $\Trm$ by quotienting: 
$\Trm = \PTrm/\equiv$. 
Finally, one proves that the preterm constructors are compatible with $\equiv$, 
which allows to define the constructors on terms:  
$\Vr : \Var \ra \Trm$, $\Ap : \Trm \ra \Trm \ra \Trm$ 
and $\Lm : \Var \ra \Trm \ra \Trm$. 
\looseness=-1

Working with terms rather than preterms has 
well-known advantages, including  
the substitution operator being well-behaved. 
This is why most formal and informal developments prefer terms. 
%
For the rest of this paper, we will focus on terms and mostly forget about preterms---the latter will show up only occasionally, when we discuss certain intuitions. 
\looseness=-1






Let $\Perm$ denote the set of finite permutations (bijections of finite support) on 
variables, 
$\{\sigma : \Var \ra \Var \mid 
\{x \mid \sigma\;x \not=x\}$ finite $\}$.
We will consider generalizations of 
some common operations and relations on 
terms, namely: 
\begin{mmyitem}
	\item the constructors 
	$\Vr : \Var \ra \Trm$, $\Ap : \Trm \ra \Trm \ra \Trm$ 
	and $\Lm : \Var \ra \Trm \ra \Trm$ 
	\item (capture-avoiding) substitution  
	$\_[\_\,/\_] : \Trm \ra \Trm \ra \Var \ra \Trm$; 
	e.g., we have 
	\\$(\Lm\;x\;(\Ap\;x\;y))\;[\Ap\;x\;x \,/\, y] \,= \Lm\;x'\;(\Ap\;x'\;(\Ap\;x\;x))$ for some $x'\not= x$
	\item (capture-avoiding) renaming  
	$\_[\_\,/\_] : \Trm \ra \Var \ra \Var \ra \Trm$,  
    the restriction of substitution to variables, i.e., it substitutes 
	variables for variables rather than terms for variables; 
	e.g., we have $(\Lm\;x\;(\Ap\;x\;y))\;[x \,/\, y] = \Lm\;x'\;(\Ap\;x'\;x)$ for some $x'\not= x$ 
	\item swapping  
	$\_[\_\sw\_] : \Trm \ra \Var \ra \Var \ra \Trm$; 
	e.g., we have 
	$(\Lm\;x\;(\Ap\;x\;y))\,[x \sw y] = 
	\Lm\;y\;(\Ap\;y\;x)$
	\item permutation $\_[\_] : \Trm \ra \Perm \ra \Trm$; 
	e.g., we have 
	$(\Lm\;x\;(\Ap\;z\;y))\,[x \mapsto y,y\mapsto z,z\mapsto x] = 
	\Lm\;y\;(\Ap\;x\;z)$
	\item free-variables 
	$\FV : \Trm \ra \Pow(\Var)$ (the powerset of $\Var$); 
	e.g., %
	we have $\FV (\Lm\;x\;(\Ap\;y\;x)) = \{y\}$ when $y\not=x$
	\looseness=-1
	\item freshness $\_\fresh\_ : \Var \ra \Trm \ra \Bool$; 
	e.g., %
	we have $x \,\fresh\, \Lm\;x\;x$, and $\neg\;x \,\fresh\, \Lm\;y\;x$ when $x\not=y$
\end{mmyitem}

 We let $x \llra y$ be the 
permutation 
 that takes $x$ to $y$, $y$ to $x$ and everything else to itself. 
Note that permutation generalizes swapping, in that $t [x \sw y] = t [x \llra y]$. 
Also, note that free variables and freshness are of course two faces of the same coin: a variable $x$ is fresh for a term $t$ (i.e., $x \,\fresh\, t$) 
 if and only if it is not free in $t$ (i.e., $x \notin \FV\;t$). 
 \looseness=-1

We will not give definitions for the above operators, but count on the reader's familiarity with them.
 The definitions can be done in 
several equivalent ways---see, e.g., \cite{bar-lam,pitts-AlphaStructural}.  
\looseness=-1

\vspace*{-0.5ex}\subsection{Nominal recursors}
\label{subsec-nominalRec}

%
Next we look at 
nominal recursors in their ``natural habitat'', using concepts and terminology used by the authors who introduced them. 
Later on, in \S\ref{sec-nomRecAsEpiRec}, we will recast them in a uniform format. 
%
%
%
For convenience, we 
refer to these recursors by the additional operators 
they are based on; e.g., the ``perm/free'', or ``swap/fresh'' recursor 
(not forgetting though that 
not only the chosen 
operators, but also 
the axioms imposed on 
them 
are 
responsible for a recursor's behavior).  
\looseness=-1



\subsubsection{The perm/free recursor} 
\label{subsubsec-alphaStructRec}
This is 
the best known nominal recursor, originating in the context of nominal logic \cite{DBLP:conf/lics/GabbayP99}. 
In the form we present here, which does not require any special 
logical foundation (e.g., axiomatic nominal set theory), 
it is due to 
\citet{pitts-AlphaStructural},  
who builds on previous work by 
\citet{DBLP:conf/lics/GabbayP99}
and 
\citet{urbanNominalRec}.     
Pitts called this recursor ``$\alpha$-structural'' 
to emphasize that 
it operates on $\alpha$-equivalence classes, i.e., 
on terms rather than preterms. But since this 
is true about all nominal recursors, 
we will instead refer to this as the  ``perm/free recursor'' because it employs the permutation and free-variable operators.  
\looseness=-1

Some preparations are needed for 
describing this recursor. 
 $(\Perm,\id,\circ)$ forms a group, where $\id$ is the identity permutation and $\circ$ is 
 composition.  
A \emph{pre-nominal set} is a set equipped with a $\Perm$-action, i.e., 
a pair $\AA = (A,\_[\_]^\AA)$ where $A$ is a set and 
$\_[\_]^\AA : A \ra \Perm \ra A$ is an action of 
$\Perm$ on $A$, i.e., is idle for identity 
 ($a[\id]^\AA = a$ for all $a\in A$) and compositional ($a[\sigma \circ \tau]^\AA = a[\tau]^\AA[\sigma]^\AA$). 
\looseness=-1

Given a pre-nominal set $\AA = (A,\_[\_]^\AA)$, an 
$a\in A$ and a set $X\su\Var$, we say that \emph{$a$ is supported by $X$}, or  
\emph{$X$ supports $a$}, if 
$a [x \!\llra\! y]^\AA = a$ holds for all $x,y\in\Var \sm X$. 
An element $a\in A$ is called \emph{finitely supported} if there exists a finite set $X$ that supports $a$. 
A \emph{nominal set} is a pre-nominal set 
where every element is finitely supported.
If $\AA = (A,\_[\_]^\AA)$ is a nominal set and $a\in A$, then the smallest set  
that supports $a$ can be shown to exist---it is denoted by $\supp^\AA(a)$ and called the \emph{support of $a$}.  
%
Given two pre-nominal sets $\AA = (A,\_[\_]^\AA)$ and $\BB = (B,\_[\_]^\BB)$,  the set $F = (A \ra B)$ of functions from $A$ to $B$ 
forms a pre-nominal set   $\FF = (F,\_[\_]^\FF)$ by defining $f[\sigma]$ to be the function that sends each $a\in A$ to $f(a[\sigma^{-1}])[\sigma]$.  
%
The set of terms 
with their $\Perm$-action, $(\Trm,\_[\_])$, forms a nominal set, where the support of a term $t$ consists 
of its free variables.
\looseness=-1

The 
recursion theorem states that it is possible to define a function $g$ from terms to any other set provided $A$ is equipped with a nominal-set structure and additionally has some ``term-like'' operators matching the variable-injection, application and $\lambda$-abstraction operator, satisfying a specific condition. Concretely, 
it states that there exists a unique function $g$ that commutes with these operators: 
\looseness=-1

\begin{thm} \rm \cite{DBLP:conf/lics/GabbayP99,pitts-AlphaStructural}
	\label{thm-pittsRec} 
	Let $\AA = (A,\_[\_]^\AA)$ be a nominal set and let $\Vr^\AA : \Var \ra A$, 
	$\Ap^\AA : A \ra A \ra A$ and 
	$\Lm^\AA : \Var \ra A \ra A$ be functions, 
	all supported by a finite set $X$ of variables and such that the following freshness condition for binders (FCB) holds: there exists $x\in\Var$ such that $x\notin X$ and $x \,\fresh^\AA\, \Lm^\AA\;x\;a$ for all $a \in A$. 

	Then there exists a unique 
	$g: \Trm \ra A$ supported by $X$ such that the following hold:
	\begin{myyitem}
		\item[(1)]  $g\,(\Vr\;x) = \Vr^\AA\,x$
		\hspace*{1ex} (2) $g\,(\Ap\;t_1\;t_2) = \Ap^\AA(g\;t_1)\,(g\;t_2)$
			\hspace*{1ex} (3) 
		$g\,(\Lm\;x\;t) = \Lm^\AA\,x\;(g\;t)$ if $x\notin X$ 
	\end{myyitem}
\end{thm}

Note that the recursor features a parameter set of variables $X$, and requires the term-like operators to be supported by $X$; in exchange, it guarantees that the defined function $g$ is also supported by $X$; moreover, the recursive clause for $\Lm$ is conditioned by the abstracted variable $x$ being fresh for $X$.
The rationale of this $X$-parametrization is the modelling of 
Barendregt's famous variable convention \cite{bar-lam}[p.26]: 
``If [the terms] $M_1,\ldots,M_n$ occur in a certain mathematical context (e.g. definition, proof), then in these terms all bound variables are chosen to be different from the free variables.''
%
According to this, functions can be defined on terms while conveniently assuming that the $\lambda$-abstracted variables do not clash with other 
variables in the context of the definition---in the perm/free recursor, the set of these other variables is over-approximated by $X$. 
\looseness=-1

\looseness=-1


\subsubsection{The swap/free recursor} 
\label{subsubsec-freeSwapRec}
The next recursor is due to 
\citet{primrecFOAS-Norrish04}, who 
takes the free-variable operator as a primitive---whereas in nominal logic this operator, called support, is defined in terms of permutation. While this distinction is not important in the concrete case of terms, it does matter when one discusses abstract ``term-like'' 
structure on target domains. 
%
Another 
difference from the perm/free recursor is in taking swapping rather than permutation as primitive.
\looseness=-1

Norrish's recursor employs 
\emph{swapping structures}, which are sets equipped with swapping- and free-variable-like operators, namely triples $\AA = (A,\_[\_\sw\_]^\AA,\FV^\AA)$ where 
$\_[\_\sw\_]^\AA : A \ra \Var \ra \Var \ra A$ and 
$\FV^\AA : A \ra \Pow(\Var)$
such that 
the following hold for all $x,y,z \in \Var$ and $a\in A$:
\looseness=-1
\begin{myitem}
	\item[(i)] $a[x \sw x]^\AA =a$
	\hspace*{22.2ex} (ii) $\;a[x \sw y]^\AA [x \sw y]^\AA = a$
	\item[(iii)]  $x,y\notin \FV^\AA\,a$ implies $a[x \sw y] = a$
\hspace*{4.53ex} (iv) $\;x\in \FV^\AA(a[y\sw z]^\AA)$ if and only if $x[y\sw z] \in \FV^\AA a$
\end{myitem}

The set of terms 
with their 
swapping and free-variable operations, 
$(\Trm,\_[\_\sw\_],\FV)$, form a swapping structure. 
The 
recursion theorem says that, given a suitable ``term-like'' infrastructure on a set $A$, which includes $A$ being a swapping structure, and 
factors in a 
set of parameter variables $X$, %
%
there exists a unique function from terms to $A$ that commutes with the term-like operators in a manner that obeys Barendregt's variables convention. 
And the 
function commutes with swapping and preserves the free variables, again in a Barendregt-convention observing manner. 
(Norrish also considers dynamic parameters,  
but	\citet[Ex.~5.6]{pitts-AlphaStructural} shows 
	how to encode 
	these using static 
	parameters.)   
	\looseness=-1

\begin{thm} \rm \cite{primrecFOAS-Norrish04}
	\label{thm-norrishRec}
	Let $\AA = (A,\_[\_\sw\_]^\AA,\FV^\AA)$  be  a swapping structure, 
	$\Vr^\AA : \Var 
	\ra A$, 
	$\Ap^\AA : (\Trm \times A) \ra (\Trm \times A) \ra A$ 
	and 
   $\Lm^\AA : \Var \ra (\Trm \times A) \ra A$ 
	some functions, and $X$ a finite set of variables such that 
	the following hold: 
	\hspace*{5ex}
	(1) $\FV^\AA(\Vr^\AA x) \su \{x\} \,\cup\,X$
	\begin{mmmyitem}
		%
		\item[(2)] 
		If $\FV^\AA a_1 \su \FV\;t_1 \,\cup\,X$ 
		and $\FV^\AA\,a_2 \su \FV\;t_2 \,\cup\,X$  
		then 
		\\$\FV^\AA(\Ap^\AA\,(t_1,a_1)\;(t_2,a_2)) \su 
		\FV\,(\Ap\;t_1\;t_2) \,\cup\,X$
		\item[(3)]  
		If $\FV^\AA\,a \su \FV\;t  \,\cup\,X$  
		then 
		$\FV^\AA(\Lm^\AA\,x\;(t,a)) \su$  
		$\FV\,(\Lm\;x\;t)  \,\cup\,X$
		\item[(4)]  If $x,y\notin X$, then $(\Vr^\AA\,z)\;[x\sw y]^\AA = 
		\Vr^\AA\,(z[x\sw y])$
		\item[(5)]  If $x,y\notin X$, then $(\Ap^\AA(t_1,a_1)\,(t_2,a_2))\;[x\sw y]^\AA = 
		\Ap^\AA(t_1[x\sw y],a_1[x\sw y]^{\AA})\, 
		(t_2[x\sw y],a_2[x\sw y]^{\AA})$
		\item[(6)]  If $x,y\notin X$, then $(\Lm^\AA\,z\;(t,a))\;[x\sw y]^\AA \;= \Lm^\AA\,(z[x\sw y])\, 
		(t[x\sw y],a[x\sw y]^{\AA})$
	\end{mmmyitem}
	Then there exists a unique function $g: \Trm \ra A$ such that the following hold:  
	\begin{myitem}
		\item[(i)]  $g\,(\Vr\;x) = \Vr^\AA\;x$
	\hspace*{23.9ex}
		(ii) $\;g\,(\Ap\;t_1\;t_2)= \Ap^\AA\,(t_1,g\;t_1)\,(t_2,g\;t_2)$
		\item[(iii)]  
		$g\,(\Lm\;x\;t) = \Lm^\AA\,x\,(t,g\;t)$ 
		if $x\notin X$
			\hspace*{5.5ex}
		(iv) 
		$\;g\,(t[x \sw y]) = (g\;t) [x \sw y]^\AA$ if $x,y\notin X$
		\item[(v)] $\FV^\AA(g\;t) \su \FV\;t \,\cup\, X$   
	\end{myitem}
\end{thm}


%
%
An enhancement  present in this recursor is 
the enabling of full-fledged (primitive) recursion rather than mere iteration---as seen in the constructor-like operators  
$\Vr^\AA$, 
$\Ap^\AA$ and 
$\Lm^\AA$ taking as inputs not only elements of $A$ but also 
terms. 
Hence the recursive clauses for 
$g$ allow 
the computed value to depend not only on the recursive results for smaller terms, but also on the smaller terms themselves. 
\looseness=-1


\subsubsection{The swap/fresh recursor} 
\label{subsubsec-alphaFreshSwapRec}
The next recursor was described by 
\citet{DBLP:journals/jar/GheriP20}.  
Similarly to the previous recursors,  it uses 
structures that generalize term operators, here freshness and swapping.  
%
It is similar to 
the swap/free recursor by its focus on swapping, but different 
in that it (a) 
uses freshness rather than free variables, 
(b) requires different properties from the models, (c) does not 
support 
Barendregt's 
convention and (d) extends full-fledged recursion  
to non-constructor operators 
(in that these operators also take additional term arguments). 
\looseness=-1

A \emph{freshness-swapping model} is a set equipped with constructor-, swapping- and freshness-like operators, namely a tuple $\AA = (A,\Vr^\AA,\Ap^\AA,\Lm^\AA,\_[\_\sw\_]^\AA,\fresh^\AA)$ where 
$\Vr^\AA : \Var \ra A$, 
$\Ap^\AA : (\Trm \times A) \ra (\Trm \times A) \ra A$, 
$\Lm^\AA : \Var \ra (\Trm \times A) \ra A$, 
$\_[\_\sw\_]^\AA : (\Trm \times A) \ra \Var \ra \Var \ra \Trm$ 
and 
$\fresh^\AA : \Var \ra (\Trm \times A) \ra \Bool$
satisfying: \hspace*{8ex}
(1) $x \not= y$ implies $x \,\fresh^\AA\, \Vr^\AA y$\
\begin{mmmyitem}
	\item[(2)] $x \,\fresh\, t_1$, 
	$x  \,\fresh^\AA\,(t_1,a_1)$, 
	$x \,\fresh\, t_2$ and 
	$x \,\fresh^\AA\,(t_2,a_2)$ implies 
	$x \,\fresh^\AA\, \Ap^\AA\,(t_1,a_1)\,(t_2,a_2)$
	\item[(3)] $y = x$  
	or [$y \,\fresh\, t$ and $y  \,\fresh^\AA\,(t,a)$] 
	implies 
	$y \,\fresh^\AA\, \Lm^\AA\,x\,(t,a)$
	\item[(4)] $(\Vr\,x,\Vr^\AA\,x)[y \sw z]^\AA = \Vr^\AA(x[y \sw z])$ 
\item[(5)] $(\Ap\,t_1\,t_2, \Ap^\AA(t_1,a_1)\,(t_2,a_2))[y \sw z]^\AA \!=\! \Ap^\AA(t_1[y \sw z],(t_1,a_1)[y \sw z]^\AA)\,
(t_2[y \sw z],(t_2,a_2)[y \sw z]^\AA)$
\item[(6)] $(\Lm\,x\,t,\,\Lm^\AA\,x\;(t,a))[y \sw z]^\AA = 
\Lm^\AA\,(x[y \sw z])\,(t[y \sw z],(t,a)[y \sw z]^\AA)$
\item[(7)] $z \not\in \{x_1,x_2\}$, 
$z \,\fresh^\AA\,(t_1,a_1)$, $z\,\fresh^\AA\,(t_2,a_2)$ 
and 
$(t_1,a_1)[z \sw x_1] = (t_2,a_2)[z \sw x_2]$  
implies 
\\$\Lm^\AA\,x_1\;(t_1,a_1) = \Lm^\AA\,x_2\;(t_2,a_2)$
\end{mmmyitem}

The 
recursion theorem states that terms are the initial freshness-swapping model 
(hence initial in a certain Horn theory), i.e., for any freshness-swapping model there exists a unique function from terms that commutes with the constructors and 
swapping,   
and preserves freshness.
\looseness=-1

\begin{thm} \rm \cite
		{DBLP:journals/jar/GheriP20} 
	\label{thm-popRecSwap}
	For any  freshness-swapping model $\AA = (A,\Vr^\AA,\alb\Ap^\AA,\Lm^\AA,\alb\_[\_\sw\_]^\A,\fresh^\AA)$, there exists a unique function $g : \Trm \ra A$  
	such that the following hold: 
	\begin{myitem}
		\item[(i)] $g\;(\Vr\;x) = \Vr^\AA\,x$
		\hspace*{12.5ex} (ii) 
		$\;g\,(\Ap\;t_1\;t_2)  = \Ap^\AA\,(t_1,g\;t_1)\,(t_2,g\;t_2)$
		\item[(iii)]   $g\,(\Lm\;x\;t) = \Lm^\AA\,x\,(t,g\;t)$ 
		\hspace*{2.3ex}
		 (iv) $\;g\,(t[x \sw y]) = (t,g\;t) [x \sw y]^\AA$ 
		 \item[(v)] $x \,\fresh\, t$ implies $x \,\fresh^\AA\, (t,g\;t)$ 
	\end{myitem}
\end{thm} 

\subsubsection{The subst/fresh recursor} 
\label{subsubsec-freshSubstRec}
The next recursor, introduced by 
 \citet{DBLP:conf/icfp/PopescuG11},  
has a similar structure to the previous one 
but uses substitution rather than swapping. 
\looseness=-1

A \emph{freshness-substitution model} is similar to a freshness-swapping model, but instead of a swapping-like operator 
it has a substitution-like operator
$\_[\_/\_]^\AA : (\Trm \times A) \ra (\Trm \times A) \ra \Var \ra \Trm$  
and: 
\looseness=-1
\begin{mmyitem}
	\item instead of clauses (4)--(6) 
	of swapping commuting 
	with the constructors, it satisfies 
	similar clauses for substitution---
	but where commutation with $\lambda$-abstraction 
	is restricted by a 
	freshness condition
	\looseness=-1
	\item instead of clause (7), it satisfies a substitution-based renaming clause for $\lambda$-abstraction. 
\end{mmyitem} 
Namely, it satisfies the following clauses: 
(4) $(\Vr\,x,\Vr^\AA\,x)[(t,a) / z]^\AA =$ 
(if $x = z$ then $a$ else $\Vr^\AA x$)
\begin{mmmyitem} 
\item[(5)] $(\Ap\,t_1\,t_2,\,\Ap^\AA\,(t_1,a_1)\,(t_2,a_2))[(s,b) / z]^\AA =\;$ \\
$ \Ap^\AA(t_1[s / z],(t_1,a_1)[(s,b) / z]^{\!\AA})\,
(t_2[s / z],(t_2,a_2)[(s,b) / z]^{\!\AA})$
\item[(6)] $x \not=z$ and $x \,\fresh^\AA (s,b)$ implies 
$(\Lm\,x\,t,\Lm^\AA x\;\alb(t,a))[(s,b) /  z]^\AA \!=\! \Lm^\AA x\,(t[s /  z],(t,a)[(s,b) /  z]^\AA)$
\item[(7)] $z \not= x$ and 
$z \,\fresh^\AA\,(t,a)$ 
implies  
$\Lm^\AA\,z\;[(t,a)[(\Vr\,z,\Vr^\AA\,z) / x]] = \Lm^\AA\,x\;(t,a)$
\end{mmmyitem}

\begin{thm} \rm \cite{
		DBLP:conf/icfp/PopescuG11}
	\label{thm-popRecSubst}
	For any freshness-substitution model $\AA = (A,\Vr^\AA,\Ap^\AA,\Lm^\AA,\alb\_[\_/\_]^\AA,\fresh^\AA)$, there exists a unique 
	$g : \Trm \ra A$  
	such that the clauses listed in Thm.~\ref{thm-popRecSwap}
	hold, 
	except that  the clause for 
	swapping 
	is replaced by a clause for substitution:  
	$g\,(t[s / y]) = (t,g\;t) [(s,g\;s) / y]^\AA$.   
	\looseness=-1
\end{thm} 

\subsubsection{The renaming recursor}  
\label{subsubsec-renamingRec}
Our last discussed recursor was introduced 
by 
\citet{DBLP:conf/cade/Popescu22}.  
It is more minimalistic than the others since, 
in addition to the constructors, it only uses one operator, 
renaming---subject to an equational theory described next. 
\looseness=-1

A \emph{constructor-enriched renset} is a tuple $\AA =\alb  (A,\_[\_/\!\_]^\AA,\Vr^\AA,\Ap^\AA,\Lm^\AA)$ where  
$\_[\_/\!\_]^\AA : A \ra \Var \ra \alb\Var  \ra A$, 
$\Vr^\AA : \Var \ra A$, $\Ap^\AA : A \ra \alb A \ra A$ and  
	$\Lm^\AA : \Var \ra A \ra A$ are 
such that the following hold: 
\ \ 
(1) 
$a[x/x]^\AA = a$
\hspace*{3ex}
(2) If $x_1\not=y$ then $a[x_1/y]^\AA[x_2/y]^\AA = a[x_1/y]$
\looseness=-1
\begin{myyitem}
	\item[(3)]  
	$y\not=x_2$ then $a[y/x_2]^\AA[x_2/x_1]^\AA[x_3/x_2]^\AA=  a[y/x_2]^\AA[x_3/x_1]^\AA$
	\item[(4)]  
	If $x_2 \not= y_1 \not= x_1 \not= y_2$ then 
	$a[x_2/x_1]^\AA [y_2/y_1]^\AA= a[y_2/y_1]^\AA[x_2/x_1]^\AA$
	\item[(5)]  
	$(\Vr^\AA\;x)[y/z]^\AA = \Vr^\AA (x[y/z])$
	\hspace*{5ex}
	(6) 
	$(\Ap^\AA\;a_1\;a_2)[y/z]^\AA = \Ap^\AA(a_1[y/z]^\AA)\,(a_2[y/z]^\AA)$
	\item[(7)]  
	 if $x\notin \{y,z\}$ then 
	$(\Lm^\AA\,x\;a)[y/z]^\AA = \Lm^\AA\,x\,(a[y/z]^\AA)$
\hspace*{2.5ex}
	 (8) $(\Lm^\AA\,x\;a)[y/x]^\AA = \Lm^\AA\,x\;a$
	\item[(9)]  
	if $z\not=y$ then $\Lm^\AA\,x\;(a[z/y]^\AA) = \Lm^\AA\,y\;(a[z/y]^\AA[y/x]^\AA)$
\end{myyitem}

Equations (1)--(3) refer to standard properties of renaming, while 
(4)--(9) connect renaming and the constructors.  The recursion theorem 
characterizes 
terms as initial model in this equational theory. 
\looseness=-1

\begin{thm} \rm \cite{DBLP:conf/cade/Popescu22} 
\label{thm-popRecRename}
For any constructor-enriched renamable set $\AA = (A,\_[\_/\!\_]^\AA,\Vr^\AA,\alb\Ap^\AA,\Lm^\AA)$, there exists a unique 
$g : \Trm \ra A$  
such that the following hold: 
\looseness=-1
\begin{myitem}
	\item[(i)] $g\;(\Vr\;x) = \Vr^\AA\,x$
	\hspace*{21ex}  (ii) $g\,(\Ap\;t_1\;t_2)  = \Ap^\AA\,(g\;t_1)\,(g\;t_2)$
	\item[(ii)]   $g\,(\Lm\;x\;t) = \Lm^\AA\,x\,(g\;t)$ 
    \hspace*{12.5ex} (iv) $g\,(t[x / y]) = (g\;t) [x / y]^\AA$   
\end{myitem}
\end{thm} 

\subsubsection{Enhancements}

The above recursors clearly have many aspects in common, but also display some essential variability 
regarding the non-constructor operators they are based on and the conditions imposed on the target-domain counterparts of these operators. Other dimensions of variability were what we called the ``enhancements'': support for Barendregt's convention and full-fledged recursion. 
It turns out that both types of enhancements 
can be made 
uniformly to all nominal recursors (as we detail in 
App.~\ref{app-addingBacknhancements}).   
So in what follows, for comparing these recursors we will strip them of  
their 
enhancements 
and focus on their essential 
variability only. 
\looseness=-1


\section{Nominal recursors as epi-recursors}
\label{sec-nomRecAsEpiRec}

In this section, we will propose 
regarding nominal recursors as mechanisms for helping recursion to proceed ``as if freely'', i.e.,  
by writing clauses for each constructor as if the datatype of terms were freely generated by the constructors.
We start by describing this 
view informally 
on an example  (\S\ref{subsec-purposeNomRec}).  
 To formalize the view, we introduce signatures and models 
 that describe uniformly the term-like operators featured in the previous section's recursion theorems (\S\ref{subsec-sigMod}). Then we define the central concept of this paper, that of an epi-recursor (\S\ref{subsec-epiRec}), which captures this view in a general category-theoretic form. Finally, we show that
 all the discussed nominal recursors, and others that are obtained as variations or combinations of these, 
 are epi-recursors 
  (\S\ref{subsec-nomrecAsEpirecFormally}).    
 
As mentioned, we will not consider the recursors in their original forms---as introduced by their authors, recalled in \S\ref{subsec-nominalRec}---but their essential cores, stripped of their full-fledged recursion and Barendregt convention enhancements.  
(The enhancements, discussed in 
App.~\ref{app-addingBacknhancements}, 
turn out to be orthogonal.)   
\looseness=-1

\vspace*{-0.5ex}\subsection{
The purpose of nominal recursors
}
\label{subsec-purposeNomRec}

Let us start with recursion 
over a 
free datatype, i.e., freely generated by the constructors, such as that of preterms (recalled in \S\ref{subsec-terms}). To define a function $g: \PTrm \ra A$ between preterms and some target domain $A$, informally speaking we write recursive clauses 
for each of the constructors:
\begin{mmyitem}
	\item %
	$g\;(\PVr\;x) \,=\, \langle \mbox{expression depending on $x$}\rangle$ 
	\item %
	$g\;(\PAp\;p_1\;p_2) \,=\, \langle \mbox{expression depending on $g\,p_1$ and $g\,p_2$}\rangle$
	\item %
	$g\;(\PLm\;x\;p) \,=\, \langle \mbox{expression depending on $x$ and $g\,p$}\rangle$
\end{mmyitem}
%

The above ``expression depending on'' formulation can be made rigorous by considering preterm-like operations on the target domain $A$. Namely, for a recursive definition like the above to be possible, we must organize $A$ as a model $\AA = (A,\PVr^\AA,\PAp^\AA,\PLm^\AA)$, where 
$\PVr^\AA : \Var \ra A$, $\PAp^\AA : A \ra A \ra A$ and 
$\PLm^\AA : \Var \ra A \ra A$. 
Now, the recursive definition of $g$ 
is nothing but the statement 
that $g$ commutes with 
the operations that correspond to each other: 
\begin{mmyitem}
	\item[] \hspace*{-2ex}$g\,(\PVr\;x) = \PVr^\AA\;x$
	\hspace*{4.2ex}
	$g\,(\PAp\;p_1\;p_2) = 
	\PAp^\AA\;(g\;p_1)\;(g\;p_2)$	
		\hspace*{4.2ex}
	$g\,(\PLm\;x\;p) = \PLm^\AA\;x\;(g\;p)$
\end{mmyitem}
In fact, we could say that the model $\AA$ 
\emph{is} the recursive definition of $g$---because it determines 
a unique 
function $g : \PTrm \ra A$ that commutes with the operations. 

\vspace*{0.5ex}
Now, let's switch from preterms to terms. 
We can summarize the purpose of all nominal recursors:  
\begin{quote} 
	\vspace*{-1ex}
to 
define functions 
$g : \Trm \ra A$ between terms and target domains $A$ 
by recursing over the 
constructors 
\emph{as if the datatype of terms was 
	freely generated,} 
\end{quote} 
\vspace*{-1ex}
%
i.e., by writing recursive clauses similarly to
those of the 
free datatype of preterms:
\begin{mmyitem}
	\item %
	$g\;(\Vr\;x) \,=\, \langle \mbox{expression depending on $x$}\rangle$ 
	\item %
	$g\;(\Ap\;t_1\;t_2) \,=\, \langle \mbox{expression depending on $g\,t_1$ and $g\,t_2$}\rangle$
	\item %
	$g\;(\Lm\;x\;t) \,=\, \langle \mbox{expression depending on $x$ and $g\,t$}\rangle$
\end{mmyitem}

But the datatype of terms is not freely generated, so 
such a definition cannot work out of the box. 
%
One needs to further underpin recursion by describing the interaction of the intended function $g$ not only with the constructors, but also with other operators.  
For example, the swap/fresh recursor described in \S
\ref{subsubsec-alphaFreshSwapRec} requires two additional clauses, for the swapping and freshness operators:  
\begin{mmyitem}
	\item %
	$g\;(t[x\sw y]) \,=\, \langle \mbox{expression depending on $x$, $y$  and $g\,t$}\rangle$
	\item %
	$x \;\fresh\;t$ \ implies \ $\langle \mbox{expression depending on $x$ and $g\,t$}\rangle$
\end{mmyitem}
%

This 
is also made rigorous using models. 
The requirement is 
to define term-like operators 
on the target domain $A$ corresponding 
not only to the constructors but also to other operators; 
i.e., in this case, organize $A$ as a model $\AA = (A,\Vr^\AA,\Ap^\AA,\Lm^\AA,\_[\_\sw \_]^\AA,\alb\fresh^\AA)$, consisting of: 
\looseness=-1
\begin{mmyitem}
	\item (as before for preterms) counterparts of the constructors,  $\Vr^\AA : \Var \ra A$, $\Ap^\AA : A \ra A \ra A$, 
	and $\Lm^\AA : \Var \ra A \ra A$ , 
	\item as well as counterparts of the swapping operation and the freshness relation, $\_[\_\sw \_]^\AA : A \ra \Var \ra \Var \ra A$ and 
	$\fresh^\AA : \Var \ra A \ra \Bool$
\end{mmyitem}
Another new requirement compared to the case of free datatypes is that 
the model $\AA$ is similar to terms not only 
in the matching arities of its operators, 
but also in 
satisfying 
specific term-like properties, i.e, $\AA$-counterparts of properties 
of the terms---e.g., swapping commuting with  
$\lambda$-abstraction. 
\looseness=-1

If the above is successfully achieved, i.e., if one provides 
a model $\AA$ satisfying the required properties, 
then the recursor guarantees the existence of a unique function $g: \Trm \ra A$ commuting with the operations (here, constructors and swapping) and preserving the relations (here, freshness).
	%
	%
	%
	%
%


%

The following 
simple example 
illustrates the above discussion. \S\ref{subsec-semInt} and App.~\ref{app-anotherExample} show  
more 
examples;  
many others can be found in the literature, 
e.g.,  \cite{primrecFOAS-Norrish04,pitts-AlphaStructural,DBLP:conf/icfp/PopescuG11}. 

\begin{exa}\rm 
	\textit{(number of free occurrences)} 
Let us consider the task of defining the function $\no : \Trm \ra (\Var \ra \Nat)$, where 
$\no\;t\;x$ counts the number of (free) occurrences of the variable $x$ in the term $t$. 
The natural recursive clauses we would wish to write are 
\begin{myitem}
	\item[(i)]
	$\no\ (\Vr\ y)\ x= \mbox{(if $x = y$ then $1$ else 0)}$
	%
\hspace*{3ex} (ii) 
	$\no\ (\Ap\ t_1\ t_2)\ x= \no\ t_1\ x + \no\ t_2\ x$ 
	\item[(iii)]
	$	\no\ (\Lm\ y\ t)\ x = 
	\mbox{(if $x = y$ then $0$ else $\no\ t\ x$)}$
\end{myitem}
As discussed, 
such a definition does not work out of the box (in that, in itself, it does not constitute a correct recursive definition) because of the non-freeness of the terms. 
To make this work, we can add clauses describing the intended behavior of $\no$ with respect to swapping and freshness: 
%
\begin{myitem}
	\item[(iv)] $\no\ (t[y_1 \sw y_2])\ x = \no\ t \ (x[y_1 \sw y_2])$
	%
	\hspace*{8ex} (v)
	$x\ \fresh\ t\ \mbox{ \rm implies } \no\ t\ x = 0$
	
\end{myitem}
This means 
organizing the target domain 
$\Var \ra \Nat$ as a model $\AA$ 
by defining 
the following 
operators: 
\looseness=-1
\begin{myitem} 
	\item $\Vr^\AA = (\lambda x.\; \mbox{if $x = y$ then $1$ else 0})$
	\hspace*{7.3ex}
	$\bullet\;$ $m\,[y_1\sw y_2]^\AA = (\lambda x.\; m\,(x[y_1 \sw y_2]))$
	
	\item $\Ap^\AA\;m_1\;m_2 = (\lambda x.\;m_1\,x + m_2\,x)$
	\hspace*{8.45ex}
	$\bullet\;$ $x\;\fresh^\AA\;m = (m\;x = 0)$
	
	\item $\Lm^\AA\;y\;m = (\lambda x.\;\mbox{if $x = y$ then $0$ else $m\ x$)}$

	
	

\end{myitem} 
After checking that $\AA$ satisfies some required properties (which in this case are trivial arithmetic properties) 
we obtain a unique function $\no$ satisfying clauses (i)--(v).

%


%

\end{exa}

\vspace*{-0.5ex}\subsection{Signatures and models}
\label{subsec-sigMod}

Next we introduce 
notation that allows us to discuss the various recursors uniformly. 
Let $\Sym$, the \emph{set of (operation or relation) symbols}, be 
$\{\vr,\ap,\lm,\alb\pm,\swp,\sbs,\ren,\fv,\fr\}$.  The symbols refer to variable, application and $\lambda$-abstraction constructors, permutation, swapping, substitution, renaming and free-variable operations, and the freshness relation, respectively. 
A \emph{signature} $\Sigma$ will be any  subset of $\Sym$. 
\looseness=-1

Given a signature $\Sigma$, a \emph{$\Sigma$-model} $\MM$ consists of a set $M$, called the \emph{carrier set}, and operations and/or relations on $M$ as indicated in the signature. 
More precisely: 
	if $\vr \in \Sigma$ 
	then $\MM$ has an operation $\Vr^\MM : \Var \ra M$; 
	if $\ap \in \Sigma$ 
	then $\MM$ has an operation  $\Ap^\MM : M \ra M \ra M$; 
	if $\lm \in \Sigma$ 
	then $\MM$ has 
	$\Lm^\MM : \Var \ra M \ra M$; 
	if $\pm \in \Sigma$ 
	then $\MM$ has 
	$\_[\_]^\MM : M \ra \Perm \ra M$; 
if $\swp \in \Sigma$  
	then $\MM$ has 
	$\_[\_\sw \_]^\MM : M \ra \Var \ra \Var \ra M$; 
	if $\sbs \in \Sigma$ 
	then $\MM$ has 
	$\_[\_ \,/ \_]^\MM : M \ra M \ra \Var \ra M$; 
if $\ren \in \Sigma$ 
then $\MM$ has 
$\_[\_ \,/ \_]^\MM : M \ra \Var \ra \Var \ra M$; 
if $\fv \in \Sigma$ 
then 
$\MM$ has 
$\FV^\MM : M \ra \Pow(\Var)$; 
if $\fr \in \Sigma$ 
then 
$\MM$ has 
$\fresh^\MM : \Var \ra M \ra \Bool$.  

Given two $\Sigma$-models $\MM$ and $\MM'$, a \emph{morphism} 
between them is a function  between their carrier sets $g: M \ra M'$
that commutes with the operations and preserves the relations. For example: 
 if $\vr \in \Sigma$, we require that $g(\Vr^\MM\,x) = \Vr^{\MM'}x$; 
 if $\lm \in \Sigma$, we require that $g(\Lm^\MM\,x\;m) = \Lm^{\MM'}x\;(g\;m)$;  
if $\fr \in \Sigma$, we require that 
	$x\;\fresh^\MM\, m$ implies $x\;\fresh^{\MM'} (g\;m)$; 
if $\fv \in \Sigma$, we require that 
$\FV^{\MM'}(g\;m) \su \FV^\MM\, m$.  
We write $g: \MM \ra \MM'$ to indicate that the function $g$ is a morphism between $\MM$ and $\MM'$. 
%
%
$\Sigma$-models and their morphisms form a category. 
We write $\TTrm(\Sigma)$ for the $\Sigma$-model whose carrier is the set of terms $\Trm$ and whose operations and relations are the standard ones for terms. 

Let $\Sigmac = \{\vr,\lm,\ap\}$ be the signature comprising the 
constructor symbols only.  
Ignoring the full-fledged recursion and Barendregt 
enhancements, 
what all the described nominal recursors have in common, which is also shared with the standard recursors over free datatypes, is that they allow one to recurse over terms using 
constructors, i.e., they  
(1) require the intended target domain to be (at least) a $\Sigmac$-model $\MM$,  and (2) ensure the existence of a function $g$ that commutes with the constructors, i.e., 
a 
morphism $g: \TTrm(\Sigmac) \ra \MM$. 
%
Also, as illustrated in  \S\ref{subsec-purposeNomRec}, another aspect that the nominal recursors have in common is that, to make recursing over terms possible, 
they (1) require extending $\MM$ to a $\Sigmae$-model $\MM'$ 
for an extended signature $\Sigmae \supseteq \Sigmac$ and verifying certain  
properties for $\MM'$, and (2) capitalize on the fact that $\TTrm(\Sigmae)$ is initial among  $\Sigmae$-models that satisfy these properties---which yields a morphism $\TTrm(\Sigmae) \ra \MM'$, i.e., a function $g$ that commutes not only with the constructors but also with the other 
operators in $\Sigmae$. 
In short, what all these recursors do is \emph{underpin 
constructor-based recursion by extending the signature and exploiting initiality of the term model there}. 
\looseness=-1

\vspace*{-0.5ex}\subsection{Epi-recursors}
\label{subsec-epiRec}

We 
capture the above phenomenon in the following concept: 
\looseness=-1

\begin{defi} \rm \label{defi-epirec}
An \emph{epi-recursor} is a tuple $r = (\Bcat,T,\Ccat,I,R)$ where:
\begin{mmyitem}
	\item $\Bcat$ is a category called \emph{the base category}
	\hspace*{9ex}
	$\bullet$ \,$T$ is an object in $\Bcat$ called \emph{the base object}
	\item $\Ccat$ is a category called \emph{the extended category}
		\hspace*{5.12ex} $\bullet$  \,$I$ is an initial object in $\Ccat$
	\item $R : \Ccat \ra \Bcat$  is a functor such that $R\;I = T$   
\end{mmyitem}
\end{defi}
%

In typical examples $\Ccat$ and $\Bcat$ will be categories of models, 
i.e., sets 
with algebraic/relational structure, 
so that  the models in $\Ccat$ have more structure than those in $\Bcat$, and 
$R$ will be a structure-forgetting functor. 
The base object $T$ will be the syntactic model of interest---such as the term model $\TTrm(\Sigmac)$ 
with  constructors only---which is the source object of the intended recursive definitions. Then $I$ is its extension to an object of $\Ccat$ that makes recursion 
possible---for our nominal recursors, this is a model $\TTrm(\Sigma_\ext)$, having other ``recursion-underpinning'' operators 
besides the constructors. 
\looseness=-1

\begin{figure}
	$$
	\xymatrix@C=4pc@R=1pc{
		\Ccat \ar[d]_R &  I \ar[r]^{\im_{I,C}} & C
		\\
		\Bcat &  T = R\,I  \ar[r]^{R\,\im_{I,C}}  & B = R\,C 
	}
	$$
		\vspace*{-1.5ex}
	\caption{Epi-recursor in action}
	\label{fig-epiRP}
	\vspace*{-2.5ex}
\end{figure}

To define a morphism $g: T \ra B$ in $\Bcat$ (to some object $B$ in $\Bcat$) using the epi-recursor $r = (\Bcat,T,\alb\Ccat,I,R)$, we do the following (
see Fig.~\ref{fig-epiRP}): 
	(1) extend $B$ to an object $C$ in $\Ccat$ (with $R\;B = C$) which gives us a morphism $\im_{I,C} : I \ra C$ in $\Ccat$ from the initiality of $I$; 
(2) take $g$ to be $R\;\im_{I,C}$, the restriction of $\im_{I,C}$ to $\Bcat$. 
%
\looseness=-1

\begin{defi}\rm \label{defi-definab}
\hspace*{-1ex}
A morphism 
$g \hspace*{-0.2ex}:\hspace*{-0.2ex} T \!\ra\! B$  is \emph{definable by the 
	epi-recursor $r$} if 
$g = R \,\im_{I,C}$ for some extension $C$ of $B$.   
\looseness=-1
\end{defi} 

So an epi-recursor defines a morphism in the base category $\Bcat$. However, 
beyond having the definition 
go through, we often want to also 
``remember what happened'' in the larger category $\Ccat$ because, 
e.g., 
properties such as 
commutation with the non-constructor operators can be useful in themselves. 
\looseness=-1

\vspace*{-0.5ex}
\subsection{Nominal recursors as epi-recursors, formally}
\label{subsec-nomrecAsEpirecFormally}

\input{figBasicProperties.tex} 

Fig.~\ref{fig-basicProps} collects the properties of the operations and relations on terms that are relevant for the recursors---%
incidentally including some that are generally useful for reasoning about terms. 
%
\SwVr{}, \SwAp{}, \SwLm{} relate swapping with the  constructors.  
\SwLm{} points to one of the main appeals of the swapping operator for developing the theory of $\lambda$-calculus: It shows that swapping commutes with $\lambda$-abstraction on terms exactly in the same way as it does for preterms, i.e., is 
oblivious to the non-injectiveness of $\lambda$-abstraction.  
\SwId{}, \SwCp{}, \SwIv{} are algebraic properties of swapping: 
identity, compositionality and involutiveness. 
\SwFr{} and \FrSw{} are properties connecting swapping to freshness (and \SwFv{} and \FvSw{} are their alternative free-variable-based formulations).  
\SwFr{} says that swapping two fresh variables has no effect on the term. 
\FrSw{} says that freshness of a variable for a swapped term is equivalent to freshness of the swapped variable for the original term---stating for the freshness predicate a variant of what in nominal logic is called \emph{equivariance}. 
\SwCg{} is a swapping-based congruence property 
describing a criterion for the equality of two $\lambda$-abstractions. 
\SwBvr{} is a property allowing the renaming 
of a $\lambda$-bound variable with any fresh variable, again via swapping. \SwCg{} and \SwBvr{} are reminiscent of preterm $\alpha$-equivalence. 
%
Most properties of swapping generalize 
to corresponding properties of permutation, those listed with ``\textsf{Pm}'' in their name.
%
%
%
\looseness=-1
\looseness=-1

\FrVr{}, \FrAp{} and \FrLm{} relate freshness with the constructors, 
corresponding to 
an inductive definition 
of freshness; and \FvVr{}, \FvAp{} and \FvLm{} are their 
 free-variable 
 counterparts. %
Note 
that 
the ``if and only if'' versions of \FrVr{}, \FrAp{} and \FrLm{}  and the equality versions of \FvVr{}, \FvAp{} and \FvLm{} also hold for terms; though 
for recursion it is not the stronger, but the weaker versions of properties that lead to stronger definitional principles---since they mean weaker constraints on 
 models. 
\looseness=-1

Like swapping, substitution 
commutes with the 
constructors, which is expressed in \SbVr{}, \SbAp{}, \SbLm{}. 
As shown by \SbLm{}, unlike in the case of 
swapping, substitution's commutation with 
$\lambda$-abstraction requires  a freshness condition. 
Substitution also enjoys congruence and bound-variable renaming properties similar to those of swapping,  as  
expressed by \SbCg{} and \SbBvr{}, and some algebraic properties, as expressed by \SbId{}, \SbIm{}, \SbCh{} 
and \SbCm{}.  The renaming 
operator of course enjoys all the properties of substitution; e.g., \RnVr{}, \RnAp{}, \RnLmO{} and \RnCg{} are the counterparts of \SbVr{}, \SbAp{}, \SbLm{} and \SbCg{}. 
One may ask why we bother considering renaming, which is a restriction of substitution;  
the reason is that, again, for expressive recursors we want \emph{less} structure and \emph{weaker} properties.  
\looseness=-1

The last group in the figure 
are nominal-logic specific properties.  
\FSupFv{} states that terms have finite support, i.e., finite set of free variables; it can 
also be expressed directly 
in terms of swapping (as in \S\ref{subsubsec-alphaStructRec}).  
\FvDPm{} and \FvDSw{} state the definability of free-variables from permutations and (alternatively) from swapping.  
%
%
%
%
%
%
\FCB{} is the \emph{freshness condition for binders} 
from  the statement of the perm/free recursion theorem (Thm.~\ref{thm-pittsRec}), but with the Barendregt set $X$ removed.  
\FCB{} is weaker than \FvLm{} since it quantifies existentially rather than universally over the bound variable, though in nominal logic they are equivalent (the ``some/any'' property \cite{pitts-AlphaStructural}).   
Finally, this last group also includes alternative, freshness-based and renaming-based formulations of 
some of the above properties. Note that, unlike \FrDSw{}, \FrDRn{} would stay true for terms if we replaced ``finite'' with ``empty''. 
\looseness=-1

Each of the properties listed in Fig.~\ref{fig-basicProps} 
is satisfied by the terms with their basic operations and relations, i.e., 
by the term model $\TTrm(\Sigma)$ for any signature $\Sigma$ that contains all the symbols referred to in the property. 
But we can 
speak of the corresponding properties in relation to any other $\Sigma$-model $\MM$, and they may or may not be satisfied by $\MM$. For example, when we say that the model $\MM$ (with carrier $M$) satisfies \SwCg{}, we mean the following:  
	For all $m_1,m_2\in M$ and $x_1,x_2,z\in \Var$, 
	if $z \not\in \{x_1,x_2\}$ and 
	$z\;\fresh^\MM\;m_1,m_2$ and 
	$m_1[z\sw x_1]^\MM = m_2[z\sw x_2]^\MM$ 
	then 
	$\Lm^\MM\;x_1\;m_1 = \Lm^\MM\;x_2\;m_2$. 
%
As another example, $\MM$ 
satisfying \FCB{} means the following: 
There exists $x \in \Var$ such that 
$x\notin \FV^\MM(\Lm^\MM x\,m)$ for all $m \in M$.  
\looseness=-1

Given a subset $\Props$ of the properties 
in Fig.~\ref{fig-basicProps} and 
a signature $\Sigma$ comprising the symbols referred to in $\Props$, any $\Sigma$-model satisfying $\Props$ will be called a \emph{$(\Sigma,\Props)$-model}. 
Now we can 
(re)formulate 
nominal recursors as epi-recurors: 
\looseness=-1

\newcommand\perFreeRec
{
\begin{tabular}{|c|}
	\hline
	$r_1$ \ (perm/free)
	\\\hline
	\PmVr{}, \PmAp{}, \PmLm{}, 
	\\
	\PmId{}, \PmCp{},  
	\\
	\FvDPm{}, \FCB{} 
	\\
	\sout{\FSupFv{}} 
	\\\hline
\end{tabular}	
}

\newcommand\perFreeRecV
{
	\begin{tabular}{|c|}
		\hline
		$r_2$ \ (perm/free variant)
		\\\hline
		\PmVr{}, \PmAp{}, \PmLm{}, 
		\\
		\PmId{}, \PmCp{},  
		\\
		\PmFv{}, \sout{\FvPm{}}, 
		\\
	    \FvVr{}, \FvAp{}, \FvLm{}  
		\\\hline
	\end{tabular}	
}

\newcommand\swapFreeRec
{
	\begin{tabular}{|c|}
		\hline
		$r_4$ \ (swap/free)
		\\\hline
		\SwVr{}, \SwAp{}, \SwLm{}, 
		\\
		\sout{\SwId{}}, \sout{\SwIv{}},  
		\\
		\SwFv{}, \sout{\FvSw{}}, 
		\\
		\FvVr{}, \FvAp{}, \FvLm{}  
		\\\hline
	\end{tabular}	
}

\newcommand\swapFreeRecV
{
	\begin{tabular}{|c|}
		\hline
		$r_3$ \ (swap/free variant)
		\\\hline
		\SwVr{}, \SwAp{}, \SwLm{}, 
		\\
		\SwId{}, \SwIv{},  \SwCp{},  
		\\
		\FvDSw{}, \FCB{} 
		\\
		\sout{\FSupFv{}} 
		\\\hline
	\end{tabular}	
}

\newcommand\swapFreshRecV
{
	\begin{tabular}{|c|}
		\hline
		$r_5$ \ (swap/fresh variant)
		\\\hline
		\SwVr{}, \SwAp{}, \SwLm{}, 
		\\
		\SwBvr{}, 
		\\
		\FrVr{}, \FrAp{}, \FrLm{}  
		\\\hline
	\end{tabular}	
}

\newcommand\swapFreshRec
{
	\begin{tabular}{|c|}
		\hline
		$r_6$ \ (swap/fresh)
		\\\hline
		\SwVr{}, \SwAp{}, \SwLm{}, 
		\\
		\SwCg{}, 
		\\
		\FrVr{}, \FrAp{}, \FrLm{}  
		\\\hline
	\end{tabular}	
}


\newcommand\substFreshRec
{
	\begin{tabular}{|c|}
		\hline
		$r_7$ \ (subst/fresh)
		\\\hline
		\SbVr{}, \SbAp{}, \SbLm{}, 
		\\
		\SbBvr{}, 
		\\
		\FrVr{}, \FrAp{}, \FrLm{}  
		\\\hline
	\end{tabular}	
}

\newcommand\renamingRec
{
	\begin{tabular}{|c|}
		\hline
		$r_8$ \ (renaming)
		\\\hline
		\RnVr{}, \RnAp{}, \RnLmO{}, 
		\\
		\RnLmT{}, \RnBvrT{}, 
		\\
		\RnId{}, \RnIm{}, \RnCh{}, \RnCm{} 
		\\\hline
	\end{tabular}	
}

\newcommand\renamingFreshRecV
{
	\begin{tabular}{|c|}
		\hline
		$r_9$ \ (renaming/fresh variant)
		\\\hline
		\RnVr{}, \RnAp{}, \RnLmO{}, 
		\\
		\RnBvr{}, 
		\\
		\FrVr{}, \FrAp{}, \FrLm{}  
		\\\hline
	\end{tabular}	
}

\begin{figure}[!t]
	\begin{center} 	
		{\small 
		\begin{tabular}{ccc}
			    \vspace*{-1.2ex}
			\perFreeRec{} & \perFreeRecV{} &
			\swapFreeRecV 
			\\ && \\    \vspace*{-1.2ex}
			\swapFreeRec  & 
			\swapFreshRecV &  \swapFreshRec 
				\\ && \\
			\substFreshRec & 
			 \renamingRec &  \renamingFreshRecV 
		\end{tabular}
	} 
	\end{center}
\vspace*{-1.5ex}
	\caption{Sets of properties 
		underlying different nominal recursors. 
	The crossed-out properties \FSupFv{} in $r_1$ 
	and \SwId{}, \SwIv{}, \FvSw{} in $r_4$ were in the original recursors but turn out not to be needed.}
	\label{fig-iter}
	\vspace*{-3ex}
\end{figure}

\begin{thm}\rm \label{thm-allNominalRecs}
Consider the  nine choices, for $i\in\{1,\ldots,9\}$, of 
tuples $r_i = (\Bcat,T,\Ccat_i,\alb I_i,R_i)$ 
given by the sets of properties $\Props_i$ shown in  
Fig.~\ref{fig-iter}. 
(E.g., $\Props_5$ is $\{\SwVr, \SwAp, \alb\SwLm, \SwBvr, \FrVr, \FrAp, \FrLm\}$.)
Namely, 
we assume that the signature $\Sigma_i$ consists of all the 
operation and relation 
symbols 
occurring in $\Props_i$, and: 
\begin{mmyitem}
	\item $\Bcat$ is the category of $\Sigmac$-models and $T=\TTrm(\Sigmac)$
	\item $\Ccat_i$ is the category of $(\Sigma_i,\Props_i)$-models and 
	$I_i$ is $\TTrm(\Sigma_i)$
	\item $R_i : \Ccat_i \ra \Bcat_i$ is the forgetful functor sending 
	 $(\Sigma_i,\Props_i)$-models to their underlying $\Sigmac$-models
	\looseness=-1 
\end{mmyitem}
\par
Then $r_i$ is an epi-recursor. 
In particular, $\TTrm(\Sigma_i)$ is the initial $(\Sigma_i,\Props_i)$-model.  
\end{thm}
%

Next we discuss this theorem's nine statements of epi-recursion principles. 
We 
distinguish between 
five ``original recursors'' from the literature and four ``variant recursors'' obtained from those. 
\looseness=-1 

\subsubsection{The original recursors}  
\label{subsubsec-origRecs}

As suggested by the names in Fig.~\ref{fig-iter}, five of these principles, 
$r_1$, $r_4$, $r_6$, $r_7$ and $r_8$, are reformulations 
of %
the (stripped down versions of) 
nominal recursors 
from \S\ref{subsec-nominalRec}.  
%
%
%
\looseness=-1

This is easy to see in the case of $r_6$ and $r_7$.  Indeed, after removing the term arguments of the operations and relations,  
 Thms.~\ref{thm-popRecSwap} (swap/fresh) and \ref{thm-popRecSubst} (subst/fresh) 
  simply state, for a suitable extension of the constructor signature $\Sigmac$, the initiality of the corresponding term model 
  among
all models satisfying $\Props_6$ or $\Props_7$. 
Moreover, Thm.~\ref{thm-popRecRename} (about renaming recursion) is easily seen to 
be exactly $r_8$.  
\looseness=-1

Seeing that $r_1$ is the stripped down version of the 
perm/free recursor (
from Thm~\ref{thm-pittsRec}) requires a bit of work. 
After removing $X$ from (i.e.,  taking $X$ to be $\emptyset$ in Thm.~\ref{thm-pittsRec}),  
we see that 
a nominal set $\AA = (A,\_[\_]^\AA)$ together with $\emptyset$-supported operations 
$\Vr^\AA$, 
$\Ap^\AA$ 
and $\Lm^\AA$ 
can be equivalently described as a $(\Sigma_1,\Props_1)$-model. 
Moreover, the properties of the unique function $g:\Trm \ra A$ 
guaranteed by Thm.~\ref{thm-pittsRec} are equivalent to those of 
$\Sigma_i$-morphisms. (App.~\ref{app-modreDetailsNomSets}   
 gives 
details.)
Thm.~\ref{thm-pittsRec} does not actually 
need 
the finite-support condition \FSupFv{} for the target domain---which is why in Fig.~\ref{fig-iter} we show it for $r_1$ (and for the variant $r_3$ discussed below) as crossed out. 
%
\looseness=-1

Seeing that $r_4$ is the stripped down version of the swap/free recursor (
from Thm.~\ref{thm-norrishRec}) is also not immediate. After removing the Barendregt parameterization on $X$ from Thm.~\ref{thm-norrishRec}, we obtain operations and relations that fit the pattern of full-fledged recursion, i.e., iteration plus additional term arguments---e.g., $\Vr^\AA : \Var \ra A$, 
$\Ap^\AA : (\Trm \times A) \ra (\Trm \times A)\alb \ra A$ and 
$\Lm^\AA : \Var \ra (\Trm \times A) \ra A$. So the situation becomes similar to that 
of $r_6$ and $r_7$ versus Thms.~\ref{thm-popRecSwap} and \ref{thm-popRecSubst}. 
However, two of Thm.~\ref{thm-norrishRec}'s assumptions, 
(2) and (3), do not 
directly fit the normal full-fledged recursion pattern. 
But after using 
that $	\FV\,(\Ap\;t_1\;t_2) = \FV\,t_1 \cup \FV\,t_2$ and $\FV\,(\Lm\;x\;t) = \FV\,t \sm \{x\}$, they are seen equivalent to: 
\looseness=-1
	\begin{myitem}
	\item[(2)] 
	If $\FV^\AA\,a_1 \su \FV\;t_1$ and $\FV^\AA\,a_2 \su \FV\;t_2$ 
	then  
	$\FV^\AA(\Ap^\AA\,(t_1,a_1)\;(t_2,a_2)) \su \FV\,t_1 \cup \FV\,t_2$
	\item[(3)]  
	If $\FV^\AA\,a \su \FV\;t$  
	then  
	$\FV^\AA(\Lm^\AA\,x\;(t,a)) \su  \FV\,t \sm \{x\}$
	\end{myitem} 
In this form, they are seen to express a kind of full-fledged recursion that is optimized for the free-variable operator. 
Indeed, they 
%
are weaker 
versions of ones that \emph{do} fit the pattern:
\looseness=-1
\begin{myitem}
	\item[(2)] 
	$\FV^\AA(\Ap^\AA\,(t_1,a_1)\;(t_2,a_2)) \su (\FV^\AA\,a_1 \cup \FV\,t_1) \cup (\FV^\AA\,a_2 \cup \FV\,t_2)$
	\item[(3)]  
	$\FV^\AA(\Lm^\AA\,x\;(t,a)) \su  (\FV^\AA\,a  \cup \FV\,t) \sm \{x\}$
\end{myitem} 
(In App.~\ref{app-addingBacknhancements} 
we show how this free-variable-specific optimization can be seen as a general enhancement available to all our discussed recursors that involve freeness or freshness.)  
Removing the term arguments from the latter turns them into Fig.~\ref{fig-basicProps}'s  
\FvAp{} and \FvLm{}; and removing the term arguments from the other assumptions in Thm.~\ref{thm-norrishRec} turns them into the other properties of $\Props_4$.  Finally, the conclusion of Thm.~\ref{thm-norrishRec} corresponds precisely to the $\Sigma_2$-morphism conditions. 
\looseness=-1

Three of the 
properties 
originally postulated by  
\citet{primrecFOAS-Norrish04} 
for the swap/free recursor, \SwId{}, \SwIv{} and \FvSw{}, are not needed, meaning that the recursion theorem holds without them (hence they are crossed out under $r_4$ in Fig.~\ref{fig-iter}.) 
%
This is a surprising result, given the 
careful analysis done by Norrish  
when distilling the required properties for his recursor to work. %
We detected this redundancy while subsuming the $(\Sigma_4,\Props_4)$-models 
to the more general $(\Sigma_5,\Props_5)$-models 
during recursor comparison (discussed in \S\ref{sec-compareNominalRec}), so this strengthening 
owes to the different 
path 
taken when 
proving $r_5$. 
\looseness=-1

\subsubsection{
	The variant recursors}  
\label{subsubsec-crossPolliRecs}
The remaining principles, $r_2,r_3,r_5$ and $r_9$, are obtained by combining  
axioms of the original recursors. They 
act as bridges between the latter 
helping their comparison, 
but are also of independent interest, e.g., $r_9$ will 
be seen to be maximal with respect to expressiveness. 
\looseness=-1

We call $r_2$ a ``perm/free variant'' because it is another recursor based on permutation and freeness, just like the original perm/free recursor $r_1$.  However $r_2$ does not follow the nominal-set route of $r_1$ (which defines the free-variable, i.e., support operator from permutation, via \FvDPm{}) but instead follows the idea of the swap/free recursor $r_4$ (using permutation instead of swapping) and postulates properties connecting the free-variable operator with  permutation (via \PmFv{} and \FvPm{}) and with the constructors (via \FvVr{}, \FvAp{} and \FvLm{}). In short, $r_2$ is a hybrid between $r_1$ and $r_4$. For the symmetry of presentation, under $r_2$ the figure also shows \FvPm{}---the permutation counterpart of \FvSw{}, but crosses it out because, like \FvSw{}, is also not needed. 
Another 
hybrid between the two is the swap-free variant $r_3$, which uses the swapping operator like $r_4$ and nominal-set-like axioms like $r_1$.  
\looseness=-1

$r_5$ 
is an $r_6$--$r_7$ hybrid, 
born from the observation that 
$r_6$ and $r_7$ 
have 
similar structures, in that they both axiomatize the interaction between 
constructors and freshness, 
and between constructors and their specific operator (either swapping or substitution); 
of course, substitution behaves differently from swapping w.r.t.\ constructors, but the respective constructor-commuting properties (\SbVr{}, \SbAp{} and \SbLm{} vs.\ \SwVr{}, \SwAp{} and \SwLm{}) have a similar flavor.
The difference between $r_6$ and $r_7$ lies in the additional property that they use to further underpin recursion over the constructors: in one case via a congruence rule \SwCg{} and in the other via a bound-variable-renaming rule \SbBvr{}. However, both these latter types of rules make sense for the other operator too, \emph{mutatis mutandis}. 
As it turns out, we can replace \SwCg{} with \SbBvr{} in the swap/fresh recursor 
$r_6$, obtaining the swap/fresh variant $r_5$. 
%
But we cannot perform the dual modification to the subst/fresh recursor $r_7$, where replacing  \SbBvr{} with \SbCg{} would not give a valid recursor; the reason is that, unlike swapping, substitution-like operators need a more delicate handling of the bound variables, which \SbBvr{} but not \SbCg{} can achieve. 
\looseness=-1

Finally, $r_9$ is a $r_7$--$r_8$ hybrid, 
in that it has axioms similar to $r_7$, but uses renaming like $r_8$ 
 rather than substitution. 
  And similarly to the case of substitution, replacing \RnBvr{} with \RnCg{} would not work. 
  \looseness=-1

We 
discovered these variant recursors 
during the Isabelle formalization of the originals---observing the roles played by different axioms in underpinning recursion, and noting that in specific contexts some operators and axioms are interchangeable.  
\looseness=-1

\smallskip
\emph{Proof idea for Thm.~\ref{thm-allNominalRecs}.}
For any $i\in \{1,\ldots,9\}$, the only non-trivial part of the statement that 
$r_i$ is an epi-recursor is the initiality theorem, i.e., the fact that $\TTrm(\Sigma_i)$ is the initial $(\Sigma_i,\Props_i)$-model. 

The initiality theorems for the original recursors already have been proved in the literature (as we discussed in \S\ref{sec-nomRec}), 
%
%
%
%
whereas 
the  variant recursors $r_2,r_3,r_5$ and $r_9$ are new. 
%
We (re)proved all these recursors via the following route: First we gave direct proofs for $r_6$ and $r_9$, and then we used the transformations underlying the expressiveness relations in Thm.~\ref{thm-expr} in order to infer (``borrow'') the initiality theorems for the others from the above two (which are at the top of Thm.~\ref{thm-expr}'s expressiveness hierarchy). 
App.~\ref{app-proofSketches} gives details. 

Next, we show the proof idea for 
$r_9$, which is a generalization/adaptation of that for $r_7$ from  \citet{DBLP:conf/cade/Popescu22}. 
Let $\MM$ be a $(\Sigma_9,\Props_9)$-model. 
We first define 
a relation $R : \Trm \ra M \ra \Bool$, with inductive clauses reflecting the desired properties 
of commutation with the constructors:
\looseness=-1
$$R\;(\Vr\;x)\;(\Vr^\MM\;x)
\hspace*{6ex}
\frac{R\;t_1\;m_1 \hspace*{4ex} R\;t_2\;m_2}{R\;(\Ap\;t_1\;t_2)\;(\Ap^\MM\;m_1\;m_2)}
\hspace*{6ex}
\frac{R\;t\;m}{R\;(\Lm\;x\;t)\;(\Lm^\MM\;x\;m)}$$
To obtain a $\Sigma_9$-morphism $f : \Trm \ra M$, it suffices 
to prove that $R$ (1) is total, (2) is functional, (3) preserves renaming and (4) preserves freshness, since then we can take $f$ to be the function induced by $R$. 
Property (1) (totality) follows easily by standard induction on terms. 
The remaining properties, (2)--(4), follow by a  simultaneous inductive proof 
using a form of ``renaming-based induction'' on terms: Given a predicate $\phi : \Trm \ra \Bool$, to show $\forall t\in \Trm.\;\phi\;t$ it suffices to show the following: (i) $\forall x\in\Var.\;\phi\;(\Vr\;x)$, (ii) $\forall t_1,t_2\in\Trm.\;\phi\;t_1 \,\&\, \phi\;t_2 \lra \phi\,(\Ap\;t_1\;t_2)$, and (iii) $\forall x\in\Var,\,t\in\Trm.\;(\forall s\in\Trm.\; \textsf{RConnect}\;t\;s \lra \phi\;s) \lra \phi\,(\Lm\;x\;t)$, where $\textsf{RConnect}\;t\;s$ means that $s$ is obtained from $t$ by a chain of renamings. (So we take $\phi$ to be the conjunction of (2)--(4).) The uniqueness of $f$ follows by induction on terms. 
%
%
%
The proof for $r_6$ is similar to that for $r_9$, but uses a corresponding swapping-based induction.   
\qed

\section{Comparing recursors} 
\label{sec-compareNominalRec}

An advantage of viewing nominal recursors as epi-recursors is 
clear sight on their relative expressiveness. 
In this section, we start with 
a 
direct means of comparing epi-recursor expressiveness 
and instantiate it to our nominal recursors (\S\ref{subsec-compareHeadToHead}). 
Then we analyze a 
problematic example, semantic interpretation (\S\ref{subsec-semInt}), which suggests a gentler comparison---yielding a much flatter expressiveness hierarchy (\S\ref{subsec-moreGentle}).  
While the kind of relationships we establish show how a recursor can replace another, they do not imply that 
the converse is not true, and indeed in some cases the converse is true, making the recursors equivalent (w.r.t.\ 
a tighter or gentler comparison); but in two cases we also know that the converse is not true, meaning the 
relation there is strict 
(\S\ref{subsec-negResRec}). 
\looseness=-1

\vspace*{-0.5ex}\subsection{A head-to-head comparison} 
\label{subsec-compareHeadToHead}

\begin{defi} \rm \label{defi-strong}
Given 
epi-recursors $r = (\Bcat,T,\Ccat,I,R)$ and $r' = \alb(\Bcat,\alb T,\Ccat',I',R')$ 
with the same base category $\Bcat$ and base object $T$, we 
call \emph{$r'$ stronger than $r$}, written $r' \geq r$, if $r'$ can define everything that  
$r$ can, 
i.e.: 
for all objects $B$ in $\Bcat$ and morphisms 
$g: T \ra B$, 
$g$ definable by $r$ implies $g$ definable by $r'$.   
\looseness=-1
\end{defi} 


\begin{figure}[t]
		$$
		\xymatrix@C=4pc@R=0.1pc{
			I & 
			\Ccat \ar[dd]_{R}  \ar[dr]^{\exists F} &  &
			\\
			& & \Ccat' \ar[dl]^{R'}  & I' = F\,I
			\\
			R\,I = R'\,I' = T& \Bcat &  &   & 
		}
	\vspace*{-2ex}
		$$
		\caption{Criterion for comparing expressiveness}
		\label{fig-critStrongerEpirec}
		\vspace*{-2ex}
\end{figure}

It is easy to see that 
$\geq$ is a preorder 
on epi-recursors.  
We 
write $r \equiv r'$ to state that $r$ and $r'$ have equal strengths, i.e., both 
$r' \geq r$ and $r \geq r'$ hold.  
%
%
We can establish $r' \geq r$ by showing how to move from $r'$ to $r$ in an initial-object preserving way, as 
depicted in Fig.~\ref{fig-critStrongerEpirec}: 

\begin{prop}\rm
	\label{prop-extCriterion} 
	Let $r = (\Bcat,T,\Ccat,I,R)$ and 
	$r' = (\Bcat,T,\Ccat',I',\alb R')$, and 
	assume $F : \Ccat \ra \Ccat'$ is a pre-functor 
	(i.e., a functor but without the requirement of preserving identity and composition of morphisms)  
	such that 
	$R' \circ F = R$ and 
	$F\;I = I'$.  
	Then $r' \geq r$. 
\end{prop}
\emph{Proof.} 
Assume $g:T \ra B$ is definable by $r$, meaning that $g = R\;\im_{I,C}$ for some $C$ in $\Ccat$. 
Let $C' = F\;C$. 
By the initiality of $I'$ and the fact that $F\;I = I'$, we have that 
$\im_{I',C'} = F\;\im_{I,C}$. Hence 
$g = R\;\im_{I,C} = R'\;(F\;\im_{I,C}) = R'\;\im_{I',C'} $, 
meaning that $g$ is definable by $r'$.  \qed

\medskip 
(In all our examples, the above initial-object preserving pre-functor condition 
will be satisfied by actual functors that are left adjoints.) 
One way to read Prop.~\ref{prop-extCriterion}'s criterion (and Fig.~\ref{fig-critStrongerEpirec}'s picture) is the following: Thinking of $R$ as a kind of ``distance" 
from the extended category $\Ccat$ (and its initial object $I$) to the base category $\Bcat$ (and the base object $B$), we have that 
the 
smaller this distance, the more expressive the recursor. 
We have applied this criterion to prove the following expressiveness hierarchy: 
\looseness=-1


\begin{thm} \rm
	\label{thm-expr}
	The epi-recursors described in Thm.~\ref{thm-allNominalRecs} (and in Fig.~\ref{fig-iter}) compare as follows with respect to their expressiveness: 
	 \ \ $r_6 \geq r_5 \geq r_4 \geq r_2 \geq r_1 \equiv r_3$ 
	\ \ and  \ \  $r_9 \geq r_8, r_7$. 
	%
\end{thm} 
\emph{Proof idea.}   
%
%
When proving each 
$r_i \geq r_j$, we instantiate Prop.~\ref{prop-extCriterion} 
taking $r' = r_i$ and $r = r_j$. So here $\Bcat$ is the category of $\Sigmac$-models, $\Ccat'$ that of $(\Sigma_i,\Props_i)$-models, 
and $\Ccat$ that of $(\Sigma_j,\Props_j)$-models; $R'$ 
is the forgetful functor from $(\Sigma_i,\Props_i)$-models to $\Sigmac$-models, and $R$ the forgetful functor from $(\Sigma_j,\Props_j)$-models to $\Sigmac$-models; 
$B=\TTrm(\Sigmac)$, 
$I'=\TTrm(\Sigma_i)$ and $I=\TTrm(\Sigma_j)$.  
In each case, we must define a pre-functor $F : \Ccat \ra \Ccat'$  
such that 
$R' \circ F = R$ and $F\;I = I'$.  
This essentially means showing how to 
transform $(\Sigma_j,\Props_j)$-models into $(\Sigma_i,\Props_i)$-models in such a manner that $\TTrm(\Sigma_j)$ 
becomes $\TTrm(\Sigma_i)$---which gives $F$'s behavior on objects, while on morphisms $F$ will be the identity. Each time, $F$ will transform models by preserving the carrier set and the constructor-like operators, and possibly defining 
(1) permutation-like from swapping-like operators or vice versa, (2) 
freshness-like from 
free-variable-like operators, or (3) renaming-like from substitution-like operators; 
these definitions are done just like for concrete terms (where, e.g., we can standardly define freshness from freeness). In each case, the only interesting fact that needs to be checked is that $F$ is well-defined on objects: when starting with a $\Sigma_j$-model satisfying $\Props_j$, the result $\Sigma_i$-model indeed satisfies $\Props_i$. Everything else 
amounts to either well-known or trivial properties. Thus, $F\;I = I'$ means that the standard inter-definability properties (1)--(3) hold for terms, e.g., $x \,\fresh\, t$ iff $x \notin \FV\;t$; and $R' \circ F = R$ (i.e., $F$ commutes with the forgetful functors to $\Sigmac$-models) follows immediately from the fact that $F$ does not change the carrier set or the constructor-like operators. 
Next, we informally discuss these transformations and highlight the intuitions behind them. 

\textit{
	$r_1 \equiv r_3$}
holds because  permutation-like and swapping-like operators 
correspond bijectively to each other, allowing one to (functorially) move back and forth between $\Props_1$-models and $\Props_3$-models \cite[Section 6.1]{pitts_2013}. 
%
For $r_2 \geq r_1$, we note that 
$r_2$ seems \emph{a priori} more flexible than $r_1$ in that it does not require the free-variable operator to be \emph{definable from} permutation, but only to be \emph{related to} permutation by some weaker properties; and indeed, 
any 
$(\Sigma_1,\Props_1)$-model can be proved to be in particular a $(\Sigma_2,\Props_2)$-model. 
%
%
\textit{
	$r_4\geq r_2$} 
holds essentially for the same reason why $r_3\geq r_1$ holds, 
since the restriction of a permutation to a swapping operator carries over to their axiomatized relationships with free-variable operators,  \PmFv{} versus \SwFv{}. (But the converse is not true because $r_4$ lacks (does not need) some of the swapping axioms that ensure extension to a permutation operator.) 
%
\textit{
	$r_5 \geq r_4$} 
follows using a model transformation that 
turns the free-variable operator of $r_4$ into a freshness operator for $r_5$, using negation; indeed, save for the straightforwardly corresponding \FvVr{}, \FvAp{}, \FvLm{} versus \FrVr{}, \FrAp{} and \FrLm{},  
the only difference between $r_4$ and $r_5$ is the replacement of 
\SwFv{} 
with 
\SwBvr{}; 
and the former axiom implies the latter in the presence of \FvLm{}. 
%
%
\textit{
	$r_6 \geq r_5$} 
follows from the fact that, in the presence of the other axioms in $\Props_5$, 
\SwBvr{} implies 
\SwCg{}. 
\textit{
	$r_9 \geq r_7$} 
holds because the axioms for substitution imply those for renaming (for 
the straightforward restriction of a substitution operator to a renaming operator). 
Finally, the proof of 
\textit{
	$r_9 \geq r_8$} 
takes advantage of the fact that, in a constructor-enriched renset (structures axiomatizing renaming that form the basis of recursor $r_8$), freshness is definable from renaming \cite{DBLP:conf/cade/Popescu22}.  
\qed
\medskip

Thus, there are two recursors 
at the top of the expressiveness hierarchy: 
the swap/fresh recursor $r_6$ and the renaming/fresh variant recursor $r_9$. 
Roughly speaking, these two recursors' expressiveness is strong because their underlying axiomatizations:
\begin{mmyitem}
\item keep freshness only loosely coupled with other operators 
such as swapping, permutation or renaming---unlike $r_1$, $r_3$ and $r_8$ which ask that freshness be definable from them; 
\looseness=-1
\item use congruence or renaming axioms that target exactly the ingredients needed for having recursion go through---unlike those of $r_1,r_2,r_3,r_4$ and $r_8$, which employ 
algebraic axiomatizations 
such as 
nominal sets, swapping structures or rensets;
\looseness=-1
\item keep the structure of their operators minimalistic and non-redundant---unlike $r_7$, whose operator emulates substitution, which is more than needed (since renaming would suffice). 
\end{mmyitem} 


Choosing between swapping and permutation as recursion primitives 
turned out to be interesting. 
 The two 
 are known to be equivalent 
 for nominal sets \cite[\S6.1]{pitts_2013}, as reflected by 
 $r_1 \equiv r_3$.  
 \looseness=-1
 
 But they are no longer equivalent when loosening the axiomatization to include freshness as a primitive---as reflected by the fact that $r_4 \geq r_2$ but (as we will show in \S\ref{subsec-negResRec}) not vice versa.  This is because the 
 proof of the $r_4$ recursor (by 
 \citet{primrecFOAS-Norrish04}) gets away without assuming swapping compositionality $\SwCp$, which is a crucial ingredient for extending swapping to permutation. Moreover, in an indirect way, we also showed the other crucial ingredients needed for this extension, namely $\SwId$ and $\SwIv$, are not required for recursion either. 
 Thus, in this case swapping-based recursion requires significantly weaker 
 assumptions   
 than permutation-based recursion.  
%
\looseness=-1

\vspace*{-0.5ex}\subsection{Semantic-interpretation example}
\label{subsec-semInt}

The notion of interpreting syntax in semantic domains is a well-known challenging example for binding-aware recursion. 
%
%
%
Let $D$ be a set and $\AP : D \ra D \ra D$ 
and $\LM: (D \ra D) \ra D$ be operators modeling semantic notions of application and abstraction. 
(Subject to some axioms that are not of interest here, the structure $(D,\AP,\LM)$ is known as a Henkin model for $\lambda$-calculus \cite{bar-lam}.) 
An environment will be a function $\xi: \Var \ra D$. 
Given $x,y\in\Var$ and $d,e\in D$, we write $\xi\<x:=d\>$ for $\xi$ updated with value $d$ for $x$, 
and write $\xi\<x:=d,y:=e\>$ instead of $\xi\<x:=d\>\<y:=e\>$. 
\looseness=-1

The semantic interpretation 
$\sem : \Trm \ra (\Var \ra D) \ra D$ should go recursively 
by the clauses: 
\looseness=-1
\begin{myitem}
	\item[(1)] $\sem\,(\Vr\;x)\,\xi = \xi\;x$
	\hspace*{12ex}
	(2) 
	$\sem\,(\Ap\;t_1\,t_2)\,\xi = \AP\,(\sem\;t_1\,\xi)\,(\sem\;t_2\,\xi)$
	\item[(3)]  $\sem\,(\Lm\;x\;t)\,\xi = \LM\,(d \mapsto \sem\;t\,(\xi\<x:= d\>))$
\end{myitem}
%

Of course, these clauses do not work out of the box (i.e., do not form a correct recursive definition yet), and here is where the nominal recursors can help.  First, let us attempt to deploy the perm/free recursor $r_1$. To this end, we try to organize the target domain $I = (\Var \ra D) \ra D$ 
as a $(\Sigma_1,\Props_1)$-model $\SS$. 
The three desired clauses above already determine constructor operations $\Vr^\SS$, $\Ap^\SS$ and $\Lm^\SS$ on the set of interpretations, $I = (\Var \ra D) \ra D$, namely:
\begin{myyyitem}
	\item[(1)]  $\Vr^\SS : \Var \ra I$ by $\Vr^\SS x\;i\;\xi = \xi\;x$
	\hspace*{5ex} (2) $\Ap^\SS : I \ra I \ra I$ by $\Ap^\SS i_1\,i_2\;\xi = \AP\,(i_1\,\xi)\,(i_2\,\xi)$
	\item[(3)]   $\Lm^\SS :\Var \ra I \ra I$ by $\Lm^\SS x\;i\;\xi = \LM\,(d \mapsto i\,(\xi\<x := d\>))$  
	\looseness=-1
\end{myyyitem} 
Thus, we already have the $\Sigmac$ component of our intended model. 
Now we must define a permutation operator on $I$. 
The definition is obtained by analyzing the desired behavior of the to-be-defined function $\sem$ w.r.t.\ 
permutation; i.e., determining the value of  $\sem(t[\sigma])$ from $\sem\;t$ and $\sigma$. 
The 
answer is 
(4) $\sem\,(t[\sigma])\,\xi = \sem\;t\;(\xi \circ \sigma)$, 
%
and leads to defining $\_[\_]^\SS$ by $i\,[\sigma]^\SS\,\xi = i\,(\xi \circ \sigma)$. 

Note that, towards the goal of building a $(\Sigma_1,\Props_1)$-model $\SS$, we had no other choice 
on defining the operators $\Vr^\SS,\Ap^\SS,\Lm^\SS$ and $[\_]^\SS$ on the target domain $I$. And the free-variable (support) operator $\FV^\SS$ is also uniquely determined by the axiom \FvDPm{} (definability of freeness from permutation). 
\looseness=-1

Finally, to deploy $r_1$ and obtain a 
function $\sem$ satisfying clauses (1)--(4), it remains to check that $\SS$ satisfies 
$\Props_1$. 
But, as it turns out, $\SS$ does not satisfy one of the axioms in $\Props_1$, namely \FCB{} (freshness condition for binders). Indeed, \FCB{} requires that there exists a variable $x$ such that 
for all $i\in I$, 
$x \notin \FV^\SS (\Lm^\SS\,x\;i)$. 
Applying \FvDPm{} and 
the definitions of $\_[\_]^\SS$ and $\Lm^\SS$, 
we see that $x \notin \FV^\SS (\Lm^\SS\,x\;i)$ means  
%
%
$\LM\;(d \mapsto i\,(\xi \< x:=d,y:=\xi\,x \>) = \LM\;(d \mapsto i\,(\xi\<x:= d\>))$   
holds for all but a finite number of variables $y$.  
The only chance for the above to be true is if $i$, when applied to an environment, say $\xi'$, ignores the value of $y$ in $\xi'$ 
for all but a finite number of variables $y$; 
in other words, $i$ 
only analyzes the values of a finite number of variables in $\xi'$%
---but this is not guaranteed to hold for arbitrary elements $i\in I$.   
%
Thus, 
$r_1$ cannot be deployed directly to define semantic interpretations. 
\looseness=-1

Other recursors in our list 
can. 
E.g., the perm-free variant $r_2$ can be deployed as follows. We use the same definitions 
for $\Vr^\SS,\Ap^\SS,\Lm^\SS$ and $[\_]^\SS\!$, but now we can 
 choose the free-variable operator $\FV^\SS\!$ more flexibly, making sure that the $(\Sigma_2,\Props_2)$-morphism condition holds for $\FV^\SS\!$ versus $\FV$, i.e., that 
	(5) $\FV^\SS\!(\sem\;t) \su \FV\,t$ holds. 
%
Namely, 
we define $\FV^\SS i$ as $\{x\in\Var \mid \exists \xi:\Var \ra D,\,d \in D.\; i\;\xi \not= i\;(\xi\<x :=d\>)\}$. The definition identifies a natural notion of what it means for a variable to ``occur freely'' in a semantic item $i \in I$: when $i$ actually depends on $x$, i.e., when changing the value of $x$ in an input environment $\xi$ makes a difference in the result of applying $i$. 
%
%
%
And indeed, with $\FV^\SS\!$ defined like this, $\SS$ forms a $(\Sigma_2,\Props_2)$-model, which gives us a unique function $\sem$ satisfying (1)--(5). 
\looseness=-1

Thus, semantic interpretation is an 
example where our 
``head-to-head'' 
comparison has a visible outcome. 
%
But there is still 
an unexplored nuance here, which we discuss next.

Above, we argued that the semantic-interpretation example cannot be defined \emph{directly} using the perm/free recursor $r_1$. 
However, as discussed by 
\citet[\S 6.3]{pitts-AlphaStructural}, it turns out that it can be defined in a more roundabout manner, after some technical hassle. 
The trick 
is to restrict the target domain $I$ to a subset $I'$ on which the above defined operators do form an $(\Sigma_1,\Props_1)$-model, and use $r_1$ to define $\sem : \Trm \ra I'$. 
It is interesting to look at Pitts's definition of the subset $I'$, because it will reveal 
a way to relax the expressiveness comparison between epi-recursors.  
$I'$ is defined as 
$\{i \in I \mid \exists V \su \Var.\; V \mbox{ finite and } \forall x\in V.\;\forall \xi,d.\;i\;\xi = i\;(\xi\<x :=d\>)\}$. 
Then one proves that $I'$ is closed under the constructors $\Vr^\SS,\Ap^\SS,\Lm^\SS$. Moreover, for $I'$ the above problem with \FCB{} disappears, roughly because all the elements of $I'$ are finitary. So $I'$, with the same operators as those we tried for $I$, now forms a $(\Sigma_1,\Props_1)$-model, and $r_1$ recursion can proceed and define $\sem : \Trm \ra I'$, hence also $\sem : \Trm \ra I$. 

Having 
different nominal recursors in front of us laid out as epi-recursors, we can view Pitts's 
trick in a new light. Remember that, when deploying $r_2$ to define $\sem$, we used the operator 
$\FV^\SS\!$, which is a laxer notion of free-variable 
than that allowed by  
$r_1$. 
An equivalent definition of $I'$ is as the set of all elements of $I$ that have $\FV^\SS\!$ finite. 
Thus, Pitts's trick can be seen as borrowing the free-variable operator from the different recursor $r_2$, in order to single out a suitable target domain for deploying $r_1$! 
One can also prove 
that, on $I'$, the nominal-logic support (defined from permutation via  \FvDPm{}) \emph{coincides} with $\FV^\SS\!$---which means that, for the target domain $I'$, 
$r_1$ works as well as $r_2$.
\looseness=-1

Thus, on a 
subset of the target domain that is closed under constructors,  the previously deemed weaker recursor $r_1$ can simulate $r_2$. 
As it turns out, this is a general phenomenon, which we can phrase 
for epi-recursors as 
a gentler expressiveness comparison.  
\looseness=-1

\vspace*{-0.5ex}\subsection{A gentler comparison}
\label{subsec-moreGentle}

Our relation $r' \geq r$ compares the strength of epi-recusors directly, 
as 
inclusion between what $r$ can define and what $r'$ can define. 
The discussion ending \S\ref{subsec-semInt} 
suggests that this relation 
may be too strict. 
More flexibly, we could check if what $r$ can define is obtainable from what $r'$ can define \emph{up to composition with a morphism} (which can be 
an inclusion, as in Pitts's trick).  
\looseness=-1

Formalizing this 
for two epi-recursors $r = (\Bcat,T,\Ccat,I,R)$ and $r' = (\Bcat,T,\Ccat',I',R')$ must make sure to avoid pathological dependencies. 
Indeed, a first attempt is: For all objects $B$ in $\Bcat$ and morphisms 
$g: T \ra B$, if $g$ is definable by $r$ then there exists an object $B_0$ 
and two morphisms $g_0: T \ra B_0$ and $h:B_0\ra B$ such that $g_0$ is definable by $r'$ and $g = h \circ g_0$.  
But this would yield a vacuous concept, rendering any epi-recursor $r'$ stronger than any other $r$: just take $B_0=T$, $g_0=1_T$ (which is obviously definable by $r'$) and $h=g$. 
So we should be careful not to allow the above 
``transition'' morphism $h$ 
to depend on the $r$-definability morphism $g$.  
Otherwise, we would use $r$-definability itself to reduce $r$-definability to $r'$-definability. 
%
%
\looseness=-1

For 
producing morphisms to objects $B$ of $\Bcat$ independently of other data, 
the following concept comes handy.  
An \emph{initial segment} of a category $\Ccat$ is a pair $(\Ccat_0,(m(C):o(C)\ra C)_{C \in \Obj{\Ccat}})$ where $\Ccat_0$ is a full subcategory 
of $\Ccat$ and, for each object $C$ of $\Ccat$, $o(C)$ is an object of $\Ccat_0$
 and $m(C)$ a morphism in $\Ccat$. 
Using an ordering metaphor, 
an initial segment of a category provides a ``smaller'' object for any of its objects.   
Now we can formulate our gentler relation for comparing strength, 
called quasi-strength:  
\looseness=-1

\begin{defi} \rm \label{defi-qstrong} 
$r'= (\Bcat,T,\alb \Ccat',I',R')$ is 
\emph{quasi-stronger} than $r = (\Bcat,T,\Ccat,I,R)$, written $r' \wgeq r$,
when there exists an initial segment $(\Bcat_0,(m(B):o(B)\ra B)_{B \in \Obj{\Bcat}})$ of $\Bcat$ such that, 
for all 
$g: T \ra B$ definable by $r$, there exists 
a morphism $g_0: T \ra o(B)$ such that $g_0$ is definable by $r'$ and $g = m(B) \circ g_0$.  
\end{defi} 

Thus, $r' \wgeq r$ says that what $r$ can define is obtainable from what $r'$ can define up to composition with a morphism that only depends on the target object 
 in the base category.  
 Note that we use initial segments to make sure that the morphisms that ``fill the gap'' between the two recursors $r$ and $r'$ are given \emph{before hand}, so that they are independent from any specific recursively defined function (in particular, preventing bogus expressiveness orderings like the one exemplified above). 
 
$\wgeq$ is a preorder 
weaker than $\geq$. We write $r \simeq r'$ to mean 
that 
$r' \wgeq r$ and $r \wgeq r'$, 
i.e., $r$ and $r'$ have quasi-equal strengths. 
\looseness=-1

While being a reasonable weakening of $\geq$, 
the 
relation $\wgeq$ is likely to be more costly to deploy 
than $\geq$. Indeed, 
as suggested by our discussion in \S\ref{subsec-semInt}, 
applying $r' \wgeq  r$, i.e., using $r'$ in lieu of $r$, in particular extracting $o(B)$ from $B$ and using $o(B)$ as a ``more precise'' target domain, 
can 
involve non-negligible formal bureaucracy in concrete situations.
%
\looseness=-1

Our effective criterion for checking $\geq$ (Prop.~\ref{prop-extCriterion}) can be generalized to deal with $\wgeq$. 
Given two categories $\Ccat$ and $\Ccat'$, each with 
initial segments $(\Ccat_0,\alb(m(C):o(C)\ra C)_{C \in \Obj{\Ccat}})$ 
and $(\Ccat_0',\alb(m'(C):o'(C)\ra C)_{C \in \Obj{\Ccat'}})$, a functor $G: \Ccat \ra \Ccat'$ is said to \emph{preserve} the indicated initial segments if $G\,o(C)=o'(G\;C)$ and  $G\,(m(C))=m'(G\;C)$ for all $C \in \Obj{\Ccat}$. 
\looseness=-1

\begin{prop}\rm
	\label{prop-WeakExtCriterion} 
	Let $r = (\Bcat,T,\Ccat,I,R)$ and 
	$r' = (\Bcat,T,\Ccat',I',\alb R')$.  
	Assume 
	$(\Bcat_0,(m(B):o(B)\ra B)_{B \in \Obj{\Bcat}})$ is an initial segment of $\Bcat$
	and $(\Ccat_0,\alb(m_1(C):o_1(C)\ra C)_{C \in \Obj{\Ccat}})$ is an initial segment of $\Ccat$ 
	such that $\Ccat_0$ contains $I$ and 
	$R$ preserves the above initial segments, 
	and $F : \Ccat_0 \ra \Ccat'$ is a 
	pre-functor such that 
	$F\;I = I'$ and $R' \circ F = R_{\restr\Ccat_0}$ 
	 (where $R_{\restr\Ccat_0}$ is the restriction of $R$ to $\Ccat_0$).  
	Then $r' \wgeq r$. 
	\looseness=-1
\end{prop}
	
The gist of this criterion (and also its proof idea) 
is shown in Fig.~\ref{fig-critQuasiStrongerEpirec}: We start with a morphism $g$ 
definable by $r$ and use the two initial segments to factor it 
as 
a morphism $g_0$ definable by $r'$  and a remainder morphism $m(B)$.  
	
\begin{figure*}[!ht]
	\vspace*{-2.8ex}
	$$
	\xymatrix@C=3.0pc@R=1.3pc{		
		\Ccat_0 \su \Ccat \ar@/_1.5pc/[dd]_{R}  \ar[d]^{\exists F} &   I 
		\ar@/_1.5pc/[rrrr]^(.5){!_{I,C}}
		\ar[rr]^{!_{I,o_1(C)}}  
		&& o_1(C) \ar[rr]^{m_1(C)}  && C
		\\
	   \Ccat' \ar[d]^{R'}  
	    & I' = F\,I   
	   \ar[rr]^(.5){F\;!_{I,o_1(C)} \,=\, !_{I',C'}}
	   && C' = F\;o_1(C)   
	   &&  
		\\
		 \Bcat_0 \su  \Bcat &    R\,I \!=\! R'\,I' \!=\! T  \ar@/_1.5pc/[rrrr]^(.5){g \,=\, R\;!_{I,C}}
		 \ar[rr]^(.6){g_0 = R\,!_{I,o_1(C)} = R'\,!_{I',C'}}
		 && o(B) \ar[rr]^{m(B) \,=\, R\;m_1(C)}  && B = R\;C
	}
  	\vspace*{-1.8ex}
	$$
	\caption{Gentler criterion for comparing expressiveness} 
	\label{fig-critQuasiStrongerEpirec}
	\vspace*{-1.7ex}
\end{figure*}

Applying the gentler comparison to our recursors (via 
Prop.~\ref{prop-WeakExtCriterion}) 
yields 
a quite surprising result: 
\looseness=-1  %

\begin{thm} \rm
	\label{thm-qexpr}
	The epi-recursors described in Thm.~\ref{thm-allNominalRecs} (and in Fig.~\ref{fig-iter}) compare as follows by quasi-strength: \ \ 
	$
	r_1 \simeq r_2 \simeq  r_3 \simeq  r_4 \simeq  r_5 \simeq  r_6  
	\wgeq 
	r_8 \simeq r_9  \wgeq r_7  
	$. 
\end{thm} 
\emph{Proof idea. } 
%
%
%
When proving each 
$r_i \wgeq r_j$, we instantiate Prop.~\ref{prop-WeakExtCriterion} 
taking $r' = r_i$ and $r = r_j$. So here $\Bcat$ is the category of $\Sigmac$-models, $\Ccat'$ that of $(\Sigma_i,\Props_i)$-models, 
and $\Ccat$ that of $(\Sigma_j,\Props_j)$-models; $R'$ 
is the forgetful functor from $(\Sigma_i,\Props_i)$-models to $\Sigmac$-models, and $R$ the forgetful functor from $(\Sigma_j,\Props_j)$-models to $\Sigmac$-models; 
$B=\TTrm(\Sigmac)$, 
$I'=\TTrm(\Sigma_i)$ and $I=\TTrm(\Sigma_j)$.  

We define the initial segment $(\Bcat_0,(m(\AA):o(\AA)\ra \AA)_{\AA \in \Obj{\Bcat}})$ of $\Bcat$ as follows: For any $\Sigmac$-model $\AA$ 
we take 
$o(\AA)$ to be its minimal submodel (subalgebra), 
i.e., the one generated by 
$\Vr^\AA$, $\Ap^\AA$ and $\Lm^\AA$; we take $m(\AA): o(\AA) \ra \AA$ to be the inclusion morphism; and we take $\Bcat_0$ to be the full subcategory given by constructor-generated models. 
Each time, we will define the initial segment $(\Ccat_0,\alb(m_1(\MM):o_1(\MM)\ra \MM)_{\MM \in \Obj{\Ccat}})$ 
so that, for each $(\Sigma_j,\Props_j)$-model $\MM$, 
$o_1(\MM)$ is a submodel of $\MM$ whose carrier is generated by the constructors ($\Vr^\MM$, $\Ap^\MM$ and $\Lm^\MM$) and will have the other operators from $\Sigma_j$ defined in specific ways; and $\Ccat_0$ will be the full subcategory given by the objects $o_1(\MM)$. 
This way, it will be guaranteed that $R$ preserves initial segments. 
\looseness=-1

To prove the $\simeq$-chain going from $r_1$ to $r_6$, 
thanks to Thm.~\ref{thm-expr} and the fact that $\wgeq$ is weaker than $\geq$, 
it suffices to prove $r_3 \wgeq r_6$.   
%
%
We proceed as follows: Given a $(\Sigma_6,\Props_6)$-model $\MM$ of carrier $M$, we take $o_1(\MM)$ to be 
a submodel $\MM'$ of $\MM$, having as carrier set the subset $M'$ of $M$ generated by the constructors $\Vr^\MM$, $\Ap^\MM$ and $\Lm^\MM$, having the constructors and swapping operators inherited from $\MM$ and having freshness defined from swapping in nominal style (as in \FrDSw{}); crucially, this definition of freshness turns out to be equivalent to an inductive one using \FrVr{}, \FrAp{} and \FrLm{}, making $\MM'$ the minimal $(\Sigma_6,\Props_6)$-submodel of $\MM$. 
Now, the pre-functor $F$ is defined on objects as follows: $F\;\MM'$ 
is the $\Sigma_3$-model having the same constructors and swapping operator as $\MM'$, and having the free-variable operator defined standardly from the freshness operator of $\MM'$, via negation. (And on morphisms, $F$ is the identity.)   $F\;\MM'$ satisfies $\Props_3$: \SwVr{}, \SwAp{}, \SwLm{} and \FvDSw{} hold by construction, and \FCB{}, \SwId{}, \SwIv{} and \SwCp{} follow by induction on the definition of $M'$. The other required properties are trivial, e.g., $F\;I = I'$ here means that the standard definition of free-variables from freshness is correct for terms; and $R' \circ F = R_{\restr\Ccat_0}$ means that $F$ commutes with the forgetful functors. 
%
%
\looseness=-1

To prove the $(\simeq,\wgeq)$-chain going from $r_{6}$ to $r_7$, 
again thanks to Thm.~\ref{thm-expr}   
it suffices to prove 
$r_8 \wgeq r_9$ and 
$r_6 \wgeq r_8$. (We will no longer show explicitly the definitions of the initial segment and the pre-functor, but give the ingredients from which they can be constructed similarly to how we did above.)
For 
	$r_8\wgeq r_9$, we 
start similarly to the proof of $r_3 \wgeq r_6$, namely for a 
$(\Sigma_{9},\Props_{9})$-model $\MM$ 
we take the 
minimal submodel $\MM'$ where freshness definable from renaming (via \FrDRn{}) turns out to coincide 
with the inductively defined version via \FrVr{}, \FrAp{} and \FrLm{}. 
Because the carrier $M'$ of $\MM'$ is the image of the unique 
$\Sigma_{9}$-morphism 
$f: \TTrm(\Sigma_{9}) \ra \MM$ 
ensured by the initiality of  
$\TTrm(\Sigma_{9})$, $\MM'$ satisfies all unconditional equations satisfied 
by $\TTrm(\Sigma_{9})$, in particular, all the $\Props_8$ properties. 
%
%
%
%

Finally, the proof of 
	$r_6 \wgeq r_8$ exploits the observation that renaming is definable from swapping 
not only for terms, but also for any $(\Sigma_8,\Props_8)$-model $\MM$ that guarantees the existence of fresh 
variables, 
i.e., having its elements finitely supported: $m\,[z_1 \wedge^{\MM} z_2]$ is defined as $
m\,[y /^{\MM} z_1]\,\alb[z_1 /^{\MM} z_2]\,[z_2 /^{\MM} y]
$ where $y$ is fresh (and, using the $\Props_8$ axioms, the choice of $y$ can be proved not to matter). 
While arbitrary $(\Sigma_8,\Props_8)$-models $\MM$ do not guarantee finite support, we can again switch to a minimal 
submodel $\MM'$ that does guarantee it---and in $\MM'$ the above definition indeed yields a swapping operator that together with 
the constructors and freshness 
satisfies $\Props_6$.  
\qed
\looseness=-1
\medskip

Thus, $\wgeq$  brings a dramatic flattening of the $\geq$ hierarchy established by Thm.~\ref{thm-expr}: 
All the swapping- and permutation-based recursors $r_1$--$r_6$ have equal quasi-strengths. The intuition for this, as we discovered during the proofs, is the following. Recall that the differences in strength (using $\geq$) between these recursors were due to: (1) looseness or tightness of their connection between swapping/permutation and freeness/freshness, (2) higher flexibility of swapping compared to permutation, and (3) more focused nature of congruence compared to an algebraic axiomatization. Remarkably, all these differences vanish if we are allowed to navigate along submodels, which $\wgeq$ enables. 
This is because (as 
explained in the 
proof of Thm.~\ref{thm-qexpr}), 
certain minimal submodels  
are much more ``term-like'' than an arbitrary model; they generalize Pitts's submodel definition for semantic interpretation, where  nominal-style freshness coincides with other, more loosely axiomatized notions of freshness.   
\looseness=-1

An interesting takeover when switching from $\geq$ to $\wgeq$ 
is the swapping/permutation-based recursors $r_1$--$r_6$ 
becoming \mbox{(quasi-)}stronger 
than the renaming-based recursors $r_8$ and $r_9$. 
Indeed, defining renaming from swapping or vice versa seems impossible in arbitrary models, meaning these two types of recursors are $\geq$-incomparable. But when switching to submodels (allowed by $\wgeq$) one direction is 
possible: The 
swapping of two variables 
can be defined in a renaming-based model similarly to how it is done for concrete terms, via picking an intermediate 
fresh variable; 
and ``picking fresh'' is possible in 
minimal submodels because everything there is finitely supported. 
\looseness=-1

\medskip
\noindent 
\textbf{Summary. }
Epi-recursors are comparable for expressiveness 
by a strict relation $\geq$, saying that everything definable by one is definable by the other, and a laxer relation $\wgeq$, saying that everything definable by one can be defined by the other with the help of an additional morphism, typically a submodel inclusion.  
The handling of the semantic-interpretation example with the nominal-logic recursor 
was our inspiration for $\wgeq$, and 
suggests the additional 
overhead incurred 
by $\wgeq$.  
%
The effective criteria we used to prove these relations for 
concrete recursors (Props.~\ref{prop-extCriterion} and \ref{prop-WeakExtCriterion}), 
can be paraphrased using ``is'' and ``has'': 
\looseness=-1
\begin{mmyitem}
\item $r' \geq r$ holds if any $r$-model $\MM$ 
\emph{is} an $r'$-model---in that it can be regarded (after defining the relevant operations, in a way that ensures functoriality) as an $r'$-model. 
\item $r' \wgeq r$ holds if any  $r$-model $\MM$ \emph{has} an $r'$-submodel---in that there exists a submodel $\MM'$ of $\MM$ that 
still satisfies the properties required by $r$, and 
can be regarded as an $r'$-model.  
\looseness=-1
\end{mmyitem}
%
%
The $\wgeq$-hierarchy 
is significantly flatter than the $\geq$-hierarchy,  
sending an 
egalitarian message: \emph{Most 
	nominal recursors turn out to have the same strength}, with the only nuance that 
those based on 
symmetric operators (swapping and permutation) are more expressive than those based on asymmetric ones (renaming and  substitution). 
\looseness=-1

\subsection{Negative results}
\label{subsec-negResRec}

Thms.~\ref{thm-expr} and \ref{thm-qexpr} establish $\geq$ and $\wgeq$
relationships between recursors, which essentially tell us that a recursor can replace/simulate another recursor (under a tighter or a looser notion of replacement). 
But how about the question of when a recursor \emph{cannot} replace another? 
The discussion in \S\ref{subsec-semInt} suggests that 
$r_1 \geq r_2$ does not hold. 
The next proposition 
states the two negative results we know so far:
\looseness=-1

\begin{prop}\rm \label{prop-negRec}
$r_1 \not\geq r_2$ (i.e., it is not the case that $r_1 \geq r_2$) 
and $r_2 \not\geq r_4$ (i.e., it is not the case that $r_2 \geq r_4$). 
\end{prop}
\vspace{-0.5ex}
\noindent 
\emph{Proof sketch. } 
To prove $r_i \not\geq r_j$, we must provide a $(\Sigma_j,\Props_j)$-model $\MM$ for which the $\Sigmac$-reduct (i.e., the $\Sigmac$-model obtained by forgetting the operators from $\Sigma_j \setminus \Sigmac$) cannot be the $\Sigmac$-reduct of any $(\Sigma_i,\Props_i)$-model. 


 
For $r_1 \not\geq r_2$, 
 we take 
 the $(\Sigma_2,\Props_2)$-model $\MM$ to have as carrier the set $M = \Trm \cup A$, where 
 $A$ consist of all the streams of variables (in $\Var^\Nat$) whose sets of occurring variables are infinite. 
 We let $(t_i)_{i\in \Nat}$ be a family of terms such that all are ground ($\FV\;t_i = \emptyset$)
 and mutually distinct. 
 We define $\MM$'s operators on $M$ by extending the standard term operators from $\Trm$ as follows, for any $\xs \in A$ (where $\map_\sigma$ is the standard stream-map operator and $\rem_y\,\xs$ removes all occurrences of $y$ from $\xs$): 
 \looseness=-1
 \\
 \hspace*{-0.5ex}
 \begin{tabular}{cc}
 	\begin{tabular}{l} 
 		$\bullet$ $\FV^\MM \xs = \Vars\;\xs$
 		\\
 		$\bullet$ $\Lm^\MM\,y\;\xs = \rem_y\,\xs$ for any $y\in\Var$
 		\\
 		$\bullet$ $\Ap^\MM\;\xs\;t_i = \Vr\;\xs_i$ for any $i\in \Nat$  
 	\end{tabular} 
 	&
 	\hspace*{-0.5ex}
 	\begin{tabular}{l} 
 		$\bullet$ $\Ap^\MM\;\xs\;m = t_0$ for any $m \in M \sm \{t_i \mid i \in \Nat\}$
 		\\
 		$\bullet$ $\Ap^\MM\;s\;\xs = t_0$ for any $s \in \Trm$ 
 		\\
 		$\bullet$ $\xs[\sigma]^\MM = \map_\sigma\,\xs$ for any $\sigma\in\Perm$
 	\end{tabular}
 \end{tabular} 

Note that, on $A$, the free-variable-like and abstraction-like operators are 
 natural, in particular $\Lm^\MM$ removes all occurrences of the abstracted variable. On the other hand, the application-like operator is contrived: the only interesting case is  $\Ap^\MM\;\xs\;t_i$, where application emulates the $i$'th projection, retrieving the $i$'th element of the stream $\xs$; in the other cases application simply returns the ground term $t_0$.  We can check that $\MM$ thus defined satisfies the $\Props_2$ properties.
One the other hand, the $\Sigmac$-reduct of $\MM$, i.e., $\Trm \cup A$ equipped with the above-defined constructor-like operators, cannot be the reduct of any $(\Sigma_1,\Props_1)$-model, i.e., there is no way to define the operators $\_[\_]'$ and $\FV'$ on $\Trm \cup A$ that, together with $\Vr^\MM$, $\Ap^\MM$ and $\Ap^\MM$, make it a $(\Sigma_1,\Props_1)$-model.  Indeed, if such operators 
$\_[\_]'$ and $\FV'$ existed, then the $\Props_1$ axioms would imply that 
$\_[\_]'$ extends the standard permutation operators from $\Trm$ and $A$, 
and then that $\FV'\,\xs = \Var$ for all $\xs$, which contradicts \FCB{}. 
\looseness=-1

For $r_2 \not\geq r_4$,  
we take the $(\Sigma_4,\Props_4)$-model $\MM$ to have as carrier the set $M = \Trm \cup \{a\}$ (where $a\not\in \Trm$), i.e., to consist of terms plus an additional element $a$. Let $x$ be a fixed variable. We define $\MM$'s operators on $M$ by extending the standard term operators from $\Trm$ as follows: 
\\
\hspace*{-2.5ex}
\begin{tabular}{cc}
\begin{tabular}{l} 
$\bullet$ $\FV^\MM a = \Var$ (the set of all variables)
\\
$\bullet$ $\Lm^\MM\,y\;a = \Lm^\MM\,y\;(\Vr\;x)$ for any $y\in\Var$
\\
$\bullet$ $\Ap^\MM\;a\;a = \Ap^\MM\;(\Vr\;x)\;(\Vr\;x)$ 
\end{tabular} 
&
\hspace*{-3.5ex}
\begin{tabular}{l}
$\bullet$ $\Ap^\MM\;a\;t = \Ap^\MM\;(\Vr\;x)\;t$ for any $t\in\Trm$
\\
$\bullet$ $\Ap^\MM\;t\;a = \Ap^\MM\;t\;(\Vr\;x)$ for any $t\in\Trm$
\\
$\bullet$ $a[z_1 \hspace*{-0.2ex}\wedge \hspace*{-0.2ex} z_2]^\MM = \Vr\,(x[z_1 \hspace*{-0.2ex}\wedge
\hspace*{-0.2ex}
 z_2])$ for any $z_1,z_2\in\Var$
\end{tabular}
\end{tabular} 

Thus, the free variables of $a$ are the entire set of variables, and the constructor and swapping operators on $a$ yield the same results as for $\Vr\;x$, i.e., have $\Vr\;x$ act in lieu of $a$.
We can check that $\MM$ satisfies $\Props_4$. On the other hand, the $\Sigmac$-reduct of $\MM$, i.e., $\Trm \cup \{a\}$ equipped with the above-defined constructor-like operators, cannot be the reduct of any $(\Sigma_2,\Props_2)$-model, i.e., there is no way to define the operators $\_[\_]'$ and $\FV'$ on $\Trm \cup \{a\}$ that, together with $\Vr^\MM$, $\Ap^\MM$ and $\Ap^\MM$, make it a $(\Sigma_2,\Props_2)$-model.  
Indeed, if such operators $\_[\_]'$ and $\FV'$ existed, 
then the axioms in $\Props_2$ would imply that 
$\_[\_]'$ extends the standard permutation operator on $\Trm$, and 
also that $\_[\sigma]'$ is bijective on $\Trm \cup \{a\}$ for any 
permutation $\sigma$; so the only possibility is that $a[\sigma]' = a$ for any $\sigma$; this together with \PmAp{} would imply that $(\Ap^\MM\,a\;a)[\sigma]' = \Ap^\MM\,(a[\sigma]')\;(a[\sigma]') = \Ap^\MM\,a\;a$, i.e., $(\Ap\;(\Vr\;x)\;(\Vr\;x))[\sigma] = \Ap\;(\Vr\;x)\;(\Vr\;x)$, 
which is false for any $\sigma$ that modifies $x$. 
\qed 
\looseness=-1

\medskip 
Note that, if we write $>$ for the strict version of $\geq$ (defined as 
$r > r'$ iff $r \geq r'$ and $r' \not\geq r$), then assuming $r \geq r'$, a negative 
result $r' \not\geq r$ is a strictness result $r > r'$.  So from Thm.~\ref{thm-expr} , Prop.~\ref{prop-negRec} and Thm.~\ref{thm-qexpr} 
we have $r_2 > r_1$ but $r_2 \simeq r_1$, and also $r_4 > r_2$ 
but $r_4 \simeq r_2$. 
We do not yet have negative/strictness results across the board, in particular, 
none for $\wgeq$.  
\looseness=-1

\section{The coinductive spectrum} 
\label{sec-co}

Next we will shift focus from the standard terms with bindings discussed so far, which were defined inductively, 
to (possibly) infinitary non-well-founded terms with bindings, defined \emph{coinductively}, where the constructors 
can be applied an infinite number of times. Unlike with the inhabitants of standard coinductive datatypes, we will still identify terms modulo  $\alpha$-equivalence. Rather than recursion, we will now study \emph{corecursion}, that is, mechanisms for defining functions having terms not as source domain, but as target domain (codomain). 
%
Building on the experience of having handled the recursors, 
we will now take a more direct route, and at a faster pace:
After recalling infinitary terms (\S\ref{subsec-infTerms}), 
we 
introduce abstract epi-corecursors (\S\ref{subsec-epicoRec}), 
then delve into the spectrum of nominal corecursor instances, connect with pre-existing nominal corecursors, and establish a hierarchy  
(\S\ref{subsec-hiarNomCorec}).
\looseness=-1
%


\vspace*{-0.5ex}\subsection{Infinitary terms with bindings}
\label{subsec-infTerms}

\newcommand\mydots{\makebox[0.8em][c]{.\hfil.\hfil.}}

Let $\Var$ be a set of variables whose cardinality is $\aleph_1$, the first uncountable  cardinal. 
(Any uncountable regular cardinal would 
do---we only care about 
the existence of fresh variables for any term.)  
%
The set $\ITrm$ of \emph{infinitary $\lambda$-terms}, \emph{iterms} for short,  
is defined by the same grammar as before, 
%
$
t ::= \Vr\;x  \mid \Ap\;t_1\;t_2  \mid \Lm\;x\;t
$, 
but interpreted \emph{coinductively}, i.e., allowing 
an infinite number of constructors. 
For example,  
$\;\ldots\,(\Ap\,(\Lm\,x_n\,(\ldots (\Ap\, (\Lm\,x_1\,(\Vr\,x_0))\,(\Vr\,x_1))\ldots))\,(\Vr\,x_n))\,\ldots\;$
 is an iterm, infinitely alternating abstractions and applications. 
%
%
Similarly to terms, iterms are equated 
modulo  $\alpha$. 
%
\looseness=-1

In more detail, the above definition means: One first defines  the set $\PITrm$ of 
\emph{pre-iterms} to be (co)freely generated by the 
grammar \ $p ::= \PVr\;x  \mid \PAp\;p_1\;p_2  \mid \PLm\;x\;p$ \ 
under the coinductive interpretation, i.e., under the assumption that constructors can be applied infinitely. 
Thus, 
$\PITrm$ is a standard coinductive datatype, given by the final coalgebra of the functor on sets taking, on objects, any set $A$ to 
$\Var + A \times A + \Var \times A$ (and operating on morphisms as expected; 
App.~\ref{appsub-preiterm} 
gives full details). 
Then one defines the $\alpha$-equivalence relation 
$\equiv\; : 
\PITrm \ra \PITrm \ra 
\Bool$  coinductively, proves that it 
is an equivalence, and defines $\ITrm$ by quotienting $\PITrm$ to 
it, i.e., takes $\ITrm = \PITrm/\equiv$. 
Finally, one proves that the pre-iterm constructors are compatible with $\equiv$, 
which allows to define the constructors on iterms:  
$\Vr : \Var \ra \ITrm$, $\Ap : \ITrm \ra \ITrm \ra \ITrm$ 
and $\Lm : \Var \ra \ITrm \ra \ITrm$.   
%
We will 
focus on iterms, forgetting about pre-iterms.

The iterms have been 
studied 
in the context of $\lambda$-calculus denotational semantics, 
e.g., 
the B\"{o}hm, L\'{e}vy-Longo and Berarducci trees  
of a $\lambda$-term \cite{bar-lam}. A bottom element $\bot$ is often included in the iterm grammar, but we omit it here since it would be entirely passive in our results.  
\looseness=-1

We also consider the usual operators (just like in the inductive case), namely 
(capture-avoiding) substitution  
	$\_[\_\,/\_] : \ITrm \ra \ITrm \ra \Var \ra \ITrm$, 
	(capture-avoiding) renaming  
	$\_[\_\,/\_] : \ITrm \ra \Var \ra \Var \ra \ITrm$,  
 swapping  
	$\_[\_\sw\_] : \ITrm \ra \Var \ra \Var \ra \ITrm$, 
permutation $\_[\_] : \ITrm \ra \Perm \ra \ITrm$, 
free-variables 
	$\FV : \ITrm \ra \Pow(\Var)$, 
and freshness $\_\fresh\_ : \Var \ra \ITrm \ra \Bool$. 

Finally, for any set $A$, let $\PPne(A)$ denote the set of nonempty subsets of $A$.
We consider the iterm \emph{destructor}, 
$\Dest : \ITrm \ra \Var + 
\ITrm \times \ITrm + \PPne(\Var\times \ITrm)$ , defined as follows, where we write $\Vv$, $\Aa$ and $\Ll$ 
for the three injections into the sum type  $\mathsf{S} = \Var + 
\ITrm \times \ITrm + \PPne(\Var\times \ITrm)$ (so that
$\Vv : \Var \ra \mathsf{S}$, 
$\Aa : \ITrm \times \ITrm \ra \mathsf{S}$ and 
$\Ll : \PPne(\Var\times \ITrm) \ra \mathsf{S}$): 
\vspace*{-0.8ex}
$$
\Dest\;t  =\left\{
\begin{array}{ll}
	\Vv\;x, & \mbox{if $t = \Vr\;x$} \\
	\Aa\,(t_1,t_2), & \mbox{if $t = \Ap\;t_1\;t_2$} \\
	\Ll\,\{(x,t') \mid \mbox{$t = \Lm\;x\;t'$}\}, & \mbox{otherwise (i.e., if $t$ is a $\Lm$-abstraction)} \\
	\end{array}
\right.
\vspace*{-0.5ex}
$$
%
$\Dest$ is the dual of the constructors, peeling off the last constructor from an iterm and returning  
its arguments.\footnote{See page \pageref{alterDest} for a discussion of alternative types for the destructor and destructor-like operators.} 
It is similar to the destructors for standard datatypes, except that on $\Lm$-abstractions it is nondeterministic. 
This is because the $\Lm$ constructor is not injective and therefore an iterm $t$ could have been built in 
different ways using $\Lm$. $\Dest$ considers all these ways, i.e., returns the set of all pairs $(x,t')$ such that $t$ has the form $\Lm\;x\;t'$.  
We thus have: $t = \Lm\;x\;t' \!\iff\! \exists K.\; \Dest\;t = \Ll\;K \mbox{ and } (x,t') \in K$. 
%
%
For iterms (and for terms too, where the destructor is defined in the same way), destructor and constructors are two faces of the same coin. 
But 
since the models for corecursion will have to emulate the destructor, 
we will 
look at 
destructor-based (re)formulations of iterm properties. 
\looseness=-1

Of the basic properties of terms listed in Fig.~\ref{fig-basicProps}, all except for the last group (the 
nominal-logic specific properties) also hold for iterms, so we will consider some of them in the context of iterms as well. 
The properties in this last group 
are tied to the finiteness of a term's free variables; 
for them 
to become true for iterms, 
we must replace ``(in)finite'' with ``(un)countable''. 
%
%

Moreover, 
Fig.~\ref{fig-basicCoProps} collects destructor-based iterm 
counterparts of some term properties from Fig.~\ref{fig-basicProps}.
%
Often, these are just (equivalent) destructor-based reformulations of the constructor-based properties. For example, this is the case of \ISwVr{}, \ISwAp{},  \ISwLm{} versus \SwVr{}, \SwAp{}, \SwLm{}.  
\looseness=-1

\input{figCoBasicProperties.tex}

However, 
sometimes we reformulate not the original property from Fig.~\ref{fig-basicProps}, but a converse (or ``almost converse'') of it. 
For example, the converse of 
\SwCg{} from Fig.~\ref{fig-basicProps} is:
%
$\Lm\;x_1\;t_1 = \Lm\;x_2\;t_2$ implies that there exists $z$ such that 
$z \not\in \{x_1,x_2\}$, 
$z\;\fresh\;t_1,t_2$, and 
$t_1[z\sw x_1] = t_2[z\sw x_2]$. 
%
This converse does hold for terms, and for iterms as well. However, we prefer to consider a weaker version of it: 
%
	$\Lm\;x_1\;t_1 = \Lm\;x_2\;t_2$ implies that there exists $z$ such that 
	($z = x_1$ or $z\;\fresh\;t_1$), ($z = x_2$ or $z\;\fresh\;t_2$), 
    and 
	$t_1[z\sw x_1] = t_2[z\sw x_2]$. 
The latter, reformulated using destructor notation, is   exactly \ISwCg{} from Fig.~\ref{fig-basicCoProps}.  
The reason why we prefer a weaker version (here due to a weaker conclusion) 
is the same as why 
we preferred a weaker version of \SwCg{} in the inductive case (there, due to a stronger hypothesis): because, to make the (co)recursors  
as expressive as possible, 
we want the models 
to have axioms as weak as possible. 
Sometimes we include in Fig.~\ref{fig-basicCoProps} two different destructor-based counterparts of a 
constructor-based property, e.g., \IRnBvr{} and \IIRnBvr{} for \RnBvr{}. 
\looseness=-1


Save for the 
finite vs.\ countable nuance in the  last group, all 
properties in Figs.~\ref{fig-basicProps} and \ref{fig-basicCoProps} 
hold 
for both terms and iterms. Their selection 
becomes 
relevant when 
regarding them as properties 
of 
models.  
%
The duality between the Fig.~\ref{fig-basicProps} and Fig.~\ref{fig-basicCoProps} properties, which informs the naming of the latter, is neither perfect nor fully systematic. But this naming 
will allow us to draw 
parallels. 
\looseness=-1


\vspace*{-0.5ex}\subsection{Epi-corecursors}
\label{subsec-epicoRec}

We introduce abstract epi-corecursors as a natural dual of epi-recursors. The idea is the same: 
A definition of a morphism in a base category is underpinned by adding more structure coming from an extended category. The difference is that the base object 
is now not the source, but the target of the to-be-defined morphism, and the underpinning 
occurs not via initiality but via finality.

\begin{defi} \rm \label{defi-epirec}
	An \emph{epi-corecursor} is a tuple $\ccr = (\Bcat,T,\Ccat,J,R)$ where:
	\begin{mmyitem} 
		\item $\Bcat$ is a category called \emph{the base category}
			\hspace*{9ex} 
		$\bullet$ $T$ is an object in $\Bcat$ called \emph{the base object}
		\item $\Ccat$ is a category called \emph{the extended category}
			\hspace*{5.12ex} 
		$\bullet$ $J$ is a final object in $\Ccat$
		\item $R : \Ccat \ra \Bcat$  is a functor such that $R\;J = T$   
	\end{mmyitem}
\end{defi}
%

Just like for epi-recursors, 
in typical epi-corecursor examples $\Ccat$ and $\Bcat$ will be categories of models, 
with the  models in $\Ccat$ having more structure than those in $\Bcat$, and 
$R$ will be a structure-forgetting functor.   
%
To define a morphism $g: B \ra T$ in $\Bcat$ (where $B$ is some object in $\Bcat$) using an  epi-corecursor $\ccr = (\Bcat,T,\Ccat,J,R)$, we (1) extend $B$ to an object $C$ in $\Ccat$ (with $R\;C = B$) yielding a morphism $\im_{C,\,J} : C \ra J$ in $\Ccat$ from the finality of $J$, 
then (2) take $g$ to be $R\;\im_{C,\,J}$, the restriction of $\im_{C,\,J}$ to $\Bcat$. 
%
Thus, we call a morphism $g: B \ra T$ \emph{definable by the 
		epi-corecursor $\ccr$} if 
	$g = R \,\im_{C,\,J}$ for some extension $C$ of $B$.    
\looseness=-1

\leftOut{
\begin{figure}
	$$
	\xymatrix@C=4pc@R=1pc{
		\Ccat \ar[d]_R &  C  \ar[r]^{\im_{C,\,J}} &  J
		\\
		\Bcat &  B = R\,C   \ar[r]^{R\,\im_{C,\,J}}  & T = R\,J
	}
	$$
	   \vspace*{-2ex}
	\caption{Epi-corecursor
	}
	\label{fig-epiCRP}
	\vspace*{-2ex}
\end{figure}
}

\vspace*{-0.5ex}\subsection{A hierarchy of nominal corecursors}
\label{subsec-hiarNomCorec}

To discuss concrete nominal corecursors, we slightly adapt the notions of signature and model 
used for nominal recursors from \S\ref{subsec-sigMod}. Namely, we use the same notions 
except that we replace the constructor symbols $\vr,\ap,\lm$ and their interpretations with 
a destructor symbol $\dest$, interpreted accordingly. 
All signatures $\Sigma$ now extend not the constructor signature $\Sigmac = \{\vr,\ap,\lm\}$, but the destructor signature 
$\Sigmad = \{\dest\}$. A $\Sigma$-model $\MM$ has a carrier set $M$,  interprets 
the signature's non-destructor symbols 
as described in \S\ref{subsec-sigMod}, and interprets $\dest$ as an operation 
$\Dest^\MM : M \ra \Var + M \times M + \PPne(\Var \times M)$. 
%
The iterm $\Sigma$-model $\ITTrm(\Sigma)$ is the $\Sigma$-model whose carrier set is  
$\ITrm$ 
and whose operations and relations are the standard ones for iterms (discussed in \S\ref{subsec-infTerms}). 

The notion of morphism of $\Sigma$-models $g : \MM \ra \MM'$ is defined like in \S\ref{subsec-sigMod}, but 
replacing commutation with the constructors by sub-commutation with the destructor:
$(1_\Var + g \times g + \img(1_\Var \times g))\,(\Dest^\MM m) \sqsubseteq \Dest^{\MM'}(g\;m) $ for all $m$ in the carrier set $M$. 
The above relation $\sqsubseteq$ on $\Var + M' \times M' + \PPne(\Var \times M')$ 
is defined by taking $u \sqsubseteq v$ to mean that: 
either $u=\Vv\;x = v$ for some $x$; 
or $u=\Aa(m_1',m_2') = v$ for some $m_1',m_2'$;
or $u=\Ll\;K$, $v=\Ll\;K'$ and $K \su K'$ for some $K,K'$. 
Thus, the \emph{sub}-commutation shows 
in the abstraction case (which is nondeterministic), where we allow inclusion instead of equality. 
To see why sub-commutation is the natural condition here,  
note that 
for a morphism that targets iterms, 
$g : \MM \ra \ITTrm(\Sigma)$, it is equivalent to the conjunction of the following three conditions:  
(1) $\Dest^{\MM}\,m = \Vv\;x$ implies $g\;m = \Vr\;x$; 
(2) $\Dest^{\MM}\,m = \Aa (m_1,m_2)$ implies $g\;m = \Ap\;(g\;m_1)\;(g\;m_2)$; 
(3) $\Dest^{\MM}\,m = \Ll \;K$ and $(x,m')\in K$ implies $g\;m = \Lm\;x\;(g\;m')$. 

Our nominal (epi-)corecursors will underpin corecursive definitions having $\ITTrm(\Sigmad)$ as target 
model by considering extensions of $\Sigmad$ to larger signatures $\Sigma$, along with certain axiomatizations of $\Sigma$-models given by subsets of the properties in Fig.~\ref{fig-basicCoProps} (interpreted not on iterms, but on $\Sigma$-models).

\newcommand\perFreeCoRec
{
	\begin{tabular}{|c|}
		\hline
		$\ccr_1$ \ (perm/free)
		\\\hline
		\IPmVr{}, \IPmAp{}, \IPmLm{}, 
		\\
		\IPmId{}, \IPmCp{},  
		\\
		\IFvDPm{}, \IPmBvr{}
		\\\hline
	\end{tabular}	
}

\newcommand\perFreeBVrCoRec
{
	\begin{tabular}{|c|}
		\hline
		$\ccr_{2}$ \ (perm/free variant)
		\\\hline
		\IPmVr{}, \IPmAp{}, \IPmLm{}, 
		\\
		\IPmId{}, \IPmCp{},  
		\\
		\IPmFv{},  
		\IPmBvr{}, 
		\\
		\IFvVr{}, \IFvAp{}, \IFvLm{}  
		\\\hline
	\end{tabular}	
}

\newcommand\swapFreeCoRecV
{
	\begin{tabular}{|c|}
		\hline
		$\ccr_3$ \ (swap/free variant)
		\\\hline
		\ISwVr{}, \ISwAp{}, \ISwLm{}, 
		\\
		\ISwId{}, \ISwIv{},  \ISwCp{},  
		\\
		\IFvDSw{}, \ISwBvrT{}
		\\\hline
	\end{tabular}	
}

\newcommand\swapFreeBVrCoRec
{
	\begin{tabular}{|c|}
		\hline
		$\ccr_{5}$ \ (swap/fresh variant)
		\\\hline
		\ISwVr{}, \ISwAp{}, \ISwLm{}, 
		\\
		\hlt{$\ISwId{}$}, \hlt{$\ISwIv{}$},  \hlt{$\ISwCp{}$},  
		\\
		\ISwFr{}, 
		\ISwBvr{}, 
		\\
		\IFrVr{}, \IFrAp{}, \IFrLm{}  
		\\\hline 
	\end{tabular}	
}

\newcommand\swapFreeCgCoRec
{
	\begin{tabular}{|c|}
		\hline
		$\ccr_{6}$ \ (swap/fresh)
		\\\hline
		\ISwVr{}, \ISwAp{}, \ISwLm{}, 
		\\
		\hlt{$\ISwId{}$}, \hlt{$\ISwIv{}$},  \hlt{$\ISwCp{}$}, 
		\\
		\ISwFr{}, \hlt{$\IFrSw{}$}, 
		\ISwCg{} 
		\\
		\IFrVr{}, \IFrAp{}, \IFrLm{}  
		\\\hline
	\end{tabular}	
}


\newcommand\substFreshCoRec
{  
	\begin{tabular}{|c|}
		\hline
		$\ccr_{7}$ \ (subst/fresh)
		\\\hline
		\ISbVr{}, \ISbAp{}, \ISbLm{}, 
		\\
		\ISbId{}, \ISbChFr{}, \ISbCm{} 
		\\
		\ISbFr{}, \IFrSb{},  \ISbBvr{}, \IISbBvr{} \!\!\!
		\\
		\IFSupFr{}, \IFrVr{}, \IFrAp{}, \IFrLm{}  \!\!\!
		\\
		 \VrInv{}, \FrVr{}
		\\\hline
	\end{tabular}	
}

\newcommand\renamingCoRec
{
	\begin{tabular}{|c|}
		\hline
		$\ccr_8$ \ (renaming)
		\\\hline
		\IRnVr{}, \IRnAp{}, \IRnLmO{}, 
		\\
		\IRnId{}, \IRnIm{}, \IRnCh{}, \IRnCm{} \!\!\!
		\\
		\IFrDRn{}, \IRnBvr{}, \IIRnBvr{}
		\\
		\IFSupFr{}, \IFrRnT
		\\\hline
	\end{tabular}	
}

\newcommand\renamingFreshVCoRec
{
	\begin{tabular}{|c|}
		\hline
		$\ccr_9$ \ (renaming/fresh variant)
		\\\hline
		\IRnVr{}, \IRnAp{}, \IRnLmO{}, 
		\\
		\IRnId{}, \IRnChFr{}, \IRnCm{}, 
		\\
      \IRnFr{}, \IFrRn{},  \IRnBvr{}, \IIRnBvr{}\!\!\!
		\\
		\IFSupFr{}, \IFrVr{}, \IFrAp{}, \IFrLm{}  
		\\\hline
	\end{tabular}	
}

\begin{figure}[!t]
\hspace*{-1.2ex}
		{\small 
			\begin{tabular}{ccc} \vspace*{-1.2ex}
				\perFreeCoRec{} & \perFreeBVrCoRec{} & \swapFreeCoRecV 
				\\  & & \\ \vspace*{-1.2ex}
				  & \swapFreeBVrCoRec &  \swapFreeCgCoRec 
					\\    	\vspace*{-1.3ex}
					  & & \\		
		\substFreshCoRec \hspace*{-2.5ex}	& \renamingCoRec  \hspace*{-2.5ex} & \renamingFreshVCoRec 
			\end{tabular}
		} 
    \vspace*{-1.5ex}
	\caption{Sets of properties 
		underlying different nominal corecursors. 
		The highlighted properties are ones that turned out to be redundant in the 
		analogous nominal recursor, but must be added back for the corecursor.}
	\label{fig-coiter}
	\vspace*{-2ex}
\end{figure}

Previous work \cite{DBLP:conf/cmcs/KurzPSV12,DBLP:journals/pacmpl/BlanchetteGPT19}  
discovered corecursive counterparts of two nominal recursors. 
Next we show that this is 
a quite pervasive phenomenon: 
 \looseness=-1

\begin{thm}\rm \label{thm-allNominalCoRecs}
	Consider the eight choices, for $i\in\{1,2,3,5,6,7,8,9\}$, of 
	tuples $\ccr_i = (\Bcat,T,\Ccat_i,\alb J_i,R_i)$  
	given by the sets of properties $\Props_i$ shown 
	in Fig.~\ref{fig-coiter}. 
	Namely (
	analogously to what we assumed in Thm.~\ref{thm-allNominalRecs}), 
	we 
	assume that $\Sigma_i$ consists of 
	the 
	operation and relation 
	symbols 
	occurring in $\Props_i$, and: 
	\looseness=-1
	\begin{mmyitem}
		\item $\Bcat$ is the category of $\Sigmad$-models and $T=\ITTrm(\Sigmad)$
		\item $\Ccat_i$ is the category of $(\Sigma_i,\Props_i)$-models and 
		$J_i$ is $\ITTrm(\Sigma_i)$
		\item $R_i : \Ccat_i \ra \Bcat_i$ is the forgetful functor sending  $(\Sigma_i,\Props_i)$-models to their underlying $\Sigmad$-models
		\looseness=-1
	\end{mmyitem}

	Then $\ccr_i$ is an epi-corecursor. 
	In particular, $\ITTrm(\Sigma_i)$ is the final $(\Sigma_i,\Props_i)$-model.  
\end{thm}

Next we unpack Thm.~\ref{thm-allNominalCoRecs}'s 
statements of epi-corecursion principles, 
exploring the connections with Thm.~\ref{thm-allNominalRecs}'s 
%
nominal epi-recursors. 
%
We used for the corecursors the same names 
as for the recursors to which they roughly correspond---although, as we 
will 
discuss, a corecursor will often ``inherit'' axioms from two different recursors. (We do not have a $\ccr_4$ corecursor  
because 
the axioms specific to $r_4$ 
were mixed into $\ccr_5$ and  $\ccr_6$; either of $\ccr_5$ and $\ccr_6$ could have alternatively been named ``$\ccr_4$''.) 
\looseness=-1

$\ccr_1$ and $\ccr_3$ are corecursors in the style of nominal logic. 
Like their recursor counterparts $r_1$ and $r_3$, they 
have the free-variable (support) operator completely determined from permutation (via \IFvDPm{}), or alternatively swapping (via \IFvDSw{}). However, these corecursors are not strictly speaking nominal-logic based, because this determination of free-variables involves not finiteness, but countability.  
%
%
Another 
difference between $\ccr_1$ / $\ccr_3$ and $r_1$ / $r_3$ is that the freshness condition for binders \FCB{} (or anything analogous to it) is no longer needed; but instead we need the (corecursive counterpart of) the bound-variable renaming axiom which was specific to the more expressive recursor $r_5$---in 
permutation or swapping form  (\IPmBvr{} or \ISwBvrT{}).   
Thus, when switching from recursion to corecursion, the nominal-logic style definitional principles trade \FCB{} for \IPmBvr{} or \ISwBvrT{}; they  
are the only ones \emph{not} to 
become axiomatically heavier during 
this switch. 
\looseness=-1

The $\ccr_{2}$ corecursor requires both the algebraic properties of permutation and freshness specific to $r_2$ (\IPmId{}, \IPmCp{} and \IPmFv{}) and the bound-variable renaming property specific to $r_5$ (
 converted from swapping 
 to permutation form, \IPmBvr{}).
 The situation is similar for $\ccr_{5}$, the swapping-based counterpart of $\ccr_{2}$, which gets axioms from both $r_4$ (with freeness 
 converted to freshness) and $r_5$. 
 All these are in sharp contrast to the recursion case, where, at the recursor $r_5$, bound-variable renaming (\SwBvr{}) was the only axiom needed (in addition to the ``unavoidable'' ones describing the  
 interaction of constructors with the other operators). 
%
%
%
Similarly to $\ccr_5$ which ``descends'' from $r_4$ and $r_5$, 
$\ccr_{6}$ ``descends'' from $r_4$ and $r_6$. 
Unlike in the recursive case where $r_4$ did not need \SwCp{} and turned out not to need \SwId{} and \SwIv{} either, here all three axioms, \ISwId{}, \ISwIv{} and \SwCp{}, are actually needed by its corecursor ``descendants''  $\ccr_{5}$ and $\ccr_{6}$. Additionally $\ccr_{6}$ requires \IFvSw{}, another axiom we had discovered to be redundant for $r_4$. 
%
%
Thus, for the principles discussed in this paragraph, the axiomatizations become heavier when switching from recursion to corecursion, 
because: (1) axioms from different recursors now need to be joined, and (2) 
previous axioms that were seen to be redundant for recursors must be added back to their corecursor counterparts.
\looseness=-1

As for the substitution- and renaming-based principles $\ccr_{7}$, $\ccr_8$ and $\ccr_9$, their axiomatizations also become heavier in a similar way,  
in that both algebraic axioms (e.g., \IRnId{}, \IRnIm{}, \IRnCh{}, \IRnCm{}) and bound-variable renaming axioms must be present. But their axiomatizations are even heavier, because they feature (1) 
two versions of the bound-variable renaming axioms (e.g., \IRnBvr{} and \IIRnBvr{} as opposed to just \IRnBvr{}) as well as (2) countable support (\IFSupFr{}). Roughly speaking, these additional axioms are needed to make corecursion go through (i.e., establish finality of the iterm model) because,  substitution/renaming not commuting unconditionally with abstractions, stronger 
bound-variable avoidance facilities must be supplied by an (arbitrary) model; 
this was not 
a problem for recursors, where fresh induction on (concrete) terms could handle that elegantly. 
%
%
%

Specific to the substitution 
corecursor $\ccr_7$ is that it features, for the variable case, not only 
the destructor freshness axiom \IFrVr{}, but also its \emph{constructor} counterpart $\FrVr{}$, 
and the implicit requirement that the signature $\Sigma_7$ contains the variable-constructor symbol $\vr$. So a $\Sigma_7$-model $\MM$ has, in addition to the destructor $\Dest^\MM\!$, a variable-constructor-like operator $\Vr^\MM : \Var \ra \MM$; the two are required to act as mutual inverses by the following axiom \VrInv{} (which, due its 
hybrid nature, fits neither Fig.~\ref{fig-basicProps} nor Fig.~\ref{fig-basicCoProps}): 
$\Dest\;t = \Vv\;x$ if and only if $t = \Vr\;x$. 
This monad-like
variable-injection setting is needed to accommodate 
the substitution of arbitrary elements $m\in M$ for variables $x$. 
\looseness=-1


\smallskip
\emph{Connection with previous corecursors. } 
Thm.~\ref{thm-allNominalCoRecs} recovers, and slightly improves on, the two existing nominal corecursors from the literature we are aware of:  
that developed by \citet{DBLP:conf/cmcs/KurzPSV12} for $\lambda$-terms and extended 
by \citet{petrisanNominalCodatatypes} to functors on nominal sets, and that developed by \citet{DBLP:journals/pacmpl/BlanchetteGPT19} in a functorial framework covering complex binders. 
Next we discuss these corecursors' instantiations to the syntax of $\lambda$-calculus.   
The Blanchette et al.\ corecursor 
corresponds to $\ccr_2$ almost exactly, with the only difference that it 
assumes \IFvPm{} which is not needed. (See \S\ref{con-blanchette}.) 
\looseness=-1

Designed 
for 
nominal logic, the Kurz et al.\ corecursor 
assumes finite support, 
and 
targets 
not the entire 
$\ITrm$ but the subset $\ITrm'\su \ITrm$ of finitely supported iterms.
Their corecursor can be obtained from our 
$\ccr_1$ by noting that, 
 if we assume the source model $\MM$ to satisfy finite support (\FSupFv{}), then 
	the image of the unique morphism $g: \MM \ra \ITTrm(\Sigma_1)$ 
	guaranteed by $\ccr_1$ is included in $\ITrm'$ (thanks to $g$'s preservation of free variables). 
So we obtain a unique morphism from $\MM$ to the submodel of $\ITTrm(\Sigma_1)$ with carrier set $\ITrm'$, i.e., the term model of Kurz et al.  
\label{alterDest}
The above summary 
ignores one technicality: 
The 
Kurz et al.\ destructor
does not have type $M \ra \Var + 
M \times M + \PPne(\Var\times M)$ like ours, but 
$M \ra \Var + 
M \times M + [\Var]M$, 
where 
$[\Var]M$ is the nominal set of abstractions, obtained by quotienting $\Var\times M$ to an $\alpha$-like equivalence relation $\sim$ defined by $(x,m) \!\sim\! (x',m')$ iff 
$m[z \llra x]^{\MM} = m[z \llra x']^{\MM}$ for some fresh $z$. 
Since $[\Var]M$ consists of $\sim$-equivalence classes, we have  
$[\Var]M \su \PPne(\Var\times M)$, so the only difference is that 
our destructor has a less constrained codomain. But our $\ccr_1$ axiom \IPmBvr{}
constrains the elements of $\PPne(\Var\times M)$ from the image of the 
destructor to contain
mutually $\sim$-equivalent items. If we also added 
\IIPmBvr{} 
to the axiomatization of $\ccr_1$, we would further constrain these to be entire $\sim$-equivalence classes, obtaining exactly the  
Kurz et al.\ models. 
Hence, due to its models being 
less constrained, 
$\ccr_1$ is (slightly) more expressive than the Kurz et al.\ corecursor. 
\looseness=-1

\smallskip
\emph{A 
	note on nominal abstractions.} 
The above recalled abstractions are a standard concept in nominal logic \cite{DBLP:conf/lics/GabbayP99}, and using abstractions as primitives is a valid alternative when introducing nominal recursors and corecursors. For the recursors, 
the $\Lm$-constructor in models would have type $[\Var]M \ra M$ rather than $\Var \ra M \ra M$. 
However, 
like the authors of the nominal recursors reviewed in \S\ref{subsec-nominalRec}, we too 
favor the abstraction-free (hence quotient-free) (co)recursors, and this is for two reasons. First, they are likely easier to deploy:  
During a recursive definition, it seems inconvenient for the user to have to provide an operator in $[\Var]M \ra M$, which usually requires making a choice and showing that the choice is immaterial; providing instead a ``free''  operator in $\Var \ra M \ra M$ and verifying an additional axiom (such as $\SwBvr$) seems more manageable. Second, 
they can be more expressive than their abstraction-based alternatives. For example, most of the recursors in Thm.~\ref{thm-allNominalRecs} do not require swapping/permutation to have the algebraic properties needed for $\sim$ to be an equivalence, so quotienting is not an option unless we 
strengthen the model axiomatization, thus 
placing a higher proof burden on the user. 
Admittedly, these advantages are 
less consequential when talking about corecursors, where the relevant algebraic properties  
are required across the board. 
\looseness=-1


\smallskip
\emph{Comparing expressiveness. }   
We use a strength relation 
that is similar to 
that from our  ``head-to-head'' comparison of epi-recursors (in \S\ref{subsec-compareHeadToHead}):
%
	Given 
	epi-corecursors $\ccr = (\Bcat,T,\Ccat,J,R)$ and $\ccr' = \alb(\Bcat,\alb T,\Ccat',J',R')$,  
	we 
	call \emph{$\ccr'$ stronger than $\ccr$}, written $\ccr' \geq \ccr$, if $\ccr'$ can define everything that  
	$\ccr$ can, 
	in that: 
	for all objects $B$ in $\Bcat$ and 
	$g: B \ra T$, 
	$g$ definable by $\ccr$ implies $g$ definable by $\ccr'$.   
	\looseness=-1
%
%
Again, we 
write $\ccr \equiv \ccr'$ to state that $\ccr$ and $\ccr'$ have equal strengths, i.e., both 
$\ccr' \geq \ccr$ and $\ccr \geq \ccr'$ hold. 
\looseness=-1
%
%
%
\leftOut{
\begin{prop}\rm
	\label{prop-extCoCriterion} 
	Let $\ccr = (\Bcat,T,\Ccat,J,R)$ and 
	$\ccr' = (\Bcat,T,\Ccat',J',\alb R')$, and 
	assume $F : \Ccat \ra \Ccat'$ is a pre-functor  
	such that 
	$R' \circ F = R$ and 
	$F\;J = J'$.  
	Then $\ccr' \geq \ccr$. 
\end{prop}
}
%

\begin{thm} \rm
	\label{thm-exprCo}
	The epi-corecursors 
	from Thm.~\ref{thm-allNominalCoRecs} (and Fig.~\ref{fig-coiter}) compare as follows w.r.t.\  expressiveness: \ 
	\\
	\hspace*{4ex} $\ccr_{2} \equiv \ccr_{5} \geq  \ccr_{6}  \geq \ccr_3 \equiv \ccr_1  \geq \ccr_8$ \  \  \  \  and \  \  \  \  
$\ccr_{5}  \geq  \ccr_9 \geq \ccr_7,\ccr_8$. 
\end{thm} 

Let us 
discuss this hierarchy in connection with 
the recursor hierarchy from Thm.~\ref{thm-expr}: 

\emph{Permutation versus swapping.} Recall that, in the recursor hierarchy, choosing between permutation and swapping 
was 
consequential to expressiveness as soon as we no longer assumed   
the tight coupling between freeness/freshness and swapping/permutation; namely, for the tight-coupling recursors $r_1$ and $r_3$ we had $r_1 \equiv r_3$, 
but for the for loose-coupling recursors $r_2$ and $r_4$
we only had $r_4 \geq r_2$.  
%
But 
on corecursors this nuance disappears: Swapping is now as expressive as permutation in both the 
tight-coupling 
($\ccr_1 \equiv \ccr_3$) 
and loose-coupling ($\ccr_{2} \equiv \ccr_{5}$) cases. This is because for swapping-based corecursors we cannot dispense with 
the algebraic axioms \ISwId{}, 
\ISwIv{} and \ISwCp{}, which 
are sufficient to ensure the extension of swapping 
to a (well-behaved) permutation operator.  
\looseness=-1

\emph{Congruence versus bound-variable renaming.} Recall that, for recursors, the congruence axiom \SwCg{} led to higher expressiveness than the bound-variable renaming axiom \SwBvr{}, yielding $r_6 \geq r_5$.
And this was because (in the presence of 
other mild axioms) \SwBvr{} implies \SwCg{}. 
The same is true here for corecursors, in that \ISwBvr{} 
implies \ISwCg{}. However, in the presence of the other $\ccr_{5}$ axioms, \ISwBvr{} is sufficient for proving a corecursion principle; whereas \ISwCg{} is not, unless we add the additional axiom \IFrSw{} (which is not needed by $\ccr_{5}$). 
And if we assume \IFrSw{} then 
\ISwCg{} also implies \ISwBvr{}. In short, congruence-based corecursion requires \IFrSw{}, and as such is less expressive than 
bound-variable renaming-based corecursion, meaning that the hierarchy gets shifted, with $\ccr_{5} \geq \ccr_{6} $.
\looseness=-1


\emph{Finiteness versus countability.}  Proving $r_2 \geq r_1$ relied on the fact that, in a pre-nominal set (i.e., a model satisfying \PmId{}, \PmCp{}) equipped with equivariant constructors (satisfying \PmVr{}, \PmAp{}, \PmLm{}) and assuming \FCB{}, if we define freshness from permutation using finiteness (via \FvDPm{}), then this freshness operator behaves well w.r.t.\ the constructors (satisfies \FrVr{}, \FrAp{}, \FrLm{}).  
This also works if we replace ``finite'' with ``countable'' and the constructors with the destructor, and use \IPmBvr{} instead of \FCB{}, which shows why we also have $\ccr_2\geq \ccr_1$ 
(and similarly for the swapping-based versions). 
\looseness=-1

\emph{Symmetric versus asymmetric operators, second round.} 
Recall that, in the strict ``head-to-head'' comparison relation, recursors based on 
symmetric operators (swapping and permutation) 
were incomparable to those based on 
asymmetric ones (renaming and  substitution), but only a laxer comparison 
deemed the symmetric ones more expressive. 
But in the case of corecursors, the symmetric ones emerge as more expressive already in a head-to-head comparison. This is not too surprising if we 
recall the reason why symmetric-operator recursors  eventually 
emerged as more expressive: because, if the model has finite support (which in the laxer criterion was possible by taking the minimal submodel), then swapping becomes definable from renaming. 
Here, our asymmetric-operator based models already have countable support (which, as discussed, seems necessary for corecursion), hence can also define swapping from renaming similarly to how this is done in the finite-support case.  
\looseness=-1


\emph{Laxer comparison relation?} It is worth asking whether (1) an analogue of the laxer comparison relation we introduced for epi-recursors is available for epi-corecursors, and whether (2) it would yield any 
flattening of the corecursor hiererchy 
(analogous to Thm~\ref{thm-qexpr}). While the answer to the first question is clearly 'yes' because a perfectly dual concept applies to epi-corecursors, 
to the second question we are inclined to answer 
`no': 
Now we would not be able to use submodels, but something akin to quotient models, and 
quotienting tends to not preserve (let alone strengthen) 
our axiomatizations. 
\looseness=-1

\medskip
\noindent 
\textbf{Summary.} 
%
Nominal corecursors can be construed and compared as epi-corecursors, following a similar methodology to that 
for nominal recursors. A corecursor axiomatization 
corresponds to one or two recursor axiomatizations via identical and quasi-dual axioms. 
%
The corecursor axiomatizations are 
heavier.  
%
%
%
We have a corecursor hierarchy that partly matches the strict-relation ($\geq$) recursor hierarchy 
but is more fine-grained, in particular it already subsumes asymmetric-operator principles to the symmetric-operator ones without the need for a laxer comparison relation (in the style of $\wgeq$). 
\looseness=-1

\section{Mechanized Results}
\label{sec-mechResults}

We have mechanized in Isabelle/HOL 
	the recursion theorem (Thm.~\ref{thm-allNominalRecs}), 
	the two recursor comparison theorems (Thms.~\ref{thm-expr} and \ref{thm-qexpr}), 
	 the two negative (strictness) results on recursor comparison (Prop.~\ref{prop-negRec}), 
	 the corecursion theorem (Thm.~\ref{thm-allNominalCoRecs}), and 
	the corecursor comparison theorem (Thm.~\ref{thm-exprCo}). 
%
What we have \emph{not} mechanized are the abstract criteria for comparing epi-recursors, namely  
Props.~\ref{prop-extCriterion} and \ref{prop-WeakExtCriterion}.
 In our mechanized results, rather than invoking these criteria, we have inlined their content on a need basis. 
 \looseness=-1

The 
mechanization is available as an archive \cite{isa-nomialEpiRec},
and is extensively documented  
in App.~\ref{app-isa}.   
It uses Isabelle's structuring mechanisms called \emph{locales} \cite{DBLP:conf/tphol/KammullerWP99,DBLP:journals/jar/Ballarin14} 
to represent the 
model axiomatizations, and uses sublocale relationships for the transformations between these axiomatizations that underlie the 
expressiveness comparisons. 
%
\looseness=-1

We have also provided a top-level, locale-free reformulation of the 
mechanized results, which 
match closely the statements from the paper, and whose inspection  does not require 
knowledge of locales. 
%
The end results about recursors, 
Thms.~\ref{thm-allNominalRecs}, \ref{thm-expr} and \ref{thm-qexpr} 
and Prop.~\ref{prop-negRec}, are mechanized in homonymous Isabelle theories,  
located in the archive's directory 
\textsf{Stripped$\_$Down/LocaleFree$\_$versions}: 
\begin{itemize} 
	\item 
Thm.~\ref{thm-allNominalRecs} is mechanized in the Isabelle theory 
\textsf{Theorem9}. 
%
That theory contains the definitions of the epi-recursor structure for each 
of the nine recursors $r_i$, and proofs that these structures indeed form epi-recursors, 
e.g., their components are categories, functors etc. The initiality theorems are named \textsf{init$\_$I$i$} where 
$i$ is a number between $1$ and $9$. 
\item Thm.~\ref{thm-expr} is mechanized in the Isabelle theory 
\textsf{Theorem12}, where the main formal theorems are named r$i\_$ge$\_$r$j$ (formalizing $r_i \geq r_j$) for the relevant choices of 
$i$ and $j$. 
\item Thm.~\ref{thm-qexpr}  is mechanized in the Isabelle theory 
\textsf{Theorem15}, where the main formal theorems are named r$i\_$quasi$\_$ge$\_$r$j$ (formalizing $r_i \wgeq r_j$), again for the relevant choices of 
$i$ and $j$. 
\item Prop.~\ref{prop-negRec}  is mechanized in the Isabelle theory 
\textsf{Prop16}, where the main formal theorems are named \textsf{not$\_$r1$\_$ge$\_$r2}  and \textsf{not$\_$r2$\_$ge$\_$r4} 
(formalizing $r_1 \not\geq r_2$ and $r_2 \not\geq r_4$). 
\end{itemize}

And similarly for corecursors, 
in directory 
\textsf{Corecursors/LocaleFree$\_$versions}: 
\begin{itemize} 
	\item 
	Thm.~\ref{thm-allNominalCoRecs} is mechanized in the Isabelle theory 
	\textsf{Theorem18}, 
	%
which contains the definitions and proofs for the epi-corecursor structure, 
	including the finality theorems named \textsf{final$\_$J$i$}. 
	\looseness=-1
	\item Thm.~\ref{thm-exprCo} is mechanized in the Isabelle theory 
	\textsf{Theorem19}, where the main formal theorems are named cr$i\_$ge$\_$cr$j$ (formalizing $\ccr_i \geq \ccr_j$) for the relevant choices of 
	$i$ and $j$. 
\end{itemize} 

App.~\ref{app-subsec-localeFree} 
gives more details about the locale-free statements of the results. 

\section{
	More Related Work}
\label{sec-relWork} 


\indent 
\hspace*{2ex}
\textit{Definitional packages for syntax with bindings.} 
A direct application of our results would be on informing the 
design of binding-aware definitional packages in proof assistants, in the style of  Nominal Isabelle \cite{nominalTwo}. 
In addition to our theoretical results on expressiveness,  
one should also consider the pragmatic aspects of how lightweight the required structure (operations and relations on the target domain) is and how easy the conditions are to solve. Ideally, in a definitional package implementation one should provide 
the maximally expressive (co)recursor as the core, but also infer from it (via "borrowing") and make available other (co)recursors which may have pragmatic advantages. For example, the recursors $r_3$ and $r_8$ 
are minimalistic in terms of structure. 
\looseness=-1


\smallskip
\textit{(Co)recursors in different paradigms.} 
Binding-aware recursors have also been developed in the other two major 
paradigms. 
%
Scope-safe versions of nameless recursion based on category theory have been studied extensively, e.g.,   \citet{fio-abs,DBLP:conf/lics/Hofmann99,BirdP-lambda,DBLP:conf/csl/AltenkirchR99,allais-bindingsByDependentTypes-agda,DBLP:conf/cpp/KaiserSS18}. 
A nameless recursor is in principle 
easier to 
deploy because the 
constructors are free; the price 
is additional index-shifting overhead \cite{DBLP:journals/entcs/BerghoferU07}.   
Nameless corecursion has been studied by  \citet{DBLP:journals/tcs/MatthesU04}, building on previous work by \citet{DBLP:journals/tcs/AczelAMV03,DBLP:journals/tcs/Moss01,DBLP:journals/mscs/GhaniLMP03}. 
\looseness=-1

Hybrid nameless/nominal solutions have also been proposed, notably 
the locally named \cite{DBLP:journals/jar/McKinnaP99,DBLP:journals/jar/PollackSR12}
and locally nameless \cite{aydemirPOPL08,locallyNameless} representations.  
\citet{pittsLocNamSets} 
introduced locally nameless sets, an algebraic axiomatization of syntax under the locally nameless representation, and characterizes the locally nameless recursor \cite{locallyNameless} using initiality in a functor category (similarly to recursors in the nameless 
setting \cite{fio-abs,DBLP:conf/lics/Hofmann99}).  
He also 
proved that the category of locally nameless sets is isomorphic to that of finitely supported rensets \cite{DBLP:conf/cade/Popescu22} and to categories given by other axiomatizations of renaming from the literature \cite{statonThesis,gabbayHofmann-nominalRenamingSets}; this suggests that the expressive power of the locally nameless recursor might be located in the vicinity of $r_8$ (which is based on rensets).  
On the way to his results, Pitts gave an alternative axiomatization of finitely supported rensets, using instead of \RnCh{} a simpler (unconditional) axiom, let us call it \RnCh{}': $t[x_2/x_1][x_3/x_2] = t[x_3/x_2][x_3/x_1]$. Replacing \RnCh{} with \RnCh{}' would yield a recursor $r_8'$ such that $r_8  \geq  r_8'$ (since \RnCh{}' implies \RnCh{} in the presence of \RnIm{}) and $r_8 \simeq r_8'$ (since 
the converse implication 
is true for finitely supported rensets, hence for a suitable minimal submodel).  
\looseness=-1

In \emph{strong HOAS}, 
as implemented in dedicated logical frameworks 
\cite{DBLP:conf/cade/PfenningS99,abellaJournalPaper,beluga}, 
the $\lambda$-constructor has type $(\Trm \ra \Trm) \ra \Trm$. Here, the difficulty with recursion is not the non-freeness of the constructors, but the fact that 
binding constructors are not recursable 
in the typical well-foundedness manner. Solutions to this have been designed using modality operators \cite{DBLP:journals/tcs/SchurmannDP01} and contextual types 
\cite{DBLP:conf/esop/0001P17}.  
Recursion mechanisms have also been designed within \emph{weak HOAS} 
\cite{weakHOAS}, 
where the $\lambda$-constructor, having  
type $(\Var \ra \Trm) \ra \Trm$, \emph{is} standardly recursable---yielding a free datatype that contains 
all terms but also additional entities referred to as ``exotic terms''.  Partly 
due to the exotic terms, this free datatype 
is not 
very 
helpful for recursively defining useful functions on terms. 
But the situation is significantly improved in a variant called 
\emph{parametric HOAS (PHOAS)} \cite{chlipalaParamHOAS2008}, 
which accommodates recursive definitions 
%
in the style of the semantic-interpretation pattern (\S\ref{subsec-semInt}).  
\looseness=-1

A  nominal/HOAS hybrid can be found in Gordon and Melham's characterization of the $\lambda$-term datatype \cite{DBLP:conf/tphol/GordonM96}, which employs the nameful constructors but 
features weak-HOAS style recursion over $\Lm$.  
\citet{primrecFOAS-Norrish04} inferred his swap/free 
recursor $r_4$ from the Gordon-Melham one. 
Weak-HOAS recursion also 
has interesting connections with nameless recursion: In presheaf toposes as in 
\citet{fio-abs}, 
\citet{DBLP:conf/lics/Hofmann99} and 
\citet{DBLP:conf/icfp/AmblerCM03},  
the function space $\Var \Ra T$ is isomorphic 
to the De Bruijn level-shifting transformation applied to $T$; this effectively equates 
the weak-HOAS and nameless recursors.  
\looseness=-1


\smallskip
\textit{Recursion over non-free datatypes.} 
Some of the discussed 
nominal recursors operate by characterizing terms as the non-free datatype determined as 
initial model of an equational theory \cite{BurrisSankappanavar1981} 
or more generally of a Horn theory \cite{DBLP:journals/jcss/Makowsky87}, employing an infinite number of axioms. In such cases, and ignoring the Barendregt enhancement, nominal recursion 
becomes a particular case of 
Horn recursion. (This 
is not true for the nominal-logic recursor $r_1$, since 
$\FvDPm$ 
is not a Horn formula.)   
Our concept of epi-recursor 
applies to general Horn recursion as well---provided one identifies a constructor-like  subsignature of the given signature, i.e., 
such that the 
initial model of the Horn theory has its carrier generated by its operations. In 
algebraic specifications, this property is called \emph{sufficient completeness} \cite{DBLP:journals/acta/GuttagH78}. 
\looseness=-1

The non-free datatypes of sets and bags are degenerate 
cases of the above, where the constructors 
form the entire signature. 
\citet{DBLP:conf/icalp/TannenS91}  
and 
\citet{DBLP:journals/tcs/BunemanNTW95} 
study 
Horn recursors for these datatypes  
when designing database languages. 
They prove connections 
between their 
axiomatizations that could be captured 
using our $\geq$ relation between epi-recursors. 
\looseness=-1


\begin{acks}
	We thank the reviewers and the artifact reviewers for the careful reading of our paper, and for their insightful comments and suggestions, which have led to improvements both in the text and in the documentation of what has been mechanized.  
	We gratefully acknowledge support from the EPSRC grant EP/X015114/1 ``Safe and secure COncurrent programming for adVancEd aRchiTectures (COVERT)''. 
\end{acks}

%
%


\input{TR_Nominal_Recursors_as_Epirecursors.bbl}

\include{appendix}

\end{document}

%% file: myCommands.tex
\newenvironment{myitem}[1][]
{\itemize[leftmargin=5.0ex,topsep=0.3ex,itemsep=1pt, #1]}
{\enditemize}
\newenvironment{myyitem}[1][]
{\itemize[leftmargin=4.0ex,topsep=0.3ex,itemsep=1pt, #1]}
{\enditemize}
\newenvironment{myyyitem}[1][]
{\itemize[leftmargin=3.7ex,topsep=0.3ex,itemsep=1pt, #1]}
{\enditemize}
\newenvironment{mmyitem}[1][]
{\itemize[leftmargin=2.15ex,topsep=0.3ex,itemsep=1pt, #1]}
{\enditemize}
\newenvironment{mmmyitem}[1][]
{\itemize[leftmargin=3.2ex,topsep=0.3ex,itemsep=1pt, #1]}
{\enditemize}

\newcommand\mystyle[1]{{\small \textsf{#1}}}

\newcommand\FrVr{\mystyle{FrVr}}

\newcommand\FrAp{\mystyle{FrAp}}
\newcommand\FrLm{\mystyle{FrLm}}
\newcommand\FvVr{\mystyle{FvVr}}
\newcommand\FvAp{\mystyle{FvAp}}
\newcommand\FvLm{\mystyle{FvLm}}

\newcommand\co[1]{#1$_\infty$}

\newcommand\IFrVr{\co{\FrVr}}

\newcommand\VrInv{\mystyle{VrInv}}
\newcommand\IFrAp{\co{\FrAp}}
\newcommand\IFrLm{\co{\FrLm}}
\newcommand\IFvVr{\co{\FvVr}}
\newcommand\IFvAp{\co{\FvAp}}
\newcommand\IFvLm{\co{\FvLm}}

\newcommand\SwVr{\mystyle{SwVr}}
\newcommand\SwAp{\mystyle{SwAp}}
\newcommand\SwLm{\mystyle{SwLm}}
\newcommand\ISwVr{\co{\SwVr}}
\newcommand\ISwAp{\co{\SwAp}}
\newcommand\ISwLm{\co{\SwLm}}

\newcommand\SwId{\mystyle{SwId}}
\newcommand\SwCp{\mystyle{SwCp}}
\newcommand\SwIv{\mystyle{SwIv}}
\newcommand\SwFr{\mystyle{SwFr}}
\newcommand\FrSw{\mystyle{FrSw}}
\newcommand\SwFv{\mystyle{SwFv}}
\newcommand\FvSw{\mystyle{FvSw}}

\newcommand\ISwId\SwId 
\newcommand\ISwCp\SwCp 
\newcommand\ISwIv\SwIv 
\newcommand\ISwFr\SwFr 
\newcommand\IFrSw\FrSw 
\newcommand\ISwFv\SwFv 
\newcommand\IFvSw\FvSw 

\newcommand\RnFr{\mystyle{RnFr}}
\newcommand\FrRn{\mystyle{FrRn}}
\newcommand\IRnFr\RnFr 
\newcommand\IFrRn\FrRn 

\newcommand\FrRnT{\mystyle{FrRn$_2$}}
\newcommand\IFrRnT\FrRnT

\newcommand\SbFr{\mystyle{SbFr}}
\newcommand\FrSb{\mystyle{FrSb}}
\newcommand\ISbFr\SbFr 
\newcommand\IFrSb\FrSb 

\newcommand\SwCg{\mystyle{SwCg}}
\newcommand\ISwCg{\co{\SwCg}}

\newcommand\SwBvr{\mystyle{SwBvr}}
\newcommand\ISwBvr{\co{\SwBvr}}
\newcommand\ISwBvrT{\mystyle{SwBvr}$_{\infty,2}$}

\newcommand\PmVr{\mystyle{PmVr}}
\newcommand\PmAp{\mystyle{PmAp}}
\newcommand\PmLm{\mystyle{PmLm}}
\newcommand\IPmVr{\co{\PmVr}}
\newcommand\IPmAp{\co{\PmAp}}
\newcommand\IPmLm{\co{\PmLm}}

\newcommand\PmId{\mystyle{PmId}}
\newcommand\PmCp{\mystyle{PmCp}}
\newcommand\PmFv{\mystyle{PmFv}}
\newcommand\FvPm{\mystyle{FvPm}}

\newcommand\IPmId\PmId 
\newcommand\IPmCp\PmCp 
\newcommand\IPmFv\PmFv 
\newcommand\IFvPm\FvPm 

\newcommand\PmBvr{\mystyle{PmBvr}}
\newcommand\IPmBvr{\co{\PmBvr}}
\newcommand\IIPmBvr{\PmBvr$_\infty'$}

\newcommand\SbVr{\mystyle{SbVr}}
\newcommand\SbAp{\mystyle{SbAp}}
\newcommand\SbLm{\mystyle{SbLm}}
\newcommand\ISbVr{\co{\SbVr}}
\newcommand\ISbAp{\co{\SbAp}}
\newcommand\ISbLm{\co{\SbLm}}

\newcommand\SbCn{\mystyle{SbCn}}
\newcommand\ISbCn\SbCn 

\newcommand\SbCg{\mystyle{SbCg}}

\newcommand\SbBvr{\mystyle{SbBvr}}
\newcommand\ISbBvr{\co{\SbBvr}}
\newcommand\IISbBvr{\SbBvr$_\infty'$}

\newcommand\RnBvr{\mystyle{RnBvr}}
\newcommand\IRnBvr{\co{\RnBvr}}

\newcommand\RnBvrT{\RnBvr$_2$}

\newcommand\IIRnBvr{\RnBvr$_\infty'$}

\newcommand\RnVr{\mystyle{RnVr}}
\newcommand\RnAp{\mystyle{RnAp}}
\newcommand\RnLmO{\mystyle{RnLm$_1$}}
\newcommand\RnLmT{\mystyle{RnLm$_2$}}

\newcommand\IRnVr{\co{\RnVr}}
\newcommand\IRnAp{\co{\RnAp}}
\newcommand\IRnLmO{\mystyle{RnLm}$_{1,\infty}$}
\newcommand\IRnLmT{\mystyle{RnLm}$_{2,\infty}$}

\newcommand\RnCg{\mystyle{RnCg}}
\newcommand\IRnCg{\co{\RnCg}}

\newcommand\RnId{\mystyle{RnId}}
\newcommand\RnIm{\mystyle{RnIm}}
\newcommand\RnCh{\mystyle{RnCh}}
\newcommand\RnCm{\mystyle{RnCm}}
\newcommand\RnChFr{\mystyle{RnChFr}}
\newcommand\SbChFr{\mystyle{SbChFr}}

\newcommand\IRnId\RnId 
\newcommand\IRnIm\RnIm 
\newcommand\IRnCh\RnCh 
\newcommand\IRnCm\RnCm 
\newcommand\IRnChFr\RnChFr 
\newcommand\ISbChFr\SbChFr 

\newcommand\SbId{\mystyle{SbId}}
\newcommand\SbIm{\mystyle{SbIm}}
\newcommand\SbCh{\mystyle{SbCh}}
\newcommand\SbCm{\mystyle{SbCm}}

\newcommand\ISbId\SbId 
\newcommand\ISbIm\SbIm 
\newcommand\ISbCh\SbCh 
\newcommand\ISbCm\SbCm 

\newcommand\FSupFv{\mystyle{FSupFv}}
\newcommand\FSupFr{\mystyle{FSupFr}}
\newcommand\FvDPm{\mystyle{FvDPm}}
\newcommand\FvDSw{\mystyle{FvDSw}}
\newcommand\FrDSw{\mystyle{FrDSw}}
\newcommand\FrDRn{\mystyle{FrDRn}}
\newcommand\FCB{\mystyle{FCB}}

\newcommand\IFvDPm{\co{\FvDPm}}
\newcommand\IFvDSw{\co{\FvDSw}}
\newcommand\IFrDSw{\co{\FrDSw}}
\newcommand\IFrDRn{\co{\FrDRn}}
\newcommand\IFSupFv{\co{\FSupFv}}
\newcommand\IFSupFr{\co{\FSupFr}}

\newtheorem{mylemma}{Lemma}
\newtheorem{prop}[mylemma]{Prop}
\newtheorem{thm}[mylemma]{Thm}
\newtheorem{defi}[mylemma]{Def}

\newtheorem{exa}[mylemma]{Example}
\newtheorem{convention}[mylemma]{Convention}
\newenvironment{proof}{{\noindent\it Proof sketch:\ }}
\newcommand{\leftOut}[1]{}

\newbox\boxA

\newcommand{\oexp}{\mbox{\hphantom{$+_{\mathsf{o}}$}}\llap{$\text{\textasciicircum}_{\mathsf{o}}\kern.1em$}}

\newcommand\hlt[1]{\mbox{\colorbox{light-gray}{$#1$}}}

\newcommand{\cexp}{\mbox{\hphantom{$+_{\mathsf{o}}$}}\llap{$\text{\textasciicircum}_{\mathsf{c}}\kern.1em$}}

\newcommand{\imageOp}{\raise.2ex\hbox{\mbox{$\scriptscriptstyle\bullet$}}}
\newcommand{\vimageOp}{\raise.2ex\hbox{\mbox{$\scriptscriptstyle-\hspace*{-0.7ex}-\kern-.2em\bullet$}}}

\newcommand\TC{\sf}

\newcommand\CHOPFROMUN{.25}
\newcommand\UN{{\setbox\boxA=\hbox{\_}\usebox\boxA\kern-\CHOPFROMUN\wd\boxA{\color{white}\vrule height 0ex depth .444ex width \CHOPFROMUN\wd\boxA}\kern-\CHOPFROMUN\wd\boxA}}

\newcommand{\sm}{\smallsetminus}

\renewcommand{\phi}{\varphi}

\newcommand\wgeq\gtrsim

\newcommand{\su}{\subseteq}
\renewcommand{\iff}{\allowbreak\mathrel{\;\Leftarrow\nobreak\kern-1.6ex\Rightarrow\;}\allowbreak} 

\newcommand{\restr}{\upharpoonright}

\newcommand{\ra}{\rightarrow}
\def\implies{\rightarrow} 

\newcommand{\Ra}{\Rightarrow}
\newcommand{\lra}{\rightarrow}
\newcommand{\llra}{\leftrightarrow}

\newcommand{\LRA}{\Longrightarrow}

\newcommand{\llam}{{\llam}}

\newcommand\alb\allowbreak
\newcommand{\<}{\langle}
\renewcommand{\>}{\rangle}

\newcommand\Obj[1]{\mathsf{Obj}(#1)}

\newcommand{\Props}{\mathsf{Props}}

\newcommand{\Bcat}{\underline{\mathcal{B}}}
\newcommand{\Ccat}{\underline{\mathcal{C}}}

\newcommand{\A}{\mathcal{A}}

\renewcommand{\SS}{\mathcal{I}}
\renewcommand{\AA}{\mathcal{A}}
\newcommand{\MM}{\mathcal{M}}

\newcommand{\TTrm}{\mathcal{T}\hspace*{-0.6ex}r}
\newcommand{\ITTrm}{\mathcal{T}\hspace*{-0.6ex}r_\infty}

\newcommand{\BB}{\mathcal{B}}

\newcommand{\FF}{\mathcal{F}}

\newcommand\Sigmac{\Sigma_{\mathsf{ctor}}}
\newcommand\Sigmad{\Sigma_{\mathsf{dtor}}}
\newcommand\Sigmae{\Sigma_{\mathsf{ext}}}
\newcommand\Sym{{\mathsf{Sym}}}
\newcommand\img{{\mathsf{image}}}
\newcommand\vr{{\mathsf{vr}}}
\newcommand\ap{{\mathsf{ap}}}
\newcommand\lm{{\mathsf{lm}}}
\newcommand\swp{{\mathsf{sw}}}
\newcommand\sbs{{\mathsf{sb}}}
\newcommand\ren{{\mathsf{ren}}}
\renewcommand\pm{{\mathsf{pm}}}
\newcommand\fv{{\mathsf{fv}}}
\newcommand\fr{{\mathsf{fr}}}
\newcommand\FVars{{\mathsf{FVars}}}

\newcommand\sem{{\mathsf{sem}}}

\newcommand{\im}{!}
\newcommand{\no}{\mathsf{{noccs}}}

\newcommand{\Vr}{\mathsf{{Vr}}}

\newcommand{\LM}{\mathsf{{LM}}}
\newcommand{\AP}{\mathsf{{AP}}}
\newcommand{\Lm}{\mathsf{{Lm}}}

\newcommand{\Ap}{\mathsf{{Ap}}}

\newcommand{\PVr}{\mathsf{{PVr}}}
\newcommand{\PLm}{\mathsf{{PLm}}}
\newcommand{\PAp}{\mathsf{{PAp}}}

\newcommand{\fresh}{\#}
\newcommand{\ifresh}{\$}

\newcommand{\getApR}{{{\mathsf{getApR}}}}
\newcommand{\getApL}{{{\mathsf{getApL}}}}
\newcommand{\enf}{{{\mathsf{enf}}}}
\newcommand{\enc}{{{\mathsf{enc}}}}
\newcommand{\ddepth}{{{\mathsf{size}}}}

  \newcommand\rep{\TC{rep}}
\newcommand\isVr{\TC{isVr}}
\newcommand\getVr{\TC{getVr}}
\newcommand\isAp{\TC{isAp}}
\newcommand\getAp{\TC{getAp}}
\newcommand\isLm{\TC{isLm}}
\newcommand\getLm{\TC{getLm}}

\newcommand{\subst}{{{\mathsf{subst}}}}
\newcommand{\psubst}{{{\mathsf{psubst}}}}

\newcommand{\Env}{{{\mathsf{Env}}}}

\newcommand{\In}{{{\mathsf{In}}}}
\newcommand{\Vv}{{{\mathsf{V}}}}
\newcommand{\Aa}{{{\mathsf{A}}}}
\newcommand{\Ll}{{{\mathsf{L}}}}

\newcommand{\VV}{{{\mathsf{V}}}}

\newcommand{\PPne}{{{\mathcal{P}_{\!\not=\emptyset}}}}

\newcommand{\pickFresh}{{{\mathsf{pickFresh}}}}

\newcommand{\FV}{{{\mathsf{FV}}}}
\newcommand{\Dest}{{{\mathsf{Dest}}}}
\newcommand{\PDest}{{{\mathsf{PDest}}}}
\newcommand{\dest}{{{\mathsf{dest}}}}

\newcommand{\Pow}{{{\mathcal{P}}}}

\newcommand{\supp}{{{\mathsf{supp}}}}

\newcommand{\ext}{{{\mathsf{ext}}}}

\newcommand{\Ct}{{{\mathsf{Ct}}}}

\newcommand{\Ttrue}{\mathsf{{True}}}

\newcommand{\Var}{\mathsf{{Var}}}

\newcommand{\Trm}{\mathsf{Tr}}
\newcommand{\ETrm}{\mathsf{ETr}}
\newcommand{\ITrm}{\mathsf{Tr}_\infty}

\newcommand{\PTrm}{\mathsf{PTr}}
\newcommand{\PITrm}{\mathsf{PTr}_\infty}

\newcommand{\Perm}{\mathsf{{Perm}}}

\newcommand{\id}{\mathsf{{id}}}

\newcommand{\sw}{\hspace*{-0.25ex}\wedge\hspace*{-0.20ex}}

\newcommand{\F}{{\TC{F}}}
\newcommand{\K}{{\TC{K}}}

\newcommand{\map}{\mathsf{{map}}}
\newcommand{\Vars}{\mathsf{{Vars}}}
\newcommand{\rem}{\mathsf{{rem}}}
\newcommand{\Fmap}{\mathsf{{Fmap}}}

\newcommand{\ccr}{{\mathit{cr}}}

\newcommand{\xs}{{\mathit{xs}}}

\newcommand{\Bool}{{\TC Bool}}

\newcommand{\Nat}{\mathbb{N}}

\newmuskip\originalthinmuskip
\originalthinmuskip=\thinmuskip
\newmuskip\tinyGapMu
\renewcommand\ldots{\mathinner{.\mskip\originalthinmuskip .\mskip\originalthinmuskip .}}

\newcommand\Smash[1]{\kern-200mm\smash{#1}\kern-200mm}
\newcommand\SubItem[2]{\indent\hbox to \leftmargini{\hfill#1\enskip}#2}

\newcommand\XDot{\raise1ex\hbox{\Large.\kern.1em}}



%% file: figBasicProperties.tex
\thispagestyle{empty}
\begin{figure*}
\hspace*{-1.9ex}
{\small
\begin{tabular}{c} 
\begin{tabular}{|c|l|}	
	\hline
\textrm{\SwVr} 
&
$(\Vr\;x)[z_1 \sw z_2] = \Vr\;(x[z_1 \sw z_2])$
\\\hline	
\textrm{\SwAp} 
&
$(\Ap\;s\;t)[z_1 \sw z_2] = \Ap\,(s[z_1 \sw z_2])\,(t[z_1 \sw z_2])$
\\\hline 	
\textrm{\SwLm}
&
$(\Lm\;x\;t)[z_1 \sw z_2] \,= $
\\&
$ \Lm\;(x[z_1 \sw z_2])\;(t[z_1 \sw z_2])$
\\\hline
	\textrm{\SwId}
&
$t[z\sw z] = t$
\\\hline
\textrm{\SwCp} 
&
$t[x\sw y][z_1\sw z_2] \,=$ 
\\
&
$(t[z_1\sw z_2]) [(x[z_1\sw z_2])\,\sw\, (y[z_1\sw z_2])]$
\\\hline
\textrm{\SwIv} 
&
$t[x\sw y][x\sw y] = t$
\\\hline
	\textrm{\SwFr} 
&
if $x \;\fresh\;t$ and $y \;\fresh\;t$ 
then $t[x\sw y] = t$
\\\hline 
\textrm{\FrSw} 
&
$z \;\fresh\;t [x\sw y]$ if and only if    
$z[x\sw y] \;\fresh\;t$
\\\hline 
	\textrm{\SwFv} 
&
if $x,y \notin \FV\;t$ then $t[x\sw y] = t$
\\\hline 
\textrm{\FvSw} 
&
$z \in \FV(t [x\sw y])$ if and only if   
\\&
$z[x\sw y] \in \FV\;t$
\\\hline 
\textrm{\SwCg} 
&
if $z \not\in \{x_1,x_2\}$ and 
$z\;\fresh\;t_1,t_2$ 
\\
&and 
$t_1[z\sw x_1] = t_2[z\sw x_2]$ 
\\
&then $\Lm\;x_1\;t_1 = \Lm\;x_2\;t_2$
\\\hline
\textrm{\SwBvr} 
&
if $x' \not= x$ and 
$x'\;\fresh\;t$ 
\\&then 
$\Lm\;x\;t = \Lm\;x'\;(t[x' \sw x])$
\\\hline
\end{tabular} 	
\\ \\  \vspace*{-1.3ex}
\begin{tabular}{|c|l|}	
	\hline
	\textrm{\RnVr}
	& 
	$(\Vr\;x)[y / z] = \Vr\,(x[y/z])$
	\\\hline 
	\textrm{\RnAp}  & 
	$(\Ap\;t_1\;t_2)[y / z] =  \Ap\,(t_1[y / z])\,(t_2[y / z])$
	\\\hline 
	\textrm{\RnLmO}  & if $x \not\in \{y,z\}$ then  
	\\&
	$(\Lm\;x\;t)[y / z] = \Lm\;x\;(t[y / z])$
	\\\hline 
	\textrm{\RnLmT}  & $(\Lm\;x\;t)[z / x] = \Lm\;x\;t$
		\\\hline 
	\textrm{\RnCg}  &
	if $z \not\in \{x_1,x_2\}$ and 
	$z\;\fresh\;t_1,t_2$ 
	\\&and 
	$t_1[z / x_1] = t_2[z / x_2]$ 
	\\&then 
	$\Lm\;x_1\;t_1 = \Lm\;x_2\;t_2$
	\\\hline 
	\textrm{\RnBvr}  &	
	if $x' \not= x$ and 
	$x'\;\fresh\;t$ 
	\\& then 
	$\Lm\;x\;t = \Lm\;x'\;(t[x' / x])$
	\\\hline 
	\textrm{\RnBvrT}  &	
	if $y \not= x'$ then 
	\\&  
	$\Lm\;x\;(t[y/x']) = \Lm\;x'\;(t[y/x'][x' / x])$
	\\\hline 
	\textrm{\RnId}
	&
	$t[z / z] = t$
	\\\hline
	\textrm{\RnIm} 
	& if $x_1\not=y$ then 
	$t[x_1/y][x_2/y]= t[x_1/y]$
	\\\hline 
	\textrm{\RnCh}  & if $y\not=x_2$ then 
	\\&$t[y/x_2][x_2/x_1][x_3/x_2]$ $=$ 
	\\& $t[y/x_2] [x_3/x_1]$
	\\\hline 
	\textrm{\RnCm}  & if $x_2 \not= y_1 \not= x_1 \not= y_2$ then 
	\\&$t[x_2/x_1] [y_2/y_1]= t[y_2/y_1][x_2/x_1]$
\\\hline 
\textrm{\RnFr} 
& if $y \,\fresh\, t$ then  $t[x/y]= t$
\\\hline 
\textrm{\FrRn} 
&
$z \;\fresh\;t [x / y]$ if and only if    
	\\&($z = y$ or $z \;\fresh\;t$) and ($y \;\fresh\;t$ or $x \not= z$) 
\\\hline 
\textrm{\FrRnT} 
&
$z[x/y] \;\fresh\;t[x/y]$ implies $z \;\fresh\;t$  
\\\hline 
\textrm{\RnChFr}  & if $x_2 \,\fresh\, t$ then 
\\&$t[x_2/x_1][x_3/x_2] = t[x_3/x_1]$
\\\hline 
\end{tabular} 
\\ \\ \vspace*{-1.3ex}
	\begin{tabular}{|c|l|}	
	\hline 
	\textrm{\FrVr} & 
	if $z \not=x$ then $z\;\fresh\;\Vr\;x$
	\\\hline
	\textrm{\FrAp} &
	if $z\;\fresh\;s$ and $z\;\fresh\;t$ then $z\;\fresh\;\Ap\;s\;t$
	\\\hline
	\textrm{\FrLm}  &
	if $z = x$ or $z\;\fresh\;t$ then $ z\;\fresh\;\Lm\;x\;t$
	\\\hline
	\textrm{\FvVr} & 
	$\FV(\Vr\;x) \su \{x\}$
	\\\hline
	\textrm{\FvAp} &
	$\FV(\Ap\;t_1\;t_2) \su \FV\,t_1 \cup \FV\,t_2$  
	\\\hline
	\textrm{\FvLm}  &
	$\FV(\Lm\;x\;t) \su \FV\,t \sm \{x\}$ 
	\\\hline
\end{tabular}
\end{tabular} 
\hspace*{-4.4ex}
\begin{tabular}{c}
	\begin{tabular}{|c|l|}	
		\hline
		\textrm{\PmVr} 
		&
		$(\Vr\;x)[\sigma] = \Vr\;(\sigma\;x)$ 
		\\\hline
		\textrm{\PmAp} 
		&
		$(\Ap\;s\;t)[\sigma] =    
		\Ap\,(s[\sigma])\,(t[\sigma])$
		\\\hline
		\textrm{\PmLm} 
		&
		$(\Lm\;x\;t)[\sigma] = 
		\Lm\;(\sigma\;x)\;(t[\sigma])$
		\\\hline
		\textrm{\PmId}  
		&
		$t[\id] = t$
		\\\hline	
		\textrm{\PmCp}  & $t[\sigma][\tau] = t[\tau \circ \sigma]$
		\\\hline
		\textrm{\PmFv} & 
		if $\supp\;\sigma \cap \FV\,t = \emptyset$ then $t[\sigma] = t$
		\\\hline
		\textrm{\FvPm}  & 
		$z \notin \FV(t[\sigma])$ if and only if  
		\\&
		$z[\sigma^{-1}] \notin \FV\,t$
		\\\hline 
		\textrm{\PmBvr} 
		&
		if $x' \not= x$ and 
		$x' \notin \FV\;t$ 
		\\&then 
		$\Lm\;x\;t = \Lm\;x'\;(t[x'  \llra x])$
		\\\hline
	\end{tabular}
\\ \\ \vspace*{-0.5ex}
\begin{tabular}{|c|l|}	
	\hline
	\textrm{\SbVr}
	& 
	$(\Vr\;x)[s / z] = $ 
	\\&
	(if $x = z$ then $s$ else $\Vr\;x$)  
	\\\hline 
	\textrm{\SbAp}  & 
	$(\Ap\;t_1\;t_2)[s / z] \,=$ 
	\\&
	$ \Ap\,(t_1[s / z])\,(t_2[s / z])$
	\\\hline 
	\textrm{\SbLm}  & if $x \not= z$ and $x\;\fresh\;s$ then 
	\\&$(\Lm\;x\;t)[s / z] = \Lm\;x\;(t[s / z])$
	\\\hline 
	\textrm{\SbCg}  &
	if $z \not\in \{x_1,x_2\}$ and 
	$z\;\fresh\;t_1,t_2$ and 
	\\& 
	$t_1[(\Vr\;z) / x_1] = t_2[(\Vr\;z) / x_2]$ 
	\\&then 
	$\Lm\;x_1\;t_1 = \Lm\;x_2\;t_2$
	\\\hline 
	\textrm{\SbBvr}  &	
	if $x' \not= x$ and 
	$x'\;\fresh\;t$ then 
	\\&
	$\Lm\;x\;t = \Lm\;x'\;(t[(\Vr\;x') / x])$
		\\\hline 
	\textrm{\SbId}
	&
	$t[z / z] = t$
	\\\hline
	\textrm{\SbIm} 
	& if $x_1\not=y$ then 
	\\& $t[(\Vr\;x_1)/y][s/y]= t[(\Vr\;x_1)/y]$
	\\\hline 
	\textrm{\SbCh}  & if $y\not=x_2$ then 
	\\&$t[(\Vr\;y)/x_2][(\Vr\;x_2)/x_1][s/x_2]$ $= $ 
	\\& $t[(\Vr\;y)/x_2] [s/x_1]$
	\\\hline 
	\textrm{\SbCm}  & if $x \not= y$, 
	$y \;\fresh\; s$ and $x \;\fresh\; t$
	then 
	\\&$t[s/x] [t/y] = t[t/y][s/x]$
	\\\hline 
	\textrm{\SbFr} 
	& if $y \,\fresh\, t$ then $t[s/y]= t$
	\\\hline 
	\textrm{\FrSb} 
	&
	$z \;\fresh\;t [s / y]$ if and only if    
	\\&($z = y$ or $z \;\fresh\;t$) and ($y \;\fresh\;t$ or $z \;\fresh\; s$) 
	\\\hline 
	\textrm{\SbChFr}  & if $x_2 \,\fresh\, t$ then 
	\\&$t[(\Vr\;x_2)/x_1][s/x_2] = t[s/x_1]$
	\\\hline 
\end{tabular} 
\\ \\ \vspace*{-0.5ex}
\begin{tabular}{|c|l|}	
	\hline 	
	\textrm{\FSupFv}  	& 
	 $\FV\;t$ is finite
	\\\hline 
	\textrm{\FvDPm}  &
	$\FV\,t =  \{x\in \Var \mid \{y \mid t[x \llra  y] \not= t\}$ 
	\\& 
	$\hspace*{17.2ex}\mbox{ is infinite} \}
	$
	\\\hline 
    \textrm{\FvDSw}  &
    $\FV\,t = \{x\in \Var \mid \{y \mid t[x \sw y] \not= t\}$ 
    \\& 
    $\hspace*{17.2ex}\mbox{ is infinite} \}$
	\\\hline 
	\textrm{\FCB}  &	
	there exists $x$ such that
	\\&
	$x \notin \FV(\Lm\;x\;t)$ for all $t$ 
	\\\hline 
	\textrm{\FSupFr}  	&   
	$\{x.\;\neg\;x \;\fresh\;t\}$ is finite
	\\\hline 
		\textrm{\FrDSw}  &
	$x\;\fresh\;t $ if and only if 
	\\& $\{y \mid t[x \sw y] \not= t\}$ is finite 
	\\\hline 
	\textrm{\FrDRn}  &
	$x\;\fresh\;t $ if and only if 
	\\& $\{y \mid t[y / x] \not= t\}$ is finite 
    \\\hline
\end{tabular}
\end{tabular} 
} 
\vspace*{-0.3ex}
\caption{Recursion-relevant properties of operations and relations on terms}
\label{fig-basicProps}
\vspace*{-4ex}
\end{figure*} 

%% file: figCoBasicProperties.tex
\begin{figure*}
\centering 
{\small
\begin{tabular}{c} 
	\vspace*{-1.5ex}
\begin{tabular}{|c|l|}	
	\hline
\textrm{\ISwVr} 
&
if $\Dest\;t = \Vv\;x$ then 
\\&
$\Dest(t[z_1\sw z_2]) = \Vv (x[z_1 \sw z_2])$
\\\hline	
\textrm{\ISwAp} 
&
if $\Dest\;t = \Aa(t_1,t_2)$ then 
\\&
$\Dest(t[z_1 \sw z_2]) =\,$ 
\\&
$\Aa(t_1[z_1 \sw z_2],t_2[z_1 \sw z_2])$
\\\hline 	
\textrm{\ISwLm}
&
if $\Dest\;t = \Ll\;K$ then there 
\\& exists $K'$ 
such that 
\\&
$\Dest\;(t[z_1\sw z_2]) = \Ll\;K'$ and 
\\&
($(x[z_1\sw z_2],t'[z_1\sw z_2]) \in K'$ 
\\&
 \ $\,$for all 
$(x,t') \in K$)
\\\hline 
\textrm{\ISwCg} 
&
if $\Dest\;t = \Ll\;K$ and 
\\&
$\{(x_1,t_1),(x_2,t_2)\} \su K$ 
\\
&
then there exists $z$ such that 
\\&
($z = x_1$ or $z\;\fresh\;t_1$), 
($z = x_2$ or $z\;\fresh\;t_2$), 
\\&
and 
$t_1[z\sw x_1] = t_2[z\sw x_2]$ 
\\\hline
\textrm{\ISwBvr} 
&
if $\Dest\;s = \Ll\;K$ and 
\\& $\{(x,t),(x',t')\} \su K$
then 
\\&($x' = x$ or $x'\;\fresh\;t$) and $t' = t[x' \sw x]$
\\\hline
\!\textrm{\ISwBvrT} 
&
same as \ISwBvr{} but with 
\\&
$x'\notin \FV\,t$ instead of $x'\;\fresh\;t$
\\\hline
\end{tabular} 	
\\ \\  \vspace*{-1.5ex}
\begin{tabular}{|c|l|}	
	\hline
	\textrm{\IRnVr}
	& 
	if $\Dest\;t = \Vv\;x$ then 
	\\&
	$\Dest(t[y / z]) = \Vv(x[y/z])$  
	\\\hline 
	\textrm{\IRnAp}  & 
	if $\Dest\;t = \Aa(t_1,t_2)$ then 
	\\&
	$\Dest(t[y / z]) = \Aa(t_1[y / z],t_2[y / z])$
	\\\hline 
	\textrm{\IRnLmO}  
	&
	if $\Dest\;t = \Ll\;K$  
	then there 
		\\&
		exists $K'$ 
	s.t. $\Dest\;(t[y/z]) = \Ll\;K'$ 
	\\&
	and ($(x,t'[y/z]) \in K'$ 
	for all 
		\\&
	 \ \ \ \hspace*{2ex} $\,(x,t') \in K$ 
	s.t. $x \not\in \{y,z\}$) 
	\\\hline 
	\textrm{\IRnLmT}  & 
	if $\Dest\;t = \Ll\;K$ and $(x,t') \in K$
	\\& then 
	$t[x / z] = t$
		\\\hline 
	\textrm{\IRnCg} 
	&
	if $\Dest\;s = \Ll\;K$ and 
	\\&
	$\{(x_1,t_1),(x_2,t_2)\} \su K$ 
	\\
	&
	then there exists $z$ such that 
	\\&
	($z = x_1$ or $z\;\fresh\;t_1$), 
	($z = x_2$ or $z\;\fresh\;t_2$), 
	\\& and 
	$t_1[z /  x_1] = t_2[z / x_2]$ 
	\\\hline 
	\textrm{\IRnBvr}  &	
	if $\Dest\;s = \Ll\;K$ and 
	\\&
	$\{(x,t),(x',t')\} \su K$ then 
	\\
	&
	($x' = x$ or $x'\;\fresh\;t$) and $t' = t[x' / x]$
	\\\hline 
   \textrm{\IIRnBvr}  &	
	if $\Dest\;s = \Ll\;K$,  $(x,t) \in K$ and 
	\\
	&
	$x'\;\fresh\;t$  then $(x',t[x' / x]) \in K$
	\\\hline 
\end{tabular} 
\\ \\ \vspace*{-1.5ex}
	\begin{tabular}{|c|l|}	
	\hline 
	\textrm{\IFrVr} & 
	if $\Dest\;t = \Vv\;x$ and $z\;\fresh\;t$ then $z\not=x$
	\\\hline
	\textrm{\IFrAp} &
	if $\Dest\;t = \Aa(t_1,t_2)$ and $z\;\fresh\;t$
	\\& then $z\;\fresh\;t_1$ and $z\;\fresh\;t_2$
	\\\hline
	\textrm{\IFrLm}  &
	if $\Dest\;t = \Ll\;K$, $(x,t')\in K$ and $z\;\fresh\;t$
	\\& then $z = x$ or $z\;\fresh\;t'$
	\\\hline
\end{tabular}
\end{tabular} 
\hspace*{-0.9ex}
\begin{tabular}{c}
	\vspace*{-1.5ex}
	\begin{tabular}{|c|l|}	
		\hline 
		\textrm{\IFvVr} & 
		if $\Dest\;t = \Vv\;x$ then $x \in \FV(t)$
		\\\hline
		\textrm{\IFvAp} &
		if $\Dest\;t = \Aa(t_1,t_2)$ then
		\\&  
		$\FV\,t_1 \cup \FV\,t_2 \su \FV\;t$  
		\\\hline
		\textrm{\IFvLm}  &
		if $\Dest\;t = \Ll\;K$ and $(x,t') \in K$ then
		\\&  
		$\FV\,t' \sm \{x\} \su \FV\,t$ 
		\\\hline
	\end{tabular}
\\ \\ \vspace*{-1.5ex}
	\begin{tabular}{|c|l|}	
		\hline
		\textrm{\IPmVr} 
		&
		if $\Dest\;t = \Vv\;x$ then 
		\\&
		$\Dest(t[\sigma]) = \Vv (x[\sigma])$
		\\\hline
		\textrm{\IPmAp} 
		&
		if $\Dest\;t = \Aa(t_1,t_2)$ then 
		\\&
		$\Dest(t[\sigma]) = \Aa(t_1[\sigma],t_2[\sigma])$
		\\\hline
		\textrm{\IPmLm} 
		&
		if $\Dest\;t = \Ll\;K$ then there 
		\\& exists $K'$ 
		such that 
				\\&
				$\Dest\;(t[\sigma]) = \Ll\;K'$ and 
		\\&
		($(x[\sigma],t'[\sigma]) \in K'$ 
		for all 
		$(x,t') \in K$)
		\\\hline
		\textrm{\IPmBvr} 
		&
		if $\Dest\;s = \Ll\;K$ and 
		\\& $\{(x,t),(x',t')\} \su K$
		then 
		\\&($x' = x$ or $x'\notin\FV\;t$) 
		\\& and $t' = t[x'  \llra x]$
				\\\hline
	 	\textrm{\IIPmBvr}  &	
		if $\Dest\;s = \Ll\;K$,  
		$(x,t) \in K$ and 
		\\
		&
		$x'\;\fresh\;t$ then $(x',t[x \llra x]) \in K$
		\\\hline 
	\end{tabular}
\\ \\ \vspace*{-1.5ex}
\begin{tabular}{|c|l|}	
	\hline
	\textrm{\ISbVr}
	& 
	if $\Dest\;t = \Vv\;x$ then 
	\\&
	$\Dest(t[s / z])$ $=$ 
	\\& (if $x = z$ then $\Dest\;s$ else $\Vv\;x$)	
	\\\hline 
	\textrm{\ISbAp}  & 
	if $\Dest\;t = \Aa(t_1,t_2)$ then 
	\\&
	$\Dest(t[s / z]) = \Aa(t_1[s / z],t_2[s / z])$ 
	\\\hline 
	\textrm{\ISbLm}   
	&
	if $\Dest\;t = \Ll\;K$  
	then there \\& 
	exists $K'$ 	
	such that 
	\\&
	$\Dest\;(t[s/z]) = \Ll\;K'$ and 
	\\&
	($(x,t'[s/z]) \in K'$ 
	 for all 
	\\& \ $\,(x,t') \in K$ 
	such that $x \not=z$ and $x \;\fresh\;s$) 
	\\\hline 
	\textrm{\ISbBvr}  &	
	if $\Dest\;s = \Ll\;K$ and 
	\\&
	$\{(x,t),(x',t')\} \su K$ then 
	\\
	&
	($x' = x$ or $x'\;\fresh\;t$) 
	and $t' = t[(\Vr\;x') / x]$
	\\\hline 
	\textrm{\IISbBvr}  &	
	if $\Dest\;s = \Ll\;K$,  
	$(x,t) \in K$ and 
	\\
	&
	$x'\;\fresh\;t$ then $(x',t[(\Vr\;x') / x]) \in K$
	\\\hline 
\end{tabular} 
\\ \\ \vspace*{-1.5ex}
\begin{tabular}{|c|l|}	
	\hline 	
	\textrm{\IFSupFv}  	& 
	$\FV\;t$ is countable
	\\\hline 
	\textrm{\IFvDPm}  &
	$\FV\,t =  \{x\in \Var \mid \{y \mid t[x \llra  y] \not= t\}$ 
	\\
	& 
	$\hspace*{17.2ex}\mbox{ is uncountable} \}
	$
	\\\hline 
    \textrm{\IFvDSw}  &
    $\FV\,t = \{x\in \Var \mid \{y \mid t[x \sw y] \not= t\}$ 
    \\
    & 
    $\hspace*{17.2ex}\mbox{ is uncountable} \}
    $
    \\\hline 
	\textrm{\IFSupFr}  	&   
	$\{x.\;\neg\;x \;\fresh\;t\}$ is countable
	\\\hline 
	\textrm{\IFrDSw}  &
	$x\;\fresh\;t $ if and only if 
	\\& $\{y \mid t[x \sw y] \not= t\}$ is countable 
	\\\hline 
	\textrm{\IFrDRn}  &
	$x\;\fresh\;t $ if and only if 
	\\& $\{y \mid t[y / x] \not= t\}$ is countable 
	\\\hline
\end{tabular}
\end{tabular} 
} 
\vspace*{-0.3ex}
\caption{Corecursion-relevant properties of 
	iterms. 
We only list properties that are 
counterparts of those from Fig.~\ref{fig-basicProps} involving constructors 
and finiteness conditions. As for the others, 
namely the algebraic properties of the operators,  
their formulation does not change, so we will use the same notation.
%
For example, \SwId{} %
from Fig.~\ref{fig-basicProps} denotes a property that makes sense not only for terms but also for iterms (and will, in due course, for our corecursor models too). 
\looseness=-1
%
%
}
\label{fig-basicCoProps}
\vspace*{-4ex}
\end{figure*} 

%% file: TR_Nominal_Recursors_as_Epirecursors.bbl

%% file: appendix.tex
\appendix

\ \\ \ \\
\begin{center}
{\huge APPENDIX}
\end{center}

\ \\ \ \par
This  appendix provides details, proof sketches and extensions for the concepts and results presented in the main paper. Specifically, it provides:
\begin{mmyitem}
	\item some technical lemmas on nominal sets that are relevant for regarding the perm/free recursor 
	 as an epi-recursor (App.~\ref{app-modreDetailsNomSets})
	 \item some additional examples of functions defined by nominal recursion 
	 (App.~\ref{app-anotherExample}) 
	\item proof sketches for all the stated results on recursors (App.~\ref{app-proofSketches})
	\item the description of a uniform way to enhance the recursors with full-fledged recursion and Barendregt's convention (App.~\ref{app-addingBacknhancements})
	\item the definition of infinitary $\lambda$-terms (iterms) and their operators, and a description of the relevant proof principles for them (App.~\ref{app-detailsIterms})
	\item more details on epi-corecursors and nominal corecursors (App.~\ref{app-moreDetailsCorec})
 	\item proof sketches for all the stated results on corecursors (App.~\ref{app-proofCoSketches})
 	\item a discussion of the notion of enhancing corecursors  (App.~\ref{app-enhCorec})
 	 \item an  example of a function defined by nominal corecursion, namely parallel substitution 
 	(App.~\ref{app-exaCorec}) 
	\item a detailed presentation of our Isabelle mechanization (App.~\ref{app-isa})
\end{mmyitem} 

\section{More details on nominal sets}
\label{app-modreDetailsNomSets}

Next, we will give details on the justification for the following claim made in the main paper: $r_1$ is the stripped down version of the perm/fresh  recursor, where here ``stripped down'' refers to removing the Barendregt parameter $X$, i.e., taking $X=\emptyset$. 

Consider the following result that gives a more direct description of the support function: 
\smallskip
\begin{mylemma}\rm \cite{pitts-AlphaStructural}
	\label{lem-charSupp}
	Let $\AA = (A,\_[\_]^\AA)$ be a nominal set. Then, for every $a\in A$, there exists the smallest set $X\su A$ that supports $a$, denoted $\supp^\AA(a)$. 
	Moreover, it holds that 
	$\supp^\AA(a) = 
	\{x\in\Var \mid \{y\in\Var \mid a [x \!\llra\! y]^\AA \not= a\} \mbox{ is infinite}\}$.   \qed 
\end{mylemma} 
\medskip 

In turn, this enables an alternative description of nominal sets: 

\smallskip
\begin{mylemma}\rm 
	\label{lem-altNomSet}
	Let $\AA = (A,\_[\_]^\AA)$. Then the following are equivalent:  
		\begin{myitem}
		\item[(1)] $\AA$ is a nominal set; 
		\item[(2)] There exists a (necessarily unique) function $\FV : A \ra \Pow(\Var)$ such that $(A,\_[\_]^\AA,\FV)$ 		
		satisfies $\PmId$, $\PmCp$, 
		$\FvDPm$ 
		and $\FSupFv$. 		
	\end{myitem} 
\end{mylemma} 
\smallskip 
\begin{proof}
(1)	 implies (2): We take $\FV$ to be $\supp^\AA$. Then $\PmId$, $\PmCp$ and $\FSupFv$ are part of the definition of nominal sets, and $\FvDPm$ is ensured by Lemma~\ref{lem-charSupp}. 

(2) implies (1): Thanks to $\PmId$ and $\PmCp$,  $\AA$  is a pre-nominal set. 

Next, we show that  the operator defined by $\FvDPm$ is the same as the support operator, i.e., for all $a\in A$, $\FV\;a$ is the smallest set of variables that supports $a$: 
\begin{description}
	\item{- $\FV\;a$ supports $a$:}  Let $x,y\in\Var \sm (\FV\;a)$.  
	If $x=y$, the desired fact, $a[x \!\llra\! y]^\AA = a$, follows from $\PmId$. So let us assume $x\not=y$. 
	Since $x\notin \FV\;a$, thanks to \FvDPm{} we have that the set $\{z\in\Var \mid a [x \!\llra\! z]^\AA \not= a\}$ is finite. Similarly, the set $\{z\in\Var \mid a [y \!\llra\! z]^\AA \not= a\}$ is finite. Then we can find $z\in \Var \sm \{x,y\}$ such that 
	$a [x \!\llra\! z]^\AA = a$ and $a [y \!\llra\! z]^\AA = a$. 
	Next, applying the properties of swapping in pre-nominal sets, we have:
	$a = a [x \!\llra\! z]^\AA [y \!\llra\! z]^\AA = 
	  a [x[y \sw z] \!\llra\! z[y\sw z]]^\AA = a [x \!\llra\! y]^\AA$, as desired.
	 \item{- $\FV\;a$ is included in any set that supports $a$:} 
	 Assume $X \su \Var$ supports $a$.  Let $X' = X \cap \FV\;a$. Since both $X$ and (as we have just proved) $\FV\;a$ support $a$, we have that $X'$ supports $a$. And since,  
	 by \FSupFv{}, $\FV\;a$ is finite, we have that $X'$ is finite. 
	 %
	 To show $\FV\;a \su X$, it suffices to show $\FV\;a \su X'$. 
	 Let $x\in \FV\;a$, meaning that the set $\{y \mid a[x \!\llra\! y]^\AA \not= a\}$ is infinite. By the finiteness of $X'$, we obtain $y\notin X'$ such that $a[x \!\llra\! y]^\AA \not= a$. Then, since $X'$ supports $a$, it cannot be the case that $x\notin X'$. Hence $x \in X'$, as desired.   	
\end{description}
We have thus proved that $\AA$ is a pre-nominal set and that 
$\FV$ coincides with $\supp^\AA$---which, thanks to \FSupFv{}, means that the finite support property holds for $\AA$.  We obtain that $\AA$ is a nominal set, as desired. \qed 
\end{proof}
\medskip 

From Lemma~\ref{lem-altNomSet}, it follows 
that 
a nominal set $\AA = (A,\_[\_]^\AA)$ together with $\emptyset$-supported operations 
$\Vr^\AA : \Var \ra A$, 
$\Ap^\AA : A \ra A \ra A$ and 
$\Lm^\AA : \Var \ra A \ra A$ is the same as a $(\Sigma_1,\Props_1)$-model.

It remains to show that the properties of the unique function $g:\Trm \ra A$ 
guaranteed by Thm~\ref{thm-pittsRec} are the same as those defining 
$\Sigma_1$-morphisms.  
Indeed, clauses (1)--(3) in Thm~\ref{thm-pittsRec} are the $\Sigmac$-part of the morphism conditions. Moreover, $g$'s commutation with permutation (in nominal terminology, equivariance) is the same as being supported by $\emptyset$, as a particular case of the following 
lemma: 
%
\smallskip 
\begin{mylemma}\rm 
	\label{lem-funSupp} 
	Let $f: A \ra B$ be a function between two nominal sets 
	$\AA = (A,\_[\_]^\AA)$ and $\BB = (B,\_[\_]^\BB)$ and $X$ a set of variables. Then the following are equivalent:
	\begin{myitem}
	\item[(1)] $f$ is supported by $X$; 
	\item[(2)] $f(a[\sigma])^\AA = (f\,a)[\sigma]^\AA$ for all $\sigma\in \Perm$ 
	such that $\supp\,\sigma \cap X = \emptyset$.  
	\end{myitem} 
\end{mylemma}
\begin{proof}
	We have the following equivalencies (where $\FF$ is the pre-nominal set of functions from $A$ to $B$): 
	\begin{center}
		(1)
		\\iff (by the definition of ``supported'')
		$\forall x,y\notin X.\;f[x \!\llra\! y]^\FF = f$
		\\iff (by the definition of swapping for functions)
		$\forall x,y\notin X.\;\forall a\in A.\;f(a[x \!\llra\! y]^\AA)[x \!\llra\! y]^\BB =
		 f\;a$
		\\iff (by the idempotency of swapping in $\BB$)
		$\forall x,y\notin X.\;\forall a\in A.\;f(a[x \!\llra\! y]^\AA) =
		(f\;a)[x \!\llra\! y]^\BB$
	\end{center}
It remains to show that the last property in the above chain of equivalencies is in turn equivalent to (2). It is clearly implied by (2), since it is a particular case of (2) for the permutation $\sigma$	being $x \!\llra\! y$. Conversely, the fact that it implies (2) follows by induction on  the finite set $\supp\,\sigma$ (employing the inductive characterization of finiteness) using the properties of permutation (including that any permutation is a composition of transpositions).  
	\qed
\end{proof}
\smallskip 

Finally, the preservation of the freshness operator (which is required by the notion of $\Sigma_1$-morphism but is not explicitly stated in Thm~\ref{thm-pittsRec}), is implied by commutation with permutation---more precisely, the following holds:
\smallskip 
\begin{mylemma}\rm 
	\label{lem-equivImplFreshnessPres} 
	Let $f: A \ra B$ be a function between two nominal sets 
	$\AA = (A,\_[\_]^\AA)$ and $\BB = (B,\_[\_]^\BB)$ that commutes with permutation (i.e., is equivariant, i.e., is supported by $\emptyset$). Then $\supp^\BB(f\;a) \su \supp^\AA\;a$ for all $a\in A$.   
\end{mylemma}
\begin{proof}
Thanks to the definition of support, it suffices to check that 
$\supp^\AA\;a$ supports $f\;a$ (in $\BB$). Indeed, assume $x,y\notin \supp^\AA\;a$; then $a[x \!\llra\! y]^\AA = a$, hence $f(a[x \!\llra\! y]^\AA) = f\;a$. And since by equivariance $f(a[x \!\llra\! y]^\AA) = (f\;a)[x \!\llra\! y]^\BB$, we obtain 
 $(f\;a)[x \!\llra\! y]^\BB = f\;a$, as desired. \qed
\end{proof}

\medskip 
This concludes the justification of the fact that $r_1$ is the stripped down version of the perm/free recursor.  

\section{Other examples of nominal recursion}
\label{app-anotherExample}

Next we show some more examples of nominal recursion taken from the literature. 
In all these examples, checking the necessary properties for the target models, i.e., 
the $\Props_i$ properties, is completely routine. 

\begin{exa}\rm 
	\textit{(the size (depth) of a term defined using $r_4$)} \cite{primrecFOAS-Norrish04}
	Consider the task of defining the size function on terms, 
	$\ddepth : \Trm \ra 
	\Nat$. 
	The desired constructor-based recursive clauses are the following: 
	\begin{myitem}
		\item[(i)] $\ddepth\;(\Vr\;x) = 1$
		%
		\hspace*{7ex}
		(ii) $\ddepth\;(\Ap\;t_1\;t_2) = 
	    \ddepth\;t_1  + \ddepth\;t_2 + 1$
		\item[(iii)] 
		$\ddepth\;(\Lm\;x\;t) = 
		\ddepth\;t + 1$
	\end{myitem}
	To make this work, we add clauses describing the intended behavior of $\ddepth$ with respect to swapping and free-variables: 
	\begin{myitem}
		\item[(iv)] $\ddepth\;(t\,[x\sw y]) = \ddepth\;t$
		\item[(v)] $\emptyset \su \FV\;t $  
		(for this particular definition, this clause is vacuous) 
	\end{myitem}
	
	This means 
	organizing the target domain 
	$\Nat$ 
	as a model $\AA = (\Nat,\Vr^\AA,\Ap^\AA,\Lm^\AA,\_[\_\,\sw \_]^\AA,\FV^\AA)$ as follows:
	$$
	\begin{array}{c} 
		\Vr^\AA\;x = 1  \hspace*{7ex}
		\Ap^\AA\;m\;n = m + n + 1 \hspace*{7ex}
		\Lm^\AA\;x\;m = m + 1 
		\\
		m\,[x\sw y]^\AA = m
		\hspace*{11ex}
		\FV^\AA\;m = \emptyset
	\end{array}
	$$
	After checking that $\AA$ satisfies the properties required by $r_4$ 
	(i.e., $\Props_4$) 
	we obtain a unique function $\ddepth$ satisfying clauses (i)--(v). 
\end{exa}

In the following example, we will use as target domain a set $\ETrm$ 
of ``extended terms'', which are defined like terms but with an 
additional constructor $\Ct\;c$, where $c$ ranges over constants from 
a set $C$; we will also assume that $C$ includes $\ap$ and $\lm$. 

\begin{exa}\rm 
	\textit{(HOAS encoding defined using $r_7$)} \cite{DBLP:conf/icfp/PopescuG11}
	Consider the task of defining a function 
	$\enc : \Trm \ra \ETrm$ that encodes terms into extended terms in a higher-order abstract syntax 
	(HOAS) fashion. This is a simplified version of HOAS encodings in logical frameworks such as 
	LF \cite{har-fra}. 
	The desired constructor-based recursive clauses are the following: 
	\begin{myitem}
		\item[(i)] $\enc\;(\Vr\;x) = \Vr\;x$
		%
		\hspace*{7ex}
		(ii) $\enc\;(\Ap\;t_1\;t_2) = \Ap\;(\Ap\;(\Ct\;\ap)\;(\enc\;t_1))\;(\enc\;t_2)$
		\item[(iii)] 
		$\enc\;(\Lm\;x\;t) = \Ap\;(\Ct\;\lm)\;(\Lm\;x\;(\enc\;t))$
	\end{myitem}
	To make this work, we add clauses describing the intended behavior of $\enc$ with respect to substitution and freshness: 
	\begin{myitem}
		\item[(iv)] $\enc\;(t\,[s/x]) = (\enc\;t)[(\enc\;s)/x]$ 
		\item[(v)] $x \,\fresh\, t$ implies $x \,\fresh\, (\enc\;t)$ 
	\end{myitem}
(Both (iv) and (v) have a stand-alone importance for HOAS encodings.) 
	
	This means 
	organizing the target domain 
	$\ETrm$ 
	as a model $\AA = (\ETrm,\Vr^\AA,\Ap^\AA,\Lm^\AA,\_[\_\,/ \_]^\AA,\fresh^\AA)$ 
	where $\Vr^\AA$, $\_[\_\,/ \_]^\AA$ and $\fresh^\AA$ are the usual variable-injection, 
	swapping and freshness on extended terms, 
	and application- and abstraction-like operators are defined as follows:
	\begin{itemize}
		\item[] $\Ap^\AA\;e_1\;e_2 = \Ap\;(\Ap\;(\Ct\;\ap)\;e_1)\;e_2$
		\hspace*{7ex} $\Lm^\AA\;x\;e = \Ap\;(\Ct\;\lm)\;(\Lm\;x\;e)$
	\end{itemize}
	After checking that $\AA$ satisfies the properties required by $r_7$ 
	(i.e., $\Props_7$) 
	we obtain a unique function $\enc$ satisfying clauses (i)--(v). 
\end{exa}

Next we show two examples that use the enhancements 
discussed in App.~\ref{app-addingBacknhancements}. 

\begin{exa}\rm
		\textit{(eta normal form using enhanced $r_4$)}  \cite{primrecFOAS-Norrish04}
	Consider the task of defining the function 
$\enf : \Trm \ra \Bool$ which checks whether a term is in $\eta$-normal form. 
We will write $\isAp : \Trm \ra \Bool$ for the function that checks whether a term is an application, 
and $\getApL, \getApR : \Trm \ra \Trm$ for the functions that return the left- and right- argument 
respectively if the term is an application (otherwise it does not matter, e,.g., they return the term itself); 
thus, $\getApL\;(\Ap\;t_1\;t_2) = t_1$ and $\getApR\;(\Ap\;t_1\;t_2) = t_2$. 
 
The desired constructor-based recursive clauses are the following: 
\begin{myitem}
	\item[(i)] $\enf\;(\Vr\;x) = \Ttrue$
	%
	\hspace*{7ex}
	(ii) $\enf\;(\Ap\;t_1\;t_2) = (\enf\;t_1 \wedge \enf\;t_2)$
	\item[(iii)] 
	$\enf\;(\Lm\;x\;t) = 
	(\enf\;t \wedge (\isAp\;t \wedge \getApR\;t = \Vr\;x \implies x \in \FV(\getApL\;t)))$
\end{myitem}
To make this work, we add clauses describing the intended behavior of $\eta$ with respect to swapping and free-variables: 
\begin{myitem}
	\item[(iv)] $\eta\;(t\,[x\sw y]) = \eta\;t$
	\item[(v)] $\emptyset \su \FV\;t$  
	(again, this clause is vacuous here) 
\end{myitem}

This means 
organizing the target domain 
$\Bool$ 
as a $(X,\Sigma_4)$-model 
$\AA = (\Bool,D = \Trm \times \Bool,\Vr^\AA,\Ap^\AA,\Lm^\AA,\_[\_\,\sw \_]^\AA,\FV^\AA)$ as follows:
$$
\begin{array}{c} 
	\Vr^\AA\;x = \Ttrue  \hspace*{7ex}
	\Ap^\AA\;(t_1,b_1)\;(t_2,b_2) = (b_1 \wedge b_2) 
	\\
	\Lm^\AA\;x\;(t,b) = (b \wedge (\isAp\;t \wedge \getApR\;t = \Vr\;x \lra x \in \FV(\getApL\;t)))
	\\
	(t,b)\,[x\sw y]^\AA = b
	\hspace*{11ex}
	\FV^\AA\;m = \emptyset 
\end{array}
$$
After checking that $\AA$ satisfies the properties required by $r_4$ 
(i.e., is an $(X,\Sigma_4)$-model satisfying $\Props_4$) 
we obtain a unique function $\eta$ satisfying clauses (i)--(v). 
\end{exa}

The above definition takes advantage of the full-recursion enhancement, 
but did not need the Barendregt enhancement. Indeed, the desired properties, 
i.e., $\Props_4$, already hold for the target model 
in the stronger form, non-relativized to the finite set of variables $X$. 
The next definition is the standard situation where the Barendregt enhancement comes handy.


\begin{exa}\rm
	\textit{(substitution defined using enhanced $r_1$)}  \cite{pitts-AlphaStructural} 
	Let $s$ and $y$ be a fixed term and a fixed variable. 
	Consider the task of defining the function $\subst_{s,y}$ that takes any term $t$ and performs the (capture-free) 
	substitution of $s$ for $y$ in $t$. 
	(Thus, $\subst_{s,y}\,t$ will be the same as $t[s/y]$.)
		The desired constructor-based recursive clauses are the following: 
	\begin{myitem}
		\item[(i)] $\subst_{s,y}\,(\Vr\;x)= (\mbox{if $x=y$ then $s$ else $\Vr\;x$})$
	    \item[(ii)] $\subst_{s,y}\,(\Ap\;t_1\;t_2) = \Ap\;(\subst_{s,y}\,t_1)\;(\subst_{s,y}\,t_2)$
		\item[(iii)] 
		$\subst_{s,y}\,(\Lm\;x\;t) = \Lm\;x\;(\subst_{s,y}\,t)$ if $x\notin \FV\,s \cup \{y\}$ 
	\end{myitem}
	To make this work, 
	we add a clause describing the intended behavior of $\subst_{s,y}$ with respect to permutation: 
	\begin{myitem}
		\item[(iv)] $(\subst_{s,y}\;t)[\sigma] = 
		\subst_{s,y}\,(t[\sigma])$ if $\supp(\sigma) \cap (\FV\,s \cup \{y\})= \emptyset$ 
		%
	\end{myitem}
	
	This means taking $X = \FV\,s \cup \{y\}$ and 
	organizing the target domain 
	$\Trm$ 
	as an $(X,\Sigma_1)$-model 
	$\AA = (\Trm,D = \Trm \times \Trm,\Vr^\AA,\Ap^\AA,\Lm^\AA,\_[\_\,\sw \_]^\AA) 
	$ as follows:
	$$
	\begin{array}{c} 
		\Vr^\AA\;x = (\mbox{if $x=y$ then $s$ else $\Vr\;x$}) \hspace*{7ex}
		\Ap^\AA\;(t_1',t_1)\;(t_2',t_2) = \Ap\;t_1\;t_2 
		\\
		\Lm^\AA\;x\;(t',t) = \Lm\;x\;t
		\hspace*{7ex}
		(t',t)\,[\sigma]^\AA = t[\sigma]
	\end{array}
	$$
	After checking that $\AA$ satisfies the properties required by $r_1$ 
	(i.e., is an $(X,\Sigma_1)$-model satisfying $\Props_1$) 
	we obtain a unique function $\subst_{s,y}$ satisfying clauses (i)--(iv). 
\end{exa}

Note that above we made crucial use of the Barendregt enhancement, but have not used 
the full-recursion enhancement (as seen in the fact that the first components of the pairs 
are ignored in the definitions of the model operators).

%



\section{Proof Sketches for the Recursor Results}
\label{app-proofSketches} 

\subsection{Proof idea for the nominal recursion theorems (Thm.~\ref{thm-allNominalRecs})}
\label{app-subsec-proofIdeaRecursionThm}

Our discussion in \S\ref{subsec-nomrecAsEpirecFormally} 
shows how 
$r_1$, $r_4$, $r_6$, $r_7$ and $r_8$ coincide with nominal recursors 
from the literature---so we are in a position to cite these literature results as justification for their share of Thm.~\ref{thm-allNominalRecs}. On the other hand, we cannot cite the literature for the variant epi-recursors $r_2$, $r_3$, $r_5$ and $r_9$. 
\looseness=-1

We will actually prove Thm.~\ref{thm-allNominalRecs} in a uniform way, 
taking advantage of the expressiveness comparisons between these recursors. We will infer the recursion theorems for all nine recursors from those of just two of them. 
We will come back to this in  App.~\ref{app-subsec-backToProofRecThms}.   

\subsection{Proofs of the recursor expressiveness comparison results}
\label{app-subsec-proofsofOthers}

\noindent 
\textbf{More detailed proof of Thm.~\ref{thm-expr}.}   
	When proving each 
	$r_i \geq r_j$, we instantiate Prop.~\ref{prop-extCriterion} 
	taking $r' = r_i$ and $r = r_j$. So here $\Bcat$ is the category of $\Sigmac$-models, $\Ccat'$ that of $(\Sigma_i,\Props_i)$-models, 
	and $\Ccat$ that of $(\Sigma_j,\Props_j)$-models; $R'$ 
	is the forgetful functor from $(\Sigma_i,\Props_i)$-models to $\Sigmac$-models, and $R$ the forgetful functor from $(\Sigma_j,\Props_j)$-models to $\Sigmac$-models; 
	$B=\TTrm(\Sigmac)$, 
	$I'=\TTrm(\Sigma_i)$ and $I=\TTrm(\Sigma_j)$.  
	In each case, we must define a pre-functor $F : \Ccat \ra \Ccat'$  
	such that 
	$R' \circ F = R$ and $F\;I = I'$.  
	This essentially means showing how to 
	transform $(\Sigma_j,\Props_j)$-models into $(\Sigma_i,\Props_i)$-models in such a manner that $\TTrm(\Sigma_j)$ 
	becomes $\TTrm(\Sigma_i)$---which gives $F$'s behavior on objects, while on morphisms $F$ will be the identity. Each time, $F$ will transform models by preserving the carrier set and the constructor-like operators, and possibly defining 
	(1) permutation-like from swapping like operators or vice versa, (2) 
	freshness-like operators from 
	free-variable-like operators, or (3) renaming-like from substitution-like operators; 
	these definitions are done just like for concrete terms (where, e.g., we can standardly define freshness from free-variables). In each case, the only interesting fact that needs to be checked is that $F$ is well-defined on objects: when starting with a $\Sigma_j$-model satisfying $\Props_j$, the result $\Sigma_i$-indeed satisfies $\Props_i$. Everything else 
	amounts to either well-known or trivial properties. Thus, $F\;I = I'$ means that the standard inter-definability properties (1)--(3) hold for terms, e.g., $x \,\fresh\, t$ iff $x \notin \FV\;t$; and $R' \circ F = R$ (i.e., $F$ commutes with the forgetful functors to $\Sigmac$-models) follows immediately from the fact that $F$ does not change the carrier set or the constructor-like operators. 

	Next, we informally discuss these transformations and highlight the intuitions behind them. 
	At the end, we also explain why we believe the stated inequalities are strict (in that the opposite inequalities do not hold), although, with the two exceptions expressed in Props.~\ref{prop-negRec}, 
	we do not yet have  
	proofs 
	 for the strictness  conjectures.  
	
	\textbf{Proof of $r_1 \equiv r_3$:} Recall that $r_1$ is the original nominal-logic recursor and $r_3$ is its variation that uses swapping rather than permutation, the difference being the use of 
	the identity, involutiveness and compositionality properties for swapping (\SwId{}, \SwIv{} and \SwCp{}) rather than the identity and compositionality for permutation (\PmId{} and \PmCp{}). $r_1 \equiv r_3$ holds because  operations with these  properties correspond bijectively and functorially to each other, allowing us to move back and forth between $\Props_1$-models and $\Props_3$-models \cite[Section 6.1]{pitts_2013}. In one direction, starting with an operation $\_[\_]^\MM : M \ra \Perm \ra M$ satisfying \PmId{} and \PmCp{}, we define $\_[\_\sw\_]^\MM : M \ra \Var \ra \Var \ra M$ as its restriction to transposition permutations---and it satisfies SwId, SwIv and 
	\SwCp{}. Conversely, starting with an operation $\_[\_\sw\_]^\MM$ satisfying \SwId{}, \SwIv{} and \SwCp{}, we define $\_[\_]^\MM$ by $m[\sigma] = m[x_1\sw y_1]\ldots[x_n\sw y_n]$ where $(x_1,y_1) \cdot \ldots \cdot (x_n,y_n)$ is any decomposition of $\sigma$ into transpositions; thanks to \SwId{}, \SwIv{} and \SwCp{}, we can prove that this definition is independent of the particular decomposition and that $\_[\_]^\MM$ satisfies \PmId{} and \PmCp{}. 
	
	\textbf{Proof of $r_2 \geq r_1$:} 
	The difference between the two recursors is the following: $r_1$ requires the definability of the free-variable (support) operator from  permutation \FvDPm{},  and the freshness condition for binders \FCB{}. By contrast, $r_2$ requires instead that the free-variable operator is related to the constructor by the usual inductive clauses \FvVr{}, \FvAp{} and \FvLm{}, and is related to permutation via \PmFv{}. 
	%
	In particular, $r_2$ is looser in that it does not require the free-variable operator to be \emph{definable from} permutation, but only to be \emph{related to} permutation by some weaker properties. And indeed, this looseness translates into higher flexibility, because any $(\Sigma_1,\Props_1)$-model can be proved to be in particular a $(\Sigma_2,\Props_2)$-model, more precisely:
	\begin{itemize} 
		\item \PmFv{} follows from \FvDPm{} in the presence of \PmId{}, \PmCp{}; 
		\item \FvVr{}, \FvAp{} and \FvLm{} follow from \FvDPm{} and \FCB{} 
		in the presence of \PmId{}, \PmCp{}, \PmVr{}, \PmAp{}, \PmLm{}.
	\end{itemize} 	
	Thus, $r_2 \geq r_1$ follows from Prop.~\ref{prop-extCriterion} taking $F$ to be the identity functor, more precisely the inclusion functor between the categories 
	of $(\Sigma_1,\Props_1)$-models and $(\Sigma_2,\Props_2)$-models. 
	(Since working with permutation is much heavier than working with swapping, 
	we preferred to do this proof while taking advantage of the permutation-swapping connection. Namely, starting with a $(\Sigma_1,\Props_1)$-model $\MM$, and already knowing from before that $\MM$ with 
	its swapping operator corresponding to permutation satisfies $\Props_3$, 
	we proved that it also satisfies $\SwFv$, 
	as well as $\FvVr,\FvAp$ and $\FvLm$ (thus essentially establishing on the way that $r_4 \geq r_3$).  
	Then, using again the permutation-swapping connection, we inferred $\PmFv$ 
	from $\SwFv$.) 
	
	\textbf{Proof of $r_4\geq r_2$:} This 
	holds essentially for the same reason why $r_3\geq r_1$ holds, i.e., because the 
	restriction to transpositions of a permutation operator satisfying \PmId{} and \PmCp{} will satisfy \SwId{}, \SwIv{} and \SwCp{}, together with the fact \PmFv{} implies \SwFv{} along this transposition-restriction operator.  
		Thus, $r_4\geq r_2$ follows from Prop.~\ref{prop-extCriterion} taking $F$ to be the same functor  
		as that used for $r_3\geq r_1$.   
	
	But note that, this time, we don't have \SwCp{} on the swapping side, in that $r_4$ does not require \SwCp{}---so the converse construction described above when proving $r_1\geq r_1$ (using decomposition into transpositions) does not hold, forbidding us from establishing the converse inequality $r_2 \geq r_4$; and indeed, we later (in  Prop.~\ref{prop-negRec}) prove that $r_2 \not\geq r_4$.  
	This reveals a perhaps unexpected phenomenon: that swapping-based recursors can be strictly more expressive than their permutation-based counterparts, as is indeed the case of $r_4$ versus $r_2$ (though not of $r_3$ versus $r_1$). 
	
	\textbf{Proof of $r_5 \geq r_4$:} This follows by applying 
	Prop.~\ref{prop-extCriterion} with the functor $F$ that transforms the 
	free-variable operator into a freshness operator using negation. Indeed, save for the straightforwardly corresponding free-variable operator's 
  properties \FvVr{}, \FvAp{} and \FvLm{} and their freshness counterparts 
	 \FrVr{}, \FrAp{} and \FrLm{}, 
	the only difference between $r_4$ and $r_5$ is the replacement of 
	\SwFv{} 
	with 
	\SwBvr{}.
	And the latter follows from the former in the presence of \FvLm{}.\footnote{\citet{primrecFOAS-Norrish04}'s original recursor also assumed, for the swapping-like operator, equivariance (\textsf{FvSw}) and two algebraic properties (\textsf{SwId}, \textsf{SwIv}), but these turn out to not be needed for his recursion theorem to hold. We discovered this redundancy when proving $r_5 \geq r_4$ and realizing that \textsf{FvSw}, \textsf{SwId}, \textsf{SwIv} are not needed to reduce $r_4$ to $r_5$
		---which means that $r_5$ can ``lend'' its recursion theorem to this axiom-lighter version of  $r_4$.} 
	Going in the other direction, 
	namely proving $r_4 \geq r_5$, does not seem possible: \SwBvr{} is a (more abstract) weaker assumption than 
	\SwFv{}, 
	even in the presence of all the other assumptions. 
	
	\textbf{Proof of $r_6 \geq r_5$:} This follows from the fact that, in the presence of the other axioms in $\Props_5$, 
	the bound-variable renaming property \SwBvr{} implies the congruence property \SwCg{} (though not the other way around). So again we use an identity functor $F$, more precisely the inclusion functor between the categories 
	of $(\Sigma_5,\Props_5)$-models and $(\Sigma_5,\Props_5)$-models. 
	
	\textbf{Proof of $r_9 \geq r_7$:} This follows by applying 
	Prop.~\ref{prop-extCriterion} with the functor $F$ that transforms the 
	substitution operator of a $(\Sigma_7,\Props_7)$-model $\MM$ into a renaming operator just like this is done for terms, by inserting the free-variable injection operator into the second argument: 
	$m[y/x]^\MM = m[(\Vr^\MM\;y)/x]^\MM$. 
	This yields a model satisfying $\Props_9$
	because the axioms for substitution straightforwardly imply those for renaming (though not the other way around, since substitution requires additional structure). 
	
	\textbf{Proof of $r_9 \geq r_8$:} The proof here takes advantage of the fact that, in a constructor-enriched renset (structures axiomatizing renaming that form the basis of recursor $r_8$), freshness is definable from renaming \cite{DBLP:conf/cade/Popescu22} in several equivalent ways, 
	including via \FrDRn{}; and with this definition, the $\Props_9$ properties 
	follow from $\Props_8$.  So $r_9 \geq r_8$ is proved using the functor that takes any $(\Sigma_8,\Props_8)$-model to a $(\Sigma_9,\Props_9)$-model that has the same carrier set, constructors and renaming operator, and has freshness defined via \FrDRn{}. 
	\qed

\ \par 
One may wonder whether any relation can be established between the strength of swapping-based recursors on the one hand, and renaming- or substitution-based recursors on the other hand. The answer seems to be negative: Substitution-like operators $\_[\_/\_]^\MM : M \ra {M} \ra \Var \ra M$ are structurally more complex than swapping-like operators 
$\_[\_\sw\_]^\MM : M \ra {\Var} \ra \Var \ra M$
as they (intuitively) refer to the replacement not of variables for variables, but of entities from the models for variables; so there seems to be no hope of defining the former from the latter in a freshness-swapping model; and it seems that not even renaming-like operators can be defined from swapping-like operators, since the latter but not the former preserve the free variables. 

Conversely, defining a swapping-like operator in a substitution-based or renaming-based model seems  superficially more plausible. However, as far as we see, the only way to achieve this while ensuring the required properties for swapping would be along the lines of the standard trick of 
employing an additional fresh variable, namely picking a fresh $z$ and defining 
$m[x \sw y]^\MM = m[z / x]^\MM [x / y]^\MM [y / z]^\MM $. But this does not work since the substitution-based and renaming-based 
models do not guarantee the existence of fresh variables---and adding axioms guaranteeing that would severely restrict the recursors' expressiveness. 
Thus, the swapping/permutation-based recursors and substitution/renaming-based recursors seem incomparable w.r.t.\ expressiveness (that is, using the ``head-to-head'' comparison relation $\geq$; but this situation changes when we switch to the laxer comparision $\wgeq$, as discussed in 
\S\ref{subsec-moreGentle}).  

\smallskip 
\ \\
\textbf{Proof of Prop.~\ref{prop-WeakExtCriterion}.}     
%
%
%
%
%
(See Fig.~\ref{fig-critQuasiStrongerEpirec} from the main paper.)  
Assume $g:T \ra B$ is definable by $r$, meaning that $g = R\;\im_{I,C}$ for some $C$ in $\Ccat$. 	
Let $g_0 = R\;\im_{I,o_1(C)}$. 	
Because $R$ preserves the initial segments and $B = R\;C$, we have 
that $o(B) = R\;o_1(C)$, hence $g_0 : T \ra o(B)$.
We must show that (i) $g_0$ is $r'$-definable and (ii) $g = m(B) \;\circ\; g_0$.   

To show (i), let $C' = F\;o_1(C)$.  From $R' \circ F = R_{\restr\Ccat_0}$, we have $R'\;C' = R'\;(F\;o_1(C)) = R\;o_1(C) = o(B)$.  
Moreover, by the initiality of $I'$ and the fact that $F\;I = I'$, we have $\im_{I',C'} = \;\im_{I',F\;o_1(C)} = F\;\im_{I,o_1(C)}$. 
Hence, using that $R' \circ F = R_{\restr\Ccat_0}$, we have $R'\;\im_{I',C'} = R'\;(F\;\im_{I,o_1(C)}) = R\;\im_{I,o_1(C)} = g_0$, which 
proves  
(i). 

Next we show (ii). 
By the initiality of $I$, we have $m_1(C) \,\circ \;\im_{I,o_1(C)}  = \;\im_{I,C}$; hence, by the functoriality of $R$, we have 
$R\;m_1(C) \circ R\;\im_{I,o_1(C)}  = R\;\im_{I,C}$. Hence, 
since $R$ preserves the initial segments which implies $R\;m_1(C) = m(R\;C) = m(B)$, we obtain $m(B) \circ g_0 = g$, as desired. 
\qed

\ \\
\textbf{More detailed proof of Thm.~\ref{thm-qexpr}.}     	
		When proving each 
		$r_i \wgeq r_j$, we instantiate Prop.~\ref{prop-WeakExtCriterion} 
		taking $r' = r_i$ and $r = r_j$. So here $\Bcat$ is the category of $\Sigmac$-models, $\Ccat'$ that of $(\Sigma_i,\Props_i)$-models, 
		and $\Ccat$ that of $(\Sigma_j,\Props_j)$-models; $R'$ 
		is the forgetful functor from $(\Sigma_i,\Props_i)$-models to $\Sigmac$-models, and $R$ the forgetful functor from $(\Sigma_j,\Props_j)$-models to $\Sigmac$-models; 
		$B=\TTrm(\Sigmac)$, 
		$I'=\TTrm(\Sigma_i)$ and $I=\TTrm(\Sigma_j)$.  
		
		We define the initial segment $(\Bcat_0,(m(\AA):o(\AA)\ra \AA)_{\AA \in \Obj{\Bcat}})$ of $\Bcat$ as follows: For any $\Sigmac$-model $\AA$ of carrier $A$ we take 
		$o(\AA)$ to be its minimal submodel (subalgebra), 
		i.e., the one generated by 
		$\Vr^\AA$, $\Ap^\AA$ and $\Lm^\AA$; we take $m(\AA): o(\AA) \ra \AA$ to be the inclusion morphism; and we take $\Bcat_0$ to be the full subcategory given by constructor-generated models. 
		Each time, we will define the initial segment $(\Ccat_0,\alb(m_1(\MM):o_1(\MM)\ra \MM)_{\MM \in \Obj{\Ccat}})$ 
		so that, for each $(\Sigma_j,\Props_j)$-model $\MM$, 
		$o_1(\MM)$ is a submodel of $\MM$ whose carrier is generated by the constructors ($\Vr^\MM$, $\Ap^\MM$ and $\Lm^\MM$) and will have the other operators from $\Sigma_j$ defined in specific ways; and $\Ccat_0$ will be the full subcategory given by the objects $o_1(\MM)$. 
		This way, it will be guaranteed that $R$ preserves initial segments. 
		\looseness=-1
		
		\textbf{Let us first consider the $\simeq$-chain going from $r_1$ to $r_6$.} 
		Since, by Thm.~\ref{thm-expr}, $r_6$ (the swap/fresh recursor), is the strongest w.r.t.\ $\geq$, and $r_1$ (the perm/free recursor) and $r_3$ (the swap/free variant recursor) are the weakest and are equivalent with each other, 
		and because the relation $\wgeq$ is weaker than $\geq$, 
		it suffices to prove $r_3 \wgeq r_6$.  
		
		\textbf{Proof of $r_3 \wgeq r_6$:} 
		To satisfy the conditions of Prop.~\ref{prop-WeakExtCriterion}, we need to define an initial segment of the category of $(\Sigma_6,\Props_6)$-models and an initial-segment-preserving pre-functor to the category of $\Sigmac$-models. 
		
		To achieve this, it turns out that the natural route goes through the properties of the swap/free recursor $r_4$---not only those in the axiom-lighter version that we are considering, but also the additional properties used by \citet{primrecFOAS-Norrish04}'s original recursor (\FvSw{}, \SwId{}, \SwIv{}).  
		We proceed as follows: Given a $(\Sigma_6,\Props_6)$-model $\MM$ of carrier $M$, we 
		define a submodel $\MM'$ of $\MM$ on the subset $M'$ of $M$ generated by the constructor-like operations $\Vr^\MM$, $\Ap^\MM$ and $\Lm^\MM$ 
		and prove that, modulo the translation between freshness and freeness (via negation), $\MM'$ is a $(\Sigma_4,\Props_4)$-model. 
		The operations on $\MM'$ should of course be inherited from the ones of $\MM$. 
		For this to be possible, we first need that $M'$ is closed under the $\Sigma_6$-operations. By definition it is closed under the constructor-like operations and, 
		thanks to $\MM$ satisfying \SwVr{}, \SwAp{} and \SwLm{}, it is also closed under the swapping-like operation.   Now, we take the freshness predicate on $\fresh^{\MM'}$ to be 
		defined in the style of nominal logic, i.e., by 
		\FrDSw{}: 
		 given $x\in\Var$ and $m\in M'$, we define 
		$x \,\fresh^{\MM'} m$ to mean that $m[x \wedge y]^{\MM'} \not= m$ for only a finite number of variables $y$.
		
		Alternatively, we could have defined $\fresh^{\MM'}$ inductively by the clauses \FrVr{}, \FrAp{} and \FrLm{}, which would make $\MM'$ the minimal submodel of $\MM$. Indeed, these two definitions will eventually turn out to be equivalent, but our proof needs to follow a delicate sequence of steps which require that we start with the nominal-like definition. For now, let us write $\ifresh^{\MM'}$ for this alternative, inductively defined freshness predicate.  Note that $\ifresh^{\MM'}$ is smaller than the restriction of $\fresh^{\MM}$ to $M'$. 
		
		To establish that $\MM'$ is a submodel of $\MM$, it remains to prove that, for items in $M'$, $\fresh^{\MM'}$ is smaller than $\fresh^{\MM}$. We will actually prove the stronger statement that $\fresh^{\MM'}$ is smaller than $\ifresh^{\MM'}$. The proof proceeds by expanding the definition of $\fresh^{\MM'}$, and needs that \FrDSw{} and \FrSw{} hold for 
		$\ifresh^{\MM'}\!$.  Both these last properties follow by induction on the definition of $\ifresh^{\MM'}\!$. (Note that a direct proof that $\fresh^{\MM'}$ is smaller than $\fresh^{\MM}$ would not work along the same lines, since neither \FrDSw{} and \FrSw{}  are guaranteed to hold for the restriction of $\fresh^{\MM}$ to $M'$; the tighter predicate $\ifresh^{\MM'}$ is actually needed.)
		
		So $\MM'$ is a submodel of $\MM$. We now need to prove that 
		$\MM'$ is a $(\Sigma_6,\Props_6)$-model. 
		$\MM'$ satisfies \SwVr{}, \SwAp{}, \SwLm{} because these are equations and its supermodel $\MM$ satisfies them.\footnote{Note that 
		we cannot say the same for \textsf{FrVr}, \textsf{FrAp} and \textsf{FrLm}.}
		That $\MM'$ satisfies \FrVr{} and \FrAp{} follows easily by applying the definition of $\fresh^{\MM'}$. 
		That $\MM'$ satisfies \SwCg{} follows immediately from $\MM$ satisfying \SwCg{}  and $\fresh^{\MM'}$ being included in $\fresh^{\MM}$. 

		For proving the next facts, we need to make heavy use of the fact that $\MM'$ satisfies the structural properties of swapping, \SwId{}, \SwIv{} and \SwCp{}. 
		All these three follow by easy induction on the definition of $\MM'$ and the fact that $\MM'$ satisfies \SwVr{}, \SwAp{}, \SwLm{}. 
		
		To prove that $\MM'$ satisfies \FrLm{}, we first prove that $x \,\fresh^{\MM'}\, \Lm^{\MM}x\,d$ holds for all $d\in M'$, in particular, that $\MM'$ satisfies \FCB{}.  
		A prerequisite for the latter is that 
		$\MM'$ satisfies \FrSw{}, which 
		follows from the definition of $\fresh^{\MM'}$ together 
		with the fact that $\MM'$ satisfies \SwCp{} and \SwIv{}. 
		Now, that $\MM'$ satisfies \FrLm{}  almost (but not quite) follows 
		from 
		\SwCg{} (for $\fresh^{\MM'}$); we actually need a stronger version of \SwCg{} to hold for $\fresh^{\MM'}$, namely: For all $x,x',z\in\Var$ and $m,m'\in M'$, if ($z=x$ or $z\fresh^{\MM'} m$), ($z=x'$ or $z\fresh^{\MM'} m'$) 
		and $m[z \sw x]^{\MM'} = m[z \sw x']^{\MM'} $, then $\Lm^{\MM'}x\,m = \Lm^{\MM'}x'\,m'$. 
		To prove the latter, we need to cover the cases not covered by \SwCg{}, namely 
		(1) that of $z=x=x'$, which follows from the fact that $\MM'$ satisfies \SwId{} and 
		(2) that of $z =x$ and $z\,\fresh^{\MM'} m'$ (and a case symmetric to it), 
		which needs that $\MM'$ satisfies 
		\SwFr{} and 
		\FrDSw{}, i.e., that \FrDSw{} holds for $\fresh^{\MM'}$ 
		(whereas so far we only know that it holds for $\ifresh^{\MM'}$). 
		That $\MM'$ satisfies 
		\SwFr{} follows from the definition of $\fresh^{\MM'}$ together 
		with the fact that $\MM'$ satisfies \SwCp{}. 
		That $\MM'$ satisfies 
		\FrDSw{} follows from the fact that (3) \FrDSw{} holds for
		$\fresh$ (the predicate on terms), and that (4) the unique $\Sigma_6$-morphism $g: \TTrm(\Sigma_6) \ra \MM'$ guaranteed by the initiality of $\TTrm(\Sigma_6)$ 
		has its image included in $M'$ and it turns out to preserve freshness not only in the form ``$x\fresh t$ implies $x\fresh^{\MM'}g\,t$'', but in the stronger form 
		``$x\fresh t$ iff $x\fresh^{\MM'}g\,t$''.  Fact (4) follows from applying 
		\FrDSw{} to both $\fresh$ and $\fresh^{\MM'}$ and using that $g$ commutes with swapping. 
		
		In summary, along the above route, we proved that $\MM'$ is a submodel of $\MM$ that satisfies the properties in $\Props_6$, meaning that the inclusion function is a morphism between $\MM'$ and $\MM$ in the category of $(\Sigma_6,\Props_6)$-models, as desired. On the way, we also proved that 
		$\MM'$ satisfies \FCB{}, \SwIv{}, \SwId{}, \SwCp{}, \FrSw{}, \SwFr{} and 
		\FvDSw{}. 
		So, in the notations of Prop.~\ref{prop-WeakExtCriterion}. we take $o_1(\MM)$ to be $\MM'$ and we take $m_1(\MM): o(\MM) \ra \MM$ to be the inclusion morphism.  The forgetful functor between $(\Sigma_6,\Props_6)$-models and $\Sigmac$-models ($R$ in Prop.~\ref{prop-WeakExtCriterion}'s notations) is easily seen to be initial-segment preserving. 
		
		We define Prop.~\ref{prop-WeakExtCriterion}'s operator $F$ to take any model $\MM$ to the model $F\;\MM$ that replaces $\fresh^{\MM}$ with $\FV^{\MM}$ standardly defined from $\fresh^{\MM}$ and to be the identity on morphisms. 
		We must show that $F\;\MM'$ satisfies the properties in $\Props_3$. This is true, because \SwVr{}, \SwAp{} and \SwLm{} are already in $\Props_6$ and all the other properties in $\Props_3$ have already been proved above. 
		It is easily seen that $F$ is a pre-functor (a functor actually) such that $R' \circ F = R_{\restr\Ccat_0}$ and 
		$F\;\TTrm(\Sigma_6) = \TTrm(\Sigma_3)$, as required by Prop.~\ref{prop-WeakExtCriterion}. 
		This concludes the proof of $r_3 \wgeq r_6$.%
	
		(Note that, on the way to proving that our constructed submodel $\MM'$ satisfies $\Props_6$
			and (via the freshness to free-variable translation) satisfies $\Props_3$, we actually proved that $\MM'$ satisfies (again via the translation) all properties of $\Props_4$ as well. So a natural question is whether the above-sketched fairly intricate proof could not be made more modular along a $r_3 \wgeq r_4 \wgeq r_6$ relationship, i.e., split into (1) a proof that any $\Props_6$-model produces a $\Props_4$-submodel and (2) any $\Props_4$-model produces a $\Props_3$-submodel. In particular, for proving (1) one may hope to avoid defining the nominal-style freshness operator and work with $\ifresh^{\MM}$ instead. After some trial and error, we came to believe that the answer to the above questions is 
			`no'. It seems that the route through nominal-style freshness is necessary even for constructing a $\Props_4$-submodel. In particular, using 
			$\ifresh^{\MM}$ works for everything except for proving the satisfaction of 
			\SwFr{} (translated as \SwFv{}). Indeed, \SwFv{} has \emph{two} hypotheses involving freshness, and we must induct on one of them; and, unlike when we work in the term model (where the constructors are ``almost free''), here we do not have any well-behaved inversion rules 
			corresponding to \FrVr{}, \FrAp{} and \FrLm{} 
			to apply to the other freshness hypothesis,  which seem necessary to make the proof go through; for example, we cannot infer $z\not= x$ from $z \,\ifresh^{\MM}\, \Vr\,x$.  
			In conclusion, if we want to build a submodel satisfying $\Props_4$, we seem to actually 
			need to build one that satisfies the  stronger properties $\Props_3$.)  

		\textbf{Let us now consider the $(\simeq,\wgeq)$-chain going from $r_{6}$ to $r_7$.} 
		Since, by Thm.~\ref{thm-expr}, we have that $r_9 \geq r_8,r_7$, 
		and again since $\wgeq$ is weaker than $\geq$, all we have to prove are two relationships: 
		$r_8 \wgeq r_9$ and 
		$r_6 \wgeq r_8$.
		
		\textbf{Proof of $r_8\wgeq r_9$:}  
		Similarly to the previous proof, 
		given a $(\Sigma_{9},\Props_{9})$-model $\MM$ of carrier $M$, we 
		define a submodel $\MM'$ of $\MM$ on the subset $M'$ of $M$ generated by the constructor-like operations $\Vr^\MM$, $\Ap^\MM$ and $\Lm^\MM$. 

		Then we prove that $f: \TTrm(\Sigma_{9}) \ra \MM$, the unique $\Sigma_{9}$-morphism ensured by the initiality of  
		$\TTrm(\Sigma_{9})$, has its image equal to $M'$, i.e., is also a morphism between $\TTrm(\Sigma_{9})$ and $\MM'$. (The proof goes smoothly: one direction by induction on terms using depth (size) as measure and the other direction by induction on the definition of $M'$.) This connection between $\TTrm(\Sigma_{9})$ and $\MM'$ will allows us to ``borrow'' any equation from terms to elements of $\MM'$; in what follows, we will refer to it as ``the term-model connection''. 
		
		So we define $M'$ to be the subset of $M$ (the carrier of $\MM$) generated by the constructor-like operations. 
		To organize $M'$ into a submodel $\MM'$ of $\MM$, we need that $M'$ is closed under the renaming operator $\_[\_\,/^{\MM'}\!\!\_]$, which follows from the term-model connection. Now we can define the model $\MM'$ to be formed on $M'$ using the restrictions of the $\MM$ operations and of the $\MM$ freshness predicate. 
		By virtue of being a full submodel of $\MM'$ (i.e., a submodel where freshness is the restriction to $M'$ of 
		the freshness predicate from the supermodel, and not just a subset of it)
		and all the properties in  $\Props_{9}$ being Horn clauses, it follows that $\MM'$ is a $(\Sigma_{9},\Props_{9})$-model. So we take, as before, $o_1(\MM) = \MM'$ and $m_1(\MM)$ to be the inclusion morphism, and the forgetful functor is easily seen to preserve the initial segments. 
		
		We define $F$ by taking $F\;\MM'$ to be the restriction of $\MM'$ obtained by forgetting the freshness operator, and taking $F$ on morphisms to be the identity. Since all the properties of $\Props_8$ are unconditional equations, we can prove that $F\,\MM'$ satisfies all of them along the term-model connection, from the corresponding properties on terms. 
		Finally, (as before) it is easily seen that $F$ is a functor such that $R' \circ F = R_{\restr\Ccat_0}$ and 
		and 
		$F\;\TTrm(\Sigma_{9}) = \TTrm(\Sigma_8)$. 
		This concludes the proof of $r_8 \wgeq r_9$.  
		
		\textbf{Proof of $r_6 \wgeq r_8$:}  
		Similarly to the previous cases, 
		given a $(\Sigma_8,\Props_8)$-model $\MM$ of carrier $M$, we 
		define a submodel $\MM'$ of $\MM$ on the subset $M'$ of $M$ generated by the constructor-like operations $\Vr^\MM$, $\Ap^\MM$ and $\Lm^\MM$. 
		%
		%
		Like in the previous proof, we establish a term-model connection, which immediately gives us that $\MM'$ is a submodel of $\MM$; and $\MM'$ satisfies $\Props_8$ because its supermodel $\MM$ does and $\Props_8$ consists of equations. 
		
		It remains to define a swapping-like operator on $\MM'$ and prove that it forms a $(\Sigma_6,\Props_6)$-model. To this end, we first define 
		a freshness-like operator on $M'$ in the style of nominal-logic, but using renaming rather than swapping \cite{DBLP:conf/cade/Popescu22}: $x\,\fresh^{\MM'}\,m$ iff $\{y \mid m[y/^{\MM'}x] \not= m\}$ is finite. 
		Next we prove that $\fresh^{\MM'}$ has finite support---by induction on the definition of $M'$. This allows us to define a ``$\pickFresh$'' 
		operator that, given any lists of variables and elements of $M'$, produces a variable that is fresh for all items in the lists. 
		Now, a swapping-like operator $\_[\_\wedge^{\MM'}\!\!\_]$ is defined as follows: 
		$$
		\begin{array}{c} 
		d[z_1 \wedge^{\MM'} z_2] = 
		d\,[y /^{\MM'} z_1]\,[z_1 /^{\MM'} z_2]\,[z_2 /^{\MM'} y]
		\\
		\mbox{ \ where $y = \pickFresh\;[z_1,z_2]\;[d]$} 
		\end{array} 
		$$
		(So we define $F\,\MM'$ to be the $\Sigma_6$-model obtained from $\MM'$ by replacing $\_[\_\,/^{\MM'}\!\!\_]$ with $\_[\_\wedge^{\MM'}\!\!\_]$ in $\MM'$, and we take $F$ to be the identity on morphisms.) 
		In other words, the definition proceeds just like we would define swapping from renaming on terms. 
		Now the only question is whether we have enough assumptions on our abstract $(\Sigma_8,\Props_8)$-model $\MM'$ in order to infer the properties of swapping from those of substitutions, like we could for concrete terms (hence prove that $F\,\MM'$ 
		satisfies $\Props_6$). 

		The rest of our proof consists of building a positive answer to this question. 
		To help with the proofs, we first establish a fresh induction principle for $\MM'$ similar to the nominal-logic one for terms \cite{pitts-AlphaStructural}. While this takes us a long way, it is not able to prove the following property stating that the definition of swapping is independent of the chosen fresh representative (which in turn is crucial for proving that swapping commutes with $\Ap$ and $\Lm$, i.e., that $\SwAp$ and $\SwLm$ hold for $\F\;\MM'$), namely: 
		\begin{center}
			For all $m\in M'$ and $y,y',z_1,z_2 \in \Var$,  
			\\
			if $y,y' \not\in \{z_1,z_2\}$ and $y,y' \fresh^{\MM'} m$, then 
			$$
			\begin{array}{c} 
			d\,[y /^{\MM'} z_1]\,[z_1 /^{\MM'} z_2]\,[z_2 /^{\MM'} y] \;=
			\\
			d\,[y' /^{\MM'} z_1]\,[z_1 /^{\MM'} z_2]\,[z_2 /^{\MM'} y'] \phantom{\;=}
			\end{array} 
			$$
		\end{center}  
		The reason why fresh induction on $\MM'$ cannot prove this property (unlike fresh induction on terms which could prove its term counterpart) is that, because of the freshness assumptions $y,y' \fresh^{\MM'} m$, we would need to apply inversion rules for freshness w.r.t.\ constructors (e.g., infer 
		$y \,\fresh^{\MM'} m_1$ from $y \,\fresh^{\MM'} \Ap^{\MM'}m_1m_2$ ), which hold for terms but not for $\MM'$. And doing some kind of induction on the freshness assumption (after proving the minimality of $\fresh^{\MM'}$) does not help either, since there are two such assumptions, and one of them would still need an inversion rule. 
		
		To deal with this difficulty, we rephrase the above definition by eliminating freshness completely.  Namely, first we prove that (*) $y\, \fresh^{\MM'} m[u/y]$ whenever $u \not= y$. This allows us to rephrase the above choice-independence property as an equation:\footnote{We display this highlighting the main additions, and crossing out what has been removed.}
		\begin{center} 
			For all $m\in M'$ and $y,y',z_1,z_2,\hlt{u} \in \Var$, 
			\\
			if $y,y' \not\in \{z_1,z_2,u\}$  
			\sout{and $y,y' \fresh^{\MM'} m$}, then 
			$$
			\begin{array}{c}
			d\,\hlt{[u /^{\MM'} y]}\,\hlt{[u /^{\MM'} y']}\,[y /^{\MM'} z_1]\,[z_1 /^{\MM'} z_2]\,[z_2 /^{\MM'} y] \\=\\
			d\,\hlt{[u /^{\MM'} y]}\,\hlt{[u /^{\MM'} y']}\,[y' /^{\MM'} z_1]\,[z_1 /^{\MM'} z_2]\,[z_2 /^{\MM'} y'] 
			\end{array} 
			$$
		\end{center} 
		
		Indeed, the previous version follows from this one, using (*). Now, this version \emph{can} be inferred from the term-model connection, using the corresponding property for terms. 
		
		After overcoming this difficulty, the desired properties in $\Props_6$ all follow smoothly either using the term-model connection or by fresh induction. Again, it is easily seen that $F$ is a functor such that $R' \circ F = R_{\restr\Ccat_0}$ and 
		$F\;\TTrm(\Sigma_8) = \TTrm(\Sigma_6)$. This concludes the proof of $r_6 \wgeq r_8$.  
		\qed

\ \\
\textbf{Proof of Prop.~\ref{prop-negRec}.}  
Recall that $\Sigmac = \{\vr,\ap,\lm\}$ is the constructor signature and all 
signatures $\Sigma_i$ extend $\Sigmac$. For proving $r_i \not\geq r_j$, we must provide a $(\Sigma_j,\Props_j)$-model $\MM$ for which the $\Sigmac$-reduct (i.e., the $\Sigmac$-model obtained by forgetting the operators from $\Sigma_j \setminus \Sigmac$) cannot be the $\Sigmac$-reduct of any $(\Sigma_i,\Props_i)$-model. 
Indeed, this would mean exhibiting a $\Sigmac$-morphism 
(namely the reduct of the unique morphism defined by $r_j$ using $\MM$) 
that is $r_j$-definable but not $r_i$-definable. 

\smallskip
We write $\Var^\Nat$ for the set of streams of variables, i.e., families $\xs = (\xs_i)_{i\in \Nat}$ with $\xs_i \in \Var$. Given $\xs \in \Var^\Nat$, $y\in\Var$ and $\sigma\in\Perm$, we write: 
$\Vars\;\xs$ for the set of all variables appearing in $\xs$, $\{\xs_i \mid i \in \Nat\}$; 
$\rem_{y}\,\xs$ for the stream obtained from $\xs$ by removing from it all occurrences of $y$; 
$\map_\sigma\,\xs$ for the stream obtained by mapping $\sigma$ on $\xs$, 
$(\sigma(\xs_i))_{i\in \Nat}$. 

To prove $r_1 \not\geq r_2$, recall that $\Sigma_1 = \Sigma_2 = \Sigmac \cup \{\pm,\fv\}$. We take the $(\Sigma_2,\Props_2)$-model $\MM$ to have as carrier the set $M = \Trm \cup A$, where $A = \{\xs \in \Var^\Nat \mid \Vars\;\xs \,\mbox{ is  infinite}\}$. Note that both $\rem_y$ and $\map_\sigma$ (for $\sigma\in \Perm$) preserve the property that $\Vars\;\xs$ is infinite. Let $(t_i)_{i\in \Nat}$ be a family of terms such that all are ground ($\FV\;t_i = \emptyset$)
and mutually distinct ($i\not=j$ implies $t_i \not=t_j$). (For example, we can take $t_i$ to be $(\Lm\;x)^i\,(\Vr\;x)$ for some fixed variable $x$.)

We define $\MM$'s operators on $M$ by extending the standard term operators from $\Trm$ as follows, for any $\xs \in A$: 
\\
\hspace*{-0.5ex}
\begin{tabular}{cc}
	\begin{tabular}{l} 
		$\bullet$ $\FV^\MM \xs = \Vars\;\xs$
		\\
		$\bullet$ $\Lm^\MM\,y\;\xs = \rem_y\,\xs$ for any $y\in\Var$
		\\
		$\bullet$ $\Ap^\MM\;\xs\;t_i = \Vr\;\xs_i$ for any $i\in \Nat$  
	\end{tabular} 
	&
	\hspace*{-0.5ex}
	\begin{tabular}{l} 
		$\bullet$ $\Ap^\MM\;\xs\;m = t_0$ for any $m \in M \sm \{t_i \mid i \in \Nat\}$
		\\
		$\bullet$ $\Ap^\MM\;s\;\xs = t_0$ for any $s \in \Trm$ 
		\\
		$\bullet$ $\xs[\sigma]^\MM = \map_\sigma\,\xs$ for any $\sigma\in\Perm$
	\end{tabular}
\end{tabular} 

Note that, on $A$, the free-variable-like and abstraction-like operators are 
 natural, in particular $\Lm^\MM$ removes all occurrences of the abstracted variable. On the other hand, the application-like operator is contrived: the only interesting case is  $\Ap^\MM\;\xs\;t_i$, where application emulates the $i$'th projection, retrieving the $i$'th element of the stream $\xs$; in the other cases application simply returns the ground term $t_0$.  We can check that $\MM$ thus defined satisfies the $\Props_2$ properties (they are known to hold form terms, so it remains to check these properties when elements $\xs$ of $A$ are involved):
\begin{itemize}
	\item \PmId{}, \PmCp{} and \PmFv{} hold on $A$ thanks to standard properties of $\map$ and $\Vars$ for streams. 
	\item For \PmVr{} and \FvVr{} there is nothing to check because they does not involve elements of $A$.
	\item \PmLm{} on $A$ 
	means $\map_\sigma(\rem_y\,\xs) = \rem_{(\sigma\,y)}(\map_\sigma \xs)$, again a standard property on streams. 
	\item To check \PmAp{}, we distinguish between three cases (according to the above definition of $\Ap^\MM$ on arguments involving items from $A$):
	\begin{itemize}
		\item $\Ap^\MM\,(\xs[\sigma]^\MM)\;(t_i[\sigma]^\MM) = 
		\Ap^\MM\,(\map_\sigma\,\xs)\;t_i = \Vr\,(\map_\sigma\,\xs)_i = \Vr\,(\sigma\;\xs_{i}) = (\Vr\,x_i) [\sigma] = (\Vr\,x_i) [\sigma]^\MM = 
		(\Ap^\MM\,\xs\,t_i)[\sigma]^\MM$. 
		\item  Assume $m \in M \sm \{t_i \mid i \in \Nat\}$, and note that we also have 
		$m[\sigma]^\MM \in M \sm \{t_i \mid i \in \Nat\}$. Then:
		$\Ap^\MM\,(\xs[\sigma]^\MM)\;(m[\sigma]^\MM)  = t_0 = t_0[\sigma] = t_0[\sigma]^\MM = (\Ap^\MM\,\xs\;m)[\sigma]^\MM$. 
		\item $\Ap^\MM\,(s[\sigma]^\MM)\;(\xs[\sigma]^\MM) = \Ap^\MM\,(s[\sigma])\;(\xs[\sigma]^\MM) = 
		t_0 = t_0[\sigma] = t_0[\sigma]^\MM = $ \\$(\Ap^\MM\,s\;\xs)[\sigma]^\MM$. 
	\end{itemize} 
	\item \FvLm{} on $A$ means $\Vars\,(\rem_{y}\,\xs) \su \Vars\;\xs \sm \{y\}$, which actually 
	holds for streams in equality form.  
	\item Finally, \FvAp{} holds trivially in cases where $\Ap^\MM$ returns $t_0$, since $\FV\,t_0 = \emptyset$; and for $\Ap^\MM\,\xs\,t_i$, it amounts to $x_i \in \Vars\;\xs$.  
	\looseness=-1
\end{itemize}

It remains to check that the $\Sigmac$-reduct of $\MM$, i.e., $\Trm \cup A$ equipped with the above-defined constructor-like operators, cannot be the reduct of any $(\Sigma_1,\Props_1)$-model, i.e., there is no way to define the operators 
$\_[\_]'$ and $\FV'$ on $\Trm \cup A$ that, together with $\Vr^\MM$, $\Ap^\MM$ and $\Ap^\MM$, make it a $(\Sigma_1,\Props_1)$-model.  So let us assume otherwise, 
i.e., that such operators $\_[\_]' : (\Trm \cup A) \ra \Perm \ra (\Trm \cup A)$ and $\FV' : (\Trm \cup A) \ra \Pow(\Var)$ exist, 
and reach a contradiction. 

We note that $\_[\_]'$ is uniquely determined on $\Trm$ because of \PmVr{}, \PmAp{} and \PmLm{}, so on $\Trm{}$ it must coincide with the standard permutation operator. 
Moreover, for each $\xs \in A$, thanks to \PmAp{} used for $\Ap^\MM\,\xs\,t_i$ we have that, for any $i$, $(\xs[\sigma]')_i = \sigma\;\xs_i$, meaning that $\_[\sigma]'$ must be $\map_\sigma$ on $A$. In other words, $\_[\_]'$ must be the same as $\_[\_]^\MM$. 

Now let $\xs$ be any element of $A$. We will show that $\FV'\,\xs$ is necessarily the entire set of variables $\Var$. Let $y\in \Var$. If $y \in \Vars\;\xs$, then $y=\xs_i$ for some $i$, hence using \FvAp{} for $\Ap^\MM\,\xs\;t_i$ we obtain $y \in \FV'\;\xs$. 
Now assume $y \in \Var \sm \Vars\;\xs$. Then, for any $z\in \Vars\;\xs$, $\xs[z \llra y]'$, i.e., 
$\map_{z \llra y}\,\xs$, is different from $\xs$. Thus, $\xs[z \llra y]' \not= \xs$ 
for an infinite number of variables $z$, which by \FvDPm{} implies $y \in \FV'\,\xs$.
 
We thus showed that $\FV'\,\xs = \Var$ for any $\xs\in A$. But this contradicts \FCB{}, according to which we must have $y \notin \FV'\;(\Lm^\MM\,y\;\xs)$.

\medskip
To prove $r_2 \not\geq r_4$, recall that $\Sigma_2 = \Sigmac \cup \{\pm,\fv\}$ and $\Sigma_4 = \Sigmac \cup \{\swp,\fv\}$. We take the $(\Sigma_4,\Props_4)$-model $\MM$ to have as carrier the set $M = \Trm \cup \{a\}$ (where $a\not\in \Trm$), i.e., to consist of terms plus an additional element $a$. Let $x$ be a fixed variable. We define $\MM$'s operators on $M$ by extending the standard term operators from $\Trm$ as follows: 
\\
\hspace*{-2.5ex}
\begin{tabular}{cc}
	\begin{tabular}{l} 
		$\bullet$ $\FV^\MM a = \Var$ (the set of all variables)
		\\
		$\bullet$ $\Lm^\MM\,y\;a = \Lm^\MM\,y\;(\Vr\;x)$ for any $y\in\Var$
		\\
		$\bullet$ $\Ap^\MM\;a\;a = \Ap^\MM\;(\Vr\;x)\;(\Vr\;x)$ 
	\end{tabular} 
	&
	\hspace*{-3.5ex}
	\begin{tabular}{l}
		$\bullet$ $\Ap^\MM\;a\;t = \Ap^\MM\;(\Vr\;x)\;t$ for any $t\in\Trm$
		\\
		$\bullet$ $\Ap^\MM\;t\;a = \Ap^\MM\;t\;(\Vr\;x)$ for any $t\in\Trm$
		\\
		$\bullet$ $a[z_1 \hspace*{-0.2ex}\wedge \hspace*{-0.2ex} z_2]^\MM = \Vr\,(x[z_1 \hspace*{-0.2ex}\wedge
		\hspace*{-0.2ex}
		z_2])$ for any $z_1,z_2\in\Var$
	\end{tabular}
\end{tabular} 

(Thus, the free variables of $a$ are the entire set of variables, and the constructor and swapping operators on $a$ yield the same results as for $\Vr\;x$, i.e., have $\Vr\;x$ act in lieu of $a$.) 

We can check that $\MM$ satisfies $\Props_4$ (they are known to hold for terms, so we only need to check these properties when $a$ is involved):
\FvVr{}, \FvAp{} and \FvLm{} immediately hold thanks to $\FV^\MM\;a$ being $
\Var$; 
moreover, 
\SwVr{}, \SwAp{} and \SwLm{} hold because any application of constructor or swapping operator turns $a$ into $\Vr\;x$; finally, 
\SwFv{} holds trivially for $a$, since, $\FV^\MM a$ being $\Var$, the hypothesis of \SwFv{} is vacuously false.
\looseness=-1

It remains to check that the $\Sigmac$-reduct of $\MM$, i.e., $\Trm \cup \{a\}$ equipped with the above-defined constructor-like operators, cannot be the reduct of any $(\Sigma_2,\Props_2)$-model, i.e., there is no way to define the operators $\_[\_]'$ and $\FV'$ on $\Trm \cup \{a\}$ that, together with $\Vr^\MM$, $\Ap^\MM$ and $\Ap^\MM$, make it a $(\Sigma_2,\Props_2)$-model.  
So let us assume otherwise, 
i.e., that such operators $\_[\_]'$ and $\FV'$ exist, 
and reach a contradiction. 

We note that $\_[\_]'$ is uniquely determined on $\Trm$ because of \PmVr{}, \PmAp{} and \PmLm{}, so on $\Trm{}$ it must coincide with the standard permutation operator. Moreover, 
\PmId{} and \PmCp{} 
imply that $\_[\sigma]'$ is bijective on $\Trm \cup \{a\}$ for any 
permutation $\sigma$. 
Hence, because the restriction $\_[\sigma]'$ to $\Trm$ is also a bijection on $\Trm$ (being equal to the standard permutation operator), the only possibility is that $a[\sigma]' = a$ for any $\sigma$. But this stands in contradiction with \PmAp{}, because together with \PmAp{} it would imply that $(\Ap^\MM\,a\;a)[\sigma]' = \Ap^\MM\,(a[\sigma]')\;(a[\sigma]') = \Ap^\MM\,a\;a$, i.e., $(\Ap\;(\Vr\;x)\;(\Vr\;x))[\sigma] = \Ap\;(\Vr\;x)\;(\Vr\;x)$, 
which is false for any $\sigma$ that modifies $x$. 
\qed 
%

\ \\
\subsection{Back to the proof of the recursion theorems} 
\label{app-subsec-backToProofRecThms} 

The heart of the proof of an epi-recursion principle, 
i.e., of the fact that a tuple $r = (\Bcat,T,\Ccat,I,R)$ forms an epi-recursor, 
is a proof of initiality, namely the initiality of the object $I$ in the category $\Ccat$. And indeed, this is the difficult part in the proof of all the nominal recursors listed in Thm~\ref{thm-allNominalRecs}. 

Next, we show how we can take advantage of the expressiveness comparisons 
to ``borrow'' a (quasi)weaker recursion principle from a (quasi)stronger one, and to infer all nominal recursors from only two of them---those located at the top of the expressiveness hierarchy. 

Indeed, the idea behind our expressiveness comparison criteria (Props.~\ref{prop-extCriterion} and \ref{prop-WeakExtCriterion}) 
has been the possibility of one recursor to \emph{simulate the behavior of another 
recursor}, so it feels natural to use this idea for ``borrowing'' purposes. 
To this end, we first introduce pre-epi-recursors, which are epi-recursors without the initiality condition, and a possible property of them called tightness: 

\medskip 
\begin{defi}\rm 
	A \emph{pre-epi-recursor} is a tuple $r = (\Bcat,T,\Ccat,I,R)$ subject to the same condition as an epi-recursor (Def.~\ref{defi-epirec}), but without the requirement that $I$ is the initial object of $\Ccat$. 
	
	A pre-epi-recursor is called \emph{tight}  
	if the following hold:
	\begin{myitem} 
		\item $T$ is a quasi-initial object in $\Bcat$ (in that for every object $B$ in $\Bcat$ there exists at most one morphism from $T$ to $B$). 
		\item The functor $R$ is faithful (in that it is injective on morphisms).   	\qed 
	\end{myitem} 
\end{defi}
\medskip

All the nominal pre-epi-recursors we discussed in this paper are tight, because $T$ is a quotient of  the term algebra (known to be quasi-initial)
and the functor $R$ is the identity on morphisms. 

Now, to make the borrowing possible, 
we take advantage of the fact that the definitions of $\geq$ and $\wgeq$ make sense, and also  
Props.~\ref{prop-extCriterion} and \ref{prop-WeakExtCriterion} hold, 
not only for epi-recursors, but also for pre-epi-recursors. 

\medskip 
\begin{prop}\rm \label{prop-borrow} 
	Assume the following:
	\begin{itemize}
		\item $(\Bcat,T,\Ccat,I,R)$ is a tight pre-epi-recursor
		\item $r'= (\Bcat,T,\Ccat',I',R') $ is an epi-recursor
		\item The hypotheses of Prop.~\ref{prop-WeakExtCriterion} hold, 
		namely we assume: 
		\begin{itemize} 
			\item a pre-functor $F : \Ccat_0 \ra \Ccat'$, 
			\item an initial segment $(\Bcat_0,(m(B):o(B)\ra B)_{B \in \Obj{\Bcat}})$ of $\Bcat$, 
			\item an initial segment $(\Ccat_0,\alb(m_1(C):o_1(C)\ra C)_{C \in \Obj{\Ccat}})$ of $\Ccat$, 
		\end{itemize}  		
		such that $\Ccat_0$ contains $I$,  
		$R$ preserves the above initial segments, 
		$F\;I = I'$ and $R' \circ F = R_{\restr\Ccat_0}$ 
		(where $R_{\restr\Ccat_0}$ is the restriction of $R$ to $\Ccat_0$).   
    	\item We additionally assume that the pre-functor $F$ is full. 
	\end{itemize}
	Then $r$ is an epi-recursor (i.e., $I$ is initial).  
\end{prop} 
\medskip
\begin{proof}
	We need to prove that $I$ is initial. To this end, let $C$ be an object in $\Ccat$. 
	Since $F$ is full and $I'=F\;I$, we obtain $h : I \ra o_1(C)$ 
	such that $F\;h = \;\im_{I',F\;o_1(C)}$.  We define $!_{I,C}$ to be $m_1(C) \circ h$.
	
	It remains to prove the uniqueness of $!_{I,C}$ as a morphism from $I$ to $C$. To this end, let $k : I \ra C$. 
	Then both $R\;!_{I,C}$ and $R\;k$ are morphisms between $R\;I = T$ and $R\;C$, hence $R\;!_{I,C} = R\;k$ by the quasi-initiality of $T$. Finally, 
	$!_{I,C} = k$ follows from the faithfulness of $R$.  \qed
	%
\end{proof}
\medskip

\medskip 
\begin{prop}\rm \label{prop-borrowStrong} 
	Assume the following:
	\begin{itemize}
		\item $(\Bcat,T,\Ccat,I,R)$ is a tight pre-epi-recursor
		\item $r'= (\Bcat,T,\Ccat',I',R') $ is an epi-recursor
		\item The hypotheses of Prop.~\ref{prop-extCriterion} hold, namely 
		we assume a pre-functor 
		$F : \Ccat \ra \Ccat'$  
		such that 
		$R' \circ F = R$ and 
		$F\;I = I'$.   
		\item We additionally assume that the pre-functor $F$ is full. 
	\end{itemize}
	Then $r$ is an epi-recursor (i.e., $I$ is initial).  
\end{prop} 
\medskip
\begin{proof}
	This already follows from Prop.~\ref{prop-borrow}, because 
	the hypotheses of Prop.~\ref{prop-extCriterion} are stronger than those of 
	Prop.~\ref{prop-WeakExtCriterion}. However, let us also give a direct proof:
	
	We need to prove that $I$ is initial. To this end, let $C$ be an object in $\Ccat$. 
	Since $F$ is full and $I'=F\;I$, we obtain $h : I \ra C$ 
	such that $F\;h = \;\im_{I',F\;C}$.  We define $!_{I,C}$ to be $h$.
	
   	It remains to prove the uniqueness of $!_{I,C}$ as a morphism from $I$ to $C$.  
   	To this end, let $k : I \ra C$. 
	Then both $R\;!_{I,C}$ and $R\;k$ are morphisms between $R\;I = T$ and $R\;C$, hence $R\;!_{I,C} = R\;k$ by the quasi-initiality of $T$. Finally, 
	$!_{I,C} = k$ follows from the faithfulness of $R$.  \qed

\end{proof} 
\medskip 

Taking advantage of Prop.~\ref{prop-borrowStrong}, our proof of Thm.~\ref{thm-allNominalRecs} goes along the following route:
\begin{itemize}
	\item We prove that $r_1$--$r_9$ are all tight pre-epi-recursors, which is immediate. 
	\item We prove that $r_6$ and $r_9$
	are epi-recursors, i.e., we do direct initiality proofs for $r_6$ and $r_9$ (as sketched in the main paper). 
	\item We use the fact that, according to Thm.~\ref{thm-expr}, $r_6$ and $r_9$ are stronger than all the others---and they are so by virtue of the criterion in Prop.~\ref{prop-extCriterion}. 
	\item We verify the additional hypothesis of Prop.~\ref{prop-borrowStrong}, namely that $F$ is full in all cases (which is again immediate).  Finally, we apply Prop.~\ref{prop-borrowStrong} to borrow all the other recursion theorems from those of $r_6$ and $r_9$. 
\end{itemize}
%

\section{Adding (Back) Enhancements to the Recursors} 
\label{app-addingBacknhancements}

It was convenient to discuss and compare the expressiveness of the stripped down versions of the recursors---since their enhancements, while useful, require some heavier notation that can clutter the main ideas.  Here, we add back the enhancements and 
show how our results generalize to cover the enhanced recursors. 

Recall from \S\ref{subsec-nominalRec} that the enhancements referred to support for the Barendregt variable convention and for full-fledged (primitive) recursion. 
In the case of full-fledged recursion, we noted that different degrees of support are possible: the additional term parameters  can affect constructor only, or the other operations as well (as we have seen with the swap/fresh and subst/fresh recursors); and they can be optimized for the freeness operator (as we have seen with the swap/free recursor).  
The table in Fig.~\ref{fig-recFeatures} summarizes the situation. 


\begin{figure}
	\centering
	\begin{tabular}{l| c | c }
		\begin{tabular}{c}
			Existing \\recursor
		\end{tabular}
		& 
		\begin{tabular}{c}
			Full-fledged \\recursion? 
		\end{tabular}
		& 
		\begin{tabular}{c}
			Barendregt \\convention? 
		\end{tabular}
		\\\hline
		Perm/free
		& No  & Yes 
		\\
		Swap/free 
		& 
		\begin{tabular}{c}
			For constructors, \\
			freeness-optimized  
		\end{tabular}
		& Yes 
		\\
		Swap/fresh 
		& 
		\begin{tabular}{c}
			For all \\operators  
		\end{tabular}
		& No
		\\
		Subst/fresh 
		& 
		\begin{tabular}{c}
			For all \\operators  
		\end{tabular}
		& No
		\\
		Renaming  
		& No  & Yes 
	\end{tabular}
	\caption{Enhancements exhibited by nominal recursors from the literature} 
	\label{fig-recFeatures}
\end{figure}

It turns out that 
all the nominal recursors have the following in common: 
\begin{myitem} 
\item[(a)] all these enhancements work on all the recursors, and 
\item[(b)] the enhanced versions can still be presented as epi-recursprs (for suitably chosen categories of models) and their expressiveness comparisons discussed in the main paper still apply. 
\end{myitem} 

Concerning point (a), 
we find it quite remarkable that the Barendregt convention enhancement can be applied democratically to recursors based on swapping/permutation and renaming/substitution alike, and also does not discriminate based on the particular axiomatization.  
Concerning the different degrees of ful-fledged recursor enhancement listed in Fig.~\ref{fig-recFeatures}---%
namely whetherit recursion affects the non-constructor operators too, and whether the  freeness optimization is being considered---we show that, in each case, the strongest version of the enhancement is applicable. 

Moreover, point (b) tells us that our general epi-recursion framework can be applied \emph{directly to the enhanced recursors}, as opposed to having to regard the enhancements as a form of ``hacks'' that are added after the fact on top of some categorically clean recursion principles. 

In what follows, we sketch the enhancement-extended version of our results. 
%
%
We first extend the notion of model from \S\ref{subsec-sigMod} as follows. 
We fix a finite set $X$ of variables (to be ``avoided'' according to the Barendregt convention). 

\medskip 
\begin{defi} \rm  \label{defi-model-extended}
Given a signature $\Sigma$, an \emph{$(X,\Sigma)$-model} $\MM$ consists of a set $M$, called the \emph{carrier set}, 
a subset $D \su \Trm \times M$ which we will call the \emph{domain}, 
and operations and/or relations on $D$ with values in $M$ according  the signature, more precisely: 
\begin{mmyitem}
	\item 
	if $\vr \in \Sigma$ 
	then $\MM$ has an operation $\Vr^\MM : \Var \ra M$; 
	\item if $\ap \in \Sigma$ 
	then $\MM$ has an operation  $\Ap^\MM : D \ra D \ra M$; 
	\item if $\lm \in \Sigma$ 
	then $\MM$ has 
	$\Lm^\MM : \Var \ra D \ra M$; 
	\item 
	if $\pm \in \Sigma$ 
	then $\MM$ has 
	$\_[\_]^\MM : D \ra \Perm \ra M$; 
	\item 
	if $\swp \in \Sigma$  
	then $\MM$ has 
	$\_[\_\sw \_]^\MM : D \ra D \ra \Var \ra M$; 
	\item 
	if $\sbs \in \Sigma$ 
	then $\MM$ has 
	$\_[\_ \,/ \_]^\MM : D \ra D \ra \Var \ra M$; 
	\item 
	if $\ren \in \Sigma$ 
	then $\MM$ has 
	$\_[\_ \,/ \_]^\MM : D \ra \Var \ra \Var \ra M$; 
	\item 
	if $\fv \in \Sigma$ 
	then the model $\MM$ has 
	$\FV^\MM : D \ra \Pow(\Var)$; 
	\item 
	if $\fr \in \Sigma$ 
	then the model $\MM$ has 
	$\fresh^\MM : \Var \ra D \ra \Bool$.  
\end{mmyitem} 

It is also required that the domain $D$ is closed under the operations modulo the avoidance of the variables in $X$ when binding or substituting, in that the following hold (if applicable, i.e., if the given operation is in the signature $\Sigma$) for all $t,t_1,t_2\in  \Trm$ and 
$m,m_1,m_2\in  M$ such that $(t,m),\;(t_1,m_1),\;(t_2,m_2) \in D$, all $x,y,z \in \Var$ such that 
$x,y \notin X$, 
and all $\sigma\in\Perm$ such that $\{u \in \Var \mid \sigma\;u \not= u\} \cap X = \emptyset$: 
\begin{mmyitem} 
	\item $(\Vr\;z,\Vr^\MM\;z) \in D$
\item $(\Ap\;t_1\;t_2,\;\Ap^\MM(t_1,m_1)\,(t_2,m_2)) \in D$; 
\item $(\Lm\;x\;t,\;\Lm^\MM x\,(t,m)) \in D$; 
\item  $(t[\sigma],\,(t,m)[\sigma]^\MM ) \in D$; 
\item  $(t[x \sw y],\,(t,m)[x \sw y]^\MM ) \in D$; 
\item  $(t[t_1 / x],\,(t,m)[(t_1,m_1) / x]^\MM ) \in D$ 
(if $\sbs \in \Sigma$); 
\item $(t[y / x],\,(t,m)[y / x]^\MM ) \in D$ 
(if $\ren \in \Sigma$).    \qed 
\end{mmyitem}
\end{defi} 
\medskip 

We can note a few things about this definition:
\begin{mmyitem} 
\item The closedness conditions only make sense for the constructor, permutation, swapping, renaming and substitution operators; and not for the free-variable operator or the freshness relation.
\item Full-fledged (primitive) recursion typically refers to having extra term arguments for the constructor operators only, but (as already pointed out) we consider them for the other operators and relations as well, since it makes the recursor more general.
\item The presence of the domain $D$ in the definition, as opposed to working with the entire product $\Trm \times M$ as would be customary, has a technical reason: In order to recover the results about the quasi-strength comparison relation $\wgeq$ (Thm.~\ref{thm-qexpr}), we must build submodels of these enhanced models; and those cannot have the form $\Trm \times M'$ for some $M' \su M$, but must be more flexible subsets $D' \su \Trm \times M$. In short, this small generalization was needed in order to close the category of models under a notion of submodel that works for generalizing our results. 
\end{mmyitem} 
\medskip 

Now the notion of morphism between $(X,\Sigma)$-models is defined in a similarly $X$-avoiding manner:

\medskip
\begin{defi}\rm \label{defi-morphismExtended}
Given two $(X,\Sigma)$-models $\MM$ and $\MM'$, a \emph{morphism} 
between them is a function  between their carrier sets $g: M \ra M'$
that preserves the domain, 
commutes with the operations and preserves the relations modulo $X$. More precisely, 
the following properties hold (if applicable, i.e., if the given operation is in the signature $\Sigma$) for all $t,t_1,t_2\in  \Trm$ and 
$m,m_1,m_2\in  M$ such that $(t,m),\;(t_1,m_1),\;(t_2,m_2) \in D$, all $x,y,z \in \Var$ such that 
$x,y \notin X$, 
and all $\sigma\in\Perm$ such that $\{u \in \Var \mid \sigma\;u \not= u\} \cap X = \emptyset$: 
\begin{mmyitem}
\item $(t,m) \in D$ implies $(t,g\;m) \in D'$;
\item $g(\Vr^\MM\,z) = \Vr^{\MM'}z$; 
\item $g(\Ap^\MM\,(t_1,m_1)\,(t_2,m_2)) = \Ap^{\MM'} (t_1,g\;m_1)\,(t_2,g\;m_2)$; 

\item $g(\Lm^\MM\,x\;(t,m)) = \Lm^{\MM'}x\;(t,g\;m)$; 

\item $g((t,m)[\sigma]^\MM) = (t,g\;m)[\sigma]^{\MM'}$; 

\item $g((t,m)[x \sw y]^\MM) = (t,g\;m)[x \sw y]^{\MM'}$; 

\item $g((t,m)[(t_1,m_1) / y]^\MM) = (t,g\;m)[(t_1,g\;m_1) / y]^{\MM'}$; 

\item $g((t,m)[x / y]^\MM) = (t,g\;m)[x / y]^{\MM'}$; 

\item $x\;\fresh^\MM\, (t,m)$ implies $x\;\fresh^{\MM'} (t,g\;m)$;

\item $\FV^{\MM'}\,(t,g\;m) \su X \cup \FV^\MM\,(t,m)$.   \qed 
\end{mmyitem}
\end{defi} 
\medskip 

$(X,\Sigma)$-models and morphisms thus defined form a category. 
We write $\TTrm(\Sigma)$ for the $(X,\Sigma)$-model whose carrier is the set of terms $\Trm$, 
whose domain is the diagonal $\{(t,t) \mid t \in \Trm\}$, 
and whose operations and relations are the obvious adaptations of the standard ones for terms.

It remains to interpret the properties in Fig.~\ref{fig-basicProps} in  $(X,\Sigma)$-models. In other words, given a signature $\Sigma$, an $(X,\Sigma)$-model $\MM$, a property $p$ from  Fig.~\ref{fig-basicProps} whose operations and relations are covered by $\Sigma$, we must state what it means for $\MM$ to satisfy $p$. The interpretation proceeds according to the following transformation rules:
\begin{myitem}
	\item[(1)] Any variable participating in $\lambda$-bindings, swappings, permutations, substitutions or freshness assertions is assumed to not belong to $X$.
	\item[(2)]  If it is the conclusion of the property's implication, any equation 
	or freshness relation 
	becomes a corresponding equation or relation referring to the operations in the model. 
	\item[(3)] Any equation in the hypotheses 
	becomes a conjunction between the term equation itself and a corresponding equation referring to the operations in the model. 
	\item[(4)] Any freshness relation in the hypotheses becomes a conjunction between 
	  \begin{itemize}
	  	\item the freshness relation itself (on terms)
	  	\item the implication between the freshness relation on terms and the corresponding one on items in the model, universally quantified on the participating variable (again assumed to not belong to $X$) 
	  \end{itemize} 
\end{myitem}

Let us illustrate the above on the same examples as those we considered in the main paper:
When we say that the $(X,\Sigma)$-model $\MM$ (with carrier $M$ and domain $D$) satisfies \SwCg{}, we mean the following:  
\begin{center} 
For all $t,t_1,t_2\in  \Trm$ and 
$m,m_1,m_2\in  M$ such that $(t,m),\;(t_1,m_1),\;(t_2,m_2) \in D$, 
\\and all
$x_1,x_2,z\in \Var$ such that $x_1,x_2,z \notin  X$, 
\\
if 
\\$z \not\in \{x_1,x_2\}$, 
$z\;\fresh\;t_1,t_2$, 
\\ 
$(\forall z.\;  z\;\fresh\;t_1 \implies z\;\fresh^\MM\;(t_1,m_1))$, 
\\
$(\forall z.\;  z\;\fresh\;t_2 \implies z\;\fresh^\MM\;(t_2,m_2))$, 
\\
$t_1[z\sw x_1] = t_2[z\sw x_1]$, and 
$(t_1,m_1)[z\sw x_1]^\MM = (t_2,m_2)[z\sw x_1]^\MM$,  
\\
then 
\\
$\Lm^\MM\;x_1\;(t_1,m_1) = \Lm^\MM\;x_2\;(t_2,m_2)$.  
\end{center} 

Notice how:
\begin{mmyitem}
	\item the variables $x_1,x_2,z$ are assumed not to be in $X$ according to the above transformation rule (1);
	\item the equation $\Lm\;x_1\;t_1 = \Lm\;x_2\;t_2$ 
	from the conclusion of \SwCg{} has become 
	$\Lm^\MM\;x_1\;(t_1,m_1) = \Lm^\MM\;x_2\;(t_2,m_2)$ according to transformation  rule (2); 
	\item the equation $t_1[z\sw x_1] = t_2[z\sw x_1]$
	from the hypotheses of \SwCg{} has become the conjunction of 
	$t_1[z\sw x_1] = t_2[z\sw x_1]$ and 
	$(t_1,m_1)[z\sw x_1]^\MM = (t_2,m_2)[z\sw x_1]^\MM$ according to transformation  rule (3); 
	\item the freshness hypothesis $z\;\fresh\;t_1$ has become  
	the conjunction of $z\;\fresh\;t_1$ and 
	$(\forall z.\;  z\;\fresh\;t_1 \implies z\;\fresh^\MM\;(t_1,m_1))$ according to transformation rule (4) 
	(and similarly for $t_2$). 	
\end{mmyitem} 

The treatment of the equations from the hypotheses---with considering both the original, 
here, $t_1[z\sw x_1] = t_2[z\sw x_1]$, and the model version $(t_1,m_1)[z\sw x_1]^\MM = (t_2,m_2)[z\sw x_1]^\MM$---honors the dual (term and semantic item) nature of full-fledged recursion. By contrast, including the original version in the conclusion too would be redundant, since that one follows anyway thanks to the properties of terms. 

The treatment of the freshness relations in the hypothesis is more involved, and departs from what would seem to be a natural rule, which is:  including both the original, $z\;\fresh\;t_1$, and 
the model version $z\;\fresh^\MM\;(t_1,m_1)$ as hypotheses. The reason why we instead include 
$z\;\fresh\;t_1$ and $(\forall z.\;  z\;\fresh\;t_1 \implies z\;\fresh^\MM\;(t_1,m_1))$ is because this way we obtain a weaker condition that still works, thus offering a stronger recursion principle. This is the freshness counterpart of the freeness optimization that is specific to the swap/free recursor (mentioned in Fig.~\ref{fig-recFeatures}). 

The properties involving the free-variable operator are interpreted using their freshness-based counterparts. 
For example, $\MM$ (with carrier $M$ and domain $D$) satisfying \FCB{} means the following: 
\begin{center}
There exists $x \in \Var$ such that 
$x\notin X$ and
\\
$x\notin \FV^\MM\,(\Lm\;x\;t)\,(\Lm^\MM x\,(t,m))$ 
\\for all $t\in\Trm$ and $m\in M$ such that 
$(t,m) \in D$.  
\end{center} 
Indeed, thinking of the conclusion of \FCB{} , namely $x\notin \FV\,(\Lm\;x\;t)$, in terms of freshness, i.e., as $x \,\fresh \alb (\Lm\;x\;t)$, we apply transformation rule (2) yielding 
$x\notin \FV^\MM\,(\Lm\;x\;t)\,(\Lm^\MM x\,(t,m))$. 
The outcome is a generalization of the standard \FCB{} used in nominal logic. 

Along the same recipe, we obtain the interpretations for \FvVr{}, \FvAp{} and \FvLm{}, where we can recognize generalizations of the ``optimized'' swap/free recursion clauses discussed in  \S\ref{subsec-nominalRec}:
\begin{mmyitem}
\item  For all $x\in \Var$, 
$\FV^\MM(\Vr\;x)\,(\Vr^\MM\;x) \su X \cup \{x\}$
\item For all $(m_1,t_1),(m_2,t_2) \in D$, 
if $\FV^\MM\;(m_1,t_1) \su X \cup \FV\;t_1$ and $\FV^\MM\;(m_2,t_2) \su X \cup \FV\;t_2$ 
then 
\\
\hspace*{-2ex}$\FV^\MM(\Ap\,t_1\,t_2)(\Ap^\MM(m_1,t_1)\,(m_2,t_2)) \su X \cup \FV\,t_1 \cup \FV\,t_2$   

\item For all $(m,t)\in D$ and $x \in  \Var \!\sm\! X$, 
if $\FV^\MM(m,t) \su X \cup \FV\;t$ 
then 
$\FV^\MM(\Lm\;x\;t)\,(\Lm^\MM x\;(t,m)) \su \FV\,t \sm \{x\}$ 
\end{mmyitem}

Indeed, these follow from the interpretations for their freshness-based counterparts 
\FrVr{}, \FrAp{} and \FrLm{}, which are produced according to the above transformation rules:
\begin{mmyitem} 
\item For all $x,z\in\Var$ such that $z\notin X$, 
 if $z \not=x$ then $z\;\fresh^\MM(\Vr\;x)\,(\Vr^\MM x)$
\item For all $z\in\Var \sm X$ and $(t,m),(s,n) \in D$, 
if $z\;\fresh\;s$ and $(\forall z\in \Var \sm X.\; z\;\fresh\;s \implies z\;\fresh^\MM\,(s,n))$, 
$z\;\fresh\;t$ and $(\forall z\in \Var \sm X.\; z\;\fresh\;t \implies z\;\fresh^\MM\,(t,m))$, 
then 
$z\;\fresh^\MM\,\Ap^\MM(s,n)\;(t,m)$
\item For all $x,z\in\Var \sm X$ and $(t,m) \in D$, 
if $z = x$ or ($z\;\fresh\;t$ and $(\forall z\in \Var \sm X.\; z\;\fresh\;t \implies z\;\fresh^\MM\,(t,m))$ ) 
then $ z\;\fresh^\MM\,\Lm^\MM\;x\;(t,m)$. 
\end{mmyitem}

All the recursion principles generalize from their stripped down versions 
to the enhanced versions: 

\medskip 
\begin{thm}\rm \label{thm-allNominalRecsEnhanced}
Thm.\ \ref{thm-allNominalRecs} still holds  if we replace 
the categories of $\Sigma_i$-models with those of $(X,\Sigma_i)$-models.  \qed
\end{thm} 
\medskip 

Recall that Thms.\ \ref{thm-pittsRec}--\ref{thm-popRecRename} list 
existing nominal recursors from the literature. 
Thm.\ \ref{thm-allNominalRecs} 's recursors 
$r_1$, $r_4$, $r_6$, $r_7$ and $r_8$ were stripped-down versions of these recursors. 
By contrast, Thm.\ \ref{thm-allNominalRecsEnhanced}'s corresponding recursors are 
\emph{further enhancements} of the original recursors, obtained by putting together all the 
enhancements---because the strongest version of each enhancement now benefits each recursor.

In order to generalize our comparison results, 
we were actually compelled to strengthen the recursors even beyond the sum of all enhancements. For example, both the perm/free recursor (specific to nominal logic) and Norrish's swap/free recursor were based on axiomatizations of permutation and swapping: forming nominal sets in the case of the perm/free recursor and entities called \emph{swapping structures}  for the swap/free recursor \cite{primrecFOAS-Norrish04}. At the same time, both recursors had Barendregt enhancements that allowed the flexibility of working modulo  $X$, meaning that some axioms on the target domains operated modulo $X$---as seen in Thms.\ \ref{thm-pittsRec} and \ref{thm-norrishRec}. However, this flexibility was not affecting the notions of nominal set or swapping structure, which did not consider $X$.  
Our systematic approach to adding performing the Barendregt enhancement, reflected in particular in the $r_1$ and $r_4$ versions of our Thm.\ \ref{thm-allNominalRecsEnhanced}, makes this flexibility pervasive, thus strengthening Thms.\ \ref{thm-pittsRec} and \ref{thm-norrishRec} with what could be called \emph{nominal sets up to $X$} and \emph{swapping structures up to $X$}. 
We have not investigated whether such stronger recursors can make a  difference in practice, but it is in principle useful to have the strongest possible recursors at our disposal. 

All the expressiveness comparisons results for the stripped down recursors carry over to the enhanced recursors as well: 

\medskip 
\begin{thm} \rm \label{thm-exprEnhanced} 
	Thms.\ \ref{thm-expr} and \ref{thm-qexpr} still hold for the enhanced recursors 
	(employing $(X,\Sigma_i)$-models) described in Thm.\ \ref{thm-allNominalRecsEnhanced}.  \qed
	\looseness=-1
\end{thm} 
\medskip 

The proofs follow 
the same lines as those we sketched for Thms.\ \ref{thm-expr} and \ref{thm-qexpr}, using the categorical criteria from Props.~\ref{prop-extCriterion} and \ref{prop-WeakExtCriterion}. One phenomenon worth mentioning is that the goal of  extending our comparison results to the enhanced recursors have forced us to perform more general enhancements than originally intended (and thought possible). For example, the aforementioned notion of performing 
Barendregt enhancement more comprehensively in  $r_1$ and $r_4$ (yielding nominal sets up to $X$ and swapping structures up to $X$) were required in order to prove that, via $r_3$, they are quasi-stronger than (the naturally enhanced version of) $r_6$.

\section{Non-well-founded infinitary terms 
}
\label{app-detailsIterms}

\subsection{Pre-iterms}
\label{appsub-preiterm}

%
We start with the set $\PITrm$ of pre-iterms, which is (co)freely generated by the grammar:
$$p ::= \PVr\;x  \mid \PAp\;p_1\;p_2  \mid \PLm\;x\;p$$

So $\PITrm$ is a standard coinductive datatype (codatatype) \cite{DBLP:journals/tcs/Rutten00,DBLP:journals/mscs/KozenS17} 
having constructors $\PVr : \Var \ra \PITrm$, $\PAp : \ITrm \ra \PITrm \ra \PITrm$ 
and $\PLm : \Var \ra \ITrm \ra \ITrm$. 

Recall from \S\ref{subsec-infTerms} that 
we wrote $\Vv$, $\Aa$ and $\Ll$ 
for the three injections into the sum type  $\mathsf{S} = \Var + 
\ITrm \times \ITrm + \PPne(\Var\times \ITrm)$ (so that
$\Vv : \Var \ra \mathsf{S}$, 
$\Aa : \ITrm \times \ITrm \ra \mathsf{S}$ and 
$\Ll : \PPne(\Var\times \ITrm) \ra \mathsf{S}$).
We will overload this notation to pre-iterms, thus writing 
$\Vv$, $\Aa$ and $\Ll$ 
also for the three injections into the sum type  $\mathsf{S'} = \Var + 
\PITrm \times \PITrm + \Var\times \PITrm$ (so that
$\Vv : \Var \ra \mathsf{S'}$, 
$\Aa : \PITrm \times \PITrm \ra \mathsf{S'}$ and 
$\Ll : \Var\times \PITrm \ra \mathsf{S'}$). 

We write $\PDest : \PITrm \ra \Var + 
\PITrm \times \PITrm + \Var\times \PITrm$ for this codatatype's destructor. Note that the destructor is the inverse of the constructors in the following sense: 
\begin{itemize}
\item $t = \PVr\;x \iff \PDest\;t = \Vv\;x$
 \item $t = \PAp\;t_1\;t_2 \iff \PDest\;t = \Aa (t_1,t_2)$
\item $t = \PLm\;x\;t' \iff \PDest\;t = \Ll(x,t')$
\end{itemize}

Thus, $(\PITrm,PDest)$ is the final coalgebra for the functor on sets defined as follows:
\begin{itemize}
	\item on objects, it takes any set $A$ to $\Var + A \times A + \Var \times A$;
	\item on morphisms, it takes any function $f:A\ra B$ to 
	$f + f \times f + \id_\Var \times f$.
\end{itemize}
Above, we used the following notation. For any two functions $u_1: A_1 \ra B_1$ and $u_2: A_2 \ra B_2$, we let:
\begin{itemize}
	\item $u_1 + u_2 : A_1 + A_2 \ra B_1 + B_2$ be the function defined by  $(u_1+u_2)(\In_1\,a_1) = \In_1 (u_1\,a_1)$ and $(u_1+u_2)(\In_2\,a_2) = \In_2 (u_2\, a_2)$, where $\In_1$ and $\In_2$ denote the two injections for the sum types; 
	\item $u_1 \times u_2 : A_1 \times A_2 \ra B_1 \times B_2$ be the function defined by  $(u_1\times u_2)(a_1,a_2) = (u_1\,a_1, u_2\,a_2)$. 
\end{itemize} 
(Thus, we use $+$ and $\times$ for the actions of the sum and product functors not only on objects, but also on morphisms.)

Like any ordinary codatatype, 
the pre-iterm codatatype features the following \emph{structural coinduction proof principle}, which 
states 
that equality is the largest destructor-bisimulation on pre-iterms:

\begin{prop} \label{prop-coind-piterm} \rm
Assume $\phi : \PITrm \ra \PITrm \ra \Bool$ is a relation on pre-iterms such that, for all $p,q\in\PITrm$, if $\phi\;p\;q$ then one of the following is true: 
\begin{mmmyitem} 
\item there exists $x$ such that $p = \PVr\;x = q$; 
\item there exist $p_1,p_2,q_1,q_2$ such that 
$p = \PAp\;p_1\;p_2$, \,$q = \PAp\;q_1\;q_2$, \,$\phi\;p_1\;q_1$ and $\phi\;p_2\;q_2$; 
\item there exist $x,p',q'$ such that 
$p = \PLm\;x'\;p'$, \,$q = \PLm\;x\;q'$ and $\phi\;p'\;q'$. 
\end{mmmyitem} 

Then $\phi$ is included in equality, in that $\forall p,q.\;\phi\;p\;q \Ra p = q$.
\end{prop} 

\smallskip 
The pre-iterm codatatype also features the following corecursion definition principle (coiteration to be more precise), which is just an expression of the fact that $(\PTrm,\PDest)$ is a final coalgebra:

\begin{prop} \label{prop-corec-piterm} \rm
If $(M,D)$ is a coalgebra of suitable type, namely 
$D :  M \ra \Var + M \times M + \Var\times M$, then there exists a unique coalgebra morphism between $(M,D)$ and $(\PTrm,\PDest)$, i.e., a unique function $g:M\ra \PTrm$ that commutes with the destructors, in that, for all $m\in M$, 
$\PDest(g\;m) = (1_\Var + g \times g + 1_\Var \times g)\,(D\;m)$. This commutation condition can also be phrased as the conjunction of three conditions, one for each summand:\footnote{In the first equation below, the first $\VV\;x$ is an element of $\Var + M \times M + \Var\times M$, whereas the second $\VV\;x$ is an element of $\Var + \PITrm \times \PITrm  + \Var\times \PITrm $. This is because of our ambiguous notation for $\Vv$; but the context should always disambiguate such situations.}
\begin{myitem} 
	\item $D\;m = \Vv\;x$ implies $\PDest(g\;m) = \Vv\;x$
	\item $D\;m = \Aa (m_1,m_2)$ implies $\PDest(g\;m) = \Aa(g\;m_1,g\;m_2)$
	\item $D\;m = \Ll (x,m')$ implies $\PDest(g\;m) = \Ll (x,g\;m')$
\end{myitem} 
and further, using pre-iterm constructors instead of destructor:
\begin{myitem} 
	\item $D\;m = \Vv\;x$ implies $g\;m = \PVr\;x$
	\item $D\;m = \Aa (m_1,m_2)$ implies $g\;m = \PAp\;(g\;m_1)\;(g\;m_2)$
	\item $D\;m = \Ll (x,m')$ implies $g\;m = \PLm\;x\;(g\;m')$
\end{myitem} 
\end{prop} 

In addition to the above, another definition and proof mechanism that is useful for concepts involving pre-iterms (and iterms as well), but that in itself is not bound to codatatypes, is \emph{rule coinduction}: Given any monotonic operator on predicates of some type (e.g., $n$-ary predicates/ relations on pre-iterms), we can take its greatest fixed point, which is also the greatest post-fixed point---whose existence (and uniqueness) is guaranteed by the Knaster-Tarski theorem \cite{knasterTarski}. Usually, this monotonic operator is described using a set of rules, and the greatest fixed point is the largest predicate that is consistent with (i.e., backwards-closed under) these rules. Hence, to prove that this greatest (post)fixed point includes another predicate $\phi$, it suffices to show that $\phi$ is consistent with these rules. 
We refer to \cite[\S21.1]{DBLP:books/daglib/0005958} for more details. 

\medskip
The swapping operator on pre-iterms, $\_[\_\sw\_] : \PITrm \ra \Var \ra \Var \ra \PITrm$, is defined corecursively by the following clause:
$$
\begin{array}{l}
p[z_1 \sw z_2] = \mbox{case $p$ of}
\\
\hspace*{5ex}\mid \PVr\;x \Ra \PVr\;(x[z_1\sw z_2]) 
\\
\hspace*{5ex}\mid \PAp\;p_1\;p_2 \Ra \PAp\;(p_1[z_1\sw z_2]) \;(p_2[z_1\sw z_2]) 
\\
\hspace*{5ex}\mid \PLm\;x\;p' \Ra \PLm\;(x[z_1\sw z_2]) \;(p'[z_1\sw z_2]) 
\end{array} 
$$
The above corecursive definition can also be expressed in destructor form:
$$
\begin{array}{l}
\PDest\;(p[z_1 \sw z_2]) = \mbox{case $\PDest\;p$ of}
\\
\hspace*{5ex}\mid \Vv\;x \Ra \Vv (x[z_1\sw z_2]) 
\\
\hspace*{5ex}\mid \Aa(p_1,p_2) \Ra \Aa(p_1[z_1\sw z_2],p_2[z_1\sw z_2]) 
\\
\hspace*{5ex}\mid \Ll(x,p') \Ra \Ll(x[z_1\sw z_2], p'[z_1\sw z_2]) 
\end{array} 
$$
What this definition means is that 
we organize the source domain of the (uncurried version of) the to-be-defined function, $M = \PITrm \times \Var \times \Var$, into a coalgebra $(M,D)$ by defining $D : M \ra \Var + M \times M + \Var\times M$ as follows: 
$$
\begin{array}{l}
D\;(p,z_1,z_2) = \mbox{case $\PDest\;p$ of}
\\
\hspace*{5ex}\mid \Vv\;x \Ra \Vv (x[z_1\sw z_2]) 
\\
\hspace*{5ex}\mid \Aa(p_1,p_2) \Ra \Aa((p_1,z_1,z_2),(p_2,z_1,z_2)) 
\\
\hspace*{5ex}\mid \Ll(x,p') \Ra \Ll(x[z_1\sw z_2],(p',z_1,z_2)) 
\end{array} 
$$
and define $p[z_1\sw z_2]$ as $g\,(p,z_1,z_2)$, where $g: (M,D) \ra (\PITrm,\PDest)$ is the unique coalgebra morphism guaranteed by the corecursion principle.
(In this particular case, the parameters $z_1$ and $z_2$ stay fixed, so they could have been left out 
of the source coalgebra's carrier.)

\medskip
The permutation operator on iterms, $\_[\_] : \PITrm \ra \Perm \ra \PITrm$, is defined similarly to swapping. (Swapping can of course be alternatively defined from permutation.)

\medskip
The freshness relation $\_ \# \_ : \Var \ra \PITrm \ra \Bool$ is defined coinductively by the following rules:
$$
\frac{z\not= x}
{z \;\fresh\;\PVr\;x}
\hspace*{13ex}
\frac{z \;\fresh\;p_1 \;\;\;\;\;\; z \;\fresh\;p_2}
{z \;\fresh\;\PAp\;p_1\;p_2}
\hspace*{13ex}
\frac{z = x \mbox{ \ or \ } z \;\fresh\;p'}
{z \;\fresh\;\PLm\;x\;p'}
$$
(What this means is that $\_\#\_$ defined to be the largest predicate $\phi: \Var \ra \PITrm \ra \Bool$ that is consistent with the above rules, in the following sense: For all $x\in\Var$ and $p\in\PITrm$, if $\phi\;z\;p$ then one of the following is true:
\begin{itemize}
	\item there exists $x$ such that $p=\PVr\;x$ and $z\not=x$;
	\item there exist $p_1,p_2$ such that $p= \PAp\;p_1\;p_2$, 
	$\phi\;z\;p_1$ and $\phi\;z\;p_2$; 
	\item there exists $x,p'$ such that $p= \PLm\;x\;p'$ 
	and (z = x or $\phi\;z\;p'$).)
\end{itemize}

\medskip
An important property that we wish to have for iterms, and to this end we first need to ensure it for pre-iterms, is that for any (pre-)iterm we have a supply of fresh variables. In fact, as a virtue of the cardinality of $\Var$ being $\aleph_1$, we have the following: For any pre-iterm $p$, since the set $\{x \mid x\;\fresh\;p\}$ is countable, there exist uncountably (in particular, infinitely) many fresh variables for $p$.

\medskip
The $\alpha$-equivalence relation on pre-iterms, $\_ \!\equiv\!\_ : \PITrm \ra \PITrm \ra \Bool$, is defined coinductively by the following rules:
$$
\PVr\;x \equiv \PVr\;x
\hspace*{13ex}
\frac{p_1 \equiv q_1 \;\;\;\;\;\; p_2 \equiv q_2 }
{\PAp\;p_1\;p_2 \equiv \PAp\;q_1\;q_2}
\hspace*{13ex}
$$
\vspace*{0.5ex}
$$
\frac{
	p[z\sw x] \equiv q[z\sw y] 
\hspace*{7ex}	
	z = x \mbox{ or } z \;\fresh\;p
\hspace*{7ex}
   z = y \mbox{ or } z \;\fresh\;q
}
{\PLm\;x\;p \equiv \PLm\;y\;q}
$$

Note that the last rule in this definition is in the style of the \SwCg{} and \ISwCg{} properties (on iterms). We obtain the same concept (i.e., we obtain the same relation $\equiv$) if we replace it with the following rule, in the style of \SwBvr{} and \ISwBvr{}: 
$$
\frac{
	p[x'\sw x] \equiv p'
	\hspace*{7ex}	
	x' = x \mbox{ or } x' \;\fresh\;p
}
{\PLm\;x\;p \equiv \PLm\;x'\;p'}
$$

Note that all these operators (swapping, permutation, freshness and $\alpha$-equivalence) would be defined in the same way for (finitary) preterms, i.e., we would write the same equations and rules, but replacing ``coinductive'' (``greatest fixed point'') with ``inductive'' (``least fixed point'') and ``recursive'' with ``corecursive''.

\subsection{Iterms}
\label{appsub-iterm}

It can be shown that $\equiv$ is an equivalence and is compatible with the swapping and permutation operations and freshness predicate, in that:
\begin{itemize}
\item $p \equiv q$ implies $p[z_1\sw z_2] \equiv q[z_1\sw z_2] $ 
\item $p \equiv q$ implies $p[\sigma] \equiv q[\sigma] $ 
\item $x \;\fresh\; p$ and $p \equiv q$ implies $x \;\fresh\; q$
\end{itemize}

We define iterms by quotienting pre-iterms, $\ITrm = \PITrm/\equiv$. For a pre-iterm $p$, let us write $p/\equiv$ for its $\alpha$-equivalence class. 

We define the corresponding operators on iterms, 
$\_[\_\sw\_] : \ITrm \ra \Perm \ra \ITrm$,  
$\_[\_] : \ITrm \ra \Perm \ra \ITrm$, $\_ \# \_ : \Var \ra \PITrm \ra \Bool$,  
 by lifting to iterms the pre-iterm operators and taking advantage of their compatibility with $\equiv$. For example, 
 given $t\in\ITrm$ we define 
 $t[z_1\sw z_2]$ to be $(p[z_1\sw z_2])/\equiv$, where $p$ is some pre-iterm such that $t = p/\equiv$ (whose choice is immaterial thanks to compatibility).  
 
 The free-variable operator $\FV: \ITrm \ra \Pow(\Var)$ is defined as expected, by $\FV\;t = \{x\in\Var \mid \neg\;x\;\fresh\;t\}$.
 
 All the properties involving swapping and/or permutation and/or equality and/or freshness/freeness on iterms from Figs.~\ref{fig-basicProps} and \ref{fig-basicCoProps} can now be proved by first establishing their pre-iterm counterparts (with $\equiv$ instead of equality) and then lifting them to iterms. 
 
 The following is the natural structural coinduction principle for iterms:
 
 \begin{prop} \rm 
 	\label{prop-coind-iterm}
 	Assume $\phi : \ITrm \ra \ITrm \ra \Bool$ is a relation on iterms such that, for all $t,s\in\ITrm$, if $\phi\;t\;s$ then one of the following is true: 
 	\begin{mmmyitem} 
 		\item there exists $x$ such that $t = \Vr\;x = s$; 
 		\item there exist $t_1,t_2,s_1,s_2$ such that 
 		$t = \Ap\;t_1\;t_2$, \,$s = \Ap\;s_1\;s_2$, \,$\phi\;t_1\;s_1$ and $\phi\;t_2\;s_2$; 
 		\item there exist $x,t',y,s'$ such that 
 		$t = \Lm\;x\;t'$, \,$s = \Lm\;y\;s'$, 
 		\,($y  = x$ or $y\;\fresh\;t'$) and 
 		$\phi\,(t'[y\sw x])\,s'$. 
 	\end{mmmyitem} 

 	Then $\phi$ is included in iterm equality, in that 
 	$\forall t,s.\;\phi\;t\;s \Ra t = s$.
 \end{prop}

Note that the above principle reflects the aforementioned alternative, \SwBvr{}/\ISwBvr{}-like definition of $\alpha$-equivalence. A principle that instead reflects the  \SwCg{}/\ISwCg{}-like definition is also possible, but is more tedious to use in proofs. 

The above principle cannot be inferred directly from the (alternative) definition of $\alpha$-equivalence, which by definition only gives us the following proof principle: 


	Assume $\phi : \PITrm \ra \PITrm \ra \Bool$ is a relation on pre-iterms such that, for all $p,q\in\ITrm$, if $\phi\;p\;q$ then one of the following is true: 
	\begin{mmmyitem} 
		\item there exists $x$ such that $p = \PVr\;x = q$; 
		\item there exist $p_1,p_2,q_1,q_2$ such that 
		$p = \PAp\;p_1\;p_2$, $q = \PAp\;q_1\;q_2$, \,$\phi\;p_1\;q_1$ and $\phi\;p_2\;q_2$; 
		\item there exist $x,p',y,q'$ such that 
		$p = \PLm\;x\;p'$, \,$q = \PLm\;y\;q'$, 
		\,($y  = x$ or $y\;\fresh\;p'$) and 
		$\phi\;(p[y\sw x])\;q'$. 
	\end{mmmyitem} 

	Then $\phi$ is included in $\alpha$-equivalence, in that 
	$\forall p,q.\;\phi\;p\;q \Ra p \equiv q$.
	%
\medskip

But in order to produce Prop.~\ref{prop-coind-iterm}, we need a stronger version whose hypotheses are weaker, in that pre-iterm equality is replaced by $\alpha$-equivalence---i.e., we need a form of $\alpha$-coinduction up to $\alpha$-equivalence. It turns out that we can prove such a stronger version, if we also assume that the predicate $\phi$ is compatible with $\alpha$-equivalence. This stronger principle is shown below, where we highlight the differences from the previous one: 

	Assume $\phi : \PITrm \ra \PITrm \ra \Bool$ is an \hlt{$\mbox{$\alpha$-equivalence-compatible}$} relation on pre-iterms 	
	such that, for all $p,q\in\ITrm$, if $\phi\;p\;q$ then one of the following is true: 
\begin{mmmyitem} 
	\item there exists $x$ such that $p \hlt{\equiv} \PVr\;x \hlt{\equiv} q$; 
	\item there exist $p_1,p_2,q_1,q_2$ such that 
	$p \hlt{\equiv} \PAp\;p_1\;p_2$, $q \hlt{\equiv} \PAp\;q_1\;q_2$, \,$\phi\;p_1\;q_1$ and $\phi\;p_2\;q_2$; 
	\item there exist $x,p',y,q'$ such that 
	$p \hlt{\equiv} \PLm\;x\;p'$, \,$q \hlt{\equiv} \PLm\;y\;q'$, 
	\,($y  = x$ or $y\;\fresh\;p'$) and 
	$\phi\;(p'[y\sw x])\;q'$. 
\end{mmmyitem} 

Then $\phi$ is included in $\alpha$-equivalence, in that 
$\forall p,q.\;\phi\;p\;q \Ra p \equiv q$.
\medskip

Now this last principle 
easily yields Prop.~\ref{prop-coind-iterm}, using the iterm to pre-iterm projection to transport the statement. 

\medskip 
The situation of the freshness predicate on iterms versus the one on pre-iterms is similar  
to one we just discussed, of iterm equality versus pre-iterm $\alpha$-equivalence. Namely, transporting directly to iterms the coinduction principle from the definition of pre-iterm freshness does not give a proof principle that is strong enough. So we need to play a game similar to the one above, working with an $\alpha$-compatible predicate and replacing equality with $\alpha$-equivalence, which then yields the desired principle (equivalent to the one we would get if we defined freshness directly on iterms coinductively, rather than defining it from pre-iterm freshness):\footnote{One may ask why we have not chosen to use this alternative definition of freshness, namely to define iterm freshness directly on iterms without using pre-iterm freshness. While we could have done that, the formal development would not have been simplified, since in order to recover some of the desired properties (those of the interaction between freshness and equality) we would have still needed to 
	connect iterm freshness with pre-iterm freshness..} 

 \begin{prop} \rm 
	\label{prop-coind-fresh}
	Assume $\phi : \Var \ra \ITrm \ra \Bool$ is a relation such that, for all $z\in\Var$ and $t\in\ITrm$, if $\phi\;z\;t$ then one of the following is true: 
	\begin{itemize}
		\item there exists $x$ such that $t=\Vr\;x$ and $z\not=x$;
		\item there exist $t_1,t_2$ such that $t= \Ap\;t_1\;t_2$, 
		\,$\phi\;z\;t_1$ and $\phi\;z\;t_2$; 
		\item there exists $x,t'$ such that $t= \Lm\;x\;t'$ 
		and ($z = x$ or $\phi\;z\;t'$).
	\end{itemize} 
	
	Then $\phi$ is included in $\_\#\_$, in that 
	$\forall x,t.\;\phi\;x\;t \Ra x \;\fresh\;t$.
\end{prop}

 \medskip
Similarly to what happens in the inductive world (for terms), the substitution operator 
 $\_[\_/\_] : \ITrm \ra \Var \ra \Var \ra \ITrm$ is not straightforward to define, because there is no corresponding well-behaved substitution operator that can be defined on pre-iterms. In fact, substitution is one of the cases where the nominal corecursors described in this paper can be deployed---see \S\ref{app-exaCorec}.  But next we describe a route that does not appeal to corecursors. 
 
 We will first define a ``pre-substitution'' operator $\_[\_/\_]' : \ITrm \ra \Var \ra \Var \ra \PITrm$ (thus targeting pre-iterms rather than iterms but still having iterms as source domain) by pre-iterm corecursion. To prepare for this definition, let us introduce the following operators:
 \begin{myitem}
 	\item $\isVr : \ITrm \ra \Bool$ and $\getVr : \ITrm \ra \Var$, where  
 	$\isVr\;t$ tests if the iterm $t$ has the form $\Vr\;x$, and in this case $\getVr\;t$ returns this unique $x$.
 	\item $\isAp : \ITrm \ra \Bool$ and $\getAp : \ITrm \ra \ITrm \times \ITrm$, where 
 	$\isAp\;t$ tests if $t$ has the form $\Ap\;t_1\;t_2$, and in this case $\getAp\;t$ returns this unique pair $(t_1,t_2)$.
 	\item $\isLm : \ITrm \ra \Bool$ and $\getLm : \ITrm \ra \Var \ra \ITrm \ra \Var \times \ITrm$, where 
 	$\isLm\;t$ tests if $t$ is a $\Lm$-abstraction, and in this case $\getLm\;s\;z\;t$ returns \emph{some} pair $(x,t')$ such that $t = \Lm\;x\;t'$, \,$x\;\fresh\;s$
 	and $x\not= z$. 
 \end{myitem}
 The correctness of all these definitions follows from the properties of iterms. For example, the possibility to write any $\Lm$-abstraction iterm as $\Lm\;x\;t'$ where $x$ is fresh for $s$ and $z$ follows from \SwBvr{} and the existence of infinitely many fresh variables for any term. 
 
We  let $\rep: \ITrm \ra \PITrm$ be the function that chooses a pre-iterm representative, i.e., such that $t = (\rep\;t)/\equiv$ for all $t\in\ITrm$. 

We are now ready for the corecursive definition of pre-substitution: 
$$
\begin{array}{l}
t[s / z]' = 
\mbox{if $\isVr\;t$ then (if $z = \getVr\;t$ then $\rep\;s$ else $\PVr\;(\getVr\;t)$)}
\\\phantom{t[s / z]' = \;} 
\mbox{else if $\isAp\;t$ then let $(t_1,t_2) = \getAp\;t$ in 
	$\PAp\;(t[s / z]' )\;(t[s / z]')$}
\\\phantom{t[s / z]' = \;} 
\mbox{else let $(x,t') = \getLm\;s\;z\;t$ in 
	$\PLm\;x\;(t'[s / z]' )$}
\end{array} 
$$

Finally, substitution on iterms is defined by $t[s / z] = t[s / z]'/\equiv$. 
The desired characteristic equations of substitution, namely:
\begin{itemize}
	\item $(\Vr\;x)[s/z] = $ (if $x = z$ then $s$ else $\Vr\;x$)
	\item $(\Ap\;t_1\;t_2)[s/z] = \Ap\;(t_1[s/z])\;(t_2[s/z])$ 
	\item $(\Lm\;x\;t)[s/z] = \Lm\;x\;(t[s/z])$ if $x\not= z$ and $x\;\fresh\;s$. 
\end{itemize}
can now be established by structural iterm coinduction (Prop.~\ref{prop-coind-iterm}) 
using the freshness-related properties of $\Lm$-abstractions.  
The proof of uniqueness, i.e., the fact that substitution is the only operator on iterms satisfying the above equations, also follows by structural iterm coinduction.\footnote{We took the trouble to sketch the development leading to the characteristic equations of iterm substitution because they have acted as an inspiration for our proof of the nominal corecursor theorem, Thm.~\ref{thm-allNominalCoRecs}---more precisely, for the direct proof of the $\ccr_2$ corecursion principle---see the \S\ref{app-proofCoSketches} proof sketch of Thm.~\ref{thm-allNominalCoRecs}.}

\medskip 
Renaming is of course a particular case of substitution, defined as 
$t[x/y] = t[\Vr\;x/y]$. After the characteristic equations of substitution have been established, all the Figs.~\ref{fig-basicProps} and \ref{fig-basicCoProps} properties involving substitution or renaming follow by iterm coinduction or rule coinduction. In these proofs, each time we need to split into cases according to the structure of an iterm, we make sure that in the abstraction case, $\Lm\;x\;t$, the binding variable $x$ is fresh for the rest of the proof context
---which in particular ensures that the above $\Lm$-clause for substitution can be applied. (This is of course a way to enforce Barendregt's variable convention. A local form of this convention, i.e., fresh cases analysis, seems sufficient in the coinductive world. On the other hand, a binding-aware coinductive notion analogous to fresh induction seems neither needed nor in fact possible.)

\medskip
\textbf{Iterms as an abstract (co)datatype. }
Given the fact that iterms are less well-known than terms, a valid question to ask is whether our definitions are correct, i.e., whether they capture correctly the notion of infinitary $\lambda$-calculus terms where the identity of bound variables does not matter. The possible uncertainty about this seems to be fed by the definitions via pre-iterms being rather low-level and tedious, not to mention that concepts such as $\alpha$-equivalence can be defined in several 
ways. Moreover, our way of defining iterms is certainly not the only way. For example, \cite{DBLP:conf/cmcs/KurzPSV12} define the iterms of finite support as the metric completion of the set of (finitary) $\lambda$-terms. 

To resolve this possible uncertainty, the concept of abstract datatype comes handy. After having proved for the above defined iterms:
\begin{myitem}
	\item all the properties listed in Figs.~\ref{fig-basicProps} and \ref{fig-basicCoProps} (with ``countable'' replacing ``finite'' for the last group in Fig.~\ref{fig-basicProps}), 
	\item the structural conduction principle described by Prop.~\ref{prop-coind-iterm} and 
	\item the corecursion principles described by Thm.~\ref{thm-allNominalCoRecs}, 
\end{myitem}
we have reached a highly redundant unique characterisation of iterms together with its operators as an abstract datatype (i.e., unique up to an operator-preserving bijection). So we can forget about how iterms were defined, in particular, can forget about pre-iterms and $\alpha$-equivalence. 
This process of ``forgetting'' is also useful from a proof development perspective, since the available proof and definition principles for iterms form a self-sufficient layer of abstraction.  

\section{More Details on Epi-Corecursors and Nominal Corecursors} 
\label{app-moreDetailsCorec}

\subsection{Miscellanea} 

The concept of epi-corecursor is depicted in Fig.~\ref{fig-epiCRP}. 

\begin{figure}
	$$
	\xymatrix@C=4pc@R=1pc{
		\Ccat \ar[d]_R &  C  \ar[r]^{\im_{C,\,J}} &  J
		\\
		\Bcat &  B = R\,C   \ar[r]^{R\,\im_{C,\,J}}  & T = R\,J
	}
	$$
	\vspace*{-2ex}
	\caption{Epi-corecursor
	}
	\label{fig-epiCRP}
	\vspace*{-2ex}
\end{figure}

The criterion that we used for proving that various epi-corecursors are more expressive than others is morphism-dual (though not functor-dual, i.e., the pre-functor's direction is not reversed) to that 
we used for recursors (Prop.~\ref{prop-extCriterion}): 

\begin{prop}\rm
	\label{prop-extCoCriterion} 
	Let $\ccr = (\Bcat,T,\Ccat,J,R)$ and 
	$\ccr' = (\Bcat,T,\Ccat',J',\alb R')$, and 
	assume $F : \Ccat \ra \Ccat'$ is a pre-functor  
	such that 
	$R' \circ F = R$ and 
	$F\;J = J'$.  
	Then $\ccr' \geq \ccr$. 
\end{prop}

In the main paper, we mentioned that a gentler/laxer comparison relation is available for epi-corecursors as well. This is indeed obtained immediately by morphism-dualizing the one from epi-recursors:

\begin{defi} \rm \label{defi-Coqstrong} 
	$\ccr'= (\Bcat,T,\Ccat',J',R')$ is 
	\emph{quasi-stronger} than $\ccr = (\Bcat,T,\Ccat,J,R)$, written $\ccr' \wgeq \ccr$,
	when there exists a final segment $(\Bcat_0,(m(B):B \ra o(B))_{B \in \Obj{\Bcat}})$ of $\Bcat$ such that, 
	for all 
	$g: B \ra T $ definable by $\ccr$, there exists 
	a morphism $g_0: o(B) \ra T$ such that $g_0$ is definable by $\ccr'$ and $g = g_0 \circ m(B) $.  
\end{defi}  

The effective criterion for checking $\wgeq$, Prop.~\ref{prop-WeakExtCriterion}, 
can also be morphism-dualized from epi-recursors to epi-corecursors:

\begin{prop}\rm
	\label{prop-CoWeakExtCriterion} 
	Let $\ccr = (\Bcat,T,\Ccat,J,R)$ and 
	$\ccr' = (\Bcat,T,\Ccat',J',\alb R')$.  
	Assume 
	$(\Bcat_0,(m(B):B\ra o(B))_{B \in \Obj{\Bcat}})$ is a final segment of $\Bcat$
	and $(\Ccat_0,\alb(m_1(C):C\ra o_1(C))_{C \in \Obj{\Ccat}})$ is a final segment of $\Ccat$ 
	such that $\Ccat_0$ contains $J$ and 
	$R$ preserves the above final segments, 
	and $F : \Ccat_0 \ra \Ccat'$ is a 
	pre-functor such that 
	$F\;J = J'$ and $R' \circ F = R_{\restr\Ccat_0}$ 
	(where $R_{\restr\Ccat_0}$ is the restriction of $R$ to $\Ccat_0$).  
	Then $\ccr' \wgeq \ccr$. 
	\looseness=-1
\end{prop}

So, as shown in Fig.~\ref{fig-CocritQuasiStrongerEpirec}, we start with a morphism $g$ 
definable by $\ccr$ and use the two final segments to factor it 
as 
a morphism $g_0$ definable by $\ccr'$  and a remainder morphism $m(B)$.  
\looseness=-1

\begin{figure*}[!ht]
	\vspace*{-2.8ex}
	$$
	\xymatrix@C=3.0pc@R=1.3pc{		
		\Ccat_0 \su \Ccat \ar@/_1.5pc/[dd]_{R}  \ar[d]^{\exists F} &  
		 C 
		\ar@/_1.5pc/[rrrr]^(.5){!_{C,\,J}}
		\ar[rr]^{m_1(C)} 
		&& o_1(C) \ar[rr]^{!_{o_1(C),\,J}} && J
		\\
		\Ccat' \ar[d]^{R'}  
		& 
		&& C' = F\;o_1(C)   
		  \ar[rr]^(.5){F\;!_{o_1(C),\,J} \,=\, !_{C',\,J'}}
		&&  J' = F\,J   
		\\
		\Bcat_0 \su  \Bcat &    
		B = R\;C		
		\ar@/_1.5pc/[rrrr]^(.5){g \,=\, R\;!_{C,\,J}}
		\ar[rr]^(.6){m(B) \,=\, R\;m_1(C)}  
		&& o(B) 
		\ar[rr]^{g_0 = R\,!_{o_1(C),\,J} = R'\,!_{C',\,J'}}
		&& 
		R\,J \!=\! R' J' \!=\! T  
	}
	\vspace*{-1.8ex}
	$$
	\caption{Gentler criterion for comparing corecursor expressiveness} 
	\label{fig-CocritQuasiStrongerEpirec}
	\vspace*{-1.7ex}
\end{figure*}

We also claimed in the main paper, that, unlike in the case of nominal recursors, this gentler comparison relation and criterion are unlikely to bring anything new in terms of concrete nominal corecursor comparisons. This is because, whereas in the case of recursor models $C$ we could fruitfully take $m_1(C):o_1(C)\ra C$ to be submodels where properties like finite support would hold and would enable the equivalence of different axiomatizations, here, in the dual case, our best bet would be to take $m_1(C):C \ra o_1(C)$ to be something like quotients---which would be unlikely to preserve even the given axiomatizations (with the \emph{conditional} equations and Horn clauses being particularly problematic), let alone produce stronger ones.

\subsection{Details on the connection with the \cite{DBLP:journals/pacmpl/BlanchetteGPT19} corecursor} 
\label{con-blanchette}

The syntax of $\lambda$-calculus is obtained by instantiating the \cite{DBLP:journals/pacmpl/BlanchetteGPT19}  binder type $F$ to the four-argument functor $F(A_1,A_2,T_1,T_2) = A_1 + T_1 \times T_1 + A_2 \times T_2$ and their binder dispatcher $\theta$ to $\{(2,2)\}$.  Here $A_1$ and $A_2$ refer to (hypothetical) types of free and  bound variables respectively, and $T_1$ and $T_2$ to (hypothetical) types of terms; and $\theta$ says that the second type of variables binds in the second type of terms. Our set $\ITrm$ of iterms and its constructors and destructor are obtained as the final solution of the equation (isomorphism)  $\ITrm \simeq_{\theta}  F(\Var,\Var,\ITrm,\ITrm)$, i.e., $\ITrm \simeq_{\theta}  \Var + \ITrm \times \ITrm + \Var \times \ITrm$, 
where the index $\theta$ indicates the quotienting modulo the $\alpha$-equivalence induced by the binder dispatcher $\theta$.  (Their equation is actually solved polymorphically in the variable type, but above we instantiated that type to $\Var$.) Their free-variable 
operator $\FVars : \ITrm \ra \Pow(\Var)$ is exactly our $\FV$. 
When restricted to permutations, 
their map operator 
$\Fmap : (\Var \ra \Var) \ra \ITrm \ra \ITrm$ 
is our permutation operator $[\_] : \ITrm \ra \Perm \ra \ITrm$ with reverse order of arguments. 

The \cite{DBLP:journals/pacmpl/BlanchetteGPT19} corecursor is described in \S7.2 of the cited paper.  Again fixing the type of variables to $\Var$ (which they instead keep polymorphic), 
their comodels (introduced in their definition 25), which are the targets for their recursors, become our $(\Sigma_2,\Props_2)$-models (for the perm/free variant corecursor $\ccr_2$) from \S\ref{subsec-hiarNomCorec}, provided we remove one of their unnecessary axioms:  
\begin{itemize}
	\item their term-like structure axioms (from their definition 20) correspond to our \PmId{}, \PmCp{}, \PmFv{} and \FvPm{} axioms;\footnote{This is one axiom more than what we assume in $\Props_2$; namely their definition 20's last axiom, which corresponds to \FvPm{}, is not in $\Props_2$ because it is not needed for our recursor.}
	\item their \textsf{DRen} axiom corresponds to our \IPmBvr{} axiom
	\item their \textsf{MD} axiom, when split according to the three summands of the underlying sum type, corresponds to our \IPmVr{}, \IPmAp{} and \IPmLm{}; 
	\item their \textsf{VD} axiom, again when split across the sum type, corresponds to our \IFrVr{}, \IFrAp{} and \IFrLm{}.
\end{itemize}
The conclusion of their corecursion theorem (Theorem 26) states the existence and uniqueness of a function subcommuting with the destructor, commuting with mapping and preserving the free variables---which in this case is the same as a morphism of $\Sigma_2$-models.

\section{Proof Sketches for the Corecursor Results} 
\label{app-proofCoSketches} 


\ \\
\textbf{Proof of Prop.~\ref{prop-extCoCriterion}.}  
The proof is dual to that of Prop.~\ref{prop-extCriterion}: 
Assume $g:B \ra T$ is definable by $r$, meaning that $g = R\;\im_{C,J}$ for some $C$ in $\Ccat$. 
Let $C' = F\;C$. 
By the finality of $J'$ and the fact that $F\;J = J'$, we have that 
$\im_{C',J'} = F\;\im_{C,J}$. Hence 
$g = R\;\im_{C,J} = R'\;(F\;\im_{C,J}) = R'\;\im_{C',J'} $, 
meaning that $g$ is definable by $r'$.  \qed

\ \\
\textbf{Proof of Prop.~\ref{prop-CoWeakExtCriterion}.}  Again, dual to that of 
Prop.~\ref{prop-WeakExtCriterion}. \qed 

\ \\
\textbf{Proof of Thm.~\ref{thm-exprCo}.}   
The proof of all inequalities $\ccr_i \geq \ccr_j$  in this theorem use Prop.~\ref{prop-extCoCriterion}, 
so we show how to (functorially) 
transform $(\Sigma_j,\Props_j)$-models to $(\Sigma_i,\Props_i)$-models in such a manner that $\ITTrm(\Sigma_j)$ 
becomes $\ITTrm(\Sigma_i)$. As before for recursors, we informally discuss these transformations and highlight the intuitions behind these expressiveness results. 

\textbf{Proof of $\ccr_3 \equiv \ccr_1$:} 
%
The correspondence between the swapping and permutation operators 
proceeds 
like in the proof of $r_3 \equiv r_1$ from Thm.~\ref{thm-expr} (which in turn is based on 
\cite[Section 6.1]{pitts_2013}). In short, just like there, we are able to move bijectively (and functorially) 
between $\Sigma_3$-models of \ISwId{}, \ISwIv{}, \ISwCp{} (i.e., pre-nominal sets axiomatized via swapping) 
and $\Sigma_1$-models of \IPmId{}, \IPmCp{} (i.e., pre-nominal sets axiomatized via permutation). 
Moreover, it is not hard to prove that, along this correspondence:
\begin{itemize}
	\item the properties expressing commutations of swapping or permutation with (the three components of) the 
	destructor, namely \ISwVr{}, \ISwAp{}, \ISwAp{} versus \IPmVr{}, \IPmAp{}, \IPmAp{}, correspond to each other;
	\item and so do the support-defining and bound-variable-renaming properties, namely \IFvDSw{}, \ISwBvrT{} versus 
	\IFvDPm{}, \IPmBvr{}. 
\end{itemize}
If we ignore the destructor part, what we ended up proving here is a variation of \citet{pitts_2013}'s result---not for nominal sets (i.e, finitely supported pre-nominal sets), but for 
countably-supported pre-nominal sets (though in the presence of an uncountable number of variables/atoms). 

\textbf{Proof of $\ccr_6 \geq \ccr_3$:}  We show that any $(\Sigma_3,\Props_3)$-model is 
a $(\Sigma_6,\Props_6)$-model via the usual translation of freeness into freshness. 
We need to show that (via this freeness-freshness translation) 
the $\Props_3$ axioms imply the $\Props_6$ axioms. 
First, we note that the conjunction of \ISwId{} and \ISwBvr{} (the latter being the freshness counterpart of \ISwBvrT{}) implies  \ISwCg{}: we fulfil the existential 
in the statement of \ISwCg{} by taking $z$ to be $x_2$. 
So we are left to show that, if we define freeness from swapping via \IFvDSw{} (i.e., employing the countability predicate), 
the ``expected'' properties that connect freeness/freshness with the destructor
(\IFrVr{}, \IFrAp{}, \IFrLm{}) and with swapping (\ISwFr{}, \IFrSw{}) hold. 
All these follow from the closure properties of countable sets and the structural properties of swapping (i.e., the pre-nomional 
set axioms). 
Thus, $\ccr_6 \geq \ccr_3$ follows from Prop.~\ref{prop-extCoCriterion} using the freeness-to-freshness translation functor. 

\textbf{Proof of $\ccr_5 \geq \ccr_6$:}  
Here, the signatures are equal ($\Sigma_5 = \Sigma_6$), and we employ the identity functor 
after showing that the $\Props_6$  axioms imply the $\Props_5$ axioms. We do this by showing 
that, in the presence of the other $\Props_6$ axioms, \ISwCg{} implies \ISwBvr{}. 
Indeed, assume $\Dest^\MM\;n = \Ll\;K$ and $\{(x,m),(x',m')\} \su K$. 
\ISwCg{} gives us a (quasi)fresh $z$ such that $m[z\sw x] = m'[z\sw x']$.
From the algebraic properties of swapping and \ISwFr{},  
we get 
$m[x'\sw x] = m[z \sw x][z \sw x'] = m'$; 
moreover, from the freshness of $z$ and \IFrSw{}, we obtain that 
$x' \not= x$ implies $x'\;\fresh^\MM\, m$, as desired for \ISwBvr{}. 
(In summary, thanks to the algebraic properties of swapping and freshness, 
we are able to use the weaker axiom \ISwCg{} to establish \ISwBvr{} by taking a roundabout 
through an ``auxiliary'' fresh variable $z$.)
%

\textbf{Proof of $\ccr_2 \equiv \ccr_5$:}  Already from the proof of 
$r_3 \equiv r_1$ (and $\ccr_3 \equiv \ccr_1$) we know that $\{$\ISwId{}, \ISwIv{}, \ISwCp{}$\}$
 and $\{$\IPmId{}, \IPmCp{}$\}$ correspond to each other (via a correspondence between swapping and permutation), in that one can move bijectively and functorially between models of one group and models of the other group. 
It is immediate to show that (if we further apply the freeness-freshness translation) this correspondence extends to \ISwFr{} versus \IPmFv{} 
(simply using that $\_[z_1\llra z_2]$ is the same as $\_[z_1\sw z_2]$ in this correspondence). Finally, by induction on the definition of permutation from swapping we can show that the correspondence also extends to $\{$\ISwVr{}, \ISwAp{}, \ISwAp{}$\}$ versus $\{$\IPmVr{}, \IPmAp{}, \IPmAp{}$\}$.

\textbf{Proof of $\ccr_{9} \geq \ccr_{8}$:} 
The signatures are the same, and we show that every $\Props_8$-model is a $\Props_9$-model. 
%
First, \IFrVr{}, \IFrAp{} and \IFrLm{} follow from the \IFrDRn{} contability-based definition of freshness from renaming, the corresponding properties of renaming (\IRnVr{}. \IRnAp{} and \IRnLmO{}), and the closure properties of countable sets. 
It remains to show that freshness (again, as defined from renaming via \IFrDRn{})  satisfies \IRnFr{}, \IRnChFr{} and \IFrRn{}. 
And indeed, \IFrDRn{} implies that, whenever 
$x \;\fresh^\MM\, m$, there exists $y\not=x$ such that $m=m[y/x]$. Using this, we can show that \IRnIm{} implies \IRnFr{}, and \IRnCh{} implies \IRnChFr{}. Finally, \IFrRnT{} already proves (in fact is equivalent to) half of \IFrRn{}, namely its left-to-right implication. The other implication, namely 
``($z = y$ or $z \;\fresh^\MM\,$) and ($y \;\fresh^\MM\,m$ or $x \not= z$) implies 
  $z \;\fresh^\MM\,m[x / y]$'', follows from  \IFrDRn{} and the closure properties of countable sets. 
 (We note that, in the analogous case of freshness from swapping, namely in the proof of $\ccr_6\geq \ccr_3$,  
   \IFrSw{} follows from \IFvDSw{} without any help of an axiom analogous to \IFrRnT{}---another virtue of swapping in comparison with renaming.)
   
\textbf{Proof of $\ccr_{9} \geq \ccr_{7}$:}  
After defining a renaming operator from the substitution operator 
as usual, we can show that all the ``\textrm{Sb}''-axioms instantiate to the corresponding ``\textrm{Rn}''-axioms.  
We note the following nuance though, which differs from the inductive case: In order to prove \IRnLmO{} from \ISbLm{}, we need to infer $x\not=z$ from $x\;\fresh^\MM\,(\Vr^\MM z)$. This not \IFrVr{} but its converse, which in turn follows from \VrInv{} and  \FrVr{}.


\textbf{Proof of $\ccr_{3} \geq \ccr_{8}$:}  
From \cite{DBLP:conf/cade/Popescu22}, we know that every renset,  i.e., model of \IRnId{}, \IRnIm{}, \IRnCh{}, \IRnCm{}  
of finite support 
gives rise (in a functorial manner) to a nominal set (in the swapping-based axiomatization \cite[Section 6.1]{pitts_2013}), i.e., a model of \ISwId{}, \ISwIv{}, \ISwCp{}  
of finite support. 
The idea is to define swapping, say, of $z_1$ with $z_2$, from renaming using the standard trick of an intermediate fresh variable $y$: first rename $z_1$ to $y$, then $z_2$ to $z_1$, and finally $y$ to $z_2$; such a fresh $y$ exists thanks to the renset being finitely supported; the nominal set properties then follow from the renset properties, after showing that the choice of $y$ does not matter. 

A similar proof works here, but using countable support (\IFSupFr{}) instead of finite support and taking advantage of the fact that we have uncountably many variables. Moreover (again using the freshness-freeness translation), \IFvDSw{} follows from \IFrDSw{}, and \ISwBvr{} follows from \IRnBvr{}.  
Finally, the destructor-commutation properties of swapping, \ISwVr{}, \ISwAp{} and \ISwLm{}, follow from the corresponding properties of renaming, \IRnVr{}, \IRnAp{} and \IRnLmO{}. 
In order to infer \ISwLm{} from (the definition of swapping from renaming and) \IRnLmO{}, we also need \IIRnBvr{}; this is because  \ISwLm{} expresses unconditional commutation, whereas \IRnLmO{} conditions commutation by freshness, and \IIRnBvr{} is needed to provide the necessary ``refresher'' to bridge this gap. 

\textbf{Proof of $\ccr_{5} \geq \ccr_{9}$:}  
The proof is similar to that of $\ccr_{3} \geq \ccr_{8}$, noting that the construction of swapping from renaming and the proof of its properties are independent from the tight coupling of freshness/freeness with renaming or swapping (via \FvDSw{} or \FrDRn{}). 
\qed

\medskip
In order to prove Thm.~\ref{thm-allNominalCoRecs} without having to prove eight different corecursion theorems, we use a similar trick to that for recursors described 
in \S\ref{app-subsec-backToProofRecThms}. Namely:
\begin{itemize}
	\item we prove the corecursion theorem only in 
	a most expressive case, $\ccr_2$; 
	\item we use a slight generalization of Thm.~\ref{thm-exprCo} (which assumes pre-epi-corecursors rather than epi-corecursors, but has essentially the same proof as the one sketched above for Thm.~\ref{thm-exprCo}) to borrow the result for $\ccr_2$ to the other seven cases, thus inferring the other seven epi-corecursors from $\ccr_2$. 
\end{itemize} 

The relevant definition and proposition follow---they are dual to those from \S\ref{app-subsec-backToProofRecThms}. 

\begin{defi}\rm 
	A \emph{pre-epi-corecursor} is a tuple $\ccr = (\Bcat,T,\Ccat,J,R)$ subject to the same condition as an epi-corecursor, but without the requirement that $J$ is the final object of $\Ccat$. 
	
	A pre-epi-corecursor is called \emph{tight}  
	if the following hold:
	\begin{myitem} 
		\item $T$ is a quasi-final object in $\Bcat$ (in that for every object $B$ in $\Bcat$ there exists at most one morphism from $B$ to $T$). 
		\item The functor $R$ is faithful (in that it is injective on morphisms).   	\qed 
	\end{myitem} 
\end{defi}
\medskip

All our pre-epi-corecursors $\ccr_i$ 
are tight. 
Indeed, the model $T=\ITTrm(\Sigmad)$ is quasi-final because its coinduction principle,  Prop.~\ref{prop-coind-iterm}, is stronger than 
that of the pre-iterm model, Prop.~\ref{prop-coind-piterm}; in other words, Prop.~\ref{prop-coind-piterm} holds for iterms as well, making 
$T$ a fully abstract (hence quasi-final) model. 
Moreover, in each case the morphism component of the functor $R$ is the identity. 

\begin{prop}\rm \label{prop-Coborrow} 
	Assume the following:
	\begin{itemize}
		\item $(\Bcat,T,\Ccat,J,R)$  is a tight pre-epi-recursor
		\item $\ccr'= (\Bcat,T,\Ccat',J',R') $ is an epi-corecursor
		\item The hypotheses of
		Prop.~\ref{prop-extCoCriterion} 
		hold, with the additional property that the pre-functor $F$ is full (i.e., it is surjecive on morphisms). 
	\end{itemize}

	Then $r$ is an epi-corecursor (i.e., $J$ is final).  
\end{prop} 
\begin{proof}
	Dual to the proof of Prop.~\ref{prop-borrow}. \qed 
\end{proof}
\medskip

\textbf{Proof of Thm.~\ref{thm-allNominalCoRecs}.}   
In light of the above discussion, is suffices to prove that $\ccr_2$ is an epi-corecursor, i.e., that $\ITTrm(\Sigma_2)$ is the final  $(\Sigma_2,\Props_2)$-model. So let $\MM$ be a $(\Sigma_2,\Props_2)$-model. We need to show that there exists a unique morphism $g : \MM \ra \ITTrm(\Sigma_2)$. 
As usual, we write $M$ for the carrier of $\MM$, $\PDest^\MM$ for its destructor, etc. 
The proof follows a similar route to (and is essentially aa generalization of) that we described for 
the substitution operator in \S\ref{appsub-iterm}. 

We first define a function to pre-iterms, $g' : M \ra \PITrm$, using (standard) pre-iterm corecursion: 
$$
\begin{array}{l}
g'\,m = \mbox{case $\PDest^\MM m$ of}
\\
\hspace*{8ex}\mid \Vv\;x \Ra \PVr\;x
\\
\hspace*{8ex}\mid \Aa\,(m_1,m_2) \Ra \PAp\;(g'\,m_1)\;(g'\,m_2)
\\
\hspace*{8ex}\mid \Ll\;K \Ra 
\mbox{let $(x,m') \in K$ in }
\PLm\;x\;(g'\,m') 
\end{array} 
$$
and then define $g : M \ra \ITrm$ by $g\;m = (g'\,m)/\equiv\,$. 
Note that the definition of $g'$, hence that of $g$ too, depends on a choice of a pair $(x,m')$ in $K$ (which is guaranteed to be non-empty).

The above definitions immediately imply that $g$ commutes with the variable and application cases of the destructor, namely
\begin{itemize}
	\item[(1)] $\Dest^\MM\,m = \Vv\;x$ implies $g\;m = \Vr\;x$, and 
	\item[(2)] $\Dest^\MM\,m = \Aa\,(m_1,m_2)$ implies $g\;m = \Ap\;(g\;m_1)\;(g\;m_2)$
\end{itemize}
but the problematic case is the abstraction case, where so far we only know:
\begin{itemize}
	\item[(3)] $\Dest^\MM\,m = \Ll\;K$ implies that \emph{there exists} $(x,m') \in K$ such that $g\;m = \Lm\;x\;(g\;m')$.
\end{itemize}
What we want for a morphism is a stronger version of (3) that replaces ``there exists'' with ``for all'. 

By freshness coinduction (Prop.~\ref{prop-coind-fresh}), using that $\MM$ satisfies \FrVr{}, \FrAp{} and \FrLm{}, we can now prove:
\begin{itemize}
\item[(4)] $g$ preserves freshness/freeness, in that $x \notin \FV^\MM\,m$ implies $x\;\fresh\;(g\;m)$; or, using free-variable notation for iterms, $x \notin \FV^\MM\,m$ implies $x \notin \FV\;(g\;m)$; i.e., 
$\FV\;(g\;m) \su \FV^\MM\,m$. 
\end{itemize}

To prove that $g$ commutes with permutation is trickier, and requires a generalization. Namely, we prove:
\begin{itemize}
	\item[(5)] $g(m[\sigma]^\MM)[\tau] = g(m)[\tau\circ \sigma]$ for all $m\in M$ and $\tau,\sigma\in\Perm$. 
	\end{itemize} 
This follows by iterm coinduction (Prop.~\ref{prop-coind-iterm}) using (1)--(4) and the fact that $\MM$ satisfies \PmVr{}, \PmAp{}, \PmLm{}, as well as  \PmCp{}, \PmFv{}, \PmBvr{}---this last group of properties is needed in addition to \PmLm{} for the case when $\Dest^\MM\,m$ is an abstraction. 

From (5) and the fact that $\MM$ satisfies \PmId{}, we immediately get 
\begin{itemize}
	\item[(6)] $g(m)[\tau] = g(m)[\tau]$ for all $m\in M$ and $\tau\in\Perm$,  
\end{itemize} 
i.e., commutation of $g$ with permutation.  
Now, from (3), (6) and the fact that $\MM$ satisfies \IPmBvr{}, we obtain the stronger version of commutation with abstractions: 
\begin{itemize}
	\item[(3')] $\Dest^\MM\,m = \Ll\;K$ implies that $g\;m = \Lm\;x\;(g\;m')$ \emph{for all} $(x,m') \in K$. 
\end{itemize}
Properties (1), (2), (3'), (4), (6) mean that $g$ is a morphism of $\Sigma_2$-models. Finally, the uniqueness of such a morphism, actually even more strongly the uniqueness of any function satisfying (1), (2) and (3'), follows by iterm coinduction (Prop.~\ref{prop-coind-iterm}). 
\qed  

\medskip
A note on the above proof: To prove the central fact (5), it was 
important to work with entire permutations rather than just swapping, essentially because the abstraction case in the proof by iterm coinduction adds a composition with a transposition (a reminiscence of the definition of $\alpha$-equivalence). This is why it seems hopeless to have a corecursor that is based on swapping (i.e., single-transposition permutation) without assuming the axioms necessary to extend swapping to permutation---which contrasts with the situation of recursors, where that was possible and yielded for swapping $\geq$-stronger recursors than for permutation (as seen with $r_4,r_5$ and $r_6$).

\section{Enhancements to the Corecursors} 
\label{app-enhCorec}

In \S\ref{app-addingBacknhancements} we discussed the notion of enhancing the nominal recursors along two main axes: (1) shifting from iteration to full recursion and (2) adding support for Barendregt's variable convention. 

The Baredregt enhancement does not seem to make sense in the case of nominal corecursors. But such an enhacement does not seems to be needed in the first place, essentially because any type of bound-variable avoidance condition can be integrated in the domain of the chosen model---indeed, unlike in the case of recursion, this is possible for corecursion because we have flexibility in the domain (rather than the codomain) of the to-be-defined function. We will illustrate this phenomenon in \S\ref{app-exaCorec} with the corecursive counterpart of the  paradigmatic situation that in the recursive case calls for Barendregt's convention: the definition of (parallel) substitution on iterms. 

\medskip
On the other hand, the enhancement of coiteration to full (structural) corecursion is possible for nominal corecursors, and is fairly straightforward: A full corecursion principle can be inferred from the coiteration principle similarly to how this is done for standard codatatypes. Below we illustrate this on the swap/fresh variant corecursor, $\ccr_5$. 

So we know from Thm.~\ref{thm-allNominalCoRecs} that 
$\ITTrm(\Sigma_5)$ is the final $(\Sigma_5,\Props_5)$-model. This means that, 
for all $(\Sigma_5,\Props_5)$-models $\MM$, there exists a unique morphism $g : \MM \ra \ITTrm(\Sigma_5)$, i.e., a unique function $g : M \ra \ITrm$ such that the following hold (where for better readability we write the sub-commutation of $g$ with the destructor in the alternative form that employs constructors for iterms): 
\begin{itemize}
	\item[(1)] $\Dest^\MM\,m = \Vv\;x$ implies $g\;m = \Vr\;x$ 
	\item[(2)] $\Dest^\MM\,m = \Aa\,(m_1,m_2)$ implies $g\;m = \Ap\;(g\;m_1)\;(g\;m_2)$
	\item[(3)] $\Dest^\MM\,m = \Ll\;K$ implies that $g\;m = \Lm\;x\;(g\;m')$ for all $(x,m') \in K$
	\item[(4)] $g(m[z_1 \sw z_2]^\MM) = g(m)[z_1 \sw z_2]$
	\item[(5)] $x\;\fresh^\MM\;m$ implies $x\;\fresh\;(g\;m)$
\end{itemize}

For full recursion, we consider \emph{generalized $\Sigma_5$-models} $\MM$, whose destructors $\Dest^{\MM}$ have type not $M \ra \Var + M \times M + \PPne(\Var \times M)$, but 
$M \ra (\Var + M \times M + \PPne(\Var \times M)) \hlt{+\, \ITrm}$. The purpose of this $\ITrm$ summand is (just like for standard codatatypes) the possibility to allow an immediate exit from the corecursion calls by returning an iterm. 
Morphisms of generalized $\Sigma_5$-models are defined as one would expect, in that commutation with the destructor now means 
$((1_\Var + g \times g + \img(1_\Var \times g))\hlt{+\, 1_{\ITrm}})\,(\Dest^{\MM} m) \sqsubseteq  \Dest^{\MM}(g\;m)$. Above, writing $\In_1$ and $\In_2$ for the two injections into the sum type, $\sqsubseteq$ on $(\Var + M' \times M' + \PPne(\Var \times M')) +\, \ITrm$ 
is defined by taking $u \leq v$ to mean that: 
either $u=\In_1(\Vv\;x) = v$ for some $x$; 
or $u=\In_1(\Aa(m_1',m_2')) = v$ for some $m_1',m_2'$;
or $u=\In_1(\Ll\;K)$, $v = \In_1(\Ll\;K')$ and $K\su K'$ for some $K,K'$; 
or $u = \In_2(t) = \In_2(t) = v$ for some $t\in\ITrm$. 
This condition is equivalent to the conjunction of the following four conditions:  
\begin{itemize}
	\item $\Dest^{\MM}\,m = \In_1(\Vv\;x)$ implies $g\;m = \Vr\;x$; 
	\item $\Dest^{\MM}\,m = \In_1(\Aa (m_1,m_2))$ implies $g\;m = \Ap\;(g\;m_1)\;(g\;m_2)$; 
	\item $\Dest^{\MM}\,m = \In_1(\Ll \;K)$ and $(x,m')\in K$ implies $g\;m = \Lm\;x\;(g\;m')$
	\item $\Dest^{\MM}\,m = \In_2\;t$ implies $g\;m = t$. 
\end{itemize}

Iterms become a generalized $\Sigma_5$-model by extending the destructor $\Dest : \ITrm \ra \Var + 
\ITrm \times \ITrm + \PPne(\Var\times \ITrm)$ to a ``generalized destructor''
$\Dest' : \ITrm \ra (\Var + 
\ITrm \times \ITrm + \PPne(\Var\times \ITrm)) \hlt{+\;\ITrm}$ defined by:
$\Dest'\;t = \In_1(\Dest\;t)$. Let us call this model $\ITTrm'(\Sigma_5)$. 

The notion of a generalized model satisfying a property from $\Props_5$ is extended from (standard) models in a straightforward manner, by simply inserting the $\In_1$ injection. For example, $\MM$ satisfying \ISwAp{} means: For all $m,m_1m_2\in M$ and $z_1,z_2\in\Var$, if $\Dest^\MM m = \In_1(\Aa(m_1,m_2))$ then 
$\Dest^\MM(m[z_1 \sw z_2]^\MM) = \Aa(m_1[z_1 \sw z_2]^\MM,m_2[z_1 \sw z_2]^\MM)$. 

Now, the full recursion principle states that $\ITTrm'(\Sigma_5)$ is the final generalized 
$(\Sigma_5,\Props_5)$-model, and can be proved from the iteration principle as follows. Let $\MM'$ be a generalized $(\Sigma_5,\Props_5)$-model. We build from it a (standard) $(\Sigma_5,\Props_5)$-model $\MM$ 
on the carrier set $M' + \ITrm$ 
by combining the operators of $\MM'$ with those of iterms: 
\begin{itemize}
	\item $M = M' + \ITrm$
	\item $\Dest^\MM\;(\In_1\;m) = \mbox{case $\Dest^{\MM'} m$ of}$
    \\$\begin{array}{l}
		\hspace*{25ex}\mid \Vv\;x \Ra \In_1(\Vv\;x) 
		\\
		\hspace*{25ex}\mid \Aa(m_1,m_2) \Ra \In_1(\Aa(m_1,m_2)) 
		\\
		\hspace*{25ex}\mid \Ll\K \Ra \In_1\,(\Ll\,\{(x,\In_1\;m) \mid (x,m') \in \Ll\;K\})
	\end{array} 
  $
  \\
  $\Dest^\MM\;(\In_2\;t) = \In_2(\Dest\;t)$
  \item $(\In_1\;m)[z_1\sw z_2]^\MM = \In_1(m [z_1\sw z_2]^{\MM'})$
  \hspace*{5ex}$(\In_2\;t)[z_1\sw z_2]^\MM = \In_2(t [z_1\sw z_2])$
   \item $z\;\fresh^\MM\,(\In_1\;m) \iff z\;\fresh^{\MM'} m $
  \hspace*{16ex}
    $z\;\fresh^\MM\,(\In_2\;t) \iff z\;\fresh\;t$
\end{itemize}
That $\MM$ satisfies $\Props_5$ follows from the fact that $\MM'$ and the model of iterms do. From coiteration, this gives us a unique morphism $g : \MM \ra \ITTrm(\Sigma_5)$, i.e., a unique function $g: M' + \ITrm = M \ra \ITrm$ satisfying clauses (1)--(5) above. Finally, we define $g': M' \ra \ITrm$ by $g'\,m = g(\In_1\;m)$. Then (1)--(5) imply that $g'$ satisfies the clauses:
\begin{itemize}
	\item[(1')] $\Dest^{\MM'}\,m = \In_1(\Vv\;x)$ implies $g'\,m = \Vr\;x$ 
	\item[(2')] $\Dest^{\MM'}\,m = \In_1(\Aa\,(m_1,m_2))$ implies $g'\,m = \Ap\;(g'\,m_1)\;(g'\,m_2)$
	\item[(3')] $\Dest^{\MM'}\,m = \In_1(\Ll\;K)$ implies that $g'\,m = \Lm\;x\;(g\;m')$ for all $(x,m') \in K$
	\item[(4')] $g'(m[z_1 \sw z_2]^{\MM'}) = g'(m)[z_1 \sw z_2]$
	\item[(5')] $x\;\fresh^{\MM'} m$ implies $x\;\fresh\;(g'\,m)$
\end{itemize}
and additionally the following follows by iterm coinduction:
\begin{itemize}
	\item[(6')] $\Dest^{\MM'}\,m = \In_2\,t$ implies $g'\,m = t$ 
\end{itemize} 
Clauses (1')--(6') mean that $g'$ is a morphism of generalized $\Sigma_5$-models, so $g' : \MM' \ra \ITTrm'(\Sigma_5)$.  The uniqueness of $g'$ follows again by iterm coinduction.

\section{Example of deploying a nominal corecursor} 
\label{app-exaCorec}

Next we show how the (capture-free) parallel substitution operator can be defined using the swap/fresh variant recursor $\ccr_5$. To keep the definition simple, we will use the full recursion enhancement of $\ccr_5$ described in \S\ref{app-enhCorec}. Let $\Env$ be the set \emph{(variable-term) environments}, which are functions $\rho:\Var \ra \ITrm$ whose support $\supp\;\rho$ is countable, where  $\supp\;\rho$ is defined to consist of all the variables $x$ that are changed by $\rho$ (in that $\rho\;x \not= \Vr\;x$) and all the free variables of the images of such variables, $\FV\;(\rho\;x)$; in other words, $\supp\;\rho = \bigcup_{x\in\Var,\rho\;x\not=\Vr\;x} (\{x\} \cup \FV(\rho\;x))$. (Note that, since iterms have countably many free variables, for $\supp\;\rho$ to be countable it suffices that the smaller set $\{x\in\Var \mid x\not=\Vr\;x\}$ is countable.) 
We wish to define $\psubst : \ITrm \ra  \Env \ra \ITrm$ satisfying the following clauses:
\begin{itemize}
	\item[(1')] $\psubst\;(\Vr\;x)\;\rho = \rho\;x$
	\item[(2')] $\psubst\;(\Ap\;t_1\;t_2)\;\rho = \Ap\;(\psubst\;t_1\;\rho)\;(\psubst\;t_2\;\rho)$
	\item[(3')] $\psubst\;(\Lm\;x\;t)\;\rho = \Lm\;x\;(\psubst\;t\;\rho)$ if $x\notin\supp\;\rho$ 
\end{itemize}
which can be reformulated as follows using the iterm destructor: 
\begin{itemize}
	\item[(1)] $\Dest\;t = \Vv\;x$ implies 
	$\psubst\;t\;\rho = \rho\;x$
	\item[(2)] $\Dest\;t = \Aa(t_1,t_2)$ implies 
	$\Dest\,(\psubst\;t\;\rho) = \Aa\, (\psubst\;t_1\;\rho,\psubst\;t_2\;\rho)$
	\item[(3)] $\Dest\;t = \Ll\;K$ implies that there exists $K'$ such that
	$\Dest\,(\psubst\;t\;\rho) = \Ll\;K'$ and $
	\{ (x,\psubst\;t'\;\rho) \mid   (x,t') \in K \mbox{ and } x\notin\supp\;\rho\} \su K'$ 
\end{itemize}
Asking how this to-be-defined function is supposed to interact with swapping and freshness, we obtain the following additional desired clauses:
\begin{itemize}
	\item[(4)] $(\psubst\;t\;\rho)[z_1\sw z_2] = 
	\psubst\;(t[z_1\sw z_2])\;(\rho[z_1\sw z_2])$ 
	\\where $\rho[z_1\sw z_2]$ is defined as $ \lambda x.\;\rho(x[z_1\sw z_2])[z_1\sw z_2]$ 
    \item[(5)] $x\;\fresh\;t$ and $x\notin\supp\;\rho$ implies 
    $x\;\fresh\;(\psubst\;t\;\rho)$
\end{itemize}

Now, clauses (1)--(5) determine the following generalized $\Sigma_5$-model $\MM$ of carrier set $M = \ITrm \times \Env$ (where clause (1) represents an ``exit'' point and thus takes advantage of the extra flexibility of full corecursion):
\begin{itemize}
	\item[(1m)] If $\Dest\;t = \Vv\;x$ then we define $\Dest^\MM(t,\rho) = \In_2\,(\rho\;x)$
	\item[(2m)] 
	If $\Dest\;t = \Aa(t_1,t_2)$ then we define 
	$\Dest^\MM (t,\rho) = \In_1(\Aa\,((t_1,\rho),(t_2,\rho)))$
	\item[(3m)] If $\Dest\;t = \Ll\;K$ then we define 
	$\Dest^\MM(t,\rho) =
	\In_1(\Ll\,\{ (x,(t',\rho)) \mid   (x,t') \in K \mbox{ and } x\notin\supp\;\rho\})$ 
	\item[(4m)] We define $(t,\rho)[z_1\sw z_2]^\MM = 
	(t[z_1\sw z_2],\rho[z_1\sw z_2])$  
	\item[(5m)] We define $x\;\fresh^\MM\;(t,\rho)$ 
	to mean 
	$x\;\fresh\;t$ and $x\notin\supp\;\rho$	
\end{itemize}

Note that the above definitions (1m)--(5m) of the operators of $\MM$ mirror the clauses (1)--(5).\footnote{This is similar to the situation we discussed for nominal recursors in \S\ref{subsec-purposeNomRec}.
Here we have some flexibility about the model's destructor in the abstraction case, because clause (3) states an inclusion; (3m) chooses the minimal solution to satisfy   (3).} Thus, stated about a presumptive function $\psubst : \ITrm \ra  \Env \ra \ITrm$, clauses (1)--(5) 
mean exactly that the curried version of $\psubst$, namely 
$\lambda (t,\rho).\;\psubst\;t\;\rho$, is a morphism of generalized $\Sigma_5$-models 
between $\ITTrm'(\Sigma_5)$ (the generalized model of iterms) and $\MM$. Thus, thanks to the recursion theorem for $\ccr_5$ (the full recursion version) 
all we need to do in order to obtain the desired function $\psubst$ satisfying (1)--(5) is to show that $\MM$ is $(\Sigma_5,\Props_5)$-model, i.e., it satisfies $\Props_5$. 

And this last fact follows routinely from the definitions and the properties of iterms. For example, the fact that $\MM$ satisfies \SwFr{} means: 
\begin{center}
For all $m\in M$, if $z_1 ,z_2\,\fresh^\MM m$ then 
$m[z_1\sw z_2]^\MM = m$. 
\end{center}
which means, using the definitions of $\MM$'s carrier and operators:
\begin{center}
	For all $t\in \ITrm$ and $\rho\in\Env$, if $z_1 ,z_2\,\fresh\,t$ and $z_1,z_2\notin\supp\;\rho$ 
	then 
	$t[z_1\sw z_2] = t$ and $\rho[z_1\sw z_2] = \rho$. 
\end{center}
This follows immediately from the fact that iterms satisfy \SwFr{} and from the definitions of swapping and support for environments. 

One may wonder where in the above development we needed that the environments have countable support: It was in ensuring that the destructor $\Dest^\MM(t,\rho)$ is well defined for the case when $\Dest\;t=\Ll\;K$, in that it returns $\Ll\;K'$ for a \emph{non-empty} set $K'$. 
Indeed, in the absence of the countable support assumption, the existence of a pair $(x,t') \in K$ such that $x\notin \supp\;\rho$ is not guaranteed. 

\medskip
Note that, similarly to term substitution, iterm  substitution must avoid the capturing of free variables in the case of $\Lm$-abstractions, as shown in clause (3'). However, here we do not need any kind of Barendregt enhancement, but have a different mechanism of ensuring that: Corecursion requires us to operate with the destructor-based clause (3) instead, which makes it clear that such avoidance conditions can be factored in the domain of the to-be-defined function (substitution).

\section{Isabelle Mechanization}
\label{app-isa}

We have mechanized our results about nominal (co)recursors as epi-(co)recursors and their comparisons in the theorem prover Isabelle/HOL \cite{LNCS2283}. More precisely, we have mechanized the following results:
\begin{itemize}
	\item the recursion theorem (Thm.~\ref{thm-allNominalRecs}), also in the enhanced-recursor version (Thm.~\ref{thm-allNominalRecsEnhanced}); 
	\item the two recursor comparison theorems (Thm.~\ref{thm-expr} and Thm.~\ref{thm-qexpr}), also in the enhanced-recursor version (Thm.~\ref{thm-exprEnhanced}); 
	\item the two negative (strictness) results on recursor comparison (Props.~\ref{prop-negRec}); 
	\item the corecursion theorem (Thm.~\ref{thm-allNominalCoRecs}); 
	\item the corecursor comparison theorem (Thm.~\ref{thm-exprCo}). 
\end{itemize}
What we have \emph{not} mechanized are the abstract criteria for comparing epi-recursors and epi-corecursors (Prop.~\ref{prop-extCriterion}, Prop.~\ref{prop-WeakExtCriterion} and Prop.~\ref{prop-extCoCriterion}). In our mechanized results, rather than invoking these criteria, we have inlined their content on a need basis, as we will explain below.

The mechanization is 
provided as a publicly available archive, 
containing the Isabelle sources as well as a browsable html version (documented by a README file and by comments in the sources).    
For the recursors, it covers both the stripped-down versions 
discussed in the main paper and their enhancements discussed in App.~\ref{app-addingBacknhancements}. 

We made heavy use of Isabelle's locales \cite{DBLP:conf/tphol/KammullerWP99,DBLP:journals/jar/Ballarin14}, which we found to be an excellent abstraction mechanism for representing the expressiveness relationships between (co)recursors. The readers not interested in locales but wishing to inspect the end mechanized results in a manner than closely matches the formulations from the paper can skip to \S\ref{app-subsec-localeFree}. 

A locale fixes some types, constants and assumptions.
One can perform definitions and prove theorems inside a locale, and everything happens relative to the entities fixed in that locale. Viewed from outside the locale, all these definitions and theorems are (1) polymorphic in that locale's fixed types, (2) universally quantified over that locale's constants, and (3) conditioned by that locale's assumptions. 
%

A locale can be \emph{interpreted} at the top level of an Isabelle theory by providing concrete types and constants for that locale's 
parameter types and constants, and verifying the locale's assumptions; after a successful interpretation, all the definitions performed and theorems proved in a locale are automatically instantiated with these concrete types and constants. 
A locale $L_2$ can also be interpreted relative to another locale $L_1$ by establishing a \emph{sublocale} relationship $L_1 \leq L_2$. This amounts to showing that the entities of $L_1$ can provide an interpretation of those of $L_2$; i.e., in the context of the fixed types, constants and assumptions of $L_1$,  one indicates some types and constants that instantiate those of $L_2$, and verifies the assumptions of $L_2$. 

The traditional application of locales is in 
modularizing the development of algebraic structures, such as groups, rings, fields etc. \cite{DBLP:journals/jar/Ballarin14,DBLP:journals/jar/Ballarin20}.  
Then (top-level) interpretations provide particular examples of such structures, e.g., interpreting the ring locale into the particular ring of integers. Moreover, sublocale relationships are useful for showing the inclusion between two types of structure, e.g., fields are particular kinds of rings, or more generally for showing that one type of structure induces another type of structure. 

Our own results in this paper are also algebraic / model-theoretic in nature. We used locales and sublocales to represent and connect our different recursor and corecursor models.

\subsection{Mechanization of the results about recursors}
\label{subapp-isa-recursors}

For each of the nine types of models underlying the nominal recursors, we have introduced a locale, as shown in Fig.~\ref{fig-localesRec}. 

\newcommand\ThAll{\small \textsf{All}}

\newcommand\PermFree{\small \textsf{PermFree$\_$model}}
\newcommand\PermFreeV{\small \textsf{PermFreeV$\_$model}}
\newcommand\SwapFreeV{\small \textsf{SwapFreeV$\_$model}}
\newcommand\SwapFree{\small \textsf{SwapFree$\_$model}}
\newcommand\SwapFreshV{\small \textsf{SwapFreshV$\_$model}}
\newcommand\SwapFresh{\small \textsf{SwapFresh$\_$model}}
\newcommand\SwapFreshSM{\small \textsf{submodel$\_$SwapFresh$\_$model}}
\newcommand\SubstFresh{\small \textsf{SubstFresh$\_$model}}
\newcommand\Renaming{\small \textsf{Renaming$\_$model}}
\newcommand\RenamingFreshV{\small \textsf{RenamingFreshV$\_$model}}

\newcommand\LambdaTerms
{\small \textsf{Lambda$\_$Terms}}
\newcommand\SwapVsPerm
{\small \textsf{Swap$\_$vs$\_$Perms}}
\newcommand\SwapSeveral
{\small \textsf{SwapFresh$\_$SwapFreshV$\_$SwapFree$\_$models}}
\newcommand\RenamingIsRenamingFreshV
{\small \textsf{Renaming$\_$model$\_$is$\_$RenamingFreshV$\_$submodel}}

\newcommand\SwapFreshHasSwapFreeV
{\small \textsf{SwapFresh$\_$model$\_$has$\_$SwapFreeV$\_$submodel}}

\newcommand\RenamingFreshVR
{\small \textsf{RenamingFreshV$\_$recursor}}

\newcommand\SwapFreshR
{\small \textsf{SwapFresh$\_$recursor}}

\newcommand\SubstFreshR
{\small \textsf{SubstFresh$\_$recursor}}

\begin{figure}
	\centering
	\begin{tabular}{c|c}
		Recursor & Corresponding locale
		\\\hline 
		$r_1$ (perm/free) & $\PermFree$
		\\
		$r_2$ (perm/free variant) &  $\PermFreeV$
		\\
		$r_3$ (swap/free variant) & $\SwapFreeV$
		\\
		$r_4$ (swap/free) & $\SwapFree$
		\\
		$r_5$ (swap/fresh variant) & $\SwapFreshV$
		\\
		$r_6$ (swap/fresh) &  $\SwapFresh$
		\\
		$r_7$ (subst/fresh) & $\SubstFresh$
		\\
		$r_8$ (renaming) & $\Renaming$
		\\
		$r_9$ (renaming fresh variant) &  $\RenamingFreshV$
		\\
	\end{tabular}
	\caption{Isabelle locales corresponding to recursors}
	\label{fig-localesRec}
\end{figure}

Each locale fixes the carrier type $M$ of a model and the operations and relations on the model: constructor, permutation, swapping, substitution, renaming, free-variable and freshness operators. 
Then it postulates the respective axioms. (In the case of the enhanced recursors, the locale also fixes the domain $D \su \Trm \times M$ and assumes that $D$ is closed under the operations, as explained in App.~\ref{app-addingBacknhancements}; it also fixes a set of variables $X$, assumes its finiteness, and the axioms are stated relative to $X$, again as explained in App.~\ref{app-addingBacknhancements}.) In short, each locale axiomatizes a class of models, namely that of $(\Sigma_i,\Props_i)$-models (and $(X,\Sigma_i,\Props_i)$-models) for each recursor $r_i$.  

\subsubsection{Mechanization of the recursor comparison results}

Recall that the results on strength comparison reported in Thm.~\ref{thm-expr} (and extended to enhanced recursors in Thm.~\ref{thm-exprEnhanced}) essentially show that 
one recursor is stronger than another, say $r_i \geq r_j$, by showing that any $(\Sigma_j,\Props_j)$-model $\MM$ is (or can be regarded as) a $(\Sigma_i,\Props_i)$-model---that is, after defining on $\MM$ the $\Sigma_i$-operations. We expressed this in Isabelle as follows: Say $L_i$ and $L_j$ are the locales for these two classes of models. Working inside locale $L_j$, we defined the $\Sigma_i$ operations and proved for them the $\Props_i$ properties. This allowed us  to prove the locale relationship $L_j \leq L_i$, which is a statement of $r_i \geq r_j$. This shallow embedding of the $\geq$ relationship allowed us to concretely borrow for $(\Sigma_j,\Props_j)$-models the $r_i$ recursor, in other words to infer the $r_j$ recursor from the $r_i$ recursor. 

\begin{figure*}
	\centering
	\includegraphics[width=\linewidth]{session_graph.pdf}
	\caption{The Isabelle theories for recursors}
	\label{fig-isabelleTheories}
\end{figure*}


To illustrate this more concretely, let us consider one of the statements of Thm.~\ref{thm-expr}, say $r_4 \geq r_2$, comparing the swap/free recursor $r_4$ with the perm/free variant recursor $r_2$. All the theories we reference below are located in the directory \textsf{Stripped$\_$Down} from the archive. The mechanisation of this result has the following components:
\begin{itemize} 
\item The class of models for each recursor corresponds to an Isabelle locale, which fixes a carrier set (as an unspecified type $'D$) and operations on it as indicated by the recursor's signature, and assumes the recursor's characteristic properties (sometimes called ``axioms'' in the paper). Namely:
\begin{itemize} 
     \item The models of $r_2$ (the perm/free variant recursor) are mechanized as the locale $\PermFreeV$ (located in theory \textsf{PermFree$\_$PermFreeV$\_$models}) which fixes the type (i.e., type variable) $'D$; and on this type it fixes constructors-like operators \textsf{VrD}, \textsf{ApD} and \textsf{LmD}, and permutation- and free-variable-like operators permD and FvarsD, and assumes the model properties required by $r_2$ -- the Isabelle notations for these properties coincide with the ones from the paper, e.g., \PmVr{}, \FvAp{}, etc. 
     \item Similarly, the models of $r_4$ (the swap/free recursor) are mechanized as the locale $\SwapFree$ (located in theory \textsf{SwapFresh$\_$SwapFreshV$\_$SwapFree$\_$models}) which again fixes the necessary model components (carrier $'D$, constructor-like operators \textsf{VrD}, \textsf{ApD} and \textsf{LmD}, and swapping- and free-variable-like operators \textsf{swapD} and \textsf{FvarsD}) and assumes the model properties required by $r_4$ (namely, $\Props_4$)---again, the Isabelle notations match the paper, e.g., \SwFv{}. 
\end{itemize} 
\item The definition of an operator, let us refer to it as $F$ (since this will represent an instance of the pre-functor $F$ from the paper's Prop.~\ref{prop-extCriterion}), that maps $r_2$-models (i.e., $(\Sigma_2,\Props_2)$-models) to $r_4$-models (i.e., $(\Sigma_4,\Props_4)$-models), was mechanized as follows:
\begin{itemize} 
      \item In the context of the $\PermFreeV$ locale, which fixes an (arbitrary) $r_2$-model, we defined a swapping-like operator \textsf{swapD} on the carrier $'D$ of that model. (This happened inside the theory \textsf{PermFreeV$\_$model$\_$is$\_$SwapFree$\_$model}.) 
      \item Then we proved that \textsf{swapD}, together with the constructor-like and free-variable-like operators (already provided by $r_2$-models), forms an $r_4$-model, i.e., satisfies the  properties $\Props_4$. This happened by first proving all the $\Props_4$ properties in the context of the locale, then using these properties to establish the sublocale relationship $\PermFreeV < \SwapFree$ via the command:
      
\begin{center}
	\begin{tabular}{l}
\textsf{sublocale PermFreeV$\_$model $<$ SwapFree$\_$model}
\\
\textsf{where swapD = swapD}
\end{tabular} 
\end{center} 

This command (which triggers a proof goal that must be discharged) 
makes the statement that, under the $\PermFreeV$ 
assumptions, i.e., for any $r_2$-model, 
the \textsf{swapD} operator just defined together with the other operators from $\PermFreeV$, 
namely \textsf{VrD}, \textsf{ApD}, \textsf{LmD} and \textsf{FvarsD}, satisfy all the $\SwapFree$ assumptions, i.e., form an 
$r_4$-model. Note that this sublocale relationship implicitly refers to the other operators, in 
other words the above is equivalent to the following command: 

\begin{center}
	\begin{tabular}{l}
		\textsf{sublocale PermFreeV$\_$model $<$ SwapFree$\_$model}
		\\
		\textsf{where swapD = swapD and VrD = VrD and ApD = ApD} 
		\\
		\textsf{and LmD = LmD and FvarsD = FvarsD}
	\end{tabular} 
\end{center} 

So this sublocale mechanizes the operator $F$. By definition, $F$ leaves unchanged the $(\textsf{VrD},\textsf{ApD},\alb \textsf{LmD})$-part of the models (i.e., factors through the forgetful operators to the constructor-only signature). Also, one can see that together with the identity on morphisms, $F$ is a functor---but we do not mechanize this fact. 
\end{itemize}
\item The above locale relationship, showing that $r_2$-models give rise to $r_4$-models (via the above operator $F$), is the core of the ordering $r_4 \geq r_2$. Indeed, taking advantage of this sublocale relationship, we showed that definability via the $r_2$ recursor implies definability via the $r_4$ recursor as follows: In the context of the $\PermFreeV$ locale, i.e., for any $r_2$-model, we showed that the unique morphism of $\Sigma_2$-models guaranteed by $r_2$ coincides (as a function) with the unique morphism to the induced $\Sigma_4$-model guaranteed by $r_4$; indeed, the latter, denoted in the formalisation by $ff0$ (and automatically made available in the $\PermFreeV$ locale via the sublocale relationship), was shown to be a morphism of  $\Sigma_2$-models. (This happened in the theory \textsf{PermFreeV$\_$model$\_$is$\_$SwapFree$\_$model}.) Note that, at this stage, $ff0$ had already been available in the context of the $\SwapFree$ locale and known to be the unique morphism between the term model and the (arbitrary) $(\Sigma_4,\Props_4)$-model fixed in the $\SwapFree$ locale.
\end{itemize} 


A similar, but slightly more involved mechanism was used for mechanizing 
the quasi-strength comparison results of Thm.~\ref{thm-qexpr} 
(extended to enhanced recursors in Thm.~\ref{thm-exprEnhanced}). Remember that $r_i \wgeq r_j$ was proved by showing that any $(\Sigma_j,\Props_j)$-model $\MM$ has a $(\Sigma_i,\Props_i)$-submodel---in that there exists a submodel $\MM'$ of $\MM$ that on the one hand still satisfies $\Props_j$, and on the other hand can be regarded as a $(\Sigma_i,\Props_i)$-model (again, via defining on $\MM'$ the $\Sigma_i$-operations). We expressed this in Isabelle as follows: %
Working inside locale $L_j$, we identified a suitable subset $M'$ of $\MM$'s carrier $M$ (and, for enhanced recursors, a suitable subset of $\MM$'s domain 
$D$) and proved that it is closed under the operations and satisfies the $\Props_j$ properties---as discussed in the proof sketch of Thm.~\ref{thm-qexpr}  
from App.~\ref{app-proofSketches}, this was in each case a minimal  set closed under the constructors, defined inductively. 
In other words, we built a $(\Sigma_j,\Props_j)$-submodel. To capture this using locales, we defined the locale $L_j'$ that extends $L_j$ with a subset $M'$ that forms a $(\Sigma_j,\Props_j)$-submodel, and proved $L_j \leq L_j'$ by defining $\MM'$ to be the aforementioned minimal submodel of $\MM$.  Then we defined the 
$\Sigma_i$ operations on $M'$ (more precisely, we defined them on the entire type and proved that $M'$ is closed under them), after which we 
proved for them the $\Props_i$ properties. This allowed us  to prove the locale relationship $L_j \leq L_i$, using the submodel $\MM'$ rather than the model $\MM$ as basis for constructing the model for $L_i$. 
These two locale inclusions together 
form a statement of $r_i \wgeq r_j$. Again, this mechanized relationship is effective, in that it allowed us to infer the $r_j$ recursor from the $r_i$ recursor.

\subsubsection{Mechanization of the recursion theorems}

As a byproduct of the above network of sublocales that allows borrowing recursors, we were able to infer all the nine recursors from just two of them, namely $r_6$ (the swap-fresh recursor) and $r_9$ (the renaming/fresh variant recursor)---as discussed in the proof sketch of Thm.~\ref{thm-allNominalRecs} from 
App.~\ref{app-proofSketches}.  For $r_6$ and $r_9$, we performed direct proofs of initiality. The initial morphism was constructed by first defining inductively a relation and then proving that it is a function and it commutes with the relevant operations and preserves freshness. This approach is distinct from (and we believe simpler than) previous techniques from the literature used to prove nominal recursion  principles. For example, \citet{primrecFOAS-Norrish04} bases the proof of his swap/free recursor on a previous recursor by 
 \citet{DBLP:conf/tphol/GordonM96}, which in turn uses the lifting of a function from preterms after proving that it respects $\alpha$-equivalence. Similarly, 
  \citet{pitts-AlphaStructural}'s proof of the perm/free recursor lifts a function from preterms. 
Our approach is simpler in that it does not delve into preterms, but operates entirely at the abstraction level of terms. 

To illustrate the borrowing process, let us give again a concrete example, considering the swap/fresh recursor $r_6$: In the context of the $\SwapFresh$ locale, i.e., fixing a $(\Sigma_6,\Props_6)$-model, consisting of a type $'D$ and some operators 
\textsf{VrD}, \textsf{ApD}, \textsf{LmD}, \textsf{swapD} and \textsf{freshD} satisfying the $\Props_6$ properties, we proved the existence and uniqueness of a $\Sigma_6$-model morphism from the term model to this (arbitrary) fixed model. This was done by defining a function $ff0$ from terms to $'D$, proving that it is is a morphism of $\Sigma_6$-models (i.e., commutes with the constructors, swapping and freshness operators), and proving that any other morphism of $\Sigma_6$-models must be equal to $ff0$. All this work was performed in the theory \textsf{Swap$\_$Fresh$\_$recursor}. The relevant theorems (as indicated via comments in the formalization) are called \textsf{ff0$\_$Vr}, \textsf{ff0$\_$Ap}, \textsf{ff0$\_$Lm}, \textsf{ff0$\_$swap} and \textsf{ff0$\_$fresh} (together stating the morphism property) and \textsf{ff0$\_$unique} (stating the uniqueness property). Note that the uniqueness property is actually stated in a stronger form: not only is $ff0$ the unique $\Sigma_6$-morphism, but is even unique among $\Sigmac$-morphisms, i.e., unique among functions commuting with the constructors. 

Now, we could have done direct proofs of soundness for all our recursors (like we did for $r_6$ and $r_9$), but we noticed that we can instead use the expressiveness relationships we discovered between them as a mechanism for \emph{borrowing} soundness from the more expressive ones. However, this was not possible with the relationships as stated in the paper because those already assumed the recursors to be sound (in fact our very notion of recursor assumed soundness); but it \emph{became} possible with a slight generalization of our results. Because this generalization does not bring much conceptually and might have distracted the reader from the main ideas, we decided not to include it in the main paper but to discuss it in the appendix (App.~\ref{app-proofSketches}) as a ``formal engineering optimization''. 

\subsubsection{Theory structure} 

The theory structure of our Isabelle development for recursors is shown in Fig.~\ref{fig-isabelleTheories}.  Everything is based on a formalization of terms as $\alpha$-equivalence classes of preterms, in the theory $\LambdaTerms$. Due to the need to borrow some properties from the terms model to arbitrary models (for given signatures) via the initial morphism (as explained in the proof sketch of Thm.~\ref{thm-qexpr} from App.~\ref{app-proofSketches}), a large theory of terms had to be formalized, comprising a wealth of results about the term operators, depth-based and fresh induction principles. 
Moreover, the auxiliary theory $\SwapVsPerm$ performs the conversions between swapping-based and permutation-based axioms: starting with the classic result on switching between the two alternative axiomatizations of nominal sets as described in Pitts's monograph \cite[Section 6.1]{pitts_2013}, and extending this correspondence in various ways as needed by the various recursors: to covering a separate freshness predicate (not reducible to swapping or permutation), to relaxing nominal sets to ``nominal sets modulo $X$'' for a more comprehensive application of Barendregt's convention (as discussed in App.~\ref{app-addingBacknhancements}), etc. 

The names of the  other theories in Fig.~\ref{fig-isabelleTheories} are 
self-explanatory. For example:
\begin{mmyitem}
\item the theory $\RenamingFreshV$ formalizes the models for the renaming/fresh variant recursor (and of course contains the locale with the same name);
\item the theory $\SwapSeveral$ formalizes the models corresponding to the swap/fresh, swap/fresh variant and swap/free recursors (and contains the corresponding locales); 
\item the theory $\RenamingIsRenamingFreshV$ proves that each renaming model is (can be regarded as) a renaming/fresh variant model, via the sublocale statement $\Renaming \leq \RenamingFreshV$; 
\item the theory $\SwapFreshHasSwapFreeV$ proves that each swap/fresh model has a swap/fresh submodel that is (can be regarded as) a swap/freee variant model, via the sublocale statements 
$\SwapFresh \leq \SwapFreshSM$ and $\SwapFresh \leq \SwapFreeV$. 
\end{mmyitem}

Note that there are three theories whose names refer explicitly to a recursor: $\RenamingFreshVR$, $\SwapFreshR$ and 
$\SubstFreshR$. The first two of these contain direct formalizations of the renaming/fresh variant and swap/fresh recursors. As discussed in the proof sketch of Thm.~\ref{thm-allNominalRecs}
in App.~\ref{app-proofSketches}, these two recursors (which are at the top of the $\geq$ hierarchy) have been used to derive all the other recursors. In all but one case, we have performed this derivation right after the sublocale result that enables it. For example, the swap/free variant recursor is derived from the swap/fresh recursor in theory  $\SwapFreshHasSwapFreeV$, right after the sublocale relationship $\SwapFresh \leq \SwapFreeV$ is established. 
The exception is the subst/fresh recursor, to which we dedicated its own theory $\SubstFreshR$---this was done in order to highlight the slightly more involved structure of the borrowing argument, which requires fresh induction. 

The theory $\ThAll$ imports all the relevant top theories (and a few of the relevant non-top ones for better documentation) and contains comments that map the formalization to the paper.\footnote{For the theories that are located in the figure below \textsf{All}, i.e., import  this theory, please see \S\ref{app-subsec-localeFree}.}

\subsection{Mechanization of the negative results} 
\label{subsec-neg}
The two negative results expressed in Prop.~\ref{prop-negRec} are mechanized in the theory \textsf{Prop16}.  
The mechanized statement follows closely the presentation from \S\ref{subsec-negResRec}, in each case stating that there exist morphisms definable by one recursor but not by the other. 
The mechanized proofs  also follow closely Prop.~\ref{prop-negRec}'s proof sketch given in the paper (and the extended proof sketch given in App.~\ref{app-proofSketches}). For each of the two $r_i \not\geq r_j$ results, we:
(1) build a $(\Sigma_j,\Props_j)$-model (and prove that it is indeed a $(\Sigma_j,\Props_j)$-model), 
and (2) show that there exists no extension of 
	the $\Sigmac$-part of that model to a $(\Sigma_i,\Props_i)$-model. 
(In each case, such an extension would consist of freeness and permutation operators that, together with the constructor-like operators, satisfy the $\Props_i$-properties.)

\subsection{Mechanization of the results about corecursors}
\label{subapp-isa-corecursors}

Our approach to mechanizing the corecursors is similar to that we took for recursors. Namely,  we have a locale for each of the eight types of models underlying the nominal corecursors, as shown in Fig.~\ref{fig-localesCoRec}. In a corecursor context, in the formalization (unlike in the paper) we use the term ``comodel'' rather than ``model''. 

\newcommand\PermFreeCo{\small \textsf{PermFree$\_$comodel}}
\newcommand\PermFreeVCo{\small \textsf{PermFreeV$\_$comodel}}
\newcommand\SwapFreeVCo{\small \textsf{SwapFreeV$\_$comodel}}
\newcommand\SwapFreshVCo{\small \textsf{SwapFreshV$\_$comodel}}
\newcommand\SwapFreshCo{\small \textsf{SwapFresh$\_$comodel}}
\newcommand\SwapFreshSMCo{\small \textsf{submodel$\_$SwapFresh$\_$comodel}}
\newcommand\SubstFreshCo{\small \textsf{SubstFresh$\_$comodel}}
\newcommand\RenamingCo{\small \textsf{Renaming$\_$comodel}}
\newcommand\RenamingFreshVCo{\small \textsf{RenamingFreshV$\_$comodel}}

\begin{figure}
	\centering
	\begin{tabular}{c|c}
		Corecursor & Corresponding locale
		\\\hline 
		$\ccr_1$ (perm/free) & $\PermFreeCo$
		\\
		$\ccr_2$ (perm/free variant) &  $\PermFreeVCo$
		\\
		$\ccr_3$ (swap/free variant) & $\SwapFreeVCo$
		\\
		$\ccr_5$ (swap/fresh variant) & $\SwapFreshVCo$
		\\
		$\ccr_6$ (swap/fresh) &  $\SwapFreshCo$
		\\
		$\ccr_7$ (subst/fresh) & $\SubstFreshCo$
		\\
		$\ccr_8$ (renaming) & $\RenamingCo$
		\\
		$\ccr_9$ (renaming fresh variant) &  $\RenamingFreshVCo$
		\\
	\end{tabular}
	\caption{Isabelle locales corresponding to corecursors}
	\label{fig-localesCoRec}
\end{figure}

The proof of Thm.~\ref{thm-exprCo} shows that 
one corecursor is stronger than another, say $\ccr_i \geq \ccr_j$, by transforming $(\Sigma_j,\Props_j)$-models to $(\Sigma_i,\Props_i)$-models; always the carrier is the same, and the specific $\Sigma_i$-operators are defined. Again this is done by working inside a locale $L_j$, which represents $(\Sigma_j,\Props_j)$-models by fixing $\Sigma_j$ operators and assuming the $\Props_j$ properties. Inside this locale, we define the $\Sigma_i$-operations and infer the $\Props_i$ properties from $\Props_j$. This allows us to prove the sublocale relationship $L_j \leq L_i$, where $L_i$ is the locale representing $(\Sigma_i,\Props_i)$-models. So $\ccr_i \geq \ccr_j$ is formalized as $L_j \leq L_i$.  This again allowed us to borrow for $(\Sigma_j,\Props_j)$-models the $r_i$ corecursor, in other words to infer the $r_j$ corecursor from the $r_i$ corecursor. 

\begin{figure*}
	\centering
	\includegraphics[width=\linewidth]{cosession_graph.pdf}
	\caption{The Isabelle theories for corecursors}
	\label{fig-isabelleCoTheories}
\end{figure*}

This sublocale hierarchy of locales, which matches exactly the $\geq$-hierarchy of Thm.~\ref{thm-exprCo}, allowed us to (1) prove the corecursion principle for $\ccr_2$ (which is at the top of the hierarchy), and infer all the others from it along sublocale relationships---as discussed in the proof sketch of Thm.~\ref{thm-allNominalCoRecs} from 
App.~\ref{app-proofCoSketches}.  

\newcommand\ParallelSubstitution
{\small \textsf{Parallel$\_$Substitution}}
\newcommand\InfinitaryLambdaTerms
{\small \textsf{Infinitary$\_$Lambda$\_$Terms}}
\newcommand\CoSwapSeveral
{\small \textsf{SwapFresh$\_$SwapFreshV$\_$comodels}}
\newcommand\RenamingIsSwapFreeV
{\small \textsf{Renaming$\_$comodel$\_$is$\_$SwapFreeV$\_$comodel}}

\newcommand\CoPermFreeVR
{\small \textsf{PermFreeV$\_$corecursor}}

\newcommand\CoSeveralR
{\small \textsf{SwapFreshV$\_$SwapFresh$\_$SwapFreeV$\_$PermFree$\_$corecursors}}

\newcommand\CoSeveralTR
{\small \textsf{Renaming$\_$RenamingFreshV$\_$SubstFresh$\_$corecursors}}

\newcommand\CoSwapFreeV{\small \textsf{SwapFreeV$\_$comodel}}

The theory structure of our Isabelle development for corecursors is shown in Fig.~\ref{fig-isabelleCoTheories}.  
The formalization of iterms as equivalence classes of pre-iterms and of all the operators and proof principles described in \S\ref{app-detailsIterms} is performed in the theory $\InfinitaryLambdaTerms$.

The names of the  other theories in Fig.~\ref{fig-isabelleCoTheories} are again
self-explanatory. For example:
\begin{mmyitem}
	\item the theory $\CoSwapFreeV$ formalizes the swap/free models (and contains the locale \\ $\SwapFreeVCo$); 
		
	\item the theory $\CoSwapSeveral$ formalizes the swap/fresh and swap/fresh variant models (and contains the corresponding locales); 

	\item the theory $\RenamingIsSwapFreeV$ proves that each renaming model is (can be regarded as) a swap/free variant model, via the sublocale statement $\RenamingCo \leq \SwapFreeVCo$. 
\end{mmyitem}

There are three theories whose names refer explicitly to  corecursors.  
$\CoPermFreeVR$ contains the direct formalization of the perm/free corecursor, i.e., the proof of the finality principle as discussed in the proof sketch of Thm.~\ref{thm-allNominalCoRecs}
from App.~\ref{app-proofCoSketches}. This corecursor (which is at the top of the $\geq$ hierarchy) has been used to derive all the other seven corecursors. These derivations happen along the corresponding sublocale relationships in the theories 
\\$\CoSeveralR$ and \\$\CoSeveralTR$.

The theory $\ParallelSubstitution$ contains the definition of parallel substitution on iterms using the swap/fresh variant corecursor $\ccr_5$ (as discussed in \S\ref{app-exaCorec}). The full-recursion enhancement of the swap/fresh variant corecursor described in \S\ref{app-enhCorec} is performed at the end of the theory \\$\CoSeveralR$. 

Again, there is a theory $\ThAll$ that imports the relevant top theories and has comments connecting the formalization to the paper.\footnote{For the theories that are located in the figure below \textsf{All}, i.e., import  this theory, we again refer the reader to  \S\ref{app-subsec-localeFree}.}

\subsection{Locale-free, top-level statements of the main results}
\label{app-subsec-localeFree}


We would like to stress that our mechanization, while walking a tight rope in order to minimize the number of recursion principles that are proved  directly (without borrowing), does \emph{not} suffer from any bootstrapping problem or incur any additional assumptions. Rather, it certifies 
Thms.~\ref{thm-allNominalRecs}, \ref{thm-expr}, \ref{thm-qexpr}, 
\ref{thm-allNominalCoRecs} and \ref{thm-exprCo} from the main paper (and also Thms.~\ref{thm-allNominalRecsEnhanced} and \ref{thm-exprEnhanced} from the appendix) 
as they are claimed in the paper, but 
using the Isabelle locale jargon. (And Prop.~\ref{prop-negRec} has a faithful formalization as well, but that does not make use of locales---see \S\ref{subsec-neg}.) 

For readers who are interested in inspecting the mechanized statements of the results but not in understanding the locale jargon, we have also reformulated the results in a manner that matches closely the statements from the main paper. 

In Fig.~\ref{fig-isabelleTheories} (for recursors), the locale-free statements are in the theories that inherit $\ThAll$, culminating with theories that have suggestive names, namely \textsf{Theorem9}, \textsf{Theorem12}, \textsf{Theorem15} and \textsf{Prop16}.  The main results in these theories, which can be found using the keyword ``theorem'', 
are statements of Thm.~\ref{thm-allNominalRecs}, Thm.~\ref{thm-expr},  Thm.~\ref{thm-qexpr} and Prop.~\ref{prop-negRec} that, just like their paper counterparts, 
refer to the nominal recursors $r_i$ using their epi-recursor structure and the notion of definability (introduced in the theory \textsf{Definability$\_$by$\_$Recursors}).  Most of the results are formalized using categories of models assumed to have the carrier sets as the entire type (which is very convenient in HOL formalizations); however, for Thm.~\ref{thm-qexpr} and one half of Prop.~\ref{prop-negRec}, we need the greater flexibility offered by considering explicit carrier sets (as subsets of the underlying types), so we formalized the set-based versions of these categories as well.  

For example, the part of Thm.~\ref{thm-allNominalRecs} that refers to $r_1$ states the recursion (initiality) principle as follows 
(in theory \textsf{Theorem9}): 

\medskip
\noindent 
\textsf{
theorem init$\_$I1:
\\
 \hspace*{3ex}isObjectC1 (VrD,ApD,LmD,permD,FvarsD) $\LRA$
\\
\hspace*{3ex}$\exists!$G. isMorphismC1 G I1 (VrD,ApD,LmD,permD,FvarsD)"
}
\medskip 

This states that for any object of the category \textsf{C1} (the category of models for $r_1$), 
there exists a unique morphism \textsf{G} from the term model \textsf{I1} to  \textsf{C1}. 
Note that an object in  \textsf{C1} is a tuple (VrD,ApD,LmD,permD,FvarsD) consisting of 
constructor-like, permutation-like and free-variable-like operators on a carrier type (not shown 
explicitly in the tuple). The predicates \textsf{isObjectC1} and \textsf{isMorphismC1} are defined appropriately 
(in particular, the models are required to satisfy the $\Props_1$ properties), and shown to form a category 
(via a separate Isabelle statement). The model \textsf{I1} is defined as the term model for the signature $\Sigma_1$:

\medskip 
\noindent
\textsf{
definition I1 where I1 = (Vr, Ap, Lm, perm, Fvars)
}
\medskip

Indeed, \textsf{Vr}, \textsf{Ap}, \textsf{Lm}, \textsf{perm}, \textsf{Fvars} are defined to be the standard operators on terms. 

In fact, we define all the components of the epi-recursor $r_1$ = \textsf{(B,T,C1,I1,R1)}: the base category \textsf{B} and its 
base object \textsf{T}, the category \textsf{C1} and its object \textsf{I1}, the functor \textsf{R1}, and wee prove that 
they are indeed categories and functor, that \textsf{R1} applied objects to \textsf{I1} yields \textsf{T}, etc. In short, 
\textsf{(B,T,C1,I1,R1)} is shown to be an epi-recursor. And the same is done for all the other recursor $r_i$. All this is done 
in theory \textsf{Theorem9}. 

The recursion theorem is alternatively expressed in combinator form (also in theory \textsf{Theorem9}): 

\medskip 
\noindent 
\textsf{
theorem REC1$\_$I1: 
\\
\hspace*{3ex}isObjectC1 (VrD,ApD,LmD,permD,FvarsD) $\LRA$ 
\\
\hspace*{3ex}isMorphismC1 (REC1 VrD ApD LmD) I1 (VrD,ApD,LmD,permD,FvarsD)
}
\medskip 

The above says that \textsf{REC1 VrD ApD LmD} is the unique morphism from the term model \textsf{I1} to the model \textsf{(VrD,ApD,LmD,permD,FvarsD)}. 
Note that the combintor \textsf{REC1} only depends on the constructor-like operators 
\textsf{VrD}, \textsf{ApD} and \textsf{LmD}, and not on the other two operators, \textsf{permD} and \textsf{FvarsD}; 
however, the fact that \textsf{REC1 VrD ApD LmD} is a morphism between \textsf{I1} and 
\textsf{(VrD,ApD,LmD,permD,FvarsD)} of course relies crucially on \textsf{permD} and \textsf{FvarsD} and the $\Props_1$ properties. 

\medskip 
Definability by recursor $r_1$ is expressed as follows (in theory \textsf{Definablity$\_$by$\_$Recursors}):\footnote{``\textsf{fun}'' is another way of introducing definitions in Isabelle. We prefer it here because, unlike ``\textsf{definition}'', it allows pattern matching.}  

\medskip 
\noindent 
\textsf{
fun definableByR1 where 
\\\hspace*{3ex}definableByR1 f (VrD,ApD,LmD) = 
\\\hspace*{3ex}$\exists$permD FvarsD. isObjectC1 (VrD,ApD,LmD,permD,FvarsD) $\wedge$ 
f = REC1 VrD ApD LmD
}
\medskip 

Thus, for a morphism \textsf{f} between the term $\Sigmac$-model and 
another $\Sigmac$-model (VrD,ApD,LmD), the definability predicate says that there exists a extension of (VrD,ApD,LmD) 
to a \textsf{C1} object \textsf{(VrD,ApD,LmD,permD,FvarsD)}  (i.e., a $(\Sigma_1,\Props_1)$-model) 
such that $f$ can be defined as the unique morphism between the term $(\Sigma_1,\Props_1)$-model 
and \textsf{(VrD,ApD,LmD,permD,FvarsD)}  (namely, using the \textsf{REC1} combinator). The above is done for all the recursors $r_i$.

Now, for example, the $r_5 \geq r_4$ part of Thm.~\ref{thm-expr} is expressed just like in the paper, 
but expanding the definition of $\geq$, 
as an implication between definabilities (in theory \textsf{Theorem12}): 

\medskip 
\noindent 
\textsf{
	theorem r5$\_$ge$\_$r4:
	\\\hspace*{3ex}definableByR4 f (VrD, ApD, LmD) $\LRA$
	\\\hspace*{3ex}definableByR5 f (VrD, ApD, LmD)
} 
\medskip 

And the $r_2 \not\geq r_4$ part of Prop.~\ref{prop-negRec} is expressed again like in the paper with the definition of $\geq$ expanded, 
stating the existence of an \textsf{f} definable by $r_4$ but not by $r_2$ (in theory \textsf{Prop16}): 

\medskip 
\noindent 
\textsf{
theorem not$\_$r2$\_$ge$\_$r4: 
	\\\hspace*{3ex}$\exists$f. definableByR4 f (VrD, ApD, LmD) $\wedge$ $\neg$ definableByR2 f (VrD, ApD, LmD)
}
\medskip 

\noindent 
\textit{Set-based versions of the concepts.}  
While working with models as tuples of operators on the entire type (i.e., assuming that the carrier of the model 
is an entire type) does not lose generality, sometimes we need more flexibility---for example, when we wish to consider 
submodels, whose carriers are usually not the entire type. Making disjoint copies and performing type definitions 
could get us by, but in these cases it is more convenient to employ a more flexible, ``set-based'' version of the models, 
with explicit carrier sets. In this more flexible setting, for example the models for $r_1$ are now 
\textsf{(D,VrD,ApD,LmD,permD,FvarsD)} which in  addition to the operators also feature a subset \textsf{D} of the carrier 
type, assumed to be closed under these operators, e.g., \textsf{ApD d1 d2} $\in$ D whenever \textsf{d1,d2} $\in$ \textsf{D}. 
We use primed notation to indicate these set-based concepts, e.g., \textsf{isObjectC1'}, 
\textsf{definableByR1'}, etc. While the primed versions are semantically equivalent to the originals, they are needed for modeling certain phenomena in a more finite-grained manner. They are formalized in 
the theories \textsf{Set$\_$Based$\_$Recursors} and  \textsf{Definablity$\_$by$\_$Recursors}.

The theory \textsf{Theorem12$\_$setBased} infers the set-based version of Thm.~\ref{thm-expr} by transferring to sets the results from theory \textsf{Theorem12}.  
For example, the $r_5 \geq r_4$ part of Thm.~\ref{thm-expr} 
is expressed as follows in the set-based setting (in theory \textsf{Theorem12$\_$setBased}): 

\medskip 
\noindent 
\textsf{
	theorem r5$\_$ge$\_$r4':
	\\\hspace*{3ex}definableByR4' f (D, VrD, ApD, LmD) $\LRA$
	\\\hspace*{3ex}definableByR5' f (D, VrD, ApD, LmD)
} 
\medskip 

The primed set-based versions yield the original versions, e.g., \textsf{r5$\_$ge$\_$r4'} yields \textsf{r5$\_$ge$\_$r4} by taking \textsf{D} to be universal set (comprising the entire type).

Because the $r_1 \not\geq r_2$ part of Prop.~\ref{prop-negRec} requires a set defined by a predicate, 
we also formalize it using the more general primed models. Namely, after defining 
a particular $(\Sigma_2,\Props_2)$-model as described in the proof of Prop.~\ref{prop-negRec}, 
which in the formalization we call \textsf{(E,VrE,ApE,LmE)}, we prove:\footnote{So in the following statement, \textsf{E,VrE,ApE,LmE} are \emph{not} universally quantified variables (as are, for example, \textsf{D,VrD,ApD,LmD} in the theorem \textsf{r5$\_$ge$\_$r4'} above), but certain defined constants.}

\medskip 
\noindent 
\textsf{
	theorem not$\_$r1$\_$ge$\_$r2: 
	\\\hspace*{3ex}$\exists$f. definableByR2' f (E, VrE, ApE, LmE) $\wedge$ $\neg$ definableByR1' f (E, VrE, ApE, LmE)
}
\medskip

The same is true for Thm.~\ref{thm-qexpr}, where we must consider initial segments formed by submodels.  For example, the 
$r_6\wgeq r_8$ part of Thm.~\ref{thm-qexpr} is formalized as follows (in theory \textsf{Theorem15}): 

\medskip 
\noindent 
\textsf{
theorem r6$\_$quasi$\_$ge$\_$r8':
\\\hspace*{3ex}definableByR8' g (D, VrD, ApD, LmD) $\LRA$  
\\\hspace*{3ex}$\exists$g0. isMorphismB' g0 T' (ob' (D, VrD, ApD, LmD)) $\wedge$ 
\\\hspace*{3ex}\phantom{$\exists$}
g = mo' (D, VrD, ApD, LmD) $\circ$ g0 $\wedge$
\\\hspace*{3ex}\phantom{$\exists$}
definableByR6' g0 (ob' (D, VrD, ApD, LmD))
}
\medskip 

Its formulation matches that from the paper, but again expands the definition of $\wgeq$: \textsf{mo'} 
and \textsf{ob'} formalize the initial segment $(\Bcat_0,(m(B):o(B)\ra B)_{B \in \Obj{\Bcat}})$ from Def.~\ref{defi-qstrong} 
(where \textsf{mo'}  is the morphism operator $m$ and \textsf{ob'} is the object operator $o$). 
For all the $\wgeq$ relationships stated in the theorem, \textsf{ob'}, when applied to an object of the category \textsf{B'}, 
returns its minimal submodel  
and \textsf{mo'} returns the inclusion morphism (as explained in the proof sketch of Thm.~\ref{thm-qexpr}). These are proved to form an initial segment of \textsf{B'}. 
So the above theorem says that any morphism \textsf{g} definable by $r_8$ can be written as a composition 
between the initial-segment morphism and a morphism \textsf{g0} definable by $r_6$ (like in the definition of $\wgeq$). 
Because this composition involves submodels, the definablity of \textsf{g0} must be expressed using the set-based version. This is why we shift to the set-based setting completely, and  
the morphisms \textsf{g0} and \textsf{g} dwell the set-based version of the base category.

\bigskip 
In Fig.~\ref{fig-isabelleCoTheories} (for corecursors), the locale-free statements are again in the theories that inherit 
$\ThAll$, culminating with theories that have suggestive names, namely \textsf{Theorem18} and \textsf{Theorem19}. 
The main results, which can again be found using the keyword ``theorem'',  
are statements of Thm.~\ref{thm-allNominalCoRecs} and Thm.~\ref{thm-exprCo}. Similarly to the case of recursors, we 
use concepts and terminology that matches the paper closely.  

The formalization follows a similar pattern to the one 
for recursors. For example, the $\ccr_5 \geq \ccr_6$ part of Thm.~\ref{thm-exprCo} is expressed as follows (in theory \textsf{Theorem19}): 

\medskip 
\noindent 
\textsf{
	theorem cr5$\_$ge$\_$cr6:
	\\\hspace*{3ex}definableByCR6 f DestD $\LRA$
	\\\hspace*{3ex}definableByCR5 f DestD
}